\documentclass[11pt]{article}

\usepackage[margin=1in]{geometry}
\usepackage[T1]{fontenc}
\usepackage{amsmath}
\usepackage{amssymb}
\usepackage{amsthm}
\usepackage{graphicx}
\usepackage{float}
\usepackage{booktabs}
\usepackage{array}
\usepackage{tabularx}
\usepackage{algorithm}
\usepackage{algpseudocode}
\usepackage{tikz}
\usetikzlibrary{arrows.meta,positioning,calc,fit,shapes.geometric}
\usepackage[font={small,sf},labelfont={bf,sf},labelsep=period,margin=8pt]{caption}
\usepackage[hidelinks]{hyperref}
\usepackage{url}
\usepackage{xurl}
\usepackage{siunitx}
\usepackage{microtype}
\usepackage[numbers,sort&compress]{natbib}

\newtheorem{theorem}{Theorem}
\newtheorem{proposition}{Proposition}
\newtheorem{corollary}[theorem]{Corollary}
\newtheorem*{fact}{Fact}
\newtheorem{remark}{Remark}[section]

\newcommand{\kasparhft}{\textsc{Kaspar-HFT}}
\newcommand{\spub}{s_{\mathrm{pub}}}
\newcommand{\findings}[1]{\par\medskip\noindent\fbox{\begin{minipage}{\dimexpr\linewidth-2\fboxsep-2\fboxrule\relax}
\textbf{Findings.}\begin{itemize}\setlength{\itemsep}{1pt}#1\end{itemize}\end{minipage}}\par\medskip}

\title{Packets, Transactions and Queues:\\
Design Principles for HFT Systems from a Measurement Study of CME Market Data}
\author{Vincent Maciejewski\thanks{M2 Tech, Montreal.
\texttt{mayeski@gmail.com}. \kasparhft{} is open source at
\url{https://github.com/vincent212/kaspar-hft}.}}
\date{September 2026}

\begin{document}
\maketitle

\begin{abstract}
This paper is a measurement study of a real exchange feed and of what its structure
implies for the design of the HFT trading system that consumes it. Over more than a year of CME
market data for the NQ front-month contract, we follow every packet and every
matching-engine transaction through the two exchange timestamps the feed carries, from
the matching engine, through the exchange's market-data publisher, to a receiver, and we
check the characterisation against a live production receiver on other instruments.

We find that packets arrive in near-critical self-exciting clusters, and
the clustering belongs to the matching engine's transactions, not to how the exchange
packs them: almost every transaction fits in one packet. The matching engine often
processes consecutive transactions within a fraction of a microsecond of each other. The
market-data publisher behaves as a queue that sends at most one packet per
\emph{publisher period}, about $7.5\,\mu\mathrm{s}$; its delay grows about tenfold when
the engine bursts, and it drains a backlog with more packets rather than larger ones. A burst at the engine therefore reaches a receiver as a train of packets one
publisher period apart, most of them one message long. The short gaps that make a
receiver queue are gaps between different transactions, and what sits on either side of
them is ordinary quote traffic, with the signature of a race after trades. Why orders
reach the engine so close together is a question about the market and is not addressed here.

For system design, a receiver that handles every packet within one publisher period
does not queue on packet arrivals, however bursty the market; this is the design target.
Above the
publisher period a queueing tail appears and grows much faster than the service time. It
is produced by the timing of transactions, not by the packet rate, whose effect is
negligible at HFT service times, and not by packet size or packetisation. In this regime
the prevailing single-thread rule reverses: cutting the servicing chain into stages on separate threads
removes most of the tail for one hop on the median, only the slowest stage matters, and
for a stateless stage a pool of cores that each take whole packets does at least as
well. Just under the publisher period, where a production receiver that hands each packet
to a decoder thread runs, the tail that remains comes from packets of several messages,
decoded one after another, and from service times that vary from packet to packet; the
levers there are the cost per message and the spread of service, not the thread count.
The same receiver decoding on its socket-reader thread runs far below the period, about
$1\,\mu\mathrm{s}$ at the median per book update, which is the design target met with
margin. We give the analytic framework, a burst-limit
throughput identity and an exact pathwise reduction of the deterministic-service tandem
to a single bottleneck server, and test it on the corpus. Code is in the \kasparhft{}
repository; the packet captures are not redistributable.
\end{abstract}

\clearpage
\tableofcontents
\clearpage

\section{Introduction}
\label{sec:intro}

\subsection{A measurement study of a real feed, and what it implies for HFT system design}
\label{sec:intro-study}

The object of study is a real exchange feed as it arrives on the wire: CME MDP3 market
data for the NQ front-month contract, over a year of trading sessions, several billion
packets and as many matching-engine transactions. Three systems produce and consume that
stream. The exchange's matching engine processes orders into transactions; the exchange's
market-data publisher packs the resulting messages into packets and sends them; a
receiver decodes the packets and applies them to its order book. The feed carries a
timestamp from each of the first two, and the receiver adds its own, so each transaction
can be followed through all three. The paper measures each stage: how the transactions
are timed at the engine, how the publisher paces and packs them, how big the packets are
and how often they come, and what a receiver of a given speed then pays in queueing.
Part~II reports the measurements; Section~\ref{sec:exchange} is the exchange side of
them.

The best practices for HFT system design that Part~III states are derived from these
measured characteristics, not from a model of the market. The oldest of those practices,
the single-threaded hot path, is the design question the measurements settle first, and
the next subsection sets it out. It is one consequence of the study, not its subject.

\subsection{The single-thread event loop against the multi-actor tandem}
\label{sec:intro-debate}

The design of the ingress-to-order pipeline in a high-frequency trading system is a
choice between two archetypes. The first is the single-threaded event loop, which
processes every step of every message on one core: SBE decode, book application, strategy
evaluation and order construction, all under one call chain, taking a total service time
we denote $T$. The second is the multi-actor pipeline, a \emph{tandem} of $N$ servers in
series. There the steps are split across $N$ actors, each on its own thread with its own
mailbox; each does its part of the work and hands the message to the next across one of
$N-1$ explicit hops, whose per-hop wall-clock cost we denote $h$. Throughout, \emph{HFT service times} means a few microseconds to a few tens of
microseconds per packet, the range a hot path occupies. The low-latency HFT community is
predominantly in the single-thread camp, and single-threaded is generally
regarded as the safe choice. The argument cites the cost of a mailbox hop: a mailbox enqueue, a cache-line transfer
between the sender's core and the receiver's core, and the cost of scheduling or spinning
the receiver. It holds that $(N-1) h$ adds to every message's end-to-end latency, so that
each additional stage is a fixed cost that no throughput or tail benefit recovers.

That argument holds when the single-server throughput $1/T$ exceeds the peak arrival
rate. On CME MDP3 the arrival process is strongly self-exciting: packets arrive in
Hawkes-clustered bursts \citep{filimonovsornette2012,hardimanbercotbouchaud2013,bacry2015review}.
The exchange publishes each matching-engine event, a transaction, as almost always one
packet, so the packet stream relays the timing of transactions without adding clustering
of its own (Section~\ref{sec:results-transactions}). The peak arrival rate has a hard
ceiling, however. The exchange's market-data publisher sends at most one packet per
\emph{publisher period}, about $7.5\,\mu\mathrm{s}$ on this feed
(Section~\ref{sec:exchange-publisher}); during a burst it sends packets exactly that far
apart. A receiver that handles every packet in less than one publisher period therefore
never queues on packet arrivals, however bursty the market: it is done with each packet
before the next can arrive. That is the design goal, and the paper's own production
receiver meets it in both configurations measured. When its socket-reader thread hands
each packet through an asynchronous mailbox to a decoder thread, a one-message packet
takes about $7\,\mu\mathrm{s}$ from socket read to book update, just under the period
(Section~\ref{sec:crossval-floor}). When the socket-reader thread decodes each packet and
updates the book itself, a book update takes about $1\,\mu\mathrm{s}$ at the median, far
below the period (Section~\ref{sec:crossval-onethread}). Below the publisher period the
single-thread argument is simply right, and splitting the chain adds hop cost for nothing.

A receiver that is slower than the publisher period, whether because its chain carries a
strategy stage or an order path, because it decodes a packet of several messages, or
because its service time occasionally spikes, falls behind on every packet of a burst.
There cluster-induced queueing produces a tail one to two orders of magnitude larger than
$h$. The paper's central empirical observation, on over a year of NQ front-month packet
arrivals, is that this cluster tail dominates end-to-end latency once the service time
is above the publisher period. The $99$th-percentile latency ($p_{99}$) of a
constant-service queue driven by real packet arrivals is about twice the service time for
a $16\,\mu\mathrm{s}$ task and between four and five times for a $32\,\mu\mathrm{s}$ task,
and more in busy windows. A Poisson stream with the same number of arrivals has no tail
at $16\,\mu\mathrm{s}$ and a small one at $32\,\mu\mathrm{s}$; over this band the tail is
almost entirely a clustering effect.

The question is then, for a chain that cannot be brought under the publisher period,
whether splitting it into $N > 1$ stages, each an independent actor on its own core with
per-stage service $T/N$, reduces the tail, and whether the additional median cost
$(N-1)h$ is small relative to the tail reduction. Neither architecture is better
unconditionally. Both conditions hold above a service time set by the feed, the
publisher period, and fail below it; the paper measures that threshold on one feed, and
its recommendation is that designers measure it on theirs. For a receiver already under
the period, the paper's second finding matters instead: the tail it still carries comes
from packets of several messages and from variable service, and the levers there are the
cost per message and the spread of service, not the thread count
(Sections~\ref{sec:crossval-spanrun} and~\ref{sec:fungibility-measured}).

\subsection{Prior recommendations for single-thread hot paths}
\label{sec:intro-singlethread-priors}

The single-thread rule has a lineage in the low-latency-systems community, summarised
here in reverse chronological order. It has three layers: the HFT-specific codification of the
last decade, the general low-latency-systems and event-driven-server culture from which
HFT inherited it, and one prominent published counter-argument that the field discarded.

\paragraph{Modern HFT-community codification (2010--present).} The founding statement of
the current HFT-community formulation is Thompson's Single Writer Principle
\citep{thompson2011singlewriter} --- ``a single-writer design consistently outperforms a
`properly' lock-free or wait-free multi-writer design on most realistic workloads'' ---
motivated by and applied in the LMAX Exchange architecture
\citep{fowler2011lmax,lmaxdisruptor}. Rigtorp's widely cited low-latency tuning guide
\citep{rigtorp_lowlatency} formalises the operational rule (``run a single thread in
\texttt{SCHED\_OTHER} per core and use busy waiting/polling''), and Chronicle Software's
public engineering materials \citep{chronicle_fasttrading} adopt the same architecture for
market-data connectors (``within a single event loop \ldots busy-spinning and affinity
locks''), pinning Chronicle Queue's design to a single writer per queue. The Aeron
messaging library \citep{aeron} inherits the same lineage. The academic side is thinner,
but \citet{kwan2023hftpatterns} codify the pattern as the standard low-latency-HFT design
pattern, and recent practitioner posts \citep{hu2024sequencer} restate it explicitly
(``all critical trading logic runs in a single-threaded event loop, eliminating context
switches and lock contention entirely''). Inside HFT specifically, the same recommendation
has been reinforced by two decades of private prop-trading arms-race practice among the
Chicago desks (Getco, Jump, Optiver, DRW; late 2000s), reported second-hand in industry
retrospectives but without an open-literature reference.

\paragraph{Older event-driven-server lineage (1995--2010).} HFT did not invent this rule;
it inherited it from the event-driven-server culture. The single-threaded event loop as an
architectural pattern dates at least to Schmidt's Reactor pattern
\citep{schmidt1995reactor} in the mid-1990s. It was sharpened by Kegel's C10K essay
\citep{kegel1999c10k}, which named the OS thread-per-connection model as the bottleneck of
scalable network servers and argued for event-driven designs. The pattern went mainstream
in production with Sysoev's nginx web server \citep{sysoev2004nginx}, whose
single-threaded worker process is explicitly justified as context-switch avoidance, and
later with Dahl's Node.js runtime \citep{dahl2009nodejs}, which pushed the single-threaded
event loop into ordinary application code.

\paragraph{Pro-pipelining minority --- outside HFT.} A legitimate multi-decade
pro-pipelining tradition exists in the datacenter and networking-systems literature; it
has been the field's non-single-thread minority throughout the period covered by the
previous two paragraphs. The earliest prominent published argument is the Staged
Event-Driven Architecture (SEDA) of \citet{welsh2001seda}, which advocated exactly the
multi-stage pipeline of independently scheduled event handlers that the present paper's
tandem design revives. Welsh's 2010 retrospective \citep{welsh2010retrospective} is often cited as a walk-back
to single-thread, but it is a hedged pro-pipelining recommendation. It advocates
decoupling stages from thread pools, grouping several stages within one thread-pool
domain where latency is critical, and putting a separate thread pool and queue in front
of any stage with long or nondeterministic runtime. That is the present paper's design
rule in different vocabulary. Beyond SEDA the
pro-pipelining tradition includes the Click Modular Router \citep{kohler2000click} for
staged packet processing; Naiad \citep{murray2013naiad} for low-latency dataflow; Snap
\citep{marty2019snap} for Google's multi-engine network stack; and the datacenter
microsecond-tail-latency scheduler line --- Shinjuku \citep{kaffes2019shinjuku}, Shenango
\citep{ousterhout2019shenango}, Caladan \citep{fried2020caladan}, Perséphone
\citep{delimitrou2021persephone} --- which argues for centralised multi-core dispatch
specifically to control the microsecond-scale tail. Perséphone in particular is targeted
at ``wide service-time distributions,'' the closest existing framing to our arrival-law
claim; but the datacenter line frames the problem as heavy-tailed \emph{service}, not
clustered \emph{arrivals}. DPDK explicitly supports both run-to-completion and pipeline
modes and its documentation recommends pipeline mode ``if some set of packets require
longer to process'' \citep{dpdk}. Inside HFT specifically the pro-pipelining position is
rare in the open literature; a Chinese patent \citep{cn105654383a} describing a
pipelined FAST market-data decoder is one commercial instance.

None of the pro-pipelining sources above ties the design recommendation to Hawkes /
self-exciting / near-critical arrival dynamics. The closest theoretical hooks are the
MMPP-input tandem work of \citet{kimchydzinski2002mmpp}, discussed further in
Section~\ref{sec:intro-priors}, and the batch-arrival tandem phase-transition analysis of
\citet{tandembatch2014}. The present paper's mechanism, that over a band of service times the recommendation
reverses when arrivals are near-critically self-exciting, is one we have not found in
prior work.

The present author's earlier work belongs on this list with a qualification.\footnote{The tail compression reported here had been observed in production before its mechanism was understood. Splitting a long task into shorter asynchronously connected stages was known to reduce the tail on a per-strategy basis. The working assumption at the time was that single-thread implementations were simply faster in the inner loop, not that the recommendation depended on the arrival law. The present paper is the explanation of that
observation.} An earlier paper by the present author
\citep{mayeski_fastsend} did not argue for the single-thread principle in its own right.
That paper addresses a different objection: the actor model had been considered
unsuitable for HFT because of asynchronous-dispatch overhead, and a synchronous send is
presented there as the mechanism that removes that overhead for co-located actors. That
send, \texttt{fast\_send} in the \kasparhft{} framework (Section~\ref{sec:kaspar}), runs
the receiving actor's handler inline on the sending thread instead of queueing a message
to another thread; the paper writes ``synchronous send'' for it outside Part~III. It makes
the actor model usable inside a servicing chain that the prevailing design placed on one
thread. Its reference architecture (reader to buffer by asynchronous \texttt{send};
decode, book and signal chained by synchronous sends on one thread; signal to order
manager by asynchronous \texttt{send}) takes the single-thread servicing chain as given
and does not compare it against more asynchronous topologies. The present paper revises
that design context, not the synchronous send: on Hawkes-clustered feeds the
servicing chain should be split into asynchronous stages once its service time exceeds
the threshold established in Part~II. The synchronous send remains the mechanism for
chaining actors placed on one thread; what changes is where the thread boundaries inside
the chain should be. The design rule is stated in Section~\ref{sec:conclusion}.

\subsection{Plan of the paper}
\label{sec:intro-plan}

Readers not familiar with point processes --- the Poisson, renewal and Hawkes processes,
the branching ratio, the Fano factor and the Hurst exponent --- should read
Appendix~\ref{app:pp-background} before Section~\ref{sec:setup-hawkes}; it defines every
term the paper uses for arrival streams from first principles.

The paper asks three groups of questions and answers each in a named place.

\begin{enumerate}
    \item \emph{What does the feed look like, on the wire and behind it?} Packet arrivals
are self-exciting clusters (Section~\ref{sec:setup-hawkes}); the clustering belongs to the
matching engine's transactions, almost every one of which fits in one packet
(Section~\ref{sec:results-transactions}); the engine processes consecutive transactions
within a fraction of a microsecond while the publisher, a queue with a period of about
$7.5\,\mu\mathrm{s}$, delivers them one period apart and packs few of them together
(Section~\ref{sec:exchange}); a live receiver on other instruments sees the same
statistics (Sections~\ref{sec:crossval-agree} and~\ref{sec:crossval-span}).
    \item \emph{What produces the queueing tail behind a single-threaded receiver?} Three
candidates are tested. The packet rate: no, at HFT service times
(Section~\ref{sec:results-conditioning}). The clustering of arrivals: yes
(Section~\ref{sec:results-nulls}), acting through consecutive transactions sent back to
back at the publisher period at short service times and through runs of transactions,
the self-exciting component, at long ones. Why the matching engine processes
transactions within a microsecond of each other, which is what backs up the publisher,
is a question about the market, outside the scope of this paper;
Section~\ref{sec:concl-cannot} states what can and cannot be claimed about it. The number
of messages per packet: just under the publisher period, where the two-stage production
receiver runs at about $7\,\mu\mathrm{s}$, it carries most of the tail, through the decode of a long packet
message by message and the queue that packet opens behind it; from $16\,\mu\mathrm{s}$ up
it adds a few percent (Sections~\ref{sec:crossval-spanrun}
and~\ref{sec:fungibility-measured}). Service times that vary from packet to packet, which
the simulation does not model, add a tail of their own on the live receiver.
    \item \emph{How should the trading system be designed?} Stay under the publisher period if
the chain allows it. Above it, are two queues better than one? Yes above a service-time
threshold set by the feed and no below it (Sections~\ref{sec:results-tail}
and~\ref{sec:results-design}), and only the slowest stage matters
(Theorem~\ref{thm:reduction}). For a stateless stage, with equal cores and every cost
charged, a dispatch pool is level with a two-stage cut and better than a cut into four
or more; the stateful order-book stage cannot be dispatched
(Section~\ref{sec:results-equal-core}). For a production receiver, which runs below the
period, the answer is no, and the levers are the cost per message and the spread of service
(Section~\ref{sec:kaspar}). Part~III turns these into a design procedure.
\end{enumerate}

The paper proceeds in five parts.

\begin{itemize}
    \item \textbf{Measurement of the feed.} The corpus is over a year of NQ front-month
packet captures with the exchange's engine and publisher timestamps
(Section~\ref{sec:setup-data}). Its clustering is characterised model-free and with a
Hawkes fit (Section~\ref{sec:setup-hawkes}); every message is grouped into its
matching-engine transaction (Section~\ref{sec:results-transactions}); and the same
transactions are followed through the engine's clock and the publisher's, giving the
publisher period, the publisher's delay under load, packet size under backlog, and what
sits on either side of a short gap (Section~\ref{sec:exchange}).
    \item \textbf{Theory.} Two closed-form results (Section~\ref{sec:theory}) predict the
pipeline behaviour. Proposition~\ref{thm:burst} shows that under any arrival process an
$N$-stage tandem has an $N$-fold throughput speedup on isolated bursts.
Theorem~\ref{thm:reduction} shows that a deterministic-service tandem reduces exactly to a
single server with the largest stage's service, and that the end-to-end wait is bounded
pathwise by $1/N$ of the single-stage wait plus $(N-1)h$, at every quantile and under
every arrival sequence. Corollary~\ref{cor:poisson} shows that under a rate-matched
Poisson null the single-stage wait at the $q$-quantile is already zero whenever the
utilisation satisfies $\rho < 1 - q$ --- which at this feed's rate covers $p_{99}$ for
every task up to $16\,\mu\mathrm{s}$ --- so over that band there is nothing to compress
and splitting costs exactly $(N-1)h$. The Hawkes--Oakes cluster representation explains why the single-stage wait on a real
feed is one to two orders of magnitude above $h$: a burst of clustered arrivals closer
together than the service time queues behind its first member (Section~\ref{sec:theory-cluster}).
This is the wait the tandem reduces.
    \item \textbf{Simulation.} We build an $N$-stage tandem Lindley recursion on real CME
MDP3 packet-arrival streams at a measured hop cost $h$ (Section~\ref{sec:setup-sweep}) and a service
time swept across $T \in \{2, 4, 8, 16, 32, 64, 128\}\,\mu\mathrm{s}$, from tasks below
the hop cost to tasks far above it. We sweep $N \in \{1, 2, 4, 8\}$ on every thirty-minute window of the corpus. Each window is also simulated under a uniform
Poisson null $P$ that redraws the same number of arrivals uniformly over the same support
(Section~\ref{sec:setup-sweep}); the null isolates the contribution of clustering, and
is not a model of what any designer assumes. Section~\ref{sec:setup-nulls} adds seven
further streams that remove the ingredients of the clustering one at a time, at the
packet level and at the level of matching-engine transactions.
    \item \textbf{Design equation for practitioners.} From the measurements and the
simulation we extract design rules for a latency-critical servicing chain: stay under the
publisher period where the chain allows it; above it, cut the chain at the boundary that
shortens its slowest stage, and settle each boundary by end-to-end measurement
(Sections~\ref{sec:kaspar} and~\ref{sec:implementation}; the full rule is the
design-rule box of Section~\ref{sec:concl-summary}).
    \item \textbf{Live cross-validation.} The corpus statistics are compared with an
independent measurement on a production receiver, on other instruments
(Section~\ref{sec:crossval}). The clustering statistics agree, and the service floor of
the two-stage production receiver coincides with the service time at which the simulated
tail appears. The sweep is
then rerun with the measured per-message decode cost
(Section~\ref{sec:crossval-spanrun}). The live ZN stream's far tail is traced to the
message count of a few large packets (Section~\ref{sec:fungibility-measured}).
\end{itemize}

\subsection{Contribution and prior-art positioning}
\label{sec:intro-priors}

The paper makes three kinds of contribution: a measurement of an exchange feed at the
packet and transaction level, a theory of what a receiver pays for it, and design rules
that follow.

\paragraph{Measurement.} Four findings about CME's NQ feed that, to our knowledge, have
not been published for any exchange.
\begin{enumerate}
    \item \emph{The publisher is a queue, and its period is the receiver's threshold.} The
exchange's market-data publisher sends at most one packet every $7.5\,\mu\mathrm{s}$; its
delay from engine to packet grows about tenfold when the engine bursts; and it drains a
backlog with more packets, not bigger ones (Section~\ref{sec:exchange-publisher},
Section~\ref{sec:exchange-drain}).
    \item \emph{Two clocks.} On the matching engine's clock one consecutive transaction in
six follows the previous one within the publisher period, often within a fraction of a
microsecond; on the publisher's clock almost none do. Most of the short gaps a receiver
sees are engine events processed almost together and delivered one period apart
(Section~\ref{sec:exchange-engine}).
    \item \emph{The transaction structure of the packet stream.} Almost every transaction
fits in one packet, a long packet is several transactions packed together, and the
packet stream inherits the clustering of the transaction stream, with the same branching
ratio (Section~\ref{sec:results-transactions}, Section~\ref{sec:crossval-span}).
    \item \emph{What sits on either side of a short engine gap.} Ordinary quote traffic,
in the same mix as across long gaps, with the signature of a race after trades
(Section~\ref{sec:exchange-bursts}).
\end{enumerate}

\paragraph{Theory and identification.} Over part of the range of HFT service times, the
choice between a single-thread receiver and an $N$-stage tandem \emph{reverses with the
arrival law}. Between the publisher period and the service time at which a Poisson stream
at the same rate begins to queue, about $10$ to $20\,\mu\mathrm{s}$ on this feed, Poisson
arrivals make the single thread the better choice, since there is no wait to compress,
while clustered arrivals make the tandem better on the tail, because the cluster wait it
removes exceeds the hop cost it adds. Outside that band there is no reversal: below it the
single thread is better under both arrival laws, and above it the tandem is better under
both, by a margin that the clustering multiplies. Two closed-form results
(Section~\ref{sec:theory}) and a tandem sweep driven by the real packet timestamps
(Section~\ref{sec:results}) establish it. Seven null streams, each removing one
ingredient of the clustering, identify what the receiver queues on at service times above the publisher period: the
timing of transactions, not their size, their packetisation, the packet rate or
second-scale load (Section~\ref{sec:setup-nulls}). Just under the period, where the
two-stage production receiver runs, long packets carry most of the tail instead
(Section~\ref{sec:crossval-spanrun}). We have found no prior statement of the reversal and no
prior test of whether an exchange's packetisation or the market's clustering is what a
receiver queues on.

\paragraph{Design.} The measurements give a system designer a threshold to measure once per feed, the publisher period; a rule for where to cut a chain above it and how to spend
cores on stateless and stateful stages; and, just under the period, where the two-stage
production receiver runs, the finding that the remaining tail is in long packets and
variable service, not in the thread count (Part~III). A receiver that decodes on its
socket-reader thread runs far below the period, about $1\,\mu\mathrm{s}$ at the median for a
book update, where it does not queue on the publisher's trains and one thread is best.

\paragraph{Prior work.} A literature search covering foundational tandem theory
\citep{jackson1957,kelly1979,burke1956,kleinrock1976,whitt1983}, tandem theory under
bursty arrivals
\citep{fosskorshunov2000,baccellifosslelarge2005,heindl2001mmpp,kimchydzinski2002mmpp,debickimandjes2015},
Hawkes single-queue and adjacent work
\citep{dawpender2018queues,chenblanchet2021,koopsboxmamandjes2018,sheldon2023,gaozhu2018jumps,gaozhu2024statedep},
empirical and applied HFT and RPC-pipeline work
\citep{byrd2020abides,frey2023jaxlob,aquilina2022arms,deanbarroso2013tailatscale,zhao2023parsimon},
the datacenter tail-latency literature, and the empirical market-microstructure work on
exchange timing found no characterisation of an exchange's outbound packet stream at
this level. The nearest neighbours measure adjacent things. Deutsche B\"orse
characterises its own T7 system from the inside \citep{deutscheboerse2026t7}: up to six
timestamps per message, a micro-burst in which inbound throughput is millions of messages per second at the gateway against a few hundred thousand at matching, and three
network timestamps per market-data packet through a timestamp file; it is the inbound
side, one example burst, and self-reported, and it does not characterise the outbound
packet spacing, packet size under backlog or transactions per packet over time.
\citet{noble2026realitygap}, on Nasdaq equities, find a sharp mode in inter-event times
at the exchange round trip, about $29\,\mu\mathrm{s}$, and about a quarter of events too
fast to be reactions to the previous one, which they read as correlated responses to a
shared signal; that is the same kind of observation as the sub-microsecond engine gaps of
Section~\ref{sec:exchange-engine}, made on event times alone, without the publisher or
the packets. \citet{bartlettmccrary2019} and \citet{dinghannahendershott2014} measure the
dissemination latency of the US equity consolidated feeds against direct feeds from
exchange timestamps, at microsecond to millisecond scale; \citet{fishehaynesonur2016}
measure message latency on CME Treasury and e-mini futures from account-level audit-trail
data; \citet{hasbrouck2013lowlatency} document millisecond runs of linked messages on
Nasdaq. None measures a direct feed packet by packet. Industry monitoring reports message
and packet rates of US feeds at one-millisecond resolution,\footnote{Exegy and the
Financial Information Forum, \emph{MarketDataPeaks},
\url{https://www.exegy.com/resources/marketdatapeaks/}.} which is rates, not spacing. On the exchange side, \citet{yoon2026matching} reports for his own matching engine that tail latency under load is set by queueing at the serialised match loop during bursts; the present paper studies the other end of the feed, with the arrival process measured rather than assumed. On the arrival-process side, the Hawkes literature on order flow
\citep{filimonovsornette2012,filimonovsornette2015,hardimanbercotbouchaud2013,achab2018branching,bacry2015review}
models order-book events for microstructure and forecasting and does not connect them to
a receiver's service queue or to the packet process that carries them. We have found no
prior study that combines a real exchange-calibrated tandem sweep with a Poisson
counterfactual and a measured framework hop cost. The full reference list appears in
Section~\ref{sec:refs}.

\subsection{Note on the use of AI in this paper}
\label{sec:intro-ai-use}

This paper was produced with Anthropic's Claude (via Claude Code) under the author's
direction for code generation (the analysis pipeline in the \kasparhft{} repository),
exploratory analysis, mathematical drafting, prior-art search, prose editing, and figure
preparation. The analytical results of Part~I, including the proof of
Theorem~\ref{thm:reduction}, were drafted with the model and then independently checked by
the author. The author independently reviewed all mathematical
arguments, experimental methodology, numerical results, and citations, and takes sole
responsibility for the contents of the paper.

\subsection{Notation}
\label{sec:notation}

The following symbols are used throughout.

\begin{center}
\begin{tabular}{@{}lp{0.78\textwidth}@{}}
\toprule
Symbol & Meaning \\
\midrule
$T$ & total service time per message on a single-stage pipeline, in \si{\micro\second} \\
$N$ & number of stages in the tandem pipeline (one actor per stage) \\
$T/N$ & per-stage deterministic service time in the $N$-stage tandem \\
$h$ & per-hop wall-clock cost of the async mailbox handoff between adjacent stages \\
$k$ & size of an arrival burst (number of packets in one cluster) \\
$\tau_k^{(N)}$ & exit time of the last message of a $k$-burst under $N$ stages \\
$s_i$, $s_{\max}$ & per-stage services of an unequal tandem, $\sum_i s_i = T$, and their maximum \\
$w_j(s)$ & Lindley waiting time of message $j$ at a single deterministic server with service $s$ \\
$W_j^{(N)}$, $W_q^{(N)}$ & end-to-end wait of message $j$ (latency minus $T$, hops included) under $N$ stages, and its $q$-quantile \\
$p_q$ & $q$-th percentile of the end-to-end latency distribution ($p_{50}$, $p_{99}$, etc.) \\
$\Delta(N)$ & clustering tail excess $p_{99}(N) - p_{50}(N)$ (\S\ref{sec:results-design}) \\
$\Delta_{\mathrm{single}}(s)$ & single-server clustering tail excess at deterministic service $s$; $\Delta(N) = \Delta_{\mathrm{single}}(T/N)$ (\S\ref{sec:results-design}) \\
$\gamma$ & fitted exponent in $\Delta(N) \approx \Delta(1) N^{-\gamma}$ \\
$\kappa$ & operator's tail-versus-median weight in the design equation \\
$N^\star$ & stage count minimising the design objective $(N-1)h + \kappa\,\Delta(N)$ (\S\ref{sec:results-design}) \\
$\mu, \alpha, \beta$ & Hawkes background rate, jump size, and decay rate (\S\ref{sec:setup-hawkes}) \\
$n = \alpha/\beta$ & Hawkes branching ratio \\
$K$ & Borel-distributed offspring cluster size (\S\ref{sec:theory-cluster}) \\
$\bar\lambda_{\mathrm{pkt}}$ & mean packet-arrival intensity per 30-min window, packets per second \\
$n_{\mathrm{pkt}}$ & branching ratio of the packet-arrival Hawkes fit per window \\
$n_{\mathrm{tx}}$ & branching ratio of the transaction-start Hawkes fit per window (\S\ref{sec:results-transactions}) \\
$n_{\mathrm{pkt}}^{\#}$ & per-window packet count \\
$\rho$ & utilisation, $\rho = \bar\lambda_{\mathrm{pkt}} \cdot T$ \\
$\tau$ & bin width of the Fano factor, in seconds (distinct from the service time $T$) \\
$F(\tau)$ & Fano factor at bin width $\tau$: $\mathrm{Var}[C_\tau] / \mathrm{E}[C_\tau]$ \\
$H$ & Hurst exponent, from log-log slope of $F(\tau)$ \\
$H$, $P$, $G$, $B$ & as stream labels: the real stream and the three packet-level null streams (\S\ref{sec:setup-nulls}); the stream label $H$ is distinct from the Hurst exponent \\
$\sigma_i$, $\bar\sigma$ & span of packet $i$ (SBE messages per datagram) and its window mean (\S\ref{sec:crossval-spanrun}) \\
$S_i$ & span-dependent service time of packet $i$ (\S\ref{sec:crossval-spanrun}) \\
\midrule
$\spub$ & publisher period: the minimum period between consecutive packets from the exchange's market-data publisher, i.e.\ the publisher's service time per packet, the inverse of its maximum send rate; about $7.5\,\mu\mathrm{s}$ on NQ, measured as the $1$st-percentile gap between packets, $7.43$--$7.53\,\mu\mathrm{s}$ (\S\ref{sec:exchange-publisher}). A receiver stage slower than $\spub$ queues on a burst; a faster one does not. \\
no-queue latency & the latency of a message that finds every queue empty, $T + (N-1)h$ \\
service floor & the two-stage production receiver's time from socket read to book publish for a one-message packet with nothing queued ahead of it, covering the decode-and-book work and the hand-off from the socket-reader thread to the decoder thread; $6.8$--$7.2\,\mu\mathrm{s}$ (\S\ref{sec:crossval-floor}) \\
\bottomrule
\end{tabular}
\end{center}

The per-window
packet count and the packet-Hawkes branching ratio would both naturally be written
$n_{\mathrm{pkt}}$; the paper reserves $n_{\mathrm{pkt}}$ for the branching ratio and
writes $n_{\mathrm{pkt}}^{\#}$ for the count.

\clearpage
\part{Theory}

\section{Analytical framework}
\label{sec:theory}

Readers not interested in the formal presentation may skip to
Section~\ref{sec:theory-summary}, which restates the results in practical terms.

Throughout, the \emph{tandem} is the paper's model of the multi-actor pipeline: the total
per-message service time $T$ is split into $N$ stages placed in series, each stage an
independent deterministic-service server running on its own thread that performs its share
of the work and passes the message to the next stage across a mailbox hop of cost $h$. The
single-thread event loop is the $N = 1$ case, in which one server performs all of $T$;
adding stages shortens each server's service but inserts a hop between neighbours.

A proposition, a theorem and a corollary underlie the simulation. Proposition~\ref{thm:burst}
is the isolated-burst identity. A message's end-to-end \emph{latency} is its
\emph{service time} $T$ plus its \emph{queueing wait}, the time it spends behind
predecessors before service begins. Running the $N$ stages concurrently raises throughput
from $1/T$ to $N/T$, so an $N$-stage tandem clears a back-to-back burst $N$ times faster:
the service time is unchanged at $T$, and the last message's queueing wait falls by a
factor of $N$. This is the classical pipeline speedup.

For example, split a service-$T$ server into two equal $T/2$ stages and feed it a burst of
$k$ back-to-back messages. All of the queueing now happens at the first stage: its last
message waits behind $k-1$ predecessors that each occupy it for only $T/2$, a wait of
$(k-1)\,T/2$ against $(k-1)\,T$ on the single server. The second stage never queues: stage-one departures arrive spaced $T/2$ apart, equal to
its service time, so each message passes through in $T/2 + h$. The wait is not spread
across two queues; it remains in one bottleneck stage that is now twice as fast. The
halving requires an equal split: cutting $T = 10$ as $7 + 3$ leaves a bottleneck of $7$,
so the wait falls only to $7/10$. In general the tail scales with the largest stage
$s_{\max}$ rather than with the number of stages; Theorem~\ref{thm:reduction} makes this
exact.
Theorem~\ref{thm:reduction} is the general statement: the tandem's end-to-end wait equals
the wait at a single server with the largest stage's service plus $(N-1)h$, exactly and on
every sample path, and that single-server wait is at most $1/N$ of the wait at service
$T$. Both are elementary and independent of the arrival law; Theorem~\ref{thm:reduction}
rests on a reduction theorem of \citet{friedman1965} and on the monotonicity of the
Lindley recursion. Corollary~\ref{cor:poisson} treats the Poisson null: the Poisson wait
at quantile $q$ is zero until $\rho = 1 - q$, which at the utilisation of a CME packet
feed covers $p_{99}$ for every task up to $16\,\mu\mathrm{s}$; over that band splitting
costs $(N-1)h$ and removes nothing. The arrival law enters only through the size of the
single-stage wait to which the factor $1/N$ is applied, and the Hawkes--Oakes cluster
representation explains why that wait is large on a real feed. Part~II measures the
slack in the bound on real arrivals.

\subsection{Burst-limit throughput theorem}
\label{sec:theory-burst}

\begin{proposition}[Burst-limit $N$-factor speedup]
\label{thm:burst}
Consider an isolated cluster of $k$ arrivals whose $j$-th member arrives no later than
$(j-1)\,T/N$ after the first, so that no server idles during the burst (an instantaneous
batch is the canonical case). Let $\tau_k^{(1)}$ denote the exit time of the last arrival
under a single-stage deterministic-service queue with service $T$, and let $\tau_k^{(N)}$
denote the exit time under an $N$-stage tandem with per-stage deterministic service $T/N$
and per-stage hop delay $h$. Then
\begin{equation}
\tau_k^{(1)} = k T,
\qquad
\tau_k^{(N)} = T + (k - 1)\,\frac{T}{N} + (N - 1)\, h,
\label{eq:burst}
\end{equation}
and in particular
\begin{equation}
\lim_{k \to \infty} \frac{\tau_k^{(1)}}{\tau_k^{(N)}} = N.
\label{eq:burst-limit}
\end{equation}
\end{proposition}

\begin{remark}
Proposition~\ref{thm:burst} is the classical pipelined-tandem throughput result
\citep{kleinrock1976,kelly1979}; we state it without proof and use it only to fix the
baseline against which the arrival-dependent tail reduction of Theorem~\ref{thm:reduction}
is measured. What is paper-specific is the numerical calibration of $T$ and $h$ on a
production HFT actor framework and the observation that under real CME MDP3 arrivals the
ratio $\tau_k^{(1)}/\tau_k^{(N)}$ is what governs the observed tail-latency compression.
\end{remark}

\paragraph{Reading Equation~\eqref{eq:burst}.} Take a typical HFT task of
$T = 10\,\mu\mathrm{s}$ on a single-thread event loop, and let $k$ be the size of a burst
of packets that hits the pipeline back to back. On a single stage the last message of the
burst leaves at
\[
\tau_k^{(1)} = k\,T ,
\]
because it waits behind $k-1$ predecessors that each spend the full $T$ on the one server;
a ten-packet burst on a $10\,\mu\mathrm{s}$ server therefore has a last-message latency of
$100\,\mu\mathrm{s}$. On an $N$-stage tandem the last message leaves at
\[
\tau_k^{(N)} = T \;+\; (k-1)\,\frac{T}{N} \;+\; (N-1)\,h ,
\]
which has three parts:
\begin{itemize}
    \item $T$, the total work every message must traverse, whatever $N$ is;
    \item $(k-1)\,T/N$, the wait behind the $k-1$ predecessors, now at throughput $N/T$
instead of $1/T$;
    \item $(N-1)\,h$, the hop cost paid once per mailbox boundary on the way through.
\end{itemize}
A single message ($k = 1$) pays only the first and the last part. Under clustered arrivals
the middle part is large on a single stage and falls by a factor of $N$ under a tandem,
while the last part is a fixed additive cost.

\paragraph{Reading Equation~\eqref{eq:burst-limit}.} The ratio $\tau_k^{(1)}/\tau_k^{(N)}$
tends to $N$ as the cluster size $k$ grows, meaning that in the burst-saturation limit a
tandem of $N$ actors clears the burst $N$ times faster than an event loop. This is the
classical pipeline throughput speedup restated for HFT: a $4$-actor pipeline drains a
large burst in one-quarter of the wall-clock time that a single-thread event loop would
take, and a $12$-actor pipeline in one-twelfth. The speedup is obtained with the same
total service work per message: the pipeline does not reduce the work per event, it
overlaps the work across cores so that the queue behind a cluster clears in parallel. The arrival law does not enter
Equation~\eqref{eq:burst-limit}; what the arrival law controls is \emph{how often} the
burst regime is entered, and Theorem~\ref{thm:reduction} shows the same factor governs
every message on every sample path.

\paragraph{On what $T$ is, and on the natural stage boundaries.} The paper adopts $T =
10\,\mu\mathrm{s}$ as a ballpark typical HFT single-thread event loop. Any real HFT pipeline of that duration is already a sequential chain of internally
atomic sub-tasks: SBE decode, order-book application, strategy touch and order
construction (Section~\ref{sec:implementation}). The $N$ stages of the tandem are placed at
the existing boundaries between those sub-tasks. The paper does not subdivide any single
sub-task, and in particular not the SBE decode, which is a linear pass over the wire bytes
of one message and cannot be split. Ten microseconds is chosen for exposition: it is a round value just above both the
service floor of the two-stage production receiver and the publisher period, and none of the paper's qualitative
results depend on it. The
relevant knob for the design equation is the ratio $T/h$: doubling $T$ at fixed $h$
doubles the cluster-wait terms of Equation~\eqref{eq:burst} and leaves the hop tax
unchanged. The service time below which splitting cannot pay
(Section~\ref{sec:results-design}) turns out to be set by the feed rather than by $h$: it
is the publisher period, and on this corpus it binds well above $2h$. The paper writes ``per-stage
service $T/N$'' as a simplifying modelling convention that assumes equal-sized stages; in
reality the per-stage services $s_i$ sum to $T$ but are unequal, and the achievable
pipeline throughput is $1/\max_i s_i$ rather than $N/T$. Theorem~\ref{thm:reduction}
handles the unequal case exactly: the factor becomes $s_{\max}/T$ in place of $1/N$, the
tail still contracts under Hawkes, and the median still grows linearly at slope $h$ per
hop.

The mechanism applies to any deterministic-service tandem regardless of the arrival
law; the arrival law determines how often large-$k$ clusters occur. The next subsection
makes the per-burst factor a pathwise bound.

\subsection{Exact reduction of the tandem to one server, and the scaling bound}
\label{sec:theory-cluster}

Proposition~\ref{thm:burst} treats one isolated burst. The result that carries the paper's
design claim is that the same $N$-factor holds pathwise for every message under every
arrival sequence, at the level of the end-to-end wait. Fix an arrival sequence $a_1 \le
a_2 \le \cdots$ (any point process, clustered or otherwise, with batches allowed as
coincident times) and write $\delta_j = a_{j+1} - a_j \ge 0$. For a single FIFO server
with deterministic service $s$ fed by this sequence, the waiting time of message $j$ obeys
the Lindley recursion \citep{lindley1952}
\begin{equation}
w_1(s) = 0, \qquad w_{j+1}(s) = \big(w_j(s) + s - \delta_j\big)^+ .
\label{eq:lindley}
\end{equation}
Define the end-to-end wait of message $j$ in the $N$-stage tandem as $W_j^{(N)} =
(\text{exit time}) - a_j - T$, i.e.\ latency in excess of the total service work; under $N
= 1$ this is $W_j^{(1)} = w_j(T)$.

\begin{theorem}[Reduction and scaling bound]
\label{thm:reduction}
Consider an $N$-stage tandem of FIFO servers with deterministic per-stage services $s_1,
\ldots, s_N$ summing to $T$, deterministic hop delay $h$ between consecutive stages,
unbounded buffers, and an arbitrary arrival sequence. Let $s_{\max} = \max_i s_i$. Then
for every message $j$ and every sample path,
\begin{align}
W_j^{(N)} &= w_j(s_{\max}) + (N-1)\,h, \label{eq:reduction} \\
w_j(s') &\le \frac{s'}{s}\; w_j(s) \qquad \text{for all } 0 < s' \le s, \label{eq:scaling}
\end{align}
and consequently
\begin{equation}
W_j^{(N)} \;\le\; \frac{s_{\max}}{T}\, W_j^{(1)} + (N-1)\,h ,
\label{eq:bound}
\end{equation}
which under equal stages, $s_i = T/N$, reads $W_j^{(N)} \le W_j^{(1)}/N + (N-1)\,h$.
Since Equation~\eqref{eq:bound} holds pathwise, it holds for every quantile of the
stationary distribution whenever one exists: $W_q^{(N)} \le W_q^{(1)}/N + (N-1)h$ for all
$q$. The same inequality holds for the empirical (sample) quantiles computed over any
finite trace, directly from the pathwise bound and without invoking stationarity; the
sample-quantile form is what Part~II measures, and it should be read as a statement about
the observed windows rather than about a stationary-distribution quantile or a
finite-sample maximum.
\end{theorem}

\begin{proof}
Equation~\eqref{eq:reduction}, equal stages. Departures from a FIFO server with
deterministic service $s$ are spaced at least $s$ apart, because each departure requires
$s$ of uninterrupted service after the previous one. A FIFO server with deterministic
service $s$ whose arrivals are spaced at least $s$ apart never queues: by induction,
message $j$ departs at its arrival time plus $s$, which is at or before the arrival of
message $j+1$. Stage~1 is a single server with service $T/N$ fed by the original sequence,
so message $j$ leaves it at $a_j + w_j(T/N) + T/N$. The arrival sequence at stage~2 is the
stage-1 departure sequence shifted by $h$, spaced at least $T/N$ apart, so stage~2 adds
exactly $h + T/N$ and its departures are again spaced at least $T/N$ apart; the same holds
at each subsequent stage. Summing, the exit time is $a_j + w_j(T/N) + N \cdot T/N +
(N-1)h$, and subtracting $a_j + T$ gives Equation~\eqref{eq:reduction} with $s_{\max} =
T/N$.

Equation~\eqref{eq:reduction}, unequal stages. By the order-invariance theorem of
\citet{friedman1965}, for a tandem of FIFO stages with constant service times, unlimited
queues, and any arrival sequence, each customer's time in system is independent of the
order of the stages. Deterministic hop delays preserve this: a hop is a constant shift of
one stage's departure sequence into the next stage's arrival sequence, and constant shifts
commute with the Lindley recursion, so reordering stages together with their hops leaves
every exit time unchanged. Reorder so that the stage with service $s_{\max}$ comes first.
Its departures are spaced at least $s_{\max}$ apart, and every subsequent stage has
service $s_i \le s_{\max}$, so by the no-queueing argument above each of them adds exactly
$s_i + h$ to every message. The exit time is $a_j + w_j(s_{\max}) + \sum_i s_i + (N-1)h$,
and subtracting $a_j + T$ gives Equation~\eqref{eq:reduction}.

Equation~\eqref{eq:scaling}. By induction on $j$ with $w_1(s') = 0 = (s'/s) w_1(s)$. The
map $x \mapsto (x + s' - \delta_j)^+$ is nondecreasing, so
\[
w_{j+1}(s') = \big(w_j(s') + s' - \delta_j\big)^+ \le \Big(\tfrac{s'}{s} w_j(s) + s' -
\delta_j\Big)^+ \le \Big(\tfrac{s'}{s} w_j(s) + s' - \tfrac{s'}{s}\delta_j\Big)^+ =
\tfrac{s'}{s}\, w_{j+1}(s),
\]
where the second inequality uses $\delta_j \ge (s'/s)\,\delta_j$, which holds because
$\delta_j \ge 0$ and $s' \le s$. Combining Equations~\eqref{eq:reduction}
and~\eqref{eq:scaling} with $s' = s_{\max}$, $s = T$ gives Equation~\eqref{eq:bound}.
\end{proof}

\begin{remark}[What the theorem says and what it does not]
Equation~\eqref{eq:reduction} says all queueing in an equal-stage deterministic tandem
happens at stage~1; stages $2, \ldots, N$ contribute exactly $T/N + h$ each to every
message. Proposition~\ref{thm:burst} is the isolated-burst special case, where
Equation~\eqref{eq:bound} holds with equality. Equation~\eqref{eq:bound} is a bound rather
than an identity because the tandem runs its bottleneck stage at utilisation $\rho/N$
while the single stage runs at $\rho$: the induction step in the proof of
Equation~\eqref{eq:scaling} gives away $(1 - 1/N)\,\delta_j$ at every message $j$ that the
single server is still serving when message $j+1$ arrives, so the tandem clears each
inter-arrival gap $N$ times more effectively than the scaled bound credits it for. The
paper's empirical $\gamma$ (Section~\ref{sec:results-design}), defined by $W_q^{(N)} -
(N-1)h \approx N^{-\gamma}\, W_q^{(1)}$, therefore satisfies $\gamma \ge 1$ in the model,
and the excess $\gamma - 1$ measures that slack. Nothing in the theorem depends on the
arrival law; what the arrival law determines is the size of $W_q^{(1)}$ that the factor
$1/N$ is applied to, which is the subject of the next remark. The hop cost $h$ enters
every quantile as exactly $(N-1)h$, in the median and in the tail alike.
\end{remark}

\begin{remark}[Unequal stages: only the largest stage matters]
\label{rem:unequal}
Real cut points do not give equal stages. Equation~\eqref{eq:reduction} states the
consequence: the end-to-end wait of any deterministic tandem equals the wait at a single
server with service $s_{\max}$, plus $h$ per hop. Three consequences follow.
\begin{enumerate}
    \item The tail factor is $s_{\max}/T$, and the stage count enters only through the hop
tax. Splitting a $10\,\mu\mathrm{s}$ task as $7 + 3$ compresses the tail excess to at most
$70\%$ of single-stage, against $50\%$ for $5 + 5$, and both pay one $h$.
    \item The order of the stages is irrelevant to end-to-end latency: $7 + 3$ and $3 + 7$
give identical latencies on every sample path. Order determines only where the queue
physically sits (below).
    \item A cut that does not lower $s_{\max}$ adds a hop and nothing else. Refining
$7 + 3$ into $7 + 2 + 1$ leaves every quantile of the wait unchanged and raises
end-to-end latency by $h$. Conversely, sub-tasks can be merged onto one thread (by
synchronous sends) without any tail cost as long as the merged service does not exceed
$s_{\max}$. The design problem is therefore to choose the partition with the smallest
$s_{\max}$ reachable at the natural cut points, and then the fewest hops that realise it.
\end{enumerate}
Appendix~\ref{app:unequal} checks all three numerically.
\end{remark}

\begin{remark}[Where the burst goes: the bottleneck absorbs the clustering]
\label{rem:dissipation}
The proof of Equation~\eqref{eq:reduction} rests on one fact: departures from a
deterministic server with service $s$ are spaced at least $s$ apart. Applied to a Hawkes
input, that fact says what the bottleneck stage does to the clustering. Upstream of the
bottleneck the stream is the raw feed, with intra-cluster inter-arrival times far below
$s_{\max}$ and instantaneous intensity far above $1/s_{\max}$. Downstream of the
bottleneck, every cluster still leaves as the same group of $K$ messages, so the
count-level branching structure is unchanged (the Borel cluster sizes, and the Fano factor
at time scales beyond $K s_{\max}$, are exactly those of the input), but the intra-cluster
spacing has been regularised to $s_{\max}$ and the instantaneous intensity is capped at
$1/s_{\max}$. The clustering is not removed; it is regularised to a spacing that any consumer at
least as fast as the bottleneck can absorb, which is why stages with service $s_i \le
s_{\max}$ never queue. If the bottleneck is stage~1, all of the waiting in the pipeline sits in its mailbox
and every downstream mailbox holds at most one message. If the bottleneck is a later
stage $b$, the stages before it queue partially, since the raw burst is still bursty
relative to their smaller services, and stage $b$ absorbs the remainder; the total is the
same by Equation~\eqref{eq:reduction}.

Four practical consequences follow. First, place the largest stage first, at ingress. End-to-end latency is unchanged by
the order, but this confines all queueing to one mailbox and keeps the downstream
mailboxes at depth one. It also places the burst-absorbing queue where its receiver is
continuously busy during a cluster, the regime in which the load-dependent hop cost of
Section~\ref{sec:results-h} is smallest. Second, any burst-aware logic, such as load shedding, coalescing of superseded
book updates, or a cancel-on-burst risk rule, can only observe the raw burst upstream of
the bottleneck; downstream it sees a rate-$1/s_{\max}$ stream and cannot tell a cluster
from steady flow except by counting. Third, a packet-Hawkes fit run on the bottleneck's
departure stream would report a lower intensity peak and a longer effective kernel while
recovering the same branching ratio in counts; operators instrumenting a pipeline should
fit on the ingress stream. Fourth, deterministic hop delay is what keeps downstream stages
queue-free; hop jitter of amplitude $\epsilon$ lets a downstream stage queue by at most
$\epsilon$ per message, so the reduction holds up to that jitter in practice.
\end{remark}

\begin{fact}[Hawkes--Oakes 1974 cluster representation \citep{hawkesoakes1974}]
A stationary Hawkes process with background rate $\mu > 0$ and branching ratio $n \in [0,
1)$ is stochastically equivalent to a Poisson-cluster process: a rate-$\mu$ Poisson
process of \emph{immigrant} arrivals, each of which independently spawns an offspring
cluster whose total size $K$ (including the immigrant) follows the Borel distribution
\begin{equation}
\Pr(K = k) = \frac{(n k)^{k-1}\, e^{-n k}}{k!}, \qquad k = 1, 2, 3, \ldots
\label{eq:borel}
\end{equation}
with mean $\mathrm{E}[K] = 1/(1-n)$ and variance $\mathrm{Var}[K] = n/(1-n)^3$
\citep{tanner1961}.
\end{fact}

\begin{remark}[Where Hawkes enters: the size of the single-stage wait]
Under Poisson arrivals at the utilisation of a CME packet feed ($\rho =
\bar\lambda_{\mathrm{pkt}} T \ll 1$), the single-stage wait is zero for almost every
message, so $W_q^{(1)} = 0$ for every quantile $q$ with $\rho < 1 - q$ and
Equation~\eqref{eq:bound} has nothing to compress; that is Corollary~\ref{cor:poisson},
and at this corpus's rate it holds at $p_{99}$ for all $T \le 16\,\mu\mathrm{s}$.
Under Hawkes arrivals the cluster representation puts $K$ near-coincident messages behind
one immigrant, and in the instantaneous-batch picture the last of them waits $(K-1)T$ on a
single stage. Stirling's formula applied to Equation~\eqref{eq:borel} gives, for large
$k$,
\begin{equation}
\Pr(K = k) \;\approx\; \frac{1}{n\sqrt{2\pi}}\; k^{-3/2}\, \big(n e^{1-n}\big)^k ,
\label{eq:borel-asymp}
\end{equation}
and $n e^{1-n} < 1$ for every $n < 1$, with $-\log(n e^{1-n}) \approx (1-n)^2/2$ near
criticality. The Borel batch size is therefore light-tailed at every fixed $n < 1$: its survival
function follows the critical $k^{-1/2}$ power law only in the window $k \ll k^\star =
(1-n)^{-2}$ and decays geometrically beyond it (Appendix~\ref{app:borel}). Three
consequences follow for the paper.
\begin{enumerate}
    \item The largest clusters, and the single-stage wait they produce, have a
characteristic scale of $T\,k^\star$. At the corpus's branching ratio of about $0.8$ that
is some twenty-five service times, and at $0.9$ it would be a hundred: one to two orders
of magnitude above the service time and the hop cost.
    \item Heavy-tail results built on subexponential batch sizes, such as the
single-big-batch principle for $M^{[K]}/G/1$ \citep{asmussen2003} and the
monotone-separable tail asymptotics of \citet{baccellifosslelarge2005}, do not apply to
this batch distribution at any fixed $n < 1$, so no closed-form power-law exponent for the
tail probability should be expected.
    \item The design claim does not depend on any of this, because
Theorem~\ref{thm:reduction} applies the factor $1/N$ to whatever single-stage wait the
feed produces. Part~II measures that wait on real arrivals.
\end{enumerate}
\end{remark}

\begin{corollary}[Poisson null]
\label{cor:poisson}
Under homogeneous Poisson arrivals at rate $\bar\lambda$ with $\rho = \bar\lambda T < 1 -
q$, the $q$-quantile of the single-stage wait is $W_q^{(1)} = 0$, because $\Pr(W^{(1)} >
0) = \rho$ for $M/D/1$ by the PASTA property (Poisson arrivals see time averages). By Theorem~\ref{thm:reduction}, $W_q^{(N)} \le (N-1)h$,
and by Equation~\eqref{eq:reduction} $W_q^{(N)} \ge (N-1)h$, so
\begin{equation}
W_q^{(N)} = (N-1)\,h \qquad \text{for all } N \text{ and all } q \le 1 - \rho .
\label{eq:poisson-cor}
\end{equation}
A CME packet feed runs at well under one percent utilisation at HFT service times, so on
a Poisson stream this covers the $p_{99}$ at every service time up to
$16\,\mu\mathrm{s}$ in the paper's sweep: the $p_{99}$ equals the median, both at the
no-queue latency, and splitting adds hops with nothing to compress. Where the per-stage
utilisation does exceed $1 - q$, the bound of Equation~\eqref{eq:bound} still holds and
the tandem reduces the Poisson wait too, for the ordinary utilisation reason rather than
for clustering. The sweep contains six such $p_{99}$ cells, all at the long-service end,
and Section~\ref{sec:results-tail} finds the wait in those cells and in no others.
\end{corollary}

\subsection{Informal statement of the two results}
\label{sec:theory-summary}

This subsection restates Proposition~\ref{thm:burst}, Theorem~\ref{thm:reduction} and
Corollary~\ref{cor:poisson} informally. The formal statements above are the reference.

\paragraph{Components of latency.} A message's end-to-end latency is the sum of three
terms: the work done on it ($T$ in total, however many stages that work is
spread over), the hop tax ($h$ for every mailbox boundary it crosses, so $(N-1)h$ in an
$N$-stage pipeline), and the time it spends sitting in queues waiting for a busy actor.
The third term is the wait, $W$. The work is fixed by the code and the hop cost by the
framework; the wait is the only term the architecture changes, and on a bursty feed it
dominates the tail.

\paragraph{Proposition~\ref{thm:burst}: one burst.} Suppose $k$ messages land at once. On
a single thread the last of them waits for the $k-1$ ahead of it, each costing $T$: it
waits $(k-1)T$. On an $N$-stage pipeline the first stage takes only $T/N$ per message, so
the last message waits $(k-1)T/N$ at stage~1 and then flows through the remaining stages
without waiting again, because each of them is fed at exactly the rate it can serve. The pipeline clears a burst $N$ times faster for the same total work, because $N$ cores
process $N$ different messages at the same instant.

\paragraph{Theorem~\ref{thm:reduction}: every message, any arrival sequence.} Two
facts.
\begin{enumerate}
    \item In a pipeline where each stage does a fixed amount of work per message, only the
slowest stage ever builds a queue. Every other stage receives messages no faster than it
can process them, because the slowest stage upstream (or, if the slowest stage is
downstream, the stages before it) meters them out. The pipeline therefore behaves as one thread whose per-message work is $s_{\max}$,
plus the hop cost. With equal stages, $s_{\max} = T/N$.
    \item Take any recording of arrival times, real or synthetic, and replay it through
two single-threaded servers: one doing $T$ of work per message and one doing $T/N$. The
second server's queue wait for message number $j$ is never more than $1/N$ of the first
server's wait for the same message. The reason is that in the faster server every backlog
is $N$ times smaller and every idle gap between arrivals clears at least as much of it.
\end{enumerate}
Combined: for every message in every replay, the pipeline's wait is at most the
single-thread wait divided by $N$, plus $(N-1)h$. Because this holds message by message,
it holds for any summary statistic: median, $p_{99}$, $p_{99.9}$, or maximum. The bound
is met with equality on an isolated instantaneous burst and is strict whenever arrivals
are spread out, because the pipeline's first stage is idle more often than the single
thread and drains more of its backlog in each idle gap. The measured compression is
written as a power of the stage count, $N^{-\gamma}$; the theorem guarantees an exponent
of at least one, and Section~\ref{sec:results} measures about three at the longest service
times. At shorter service times the tandem reduces the excess to zero and no exponent is
defined.

\paragraph{Role of the arrival law.} The arrival law does not enter either
theorem. It enters only through the size of the single-thread wait to which the factor
$1/N$ is applied. Under Poisson traffic at the load of a CME packet feed, the single thread is
idle when almost every message arrives, so for tasks up to $16\,\mu\mathrm{s}$ its
$p_{99}$ wait is zero; there is nothing to divide by $N$, and the pipeline is simply
$(N-1)h$ slower (Corollary~\ref{cor:poisson}). Under the real feed, transactions arrive
in self-exciting clusters and the packets relay them, the single thread's $p_{99}$ is two
to about forty times $T$ for tasks of $16$ to $128\,\mu\mathrm{s}$, and the pipeline divides
the excess by $N$ or better.
The Hawkes cluster representation gives a model-implied cluster size: Borel-distributed,
with large clusters of order $(1-n)^{-2}$ transactions, about $25$ at the corpus's
$n \approx 0.8$. What the receiver waits behind depends on the service time: at
$8$--$16\,\mu\mathrm{s}$ the wait is usually behind one predecessor sent
$7.5$--$16\,\mu\mathrm{s}$ earlier,
and only from about $64\,\mu\mathrm{s}$ does a run of cluster members fall within one
service time (Sections~\ref{sec:results-nulls} and~\ref{sec:evaluation}).

\begin{center}
\fbox{\begin{minipage}{0.92\linewidth}
\textbf{Design rules (from Proposition~\ref{thm:burst} and Theorem~\ref{thm:reduction}).}
\begin{itemize}
    \item The tail improvement is set by the slowest stage, $s_{\max}/T$, and by nothing
else about the decomposition. Balanced stages are the best case for a given number of
hops.
    \item A cut that does not reduce the slowest stage costs $h$ at every quantile and
improves nothing. Sub-tasks can be merged onto one thread until their combined work
reaches $s_{\max}$ at no cost to the tail.
    \item The order of the stages does not change any message's latency. It changes only
where the queue physically sits; putting the slowest stage first keeps all queueing in one
mailbox.
    \item Every hop costs $h$ on every message, at the median and in the tail.
    \item For a stream that must be processed in order, whether to split reduces to one
comparison: the single-thread wait at the quantile of interest, times $(1 - s_{\max}/T)$,
against $h$ per hop added. On a Poisson feed at HFT service times the left side is zero;
on a clustered feed at the same rate it is tens to thousands of microseconds. With the
same cores, a stateless stage on a stream that tolerates reordering can instead dispatch
whole packets to $N$ servers; Section~\ref{sec:results-equal-core} compares the two, first without and then with
dispatch's ingress, egress and resequencing costs.
\end{itemize}
\end{minipage}}
\end{center}

The guarantees above are exact within a specific model. Its assumptions follow.

\paragraph{Fixed work per message.} Each stage is assumed to spend the same amount of time
on every message. That is what makes only the slowest stage queue: a stage that sometimes
takes longer than its upstream neighbour would occasionally fall behind and build its own
queue, and then the pipeline's wait would exceed the single-server wait at $s_{\max}$. The simulator of Part~II uses constant per-packet service, so the theorems describe it
exactly; Section~\ref{sec:crossval-spanrun} relaxes this for the message count per
packet. Real actors also have work that varies with book state, and for them the theorems
are a model to be checked by measurement (Section~\ref{sec:fungibility}).

\paragraph{No drops, no bounds, in-order.} Mailboxes are unbounded, messages are never
discarded or coalesced, and each stage processes in arrival order. A framework that drops
or merges messages under load changes the arithmetic (in the pipeline's favour, at the
cost of lost messages); a priority queue changes which message pays the wait but not the
total.

\paragraph{Fixed hop delay.} The hop is modelled as a constant $h$. A jittery hop lets a
downstream stage queue by at most the jitter amplitude per message, so the reduction holds
up to that amplitude. The paper's calibration uses the low-load, wake-from-sleep value of
$h$, which Section~\ref{sec:results-h} argues is an upper bound inside a burst.

\paragraph{Equal stages, no cross-core cost beyond $h$.} The $N$ stages are taken to
split the work evenly and to exchange nothing but the message: moving book state or cache
lines between the cores that run adjacent stages is charged nothing beyond $h$. Unequal
stages are handled by Remark~\ref{rem:unequal}; a cross-core cost would add to $h$ and
is a matter for measurement (Section~\ref{sec:results-h}).

\paragraph{Proposition~\ref{thm:burst} needs a dense burst.} The one-burst identity
assumes the $k$ messages arrive faster than the first stage can drain them, i.e.\ faster
than one per $T/N$. For a burst that trickles in, the ratio in
Equation~\eqref{eq:burst-limit} is $kN / (k + N - 1 + N(N-1)h/T)$, which climbs toward $N$
as the burst grows. Theorem~\ref{thm:reduction} does not need this assumption; it covers
every arrival pattern.

\paragraph{Limits of the theory.} It does not say how large the single-thread
wait is on a given feed. That number depends on the feed's clustering and utilisation and,
as the remark after the Borel fact explains, has no clean closed form at a fixed branching
ratio: the Borel cluster size is light-tailed for every $n < 1$, so the heavy-tail
formulas that would give a power-law exponent do not apply. The paper measures the
single-thread wait on real CME arrivals instead (Part~II). Nor does the theory predict how
far below the $1/N$ bound the pipeline lands; that slack is also measured. What the theory guarantees is that the pipeline's wait is at most $1/N$ of the
single-thread wait plus the hop cost, on any arrival sequence.

\clearpage
\part{Measurement of the feed and simulation of the receiver}

\section{Experiment setup}
\label{sec:experiment-setup}

This section describes the empirical experiment that tests the analytical predictions of
Section~\ref{sec:theory}: the data corpus, the confirmation that packet arrivals are
Hawkes-clustered (the reason the single-stage wait that Theorem~\ref{thm:reduction}
divides by $N$ is large at all), the tandem-Lindley simulator, and the sweep design that
produces the numbers reported in Section~\ref{sec:results}. Figure~\ref{fig:experiment}
shows the end-to-end data flow of a single 30-minute window; Algorithm~\ref{alg:sweep}
states the full per-window procedure in pseudocode.

\begin{figure}[H]
\centering
\begin{tikzpicture}[
    node distance=6mm and 8mm,
    box/.style={draw, rounded corners=1pt, align=center, minimum height=8mm, inner sep=3pt, font=\small},
    tape/.style={box, fill=black!5, minimum width=32mm},
    arm/.style={box, fill=blue!5, minimum width=32mm},
    sim/.style={box, fill=orange!10, minimum width=42mm},
    agg/.style={box, fill=green!10, minimum width=42mm},
    arrow/.style={-{Latex[length=2mm]}},
]
\node[tape] (tape) {CME MDP3 tape\\ (one session)};
\node[box, right=of tape, minimum width=32mm] (window) {30-min window\\ $[T_0, T_1]$};
\node[box, right=of window, minimum width=38mm] (packets) {group by \texttt{packet\_seq}\\ $\Rightarrow$ packet arrivals};

\node[arm, below=8mm of packets, xshift=-30mm] (hawkes) {real arrivals};
\node[arm, below=8mm of packets, xshift=+30mm] (poisson) {Poisson-null arm\\ uniform shuffle};

\node[sim, below=8mm of hawkes] (simH) {Tandem Lindley\\ $N \in \{1,2,4,8\}$};
\node[sim, below=8mm of poisson] (simP) {Tandem Lindley\\ $N \in \{1,2,4,8\}$};

\node[agg, below=10mm of $(simH)!0.5!(simP)$] (agg) {Per-window quantiles\\ $p_{50}, p_{95}, p_{99}, p_{99.9}, \max$};
\node[box, below=6mm of agg, minimum width=42mm, fill=gray!10] (corpus) {corpus-median across 3512 windows};

\draw[arrow] (tape) -- (window);
\draw[arrow] (window) -- (packets);
\draw[arrow] (packets.south) -- ++(0,-3mm) -| (hawkes.north);
\draw[arrow] (packets.south) -- ++(0,-3mm) -| (poisson.north);
\draw[arrow] (hawkes) -- (simH);
\draw[arrow] (poisson) -- (simP);
\draw[arrow] (simH.south) -- ++(0,-3mm) -| (agg.north);
\draw[arrow] (simP.south) -- ++(0,-3mm) -| (agg.north);
\draw[arrow] (agg) -- (corpus);
\end{tikzpicture}
\caption{Per-window data flow of the empirical experiment. A single session tape is
chopped into 30-minute RTH windows. Each window's messages are grouped by
\texttt{packet\_seq} into a real packet-arrival stream (Hawkes arm). A rate-matched
Poisson-null arm is generated by uniformly shuffling the same arrival count over the same
window. Both arms are run through the $N$-stage tandem Lindley simulator at $N \in \{1, 2,
4, 8\}$. Per-window quantiles are recorded and aggregated by corpus median across all 3512
windows in the 276-session corpus.}
\label{fig:experiment}
\end{figure}
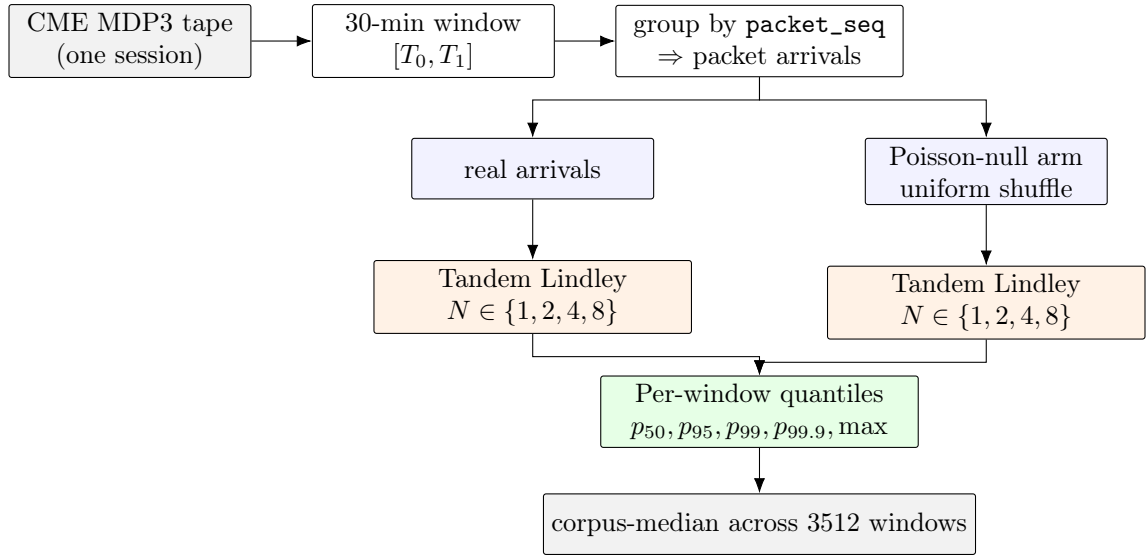

\begin{algorithm}[H]
\caption{Per-window procedure of the empirical experiment.}
\label{alg:sweep}
\begin{algorithmic}[1]
\State \textbf{input:} session tape $S$, service $T$, hop $h$, stage counts $\mathcal{N}$
\For{each 30-min RTH window $W$ in $S$ with $[T_0, T_1]$}
    \State $\{m_i\} \gets$ messages in $W$
    \State $\{(\text{seq}, t)\} \gets \{(\text{packet\_seq}(m_i), \text{sendingTime}(m_i))\}$
    \State $\text{arr}_H \gets \text{sort}\big(\{\min_i t_i : \text{seq}_i = s\}_s\big)$ \Comment{Hawkes arm (real arrivals)}
    \State $n_{\mathrm{pkt}}^{\#} \gets |\text{arr}_H|$
    \State $\text{arr}_P \gets \text{sort}\big(\text{Uniform}(\min \text{arr}_H, \max \text{arr}_H)^{n_{\mathrm{pkt}}^{\#}}\big)$ \Comment{Poisson-null arm (seeded per window)}
    \For{regime $\in \{H, P\}$, $N \in \mathcal{N}$}
        \State $\text{lat} \gets \textsc{TandemLindley}(\text{arr}_{\text{regime}}, N, T, h)$
        \State record $\big(S, W, \text{regime}, N, \{p_{50}, p_{95}, p_{99}, p_{99.9}, \max\}(\text{lat})\big)$
    \EndFor
\EndFor
\State \Return corpus-median of each recorded quantile across all $(S, W)$ pairs
\Function{TandemLindley}{$\text{arr}, N, T, h$}
    \State $\text{dep} \gets \text{arr}$
    \For{stage $k = 1, \ldots, N$}
        \If{$k > 1$} $\text{dep} \gets \text{dep} + h$ \EndIf
        \State $\text{dep} \gets \text{dep} + \textsc{Lindley}(\text{dep}, T/N)$
    \EndFor
    \State \Return $\text{dep} - \text{arr}$
\EndFunction
\end{algorithmic}
\end{algorithm}

\subsection{Data corpus}
\label{sec:setup-data}

The corpus is 281 sessions of CME MDP3 NQ front-month pcap-derived tapes
covering 2025-01 through 2026-02, RTH only, holidays and half-days excluded. Each MBO
record carries a \texttt{packet\_seq} (the UDP packet sequence number) and two
nanosecond timestamps: \texttt{transactTime}, the CME matching-engine event time carried
in the MDP3 message body, and \texttt{sendingTime}, the publisher's per-packet header
timestamp. Messages are grouped into packets by \texttt{packet\_seq}, and a packet's
arrival time is its \texttt{sendingTime}.

\texttt{transactTime} is not a suitable arrival clock. It is the matching engine's event
time, not the time at which anything reaches the receiver: the $1.7\%$ of transactions
that the publisher splits across two or more packets
(Section~\ref{sec:results-transactions}) carry one \texttt{transactTime} on packets that
leave about $15\,\mu\mathrm{s}$ apart, so keying arrivals on it collapses distinct
packets onto one instant. \texttt{sendingTime} is emitted once per packet at nanosecond
resolution and is the receiving pipeline's input clock; ordered by \texttt{packet\_seq}
it produces no ties and no sequence inversions on the corpus. All results below are keyed
on \texttt{sendingTime}. Section~\ref{sec:results-transactions} shows that merging each
transaction into one packet leaves the tail within $\pm 15\%$, so the choice of clock
does not affect the results.

Each session tape is chopped into
non-overlapping 30-minute RTH windows; for each window we compute the packet-arrival
stream and the per-window quantities listed in Section~\ref{sec:notation}. The corpus
yields 3512 per-window records from the 276 sessions that contain at least one complete
RTH window, roughly 13 windows per session; the remaining five of the 281 tapes are
holiday or half-day sessions with none.

\subsection{Hawkes clustering of packet arrivals (empirical prerequisite)}
\label{sec:setup-hawkes}

Readers not familiar with point processes will find the Poisson, renewal and Hawkes
processes, the branching ratio, the Fano factor and the Hurst exponent defined in
Appendix~\ref{app:pp-background}.

The paper's simulator input is the packet-arrival stream, so the first empirical question
is whether that stream is clustered or Poisson. Three model-free diagnostics and one
model-based fit, on every window of the corpus, reject Poisson; Table~\ref{tab:hawkes-diag}
gives the values. Nearly three quarters of the gaps between packets are shorter than a
tenth of the mean gap, where a Poisson stream of any rate has under a tenth, and no window
comes close to the Poisson value. The Fano factor of packet counts,
$F(\tau) = \mathrm{Var}[C_\tau] / \mathrm{E}[C_\tau]$ in bins of width $\tau$, is $1$ at
every bin width for Poisson; here it grows from about $50$ at a tenth of a second to about
a thousand at a minute, and the Hurst exponent from its log-log slope is about $0.75$
against $0.5$. These are within-window values; pooling a whole session instead raises the
Fano level without changing the shuffle ratio, for reasons taken up in
Section~\ref{sec:crossval-disagree}. The same estimators applied to the paper's own
Poisson null return the Poisson values at every bin width, so the figures are not an
artefact of binning or of window length.

\begin{table}[H]
\centering
\small
\begin{tabular}{@{}lrrr@{}}
\toprule
statistic & real stream & Poisson & real stream, gaps shuffled \\
\midrule
gaps shorter than a tenth of the mean gap & $73\%$ & $9.5\%$ & $73\%$ \\
Fano factor, $0.1\,\mathrm{s}$ bins & $52$ & $1$ & $16$--$18$ \\
Fano factor, $1\,\mathrm{s}$ bins & $108$ & $1$ & $17.5$ \\
Fano factor, $5\,\mathrm{s}$ bins & $260$ & $1$ & $17.6$ \\
Fano factor, $60\,\mathrm{s}$ bins & $1027$ & $1$ & $16$--$18$ \\
Hurst exponent & $0.74$ & $0.50$ & $0.50$ \\
Hawkes branching ratio $n$ & $0.798$ & $0$ & --- \\
Hawkes kernel decay time $1/\beta$ & $167\,\mu\mathrm{s}$ & --- & --- \\
\bottomrule
\end{tabular}
\caption{Clustering diagnostics of the packet stream, corpus median over windows, against
the Poisson values and against the same stream with its gaps shuffled within each window.
The kernel decay time is from 254 windows in 20 sessions. What the table shows: every
diagnostic is far from Poisson. Shuffling the gaps keeps the share of short gaps, by
construction, and removes all the growth of the Fano factor with bin width, leaving a flat
value of about $17$. What follows: the long-range dependence lives in the order of the
gaps, which is what distinguishes a self-exciting process from a renewal one, while the
flat residual comes from the gap lengths themselves, which are heavy-tailed.}
\label{tab:hawkes-diag}
\end{table}

Fitting the standard exponential-kernel Hawkes process \citep{hawkes1971spectra}
$\lambda(t) = \mu + \sum_{t_i < t} \alpha\, e^{-\beta(t - t_i)}$ to each window by maximum
likelihood \citep{ozaki1979} (Ogata log-likelihood with the standard recursive intensity,
closed-form compensator, Nelder--Mead over $(\mu, \alpha, \beta)$) gives a branching ratio
$n = \alpha/\beta$ of about $0.8$, varying little between windows, with every fit
converged; a reparameterised L-BFGS-B fit finds the same optimum on spot-checked windows.
The kernel decays over about $160\,\mu\mathrm{s}$, so the fitted process triggers at lags of
order a hundred microseconds and places only a few percent of its triggering within the
publisher period. In plain terms, each transaction makes further transactions more
likely, as traders and algorithms react to it, and that effect fades like an echo: about a
third of it remains after $160\,\mu\mathrm{s}$ and almost none after half a millisecond. A
branching ratio of $0.8$ means that each transaction triggers on average $0.8$ further
transactions directly. The reactions are therefore spread over the next few hundred
microseconds; only about $4\%$ of them arrive within the $7.5\,\mu\mathrm{s}$ publisher
period and about $9\%$ within $16\,\mu\mathrm{s}$. The self-excitation is what builds runs
of transactions over hundreds of microseconds, not what places two transactions a few
microseconds apart. This is the near-critical regime that \citet{filimonovsornette2012} and
\citet{hardimanbercotbouchaud2013} document on ES, at a somewhat lower level than their
$n \approx 0.9$; they fit trade arrivals rather than packet arrivals, and on this corpus
the fit on matching-engine transactions gives the same value as on packets
(Section~\ref{sec:results-transactions}). \citet{filimonovsornette2015} show that an exponential-kernel fit with a constant
background over a fixed window overstates $n$ when the background intensity varies within
the window, and that ES estimates fall well below criticality once that variation is
modelled. The $n$ reported here is descriptive and subject to the same caveat, and none of
the paper's conclusions rests on its value. The
model-free diagnostics hold whatever the fitted $n$, and the nonparametric estimates of
\citet{achab2018branching} on EUREX order books give the same near-critical picture
without a kernel assumption.

A shuffle of the gaps within each window keeps every gap and destroys only their order.
It removes most but not all of the excess over Poisson: all of the growth of the Fano
factor with bin width goes, leaving a flat over-dispersion of about $17$ and a Hurst
exponent of $0.5$ (Table~\ref{tab:hawkes-diag}). All of the long-range dependence therefore
lives in the ordering, which is what distinguishes a self-exciting process from a renewal
process, and the flat residual is contributed by the gap lengths, which are heavy-tailed.
The corpus is neither a renewal process nor a pure Hawkes process with light-tailed gaps,
but a self-exciting process with heavy-tailed gaps. This bounds what the Hawkes fit
describes: the fitted $n$ characterises the ordering component, and the simulator is
driven by the recorded arrivals, not by the fitted process. The simulator input is
therefore an arrival process far from Poisson, and the Poisson null at the same rate is a
strict counterfactual. The tandem reduction of Section~\ref{sec:theory} does not depend on
the Hawkes assumption, since any arrival process with the same clustering drives the same
single-server queue. The model-free diagnostics, not the fit, establish the empirical
premise.

\findings{
    \item \emph{Packet arrivals are far from Poisson.} Nearly three quarters of gaps are shorter than a tenth of the mean gap, against under a tenth for Poisson.
    \item \emph{Clustering grows with time scale}: the Fano factor rises from about $50$ at a tenth of a second to about a thousand at a minute (Poisson: $1$).
    \item \emph{A Hawkes fit describes it as near-critical}, with branching ratio about $0.8$ and a kernel acting over about $160\,\mu\mathrm{s}$, mostly beyond the publisher period.
    \item \emph{Two ingredients}: shuffling the gaps removes the growth with time scale (the ordering) and leaves a flat over-dispersion from heavy-tailed gaps.
}

\subsection{Simulator}
\label{sec:setup-simulator}
\label{sec:sim-model}

Packet arrivals $t_1 < t_2 < \cdots$ enter stage 1. Each stage $k$ is a single-server FIFO
queue with deterministic service $T/N$ (own actor, own thread), and departures from stage
$k$ are delayed by $h$ before arriving at stage $k+1$. Standard Lindley recursion per
stage:
\begin{equation}
W^{(k)}_{i+1} = \max\!\left(0,\; W^{(k)}_i + \frac{T}{N} - (a^{(k)}_{i+1} - a^{(k)}_i)\right),
\end{equation}
with $a^{(k+1)}_i = a^{(k)}_i + W^{(k)}_i + T/N + h$.

Service is the same for every packet. Section~\ref{sec:crossval-spanrun} makes it grow
with the number of messages in the packet; nothing in the paper makes it random. A real
receiver's service time varies from packet to packet for reasons other than message
count (Section~\ref{sec:fungibility-measured}), and for a single server, more variable
service at the same mean lengthens the mean wait (the Pollaczek--Khinchine formula for
Poisson arrivals). The simulated tails are therefore what the arrival timing and the
message count produce on their own, not everything a real receiver pays.

\subsection{How the simulator models concurrent stages}
\label{sec:sim-pipeline}

Each stage's actor runs on its own thread, so stage 2 can be processing message $k-1$
while stage 1 is processing message $k$. The simulator captures this because Lindley is
applied separately per stage, each with its own busy-idle timeline. What the split buys
is not the overlap as such but the shorter queue in front of each stage: a burst that one
server of service $T$ holds in a queue of length $k$ is held by the first stage at $T/N$
and drained $N$ times faster.

Worked example. Three packets arriving at $t = 0, 1, 2\,\mu\mathrm{s}$; single-stage total
service $T = 8\,\mu\mathrm{s}$ (a service time in the sweep of
Table~\ref{tab:main-corpus}); two-stage tandem with per-stage service $T/2 =
4\,\mu\mathrm{s}$ and hop $h = 1.7\,\mu\mathrm{s}$:

\begin{center}
\begin{tabular}{lrrr}
\toprule
 & msg 1 & msg 2 & msg 3 \\
\midrule
single-stage arrival & 0 & 1 & 2 \\
single-stage latency ($\mu$s) & 8.0 & 15.0 & \textbf{22.0} \\
\midrule
tandem arrive stage 1 & 0.0 & 1.0 & 2.0 \\
tandem wait stage 1 & 0.0 & 3.0 & 6.0 \\
tandem leave stage 1 & 4.0 & 8.0 & 12.0 \\
tandem arrive stage 2 ($+h$) & 5.7 & 9.7 & 13.7 \\
tandem wait stage 2 & 0.0 & 0.0 & 0.0 \\
tandem leave stage 2 & 9.7 & 13.7 & 17.7 \\
\textbf{tandem latency ($\mu$s)} & \textbf{9.7} & \textbf{12.7} & \textbf{15.7} \\
\bottomrule
\end{tabular}
\end{center}

At wall-clock $t \approx 7$, stage 1 is processing message 2 while stage 2 is
processing message 1. Message 1, which never queued, pays the hop and is slower by exactly $h$; message 3, at
the end of the burst, is faster by a third. The median rises and the tail falls, at the
same message count, because the two stages together clear the burst twice as fast. This is the mechanism of Proposition~\ref{thm:burst}.

A chain of synchronous sends on one thread shortens no queue and pays
$(N-1)h$; only stages on their own threads do (Section~\ref{sec:implementation}).

\paragraph{Calibration and cost.} The paper reports numerical results at $h =
1.7\,\mu\mathrm{s}$ --- the one-way async mailbox cost obtained by halving the round-trip
figure in Table~4 of \citet{mayeski_fastsend} --- swept across $T \in
\{2, 4, 8, 16, 32, 64, 128\}\,\mu\mathrm{s}$ to cover the sub-hop regime through the slow-servicing
regime. The simulator is deterministic given the tape and $(T, h, N)$; one $O(n)$ pass per
stage per window makes its wall-clock cost negligible next to the tape read.

\subsection{Sweep design}
\label{sec:setup-sweep}

The service times are a power-of-two grid; $T = 8\,\mu\mathrm{s}$ is the sweep point just
above the publisher period $\spub \approx 7.5\,\mu\mathrm{s}$ (Section~\ref{sec:exchange-publisher}),
and ``the sweep point just above the publisher period'' means that cell throughout.
The design asks whether an $N$-stage tandem reduces $p_{99}$ latency under real Hawkes
packet arrivals and whether the same architecture achieves the same compression under a
rate-matched Poisson null. For each 30-minute window we run the tandem Lindley of Section~\ref{sec:setup-simulator}
at $N \in \{1, 2, 4, 8\}$ under two arrival regimes: the real packet-arrival timestamps
from the tape (stream $H$), and a uniform redraw of the same number of arrivals over the
same support (the uniform Poisson null $P$). The further null streams that separate the
ingredients of the clustering are defined in Section~\ref{sec:setup-nulls}. The
support is the observed packet span $[\min \text{arr}_H, \max \text{arr}_H]$ rather than
the nominal window $[T_0, T_1]$, so the null is rate-matched to the arrivals it replaces
rather than to wall-clock; the two coincide on busy windows and the former is the stricter
counterfactual on sparse ones. The null is therefore a \emph{conditional
homogeneous-Poisson} process: a homogeneous Poisson process at rate $n_{\mathrm{pkt}}^{\#}
/ (\max \text{arr}_H - \min \text{arr}_H)$ conditioned on the observed count over the
observed support. Per (window, $N$, regime) we record $p_{50}, p_{95}, p_{99}, p_{99.9},
\max$ of end-to-end latency. Aggregation is corpus-median across all windows; the two
headline curves per scenario are $p_{99}(N)$ and $p_{50}(N)$ under Hawkes and Poisson
overlaid.

The calibration is the $h = 1.7\,\mu\mathrm{s}$ hop and the service-time grid of
Section~\ref{sec:sim-model}. The theory of Section~\ref{sec:theory} makes three
predictions that the sweep can falsify:
\begin{enumerate}
    \item Theorem~\ref{thm:reduction}: on every window, the tandem's tail excess is at most
the single-thread excess divided by the stage count,
\[
p_{99}(N) - (N-1)h - T \;\le\; \frac{p_{99}(1) - T}{N} ,
\]
so a fitted decay exponent $\gamma$ in $\Delta(N) \approx \Delta(1)\,N^{-\gamma}$ must be
at least $1$. A window that violates the bound falsifies the theorem.
    \item Corollary~\ref{cor:poisson}: on the Poisson stream, every quantile equals the
no-queue latency $T + (N-1)h$ in every cell whose per-stage utilisation $\rho/N$ is
below $1\%$, and a wait appears in every cell where it is not. The prediction is two-sided:
a Poisson $p_{99}$ at the no-queue latency above the boundary falsifies it as much as one
above the no-queue latency below it.
    \item Proposition~\ref{thm:burst}: the median grows by exactly one hop per added stage.
A median growing much faster than $h$ per stage falsifies the paper's model of the
simulator.
\end{enumerate}

\subsection{Null streams and the ingredient each removes}
\label{sec:setup-nulls}

The claim under test is causal: the clustering of arrivals produces the queueing tail.
A tail in the Lindley recursion can be built by four distinct properties of an arrival
stream, and the real stream has all four at once:
\begin{enumerate}
    \item the \emph{rate}, the number of packets per second;
    \item the \emph{gap distribution}: how often two consecutive packets land closer together than the service time $T$
(a \emph{tight} gap), regardless of the order in which those gaps occur;
    \item the \emph{ordering}: whether the tight gaps come in runs (a burst of ten packets one publisher period apart)
or lie singly among long gaps;
    \item \emph{slow rate variation}: busy and quiet stretches at the scale of seconds
within a window.
\end{enumerate}
A null stream is a synthetic stream that keeps some of these and destroys the rest.
Running the identical simulator on the real stream and on the null, and comparing the
tails, measures the share of the tail that the destroyed property was carrying. Each
null is a controlled subtraction of one ingredient, and the design is factorial in the
sense that the four properties are removed in different combinations by different nulls.
The reference quantity throughout is the single-stage $p_{99}$ excess, $p_{99} - T$, and
the statistic reported is the fraction of the real stream's excess that survives on the
null, corpus-median over the 3512 windows. All nulls are built per window from that
window's own real stream, with seeds derived from the session and window identifiers so
that every run is reproducible. The 3512 windows are nested in 276 sessions, so the
corpus medians reported here are point estimates and no interval is attached to them.

\paragraph{Packet-level nulls.} The uniform Poisson null $P$ of
Section~\ref{sec:setup-sweep} redraws every packet time uniformly over the window: it
keeps (1) and destroys (2), (3) and (4) together, so a tail that vanishes under it is
attributable to clustering in some form but not to any particular form. The \emph{gap
shuffle} $G$ takes the window's sequence of interarrival gaps, permutes it at random
(Fisher--Yates, the renewal null already used for the Fano diagnostic of
Section~\ref{sec:setup-hawkes}), and re-accumulates arrival times: the same packet count,
the same multiset of gaps and therefore the same tight-gap fraction and the same
shortest gap, $\spub$, with the ordering destroyed. $G$ keeps (1) and (2) and removes
(3) and (4). A tail that survives $G$ is set by how many tight gaps the stream contains; a
tail that $G$ removes needs the runs. The \emph{1\,s-binned Poisson null} $B$ cuts the window
into one-second bins, keeps each bin's real packet count, and redraws the arrivals
uniformly inside the bin: it keeps (1) and (4) and removes (2) and (3). A tail that
survives $B$ is a busy-second load effect, not a microsecond one. Section~\ref{sec:results-nulls}
reports all three.

\paragraph{Transaction-level nulls.} The packet nulls cannot distinguish two accounts
of the tight-gap surplus: that the exchange splits one matching-engine event across
back-to-back packets, or that the events themselves arrive clustered and the packets
relay them. Section~\ref{sec:results-transactions} therefore groups every message by its
\texttt{transactTime}, the identifier of one matching-engine event (a
\emph{transaction}), and builds five further streams in which each transaction is a
block --- its packets, their message counts and their internal spacing --- and exactly
one ingredient is removed. Service in these runs is charged per packet at its measured
cost, $S_i = T\,(1 + r(\sigma_i - 1))$ with $\sigma_i$ the message count and $r$ the ratio of the live receiver's per-message cost to its service floor (its time for a
one-message packet, Section~\ref{sec:crossval-floor}), so that a large packet adds work
rather than redistributing it; if the size of a transaction matters to the tail, these runs can show
it. Table~\ref{tab:null-design} lists all nine streams.

\begin{table}[tbp]
\centering
\small
\begin{tabular}{@{}l >{\raggedright\arraybackslash}p{5.2cm} >{\raggedright\arraybackslash}p{4.3cm} >{\raggedright\arraybackslash}p{4.6cm}@{}}
\toprule
stream & kept & destroyed & question answered \\
\midrule
H  & the real stream & --- & reference \\
\midrule
P  & packet count, support & gap distribution, ordering, slow rate variation & is the tail due to clustering in any form? \\
G  & count; the multiset of gaps, hence the tight-gap fraction and the shortest gap & ordering of gaps into runs; slow rate variation & is the tail set by how many tight gaps there are, or by their arrangement? \\
B  & count per one-second bin & microsecond gap structure and ordering & is the tail busy-second load? \\
\midrule
TG & every transaction intact, in its original order & the idle gaps between transactions, permuted & does the clustering of transactions into runs build the tail? \\
TP & every transaction intact & transaction start times, redrawn uniformly & does transaction timing matter at all? \\
TS & real start times & which transaction shape sits at which start (shapes permuted across starts) & do large transactions arriving at particular moments matter? \\
TM & real start times & the split of a transaction into packets (each transaction becomes one packet) & does the exchange's packetisation matter? \\
TW & real start times, the split into packets & the tight spacing inside a split transaction (packets spread by the median idle gap) & does back-to-back spacing inside a transaction matter? \\
\bottomrule
\end{tabular}
\caption{The nine arrival streams. $P$, $G$ and $B$ act on packet times with constant
service and are reported in Table~\ref{tab:nulls}; \textsf{TG}--\textsf{TW} act on transaction blocks
with span-dependent service and are reported in Table~\ref{tab:tx-arms}. Reading the
results together: $P \approx 0$ and $\textsf{TP} \approx 0$ say the tail is produced by the timing of the
transaction process; $\textsf{TS}$, $\textsf{TM}$, $\textsf{TW} \approx 1$ say sizes, packetisation and
intra-transaction spacing are not the mechanism; $G$ and \textsf{TG}, which remove the same
thing at two levels of description, give the same profile across $T$ and split the
clustering into the two channels named in Section~\ref{sec:results-transactions}.}
\label{tab:null-design}
\end{table}

The logic of the two families is the same and they are read together. If $P$ and \textsf{TP}
remove the tail, the tail is produced by the timing of the transaction process. If \textsf{TS}, \textsf{TM} and \textsf{TW} leave it, then the size of transactions,
their division into packets and the spacing of those packets are not the mechanism: the
exchange is relaying events, not creating bursts. $G$ and \textsf{TG} remove the same ingredient --- the arrangement of events
into runs --- at the packet and at the transaction level respectively; agreement between
them identifies the packet-level ordering channel with the clustering of transactions.
What survives $G$ is the surplus of tight gaps over a rate-matched Poisson
stream, which Table~\ref{tab:tx-gaps} measures directly. That surplus shows the gap
distribution is far from exponential; it does not show what produces it, and $G$ itself is
a stream with the surplus and no self-excitation. What $G$ removes, the runs, is the
self-exciting component.

\section{Empirical results}
\label{sec:results}

\subsection{Single-thread and tandem tails on the real stream and on the Poisson null}
\label{sec:results-tail}

This section presents the main result. It runs the simulated pipeline over the whole
corpus and reports, for each task length and each stage count, how long a typical message
takes and how long the slowest one in a hundred takes. Each case is run twice: once on the
real packet arrivals, which come in bursts, and once on a control in which the same
number of packets is scattered uniformly at random over the window. The comparison
isolates the effect of bursts. The findings
are these. On randomly scattered arrivals there is almost no slow tail at the task lengths a
trading hot path uses, so splitting a task across threads gains nothing and costs a hop per
stage. On the real, bursty arrivals a slow tail appears once a task is longer than about
seven microseconds, and it grows quickly with task length, until the slowest messages
take tens of times longer than the task itself. There, splitting the task into two or
more stages cuts the tail sharply for a small, fixed cost on the typical message, and the
longer the task, the more stages pay off. In short, it is the bursts, not the average
load, that create the tail, and that is why splitting helps on real data and not on the
random control.

Table~\ref{tab:main-corpus} gives the corpus-median latency for every service time $T$
and stage count $N$ at $h = 1.7\,\mu\mathrm{s}$, and Figure~\ref{fig:curves} plots it. The
Poisson columns are the uniform null $P$ of Section~\ref{sec:setup-sweep}, with the same
packet count per window.

\begin{table}[tbp]
\centering
\small
\begin{tabular}{@{}rr rrr rrr@{}}
\toprule
 & & \multicolumn{3}{c}{real arrivals} & \multicolumn{3}{c}{Poisson null} \\
\cmidrule(lr){3-5} \cmidrule(lr){6-8}
$T$ & $N$ & $p_{50}$ & $p_{99}$ & $p_{99.9}$ & $p_{50}$ & $p_{99}$ & $p_{99.9}$ \\
\midrule
 2 & 1 &   2.00 &   2.00 &   2.00 &   2.00 &   2.00 &   2.00 \\
 2 & 2 &   3.70 &   3.70 &   3.70 &   3.70 &   3.70 &   3.70 \\
 2 & 4 &   7.10 &   7.10 &   7.10 &   7.10 &   7.10 &   7.10 \\
 2 & 8 &  13.90 &  13.90 &  13.90 &  13.90 &  13.90 &  13.90 \\
\midrule
 4 & 1 &   4.00 &   4.00 &   4.00 &   4.00 &   4.00 &   5.78 \\
 4 & 2 &   5.70 &   5.70 &   5.70 &   5.70 &   5.70 &   5.70 \\
 4 & 4 &   9.10 &   9.10 &   9.10 &   9.10 &   9.10 &   9.10 \\
 4 & 8 &  15.90 &  15.90 &  15.90 &  15.90 &  15.90 &  15.90 \\
\midrule
 8 & 1 &   8.00 &   8.61 &  10.14 &   8.00 &   8.00 &  13.79 \\
 8 & 2 &   9.70 &   9.70 &   9.70 &   9.70 &   9.70 &  11.48 \\
 8 & 4 &  13.10 &  13.10 &  13.10 &  13.10 &  13.10 &  13.10 \\
 8 & 8 &  19.90 &  19.90 &  19.90 &  19.90 &  19.90 &  19.90 \\
\midrule
16 & 1 &  16.00 &  34.59 &  53.97 &  16.00 &  16.00 &  29.84 \\
16 & 2 &  17.70 &  18.31 &  19.84 &  17.70 &  17.70 &  23.49 \\
16 & 4 &  21.10 &  21.10 &  21.10 &  21.10 &  21.10 &  22.88 \\
16 & 8 &  27.90 &  27.90 &  27.90 &  27.90 &  27.90 &  27.90 \\
\midrule
32 & 1 &  32.00 & 146.14 & 259.11 &  32.00 &  42.02 &  62.02 \\
32 & 2 &  33.70 &  52.29 &  71.67 &  33.70 &  33.70 &  47.54 \\
32 & 4 &  37.10 &  37.71 &  39.24 &  37.10 &  37.10 &  42.89 \\
32 & 8 &  43.90 &  43.90 &  43.90 &  43.90 &  43.90 &  45.68 \\
\midrule
64 & 1 &  98.88 & 820.70 & 1548.68 &  64.00 & 106.73 & 126.72 \\
64 & 2 &  65.70 & 179.84 & 292.81 &  65.70 &  75.72 &  95.72 \\
64 & 4 &  69.10 &  87.69 & 107.07 &  69.10 &  69.10 &  82.94 \\
64 & 8 &  75.90 &  76.51 &  78.04 &  75.90 &  75.90 &  81.69 \\
\midrule
128 & 1 & 364.45 & 5316.85 & 12487.72 & 128.00 & 237.65 & 289.07 \\
128 & 2 & 164.58 & 886.40 & 1614.38 & 129.70 & 172.43 & 192.42 \\
128 & 4 & 133.10 & 247.24 & 360.21 & 133.10 & 143.12 & 163.12 \\
128 & 8 & 139.90 & 158.49 & 177.87 & 139.90 & 139.90 & 153.74 \\
\bottomrule
\end{tabular}
\caption{Corpus-median latency quantiles in \si{\micro\second} for each service time $T$ and
stage count $N$ at $h = 1.7\,\mu\mathrm{s}$, on the real NQ arrivals and on a Poisson
stream with the same packet count per window. What the table shows: the median is
$T + (N-1)h$ in almost every cell, so each stage costs one hop. On the real arrivals the
single-thread $p_{99}$ equals $T$ up to $4\,\mu\mathrm{s}$ and then grows much faster than
$T$, to about forty times $T$ at $128\,\mu\mathrm{s}$; on the Poisson stream it stays at
$T$ up to $16\,\mu\mathrm{s}$. One cut removes most of the real stream's tail for one hop,
and further cuts keep helping until each stage is shorter than the publisher period.
What follows: the tail comes from how the real packets are bunched, not from load. A
server queues only when consecutive gaps are shorter than its service time, and the real
stream has many times the Poisson share of such gaps (Table~\ref{tab:tx-gaps}) and
arranges them in runs.}
\label{tab:main-corpus}
\end{table}

\begin{figure}[H]
\centering
\includegraphics[width=\textwidth]{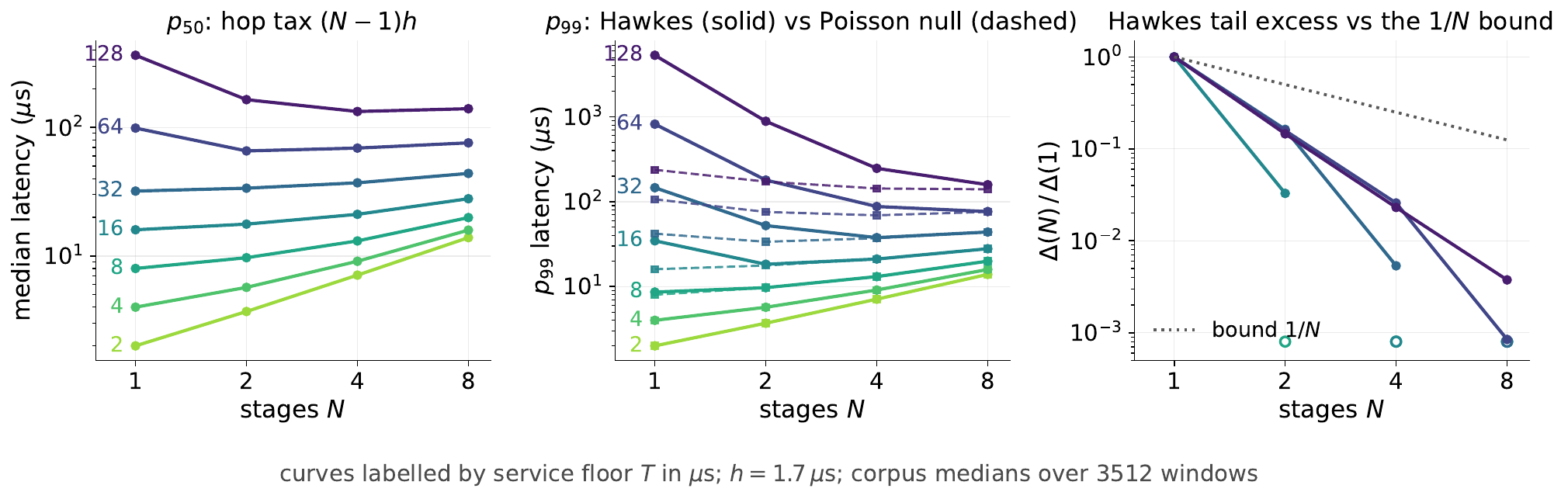}
\caption{Corpus-median latency versus stage count $N$ for each service time $T$ (one
shade per $T$, each curve labelled at its $N = 1$ end). Left: median. Centre: $p_{99}$,
Hawkes solid and Poisson null dashed. Right: the Hawkes tail excess ratio
$\Delta(N)/\Delta(1)$ against the Theorem~\ref{thm:reduction} bound $1/N$ (dotted); open
markers are cells where the tandem drives $\Delta(N)$ to exactly zero, drawn at the bottom of
the panel because zero cannot be placed on a log axis. The median rises by exactly $h$ per hop
wherever the queue is empty at the median, which is every cell except the three
longest-service Hawkes cells. The Hawkes $p_{99}$ falls with $N$ only at $T \ge
16\,\mu\mathrm{s}$; at $T \le 8\,\mu\mathrm{s}$ there is no tail to compress and it rises
with the hop count, tracking the median. The Poisson $p_{99}$ likewise tracks the median
up to $T = 16\,\mu\mathrm{s}$ and separates from it at $T \ge 32\,\mu\mathrm{s}$, where
utilisation queueing begins.}
\label{fig:curves}
\end{figure}

On the Poisson stream almost every message finds the queue empty, so every quantile
equals the no-queue latency up to $T = 16\,\mu\mathrm{s}$. A wait appears first at
$32\,\mu\mathrm{s}$, where utilisation passes the $1\%$ at which
Corollary~\ref{cor:poisson} predicts it at $p_{99}$; from there on the Poisson stream
behaves as an ordinary queue, and splitting helps it too. The $p_{99.9}$ column shows the
same pattern at utilisations ten times lower (Section~\ref{sec:evaluation}). Under Poisson
arrivals, then, splitting is not useless in general; it is useless over the service times
an HFT hot path occupies.

On the real stream a tail appears only above a threshold. At $T = 2$ and
$4\,\mu\mathrm{s}$ there is none, and the real stream cannot be told apart from the
Poisson one. Between $4$ and $8\,\mu\mathrm{s}$ a tail appears, and it then grows much
faster than the service time. Splitting into two stages removes most of it at every
service time from $16\,\mu\mathrm{s}$ up, for one hop on the median, and wins on every
window there. The best stage count rises with $T$, from one at $8\,\mu\mathrm{s}$ and below
to eight at $64\,\mu\mathrm{s}$ and above (Table~\ref{tab:main-corpus}).
Section~\ref{sec:results-threshold} locates the onset more precisely and
Section~\ref{sec:results-conditioning} explains what sets it.

\findings{
    \item \emph{No tail below the publisher period.} At $T \le 4\,\mu\mathrm{s}$ the $p_{99}$ equals $T$ on real and Poisson arrivals alike; only the $p_{99.9}$ of the Poisson stream shows the small utilisation wait Corollary~\ref{cor:poisson} predicts.
    \item \emph{Above it, the real stream's tail grows much faster than $T$}, to about forty times $T$ at $128\,\mu\mathrm{s}$.
    \item \emph{The Poisson stream at the same rate has no tail up to $16\,\mu\mathrm{s}$}, and above that only a much smaller utilisation tail.
    \item \emph{One cut removes most of the tail for one hop} and wins on every window from $16\,\mu\mathrm{s}$.
    \item \emph{The best stage count rises with $T$}, from one at $8\,\mu\mathrm{s}$ to eight at $64\,\mu\mathrm{s}$ and above.
}

\subsection{Gap-shuffle and 1\,s-binned nulls}
\label{sec:results-nulls}

The Poisson null of Table~\ref{tab:main-corpus} redraws every arrival uniformly over the
whole window. That removes three things at once: the ordering of gaps into runs, the slower rate
variation across the window, and the shape of the gap distribution, including its
shortest gap, the publisher period. The tail difference in Table~\ref{tab:main-corpus}
could come from any of the three. The mechanism of Section~\ref{sec:results-conditioning}
runs on the third: a queue forms when
consecutive gaps are shorter than the service time, and that is a property of the gap
distribution, which a reordering of the same gaps leaves untouched. Two further nulls
separate the three. The \emph{gap shuffle} permutes each window's interarrival gaps at
random (the renewal null of Section~\ref{sec:setup-hawkes}): same count, same gap
distribution, same publisher period, no ordering. The \emph{1\,s-binned Poisson} keeps
each one-second bin's own packet count and redraws the arrivals uniformly inside the bin:
same second-scale rate variation, no microsecond clustering.

\begin{table}[tbp]
\centering
\small
\begin{tabular}{@{}rr rrrr@{}}
\toprule
$T$ & $N$ & real & uniform & gap shuffle & 1\,s-binned \\
\midrule
 2 & 1 &    2.00 &    2.00 &    2.00 &    2.00 \\
 2 & 2 &    3.70 &    3.70 &    3.70 &    3.70 \\
 2 & 4 &    7.10 &    7.10 &    7.10 &    7.10 \\
 2 & 8 &   13.90 &   13.90 &   13.90 &   13.90 \\
\midrule
 4 & 1 &    4.00 &    4.00 &    4.00 &    4.00 \\
 4 & 2 &    5.70 &    5.70 &    5.70 &    5.70 \\
 4 & 4 &    9.10 &    9.10 &    9.10 &    9.10 \\
 4 & 8 &   15.90 &   15.90 &   15.90 &   15.90 \\
\midrule
 8 & 1 &    8.61 &    8.00 &    8.54 &    8.00 \\
 8 & 2 &    9.70 &    9.70 &    9.70 &    9.70 \\
 8 & 4 &   13.10 &   13.10 &   13.10 &   13.10 \\
 8 & 8 &   19.90 &   19.90 &   19.90 &   19.90 \\
\midrule
16 & 1 &   34.59 &   16.00 &   30.75 &   16.00 \\
16 & 2 &   18.31 &   17.70 &   18.24 &   17.70 \\
16 & 4 &   21.10 &   21.10 &   21.10 &   21.10 \\
16 & 8 &   27.90 &   27.90 &   27.90 &   27.90 \\
\midrule
32 & 1 &  146.14 &   42.02 &   99.61 &   47.36 \\
32 & 2 &   52.29 &   33.70 &   48.45 &   33.70 \\
32 & 4 &   37.71 &   37.10 &   37.64 &   37.10 \\
32 & 8 &   43.90 &   43.90 &   43.90 &   43.90 \\
\midrule
64 & 1 &  820.70 &  106.73 &  340.62 &  112.62 \\
64 & 2 &  179.84 &   75.72 &  133.31 &   81.06 \\
64 & 4 &   87.69 &   69.10 &   83.85 &   69.10 \\
64 & 8 &   76.51 &   75.90 &   76.44 &   75.90 \\
\midrule
128 & 1 & 5316.85 &  237.65 & 1306.63 &  246.03 \\
128 & 2 &  886.40 &  172.43 &  406.32 &  178.32 \\
128 & 4 &  247.24 &  143.12 &  200.71 &  148.46 \\
128 & 8 &  158.49 &  139.90 &  154.65 &  139.90 \\
\bottomrule
\end{tabular}
\caption{Corpus-median $p_{99}$ in \si{\micro\second} under the real arrivals and the three
packet-level null streams of Section~\ref{sec:setup-nulls}, for every service time $T$
and stage count $N$. Uniform is the Poisson null; the gap shuffle keeps every gap and
destroys their order; the 1\,s-binned stream keeps each second's packet count and destroys
the microsecond structure. What the table shows: at $T \le 4\,\mu\mathrm{s}$ every stream's $p_{99}$ sits at the no-queue latency, because no gap between packets is that short. For a single thread above
that, the uniform and binned streams lose almost all of the real stream's tail, while the
gap shuffle keeps most of it at short service times and less as $T$ grows
(Figure~\ref{fig:bars-nulls}). Once a cut takes each stage below the publisher period,
all four streams meet at the no-queue latency. What follows: the tail is not busy-second
load. At short service times it is set by how many gaps are shorter than $T$, since a
single short gap costs the next packet one wait whether or not its neighbours are short.
At long service times it needs runs of short gaps, and runs are what the shuffle
destroys. Session-bootstrap $95\%$ intervals on the shares are within about two points of
the reported values.}
\label{tab:nulls}
\end{table}

\begin{figure}[H]
\centering
\includegraphics[width=0.8\textwidth]{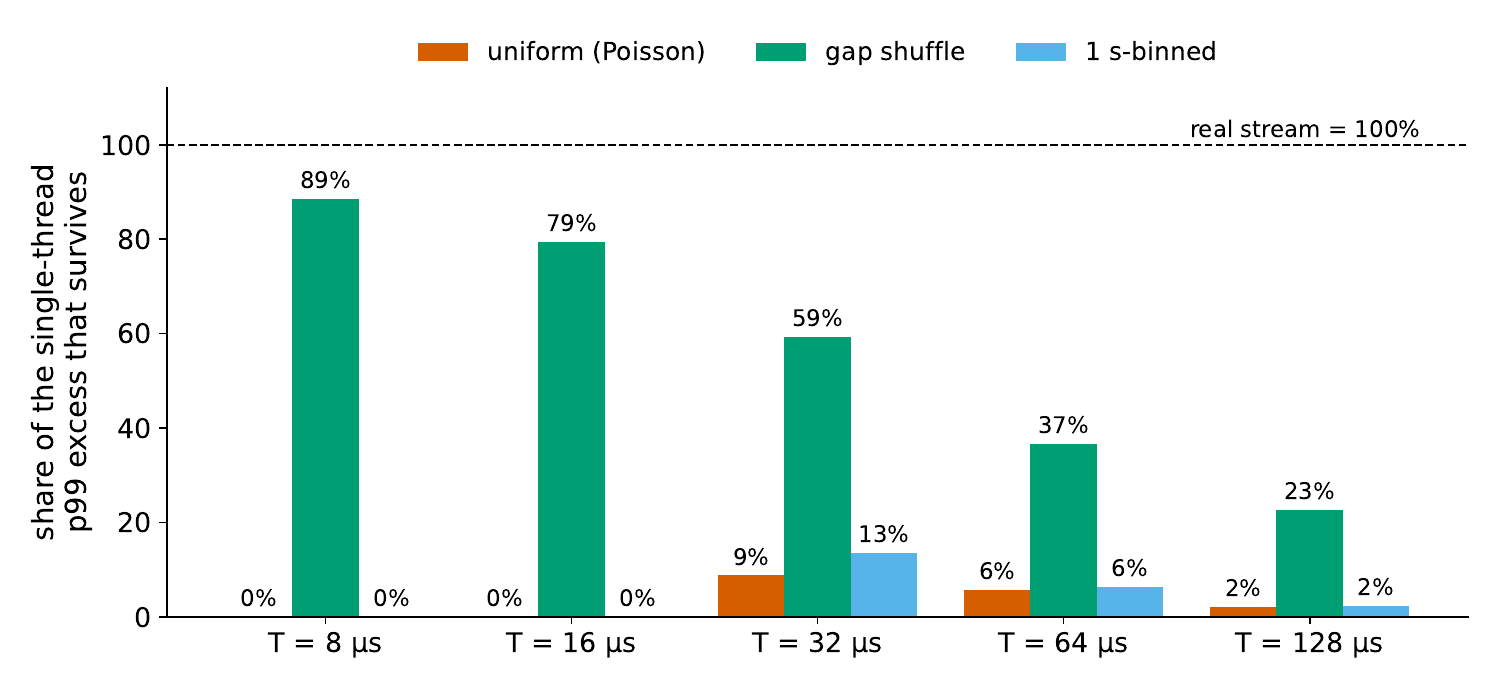}
\caption{Table~\ref{tab:nulls} as bars: the share of the single-thread $p_{99}$ excess
($p_{99} - T$) that survives each null stream, the real stream being $100\%$. The gap
shuffle (green), which keeps every gap and destroys their order, keeps most of the tail at
$8$--$16\,\mu\mathrm{s}$ and less as $T$ grows: at short service times the tight gaps carry
the tail, at long ones their arrangement into runs. The uniform and 1\,s-binned streams
keep almost none at any $T$.}
\label{fig:bars-nulls}
\end{figure}

Table~\ref{tab:nulls} and Figure~\ref{fig:bars-nulls} give the result. The 1\,s-binned
stream keeps none of the tail at $8$ and $16\,\mu\mathrm{s}$ and little above, so the tail
is not busy-second load. The gap shuffle keeps about $80\%$ of it at $16\,\mu\mathrm{s}$
and less as the service time grows, about a quarter at $128\,\mu\mathrm{s}$. At the
service times of an HFT hot path, then, the tail is set mainly by how many gaps are shorter
than the service time, and the arrangement of those gaps into runs matters only from
about $64\,\mu\mathrm{s}$. Section~\ref{sec:results-transactions} traces both to the timing
of matching-engine transactions: the short gaps are consecutive transactions sent back to
back at the publisher period (Table~\ref{tab:two-clocks}), and the runs are the
self-exciting component. Which process places two transactions within microseconds of
each other at the engine is not identified by any stream run here
(Section~\ref{sec:concl-cannot}).

\findings{
    \item \emph{The tail is not busy-second load.} Keeping every second's packet count but placing packets at random within it leaves almost no tail.
    \item \emph{At short service times the number of short gaps sets the tail.} Keeping every gap but shuffling their order keeps about $80\%$ of it at $16\,\mu\mathrm{s}$.
    \item \emph{At long service times the runs of short gaps set it.} The same shuffle keeps only about a quarter at $128\,\mu\mathrm{s}$.
}

\subsection{Grouping packets into matching-engine transactions}
\label{sec:results-transactions}

The nulls of Section~\ref{sec:results-nulls} act on packet times. A packet is the
exchange's unit, not the market's. CME MDP3 publishes each matching-engine event --- a
\emph{transaction}, identified by the common \texttt{transactTime} carried by every message
it generates --- as one or more UDP datagrams, and a datagram may carry the end of one
transaction and the start of the next. Two accounts of the tight-gap surplus are therefore
open. Either the exchange splits transactions across back-to-back packets, in which case
the clustering is an artefact of packetisation and a receiver that reassembled transactions
would see none of it; or the transactions themselves arrive clustered and the packets relay
them. The corpus decides between the two. Every message in the corpus was grouped by
\texttt{transactTime} and every packet labelled with the transactions it carries, about
$3 \times 10^9$ of each. About $3\%$ of packets carry parts of two transactions: the
exchange's publisher packs transactions that reach it while it is busy into one packet
(Section~\ref{sec:exchange-drain}).

\paragraph{Consecutive transactions and the gap between them.} A \emph{transaction} is one
matching-engine event: all messages that carry the same \texttt{transactTime}. One incoming
order that trades against three resting orders is one transaction, however many messages
it produces. Within a session the NQ front-month transactions are put in the order the
engine processed them, $A, B, C, D, \ldots$; two transactions are \emph{consecutive} when
one directly follows the other in that order, with no other NQ front-month transaction
between them ($A$ and $B$, $B$ and $C$, and so on). The \emph{gap between consecutive
transactions}, their inter-arrival time, is the time from one to the next. It is measured
on one of two clocks, always named: the matching engine's (the difference of the two
\texttt{transactTime}s) or the publisher's (the difference of the \texttt{sendingTime}s of
the packets that carry them, zero if one packet carries both). A session with $m$
transactions has $m - 1$ such gaps. Every statement in the paper about how close
transactions are is a statement about these gaps; a \emph{tight} gap is one below
$16\,\mu\mathrm{s}$.

\begin{table}[H]
\centering
\small
\begin{tabular}{@{}l rrrrrrrr@{}}
\toprule
 & 1 & 2 & 3 & 4 & 5 & 6--10 & 11--20 & 21+ \\
\midrule
messages per packet      & 95.897 & 2.439 & 0.869 & 0.401 & 0.188 & 0.171 & 0.027 & 0.009 \\
packets per transaction  & 98.295 & 1.689 & 0.011 & 0.002 & 0.001 & 0.001 & 0.000 & 0.000 \\
messages per transaction & 95.362 & 2.729 & 0.822 & 0.558 & 0.195 & 0.280 & 0.038 & 0.016 \\
\bottomrule
\end{tabular}
\caption{Size distributions, percent of the pooled corpus (about $3 \times 10^9$ packets and as
many transactions). What the table shows: almost every packet carries one message and
almost every transaction fits in one packet; a transaction split across three or more
packets is very rare. The message count per transaction is the message count per packet
with the split transactions folded back in.}
\label{tab:tx-counts}
\end{table}

\begin{figure}[H]
\centering
\includegraphics[width=\textwidth]{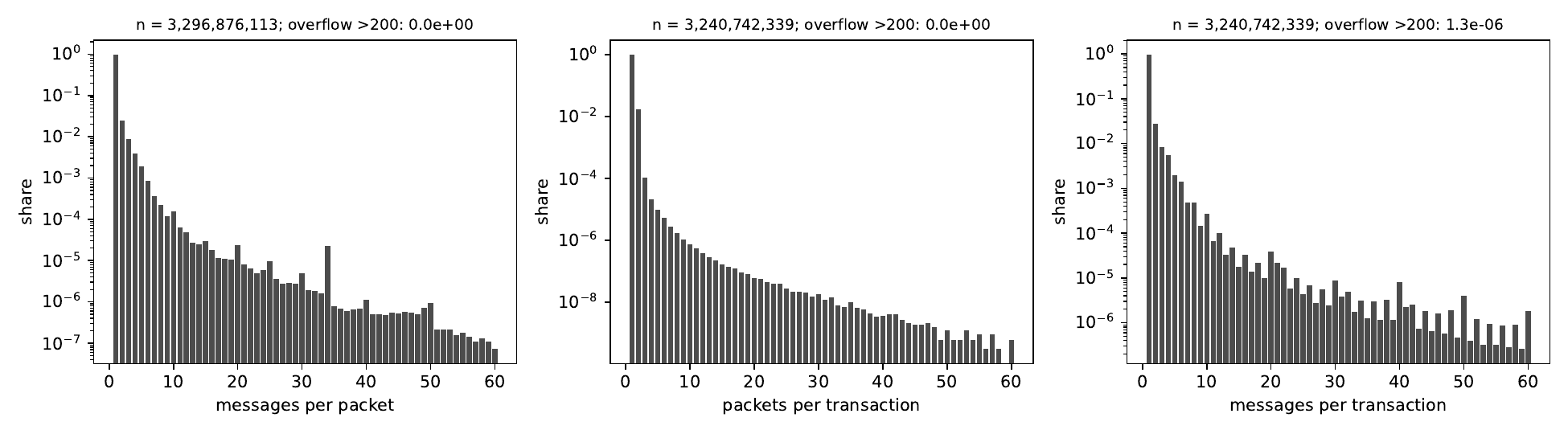}
\caption{The three size distributions of Table~\ref{tab:tx-counts} on a logarithmic
share axis, counts $1$--$60$. Messages per packet and per transaction have the same
geometric decay; packets per transaction has almost none beyond two.}
\label{fig:tx-counts}
\end{figure}

\paragraph{Packetisation and the tight gaps.} Table~\ref{tab:tx-counts} gives the three
size distributions, and Figure~\ref{fig:tx-counts} plots them. A split transaction is
rare, and its packets are not sent back to back: they arrive about $15\,\mu\mathrm{s}$
apart, above the publisher period, and almost never closer (Figure~\ref{fig:tx-gaps},
left, and Figure~\ref{fig:tx-starts}, right). The tight gaps lie between different
transactions (Table~\ref{tab:tx-gaps}).

\begin{table}[H]
\centering
\small
\begin{tabular}{@{}r rrr@{}}
\toprule
gap shorter than & share of all packet gaps & of which inside one transaction & Poisson at the same rate \\
\midrule
$7.9\,\mu\mathrm{s}$ (period) & $2.15\%$ & $0.1\%$ & $0.56\%$ \\
$17.8\,\mu\mathrm{s}$  & $20.6\%$ & $5.0\%$ & $1.24\%$ \\
$35.5\,\mu\mathrm{s}$  & $34.1\%$ & $4.6\%$ & $2.46\%$ \\
\bottomrule
\end{tabular}
\caption{Short packet gaps, pooled over the corpus. The rows are the bin edges of the
logarithmic grid of Figure~\ref{fig:tx-gaps} nearest to the publisher period and to the $16$ and $32\,\mu\mathrm{s}$ service times. What the table shows: the
real stream has several times the Poisson share of gaps below each of these service
times, about one gap in five below $16\,\mu\mathrm{s}$ against about one in a hundred, and
almost none of these short gaps lie inside a single transaction. What follows: the short
gaps that make a server queue are gaps between different matching-engine events, not the
exchange's division of one event into several packets, so a receiver that reassembled
transactions before decoding would remove almost none of them. Section~\ref{sec:exchange}
shows where they come from: the engine processes many transactions within a microsecond
of each other, and the publisher sends them one publisher period apart
(Table~\ref{tab:two-clocks}). Figure~\ref{fig:bars-tx-gaps} shows the same numbers as bars.}
\label{tab:tx-gaps}
\end{table}

\begin{figure}[H]
\centering
\includegraphics[width=0.72\textwidth]{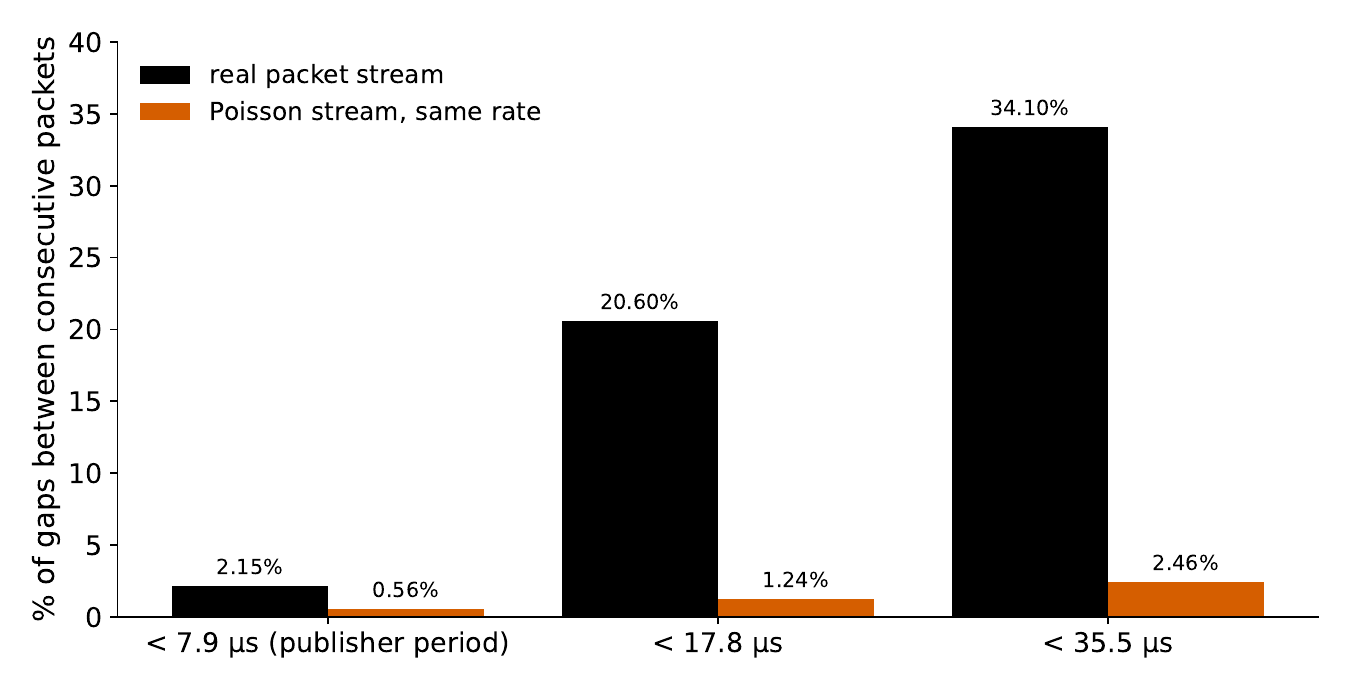}
\caption{Table~\ref{tab:tx-gaps} as bars: the share of gaps between consecutive packets
shorter than the publisher period, than about $16\,\mu\mathrm{s}$ and than about
$32\,\mu\mathrm{s}$, on the real stream (black) and on a Poisson stream with the same
packet count (vermillion). Below $16\,\mu\mathrm{s}$ the real stream has about one gap in
five against about one in a hundred. A server with a service time in this range queues on
each of these gaps, which is why the real stream has a tail and the Poisson stream does
not.}
\label{fig:bars-tx-gaps}
\end{figure}

\begin{figure}[H]
\centering
\includegraphics[width=\textwidth]{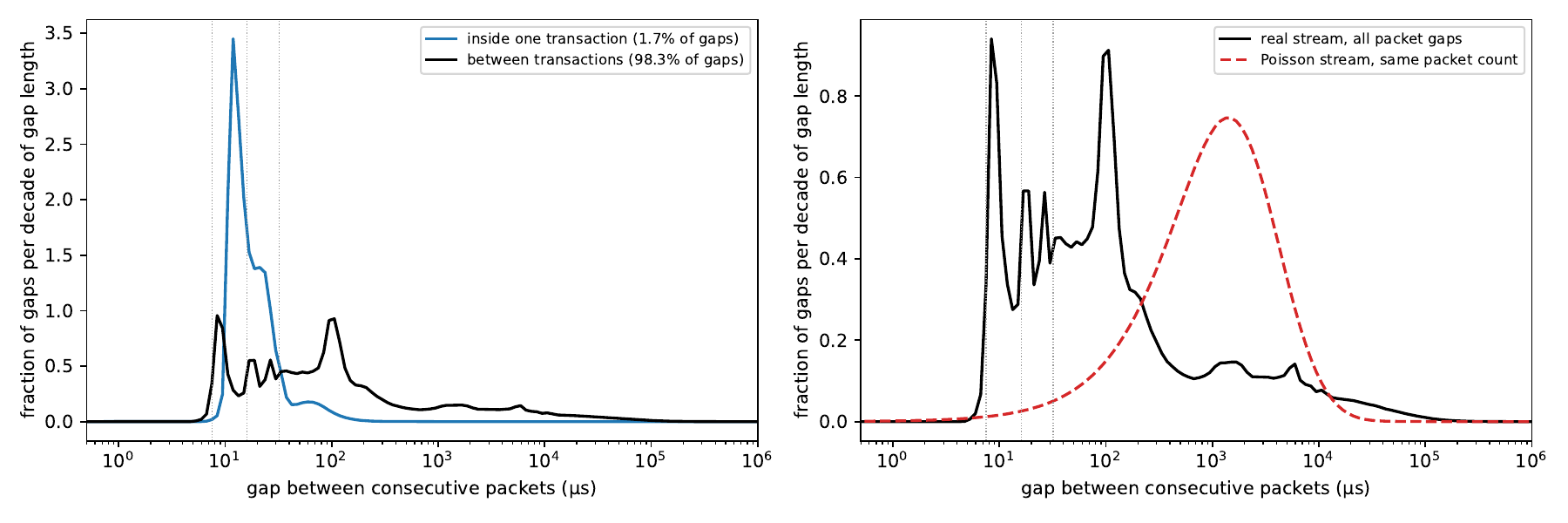}
\caption{Packet gaps, pooled over the corpus. Left: gaps between the packets of one split
transaction (blue) against gaps between packets of different transactions (black). The
pieces of a split transaction arrive about $15\,\mu\mathrm{s}$ apart and almost never
closer than the publisher period; the region from the publisher period to about
$12\,\mu\mathrm{s}$ is filled by the black curve alone. So the gaps that make a receiver
queue are between packets of different events. Right: all packet gaps of the real stream
(black) against a Poisson stream with the same packet count (red, dashed), the null of
Table~\ref{tab:main-corpus} drawn in gap space. The Poisson density is one hump near
$1\,\mathrm{ms}$; the real density puts about a third of its gaps below
$32\,\mu\mathrm{s}$, with an edge at the publisher period and a mode near
$100\,\mu\mathrm{s}$ that the null lacks. That region of short gaps is why the real stream
has a tail at HFT service times and the Poisson stream does not. Axes: logarithmic gap
axis; vertical axis the fraction of gaps per decade of gap length, so the area under a
curve over a range is the fraction of its gaps in that range. Dotted lines at $7.5$, $16$
and $32\,\mu\mathrm{s}$.}
\label{fig:tx-gaps}
\end{figure}

\paragraph{Self-excitation of the transaction stream.} Treat each transaction as one
event at the time of its first packet. The transaction stream then carries the same
clustering as the packet stream (Table~\ref{tab:tx-fano}): the variability of transaction
counts grows with the counting window just as that of packet counts does, and shuffling
the idle gaps between transactions removes most of the growth. The Hawkes fit gives the
same branching ratio, about $0.8$, on both. The gap density of transaction starts
(Figure~\ref{fig:tx-starts}, left) is the same as that of packets and unlike the Poisson
curve at every scale. The exchange publishes each matching event as it occurs, and the
packet stream inherits the clustering of the events. Why the events cluster is a question
about the market, outside the scope of this paper (Section~\ref{sec:concl-cannot}).

\begin{table}[H]
\centering
\small
\begin{tabular}{@{}r rrr@{}}
\toprule
count window $\tau$ & transactions & transactions, idle gaps shuffled & packets \\
\midrule
$1\,\mathrm{ms}$   &   7.6 &  5.5 &   7.8 \\
$10\,\mathrm{ms}$  &  21.6 & 10.8 &  22.2 \\
$100\,\mathrm{ms}$ &  50.5 & 15.8 &  51.8 \\
$1\,\mathrm{s}$    & 105.4 & 17.1 & 108.4 \\
\midrule
branching ratio $n$ & 0.795 & --- & 0.798 \\
\bottomrule
\end{tabular}
\caption{Fano factor $\mathrm{Var}[C_\tau]/\mathrm{E}[C_\tau]$ of event counts, corpus
median over 3512 windows, for transaction starts, for the same starts with the idle gaps
between transactions randomly permuted (transactions kept intact), and for packets; Poisson
gives $1$ at every $\tau$. Last row: exponential-kernel Hawkes branching ratio, corpus
median, all fits converged. The within-window Fano values agree with those of
Section~\ref{sec:setup-hawkes} at the same bin width. Figure~\ref{fig:bars-tx-fano} shows the same numbers as bars.}
\label{tab:tx-fano}
\end{table}

\begin{figure}[H]
\centering
\includegraphics[width=0.8\textwidth]{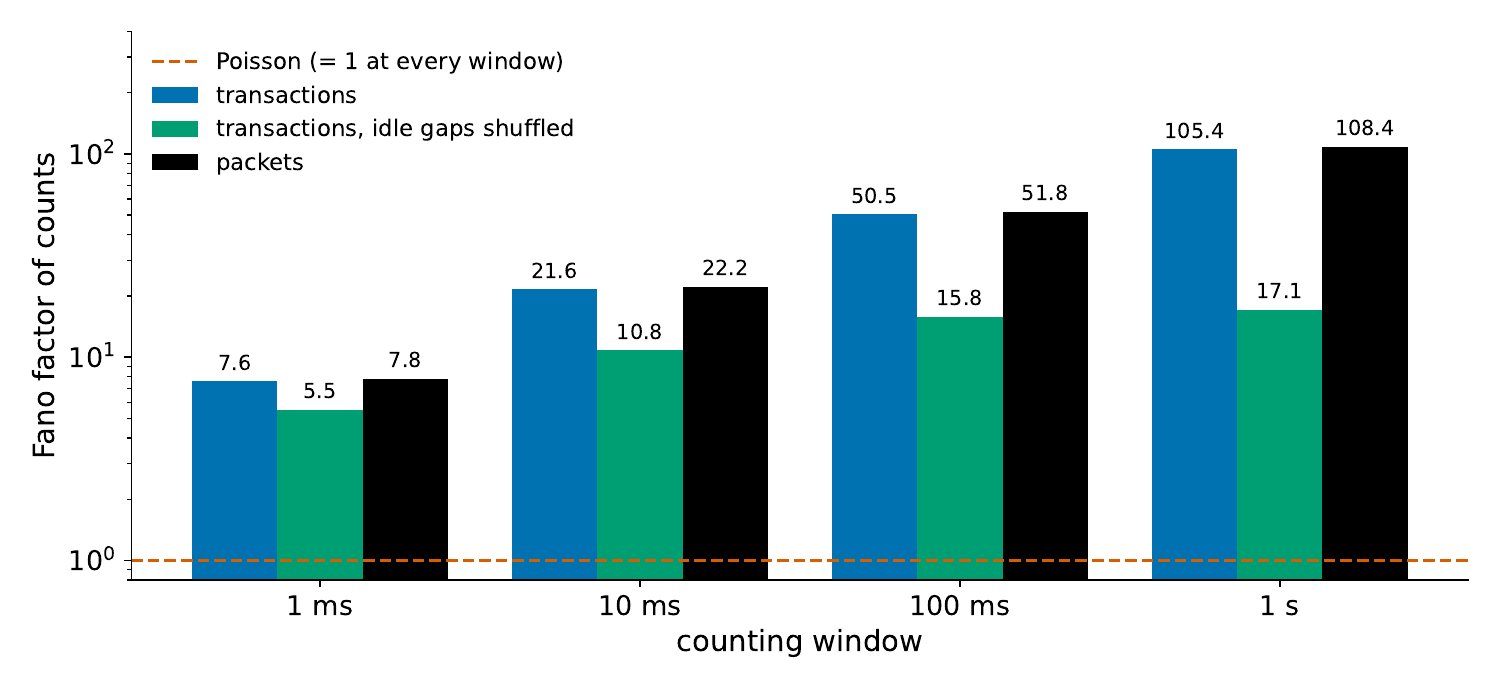}
\caption{Table~\ref{tab:tx-fano} as bars, log scale: how much more variable the counts are than
Poisson (Fano factor, $1$ for Poisson, dashed) in counting windows from $1\,\mathrm{ms}$ to
$1\,\mathrm{s}$. Transactions (blue) and packets (black) are almost identical and grow from
about $8$ to over $100$: the packet stream inherits its clustering from the transactions.
Shuffling the idle gaps between transactions (green) removes most of the growth: what is
left is the over-dispersion of the gap lengths themselves.}
\label{fig:bars-tx-fano}
\end{figure}

\begin{figure}[H]
\centering
\includegraphics[width=\textwidth]{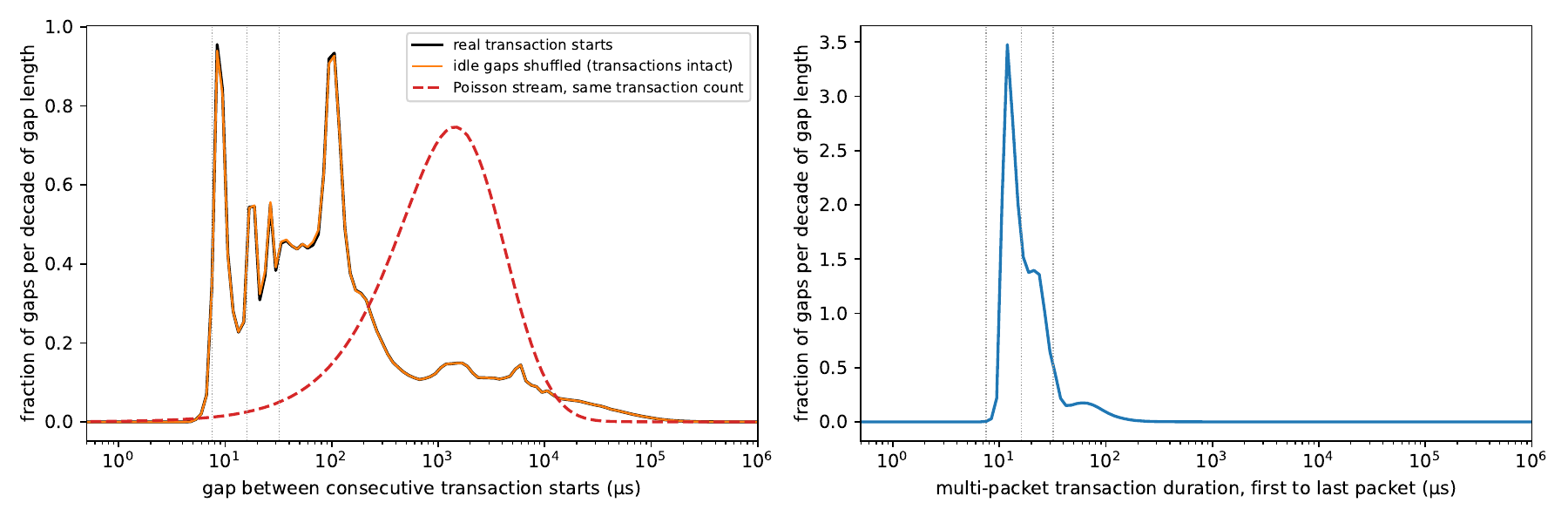}
\caption{Gaps between consecutive \emph{transaction starts}, to be read against
Figure~\ref{fig:tx-gaps}, whose gaps are between consecutive \emph{packets}. Left: each
transaction collapsed to one event at the time of its first packet, for the real stream
(black), for the same events with the idle gaps between transactions shuffled (orange),
and for a Poisson stream with the same transaction count (red, dashed). The black curve is
the black curve of Figure~\ref{fig:tx-gaps}: the same edge at the publisher period and
the same mode near $100\,\mu\mathrm{s}$, because almost every transaction is one packet.
The orange curve lies on the black one by construction, since shuffling keeps the gaps.
The Poisson curve has a different shape. What follows: the clustering seen at the packet
level is present unchanged at the level of matching-engine events, which is what
Table~\ref{tab:tx-fano} shows numerically and Table~\ref{tab:tx-arms} shows in the
simulator. Right: time from the first to the last packet of a split transaction, about
$15\,\mu\mathrm{s}$ at the median. Axes as in Figure~\ref{fig:tx-gaps}.}
\label{fig:tx-starts}
\end{figure}

\paragraph{The tail under the transaction-level nulls.} To separate the parts of the
transaction process that build the queue, each window is re-simulated at $N = 1$ with
service charged per packet at its measured absolute cost, $S_i = T\,(1 + r(\sigma_i - 1))$
with $\sigma_i$ the packet's message count and $r$ the ratio of the live receiver's
per-message cost to its service floor (Section~\ref{sec:crossval-floor}), so that a large
packet adds work rather than redistributing it. Five counterfactual streams
are built from the real one (H), each removing one ingredient. TG keeps every transaction
intact --- its packets, their sizes and their spacing --- and permutes the idle gaps
between transactions, removing the clustering of transactions. TP keeps
the transactions intact and redraws their start times uniformly over the window. TS keeps
the real start times and permutes the transactions' shapes among them, decoupling size
from timing. TM merges each transaction into a single packet at its start. TW keeps the
starts and spreads a split transaction's packets by the median idle gap instead of their
real spacing. Table~\ref{tab:tx-arms} reports the corpus-median $p_{99}$ of each stream
and the share of the real stream's excess over $T$ that it retains.

\begin{table}[H]
\centering
\small
\begin{tabular}{@{}r r rrrrr@{}}
\toprule
$T$ & H & TG & TP & TS & TM & TW \\
\midrule
$16$  &   34.8 &   31.0 (80) &   20.7 (25) &   37.4 (114) &   34.9 (101) &   35.6 (104) \\
$32$  &  146.8 &  100.5 (60) &   52.9 (18) &  155.9 (108) &  146.3 (100) &  148.1 (101) \\
$64$  &  832   &  345 (37)   &  118 (7)    &  839 (101)   &  795 (95)    &  833 (100) \\
$128$ & 5460   & 1319 (22)   &  251 (2)    & 5218 (95)    & 5035 (92)    & 5459 (100) \\
\bottomrule
\end{tabular}
\caption{Which part of the transaction process builds the queue. Each row is one service time
$T$; each column is the single-thread $p_{99}$ in \si{\micro\second} (corpus median,
service charged per packet at NQ's cost per message) when the simulator is driven by the
real stream H or by one of five rebuilt streams, each identical to the real one except
for one removed ingredient. The number in parentheses is the share, in percent, of the
real stream's tail ($p_{99} - T$) that the rebuilt stream still produces: $100$ means the
ingredient did not matter, $0$ that it was the whole tail. The columns ask: TG, if the
transactions are kept intact but no longer come in bunches, does the tail survive? TP, if
each transaction is put at a random moment? TS, if big and small transactions swap
places? TM, if every transaction is sent as one packet? TW, if the pieces of a split
transaction are spread apart? What the table shows: TS, TM and TW keep about all of the
tail in every row, so the size of transactions and how the exchange packs them do not
build the queue. TP removes most of it at every service time, so the timing of
transactions does. TG removes little at $16\,\mu\mathrm{s}$ and most at
$128\,\mu\mathrm{s}$: a short server is hurt by each short gap on its own, a long server
only by sustained runs, and runs are what TG breaks up. This is the profile the gap
shuffle gave at the packet level (Table~\ref{tab:nulls}). At $T \le 8\,\mu\mathrm{s}$ the
real tail is under $1\,\mu\mathrm{s}$ and the shares are not meaningful. Using ZN's cost per
message instead of NQ's moves every cell by at most a few points, and session-bootstrap
intervals are within about two points. Figure~\ref{fig:bars-tx-arms} shows the same numbers as bars.}
\label{tab:tx-arms}
\end{table}

\begin{figure}[H]
\centering
\includegraphics[width=0.95\textwidth]{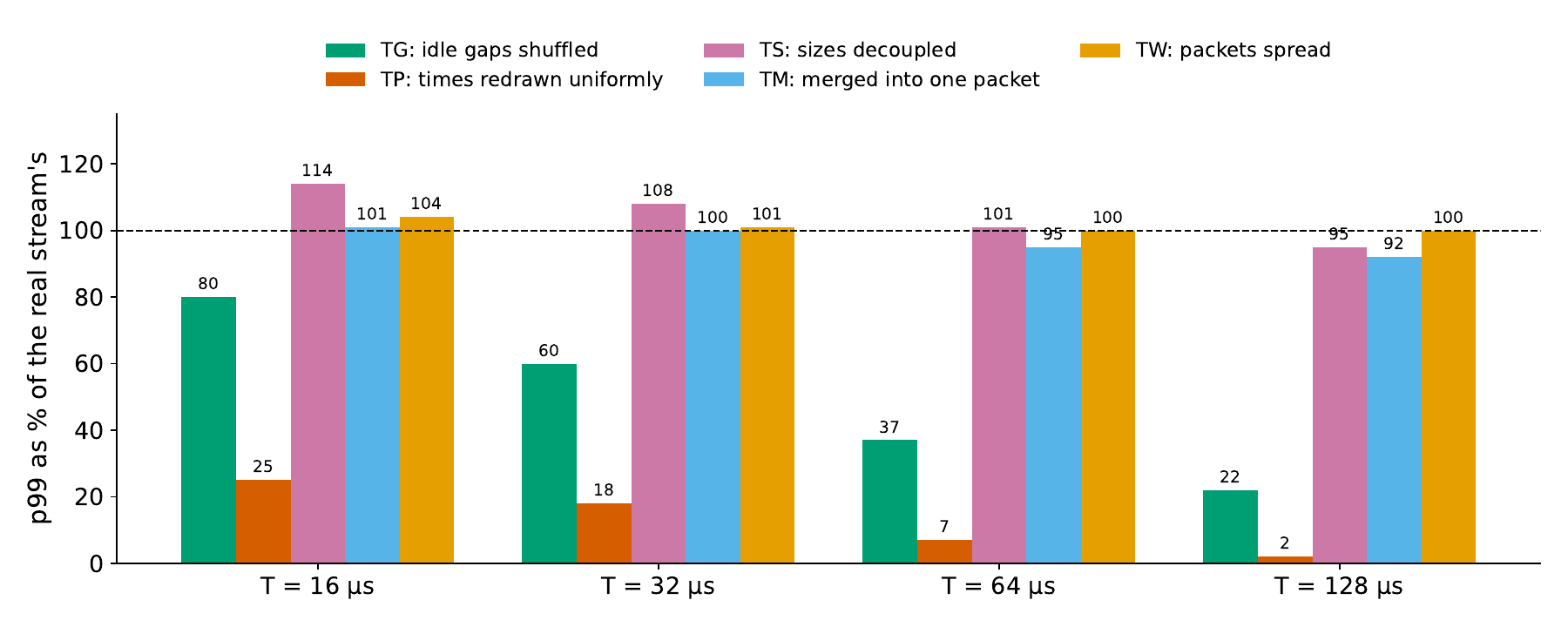}
\caption{Table~\ref{tab:tx-arms} as bars: each transaction-level arm's $p_{99}$ as a percentage of
the real stream's. Redrawing transaction times uniformly (TP, vermillion) removes the tail;
shuffling idle gaps (TG, green) removes more of it the longer the service time. Decoupling
sizes (TS), merging each transaction into one packet (TM) or spreading its packets (TW)
leaves it near $100\%$ from $16\,\mu\mathrm{s}$ up.}
\label{fig:bars-tx-arms}
\end{figure}

\paragraph{Summary in plain terms.} Each rebuilt stream changes one thing about the real
transactions and keeps everything else. What removes the tail, and at which service time
$T$:
\begin{itemize}
    \item \emph{Putting each transaction at a random moment (TP) removes the tail at every
$T$}: about three quarters of it at $16\,\mu\mathrm{s}$ and almost all of it from
$64\,\mu\mathrm{s}$. When transactions happen is what builds the queue.
    \item \emph{Keeping every gap between transactions but shuffling their order (TG)
removes little at short $T$ and most at long $T$}: about $80\%$ of the tail remains at
$16\,\mu\mathrm{s}$, about a fifth at $128\,\mu\mathrm{s}$. A fast server is hurt by each
short gap on its own, and the shuffle keeps every short gap. A slow server is hurt by long
runs of transactions close together, and the shuffle breaks the runs up.
    \item \emph{Nothing about what the transactions contain or how the exchange packs them
removes the tail from $16\,\mu\mathrm{s}$ up.} Swapping big and small transactions around
(TS), sending each transaction as one packet (TM), or spreading a transaction's packets
apart (TW) leaves the tail about as it was.
    \item \emph{At $T \le 8\,\mu\mathrm{s}$, the sweep points at or just above the publisher period, these rebuilt streams say nothing.} The real
stream's tail there is too small to split into shares. At the service floor of the
two-stage production receiver the tail comes mostly from decoding long packets one message after
another (Section~\ref{sec:crossval-spanrun}), which these streams were not built to test.
\end{itemize}

Three things follow. The tail is the transaction process: putting transactions at random
moments removes most of it. From $16\,\mu\mathrm{s}$ up, sizes and packetisation do not
enter, at ZN's cost per message as at NQ's. And the ordering of transactions into runs,
which TG destroys, carries a small part of the tail at $16\,\mu\mathrm{s}$ and most of it
at $128\,\mu\mathrm{s}$, the same profile the packet-level gap shuffle gave. The rest,
dominant at HFT service times, is the surplus of short gaps between transactions over a
Poisson stream at the same rate (Table~\ref{tab:tx-gaps}), which on the publisher's clock
sit at its period (Table~\ref{tab:two-clocks}). That surplus says the gap distribution is
far from exponential; it does not say what put the mass there, and a stream with the same
gaps and no self-excitation (the gap shuffle) reproduces it and most of the tail at
$16\,\mu\mathrm{s}$. The attribution is therefore the following. The queueing tail behind a
single-threaded receiver on this feed is produced by the timing of matching-engine
transactions; from $16\,\mu\mathrm{s}$ up it is not produced by packetisation, by the
message count per packet, or by the packet rate (at the service floor of the two-stage
production receiver the message count carries most of it, Section~\ref{sec:crossval-spanrun}). It
acts on a short service time through how often two transactions land within one service
time of each other, and on a long one through how long the runs of them last, which is the
self-excitation. Why transactions reach the engine within microseconds of each other is a
question about the market and is left open (Section~\ref{sec:concl-cannot}).

\findings{
    \item \emph{The exchange does not split transactions into bursts.} Almost every transaction fits in one packet; the pieces of a split one arrive about $15\,\mu\mathrm{s}$ apart; the short gaps lie between different transactions.
    \item \emph{The transaction stream carries the clustering}: transactions and packets have the same count variability and the same branching ratio, about $0.8$.
    \item \emph{When transactions happen builds the tail.} Placing them at random moments removes most of it.
    \item \emph{From $16\,\mu\mathrm{s}$ up, what they contain and how they are packed does not}: swapping sizes, merging into one packet or spreading packets leaves the tail about as it was.
}

\subsection{Tandem, sharding and dispatch on the same $N$ cores}
\label{sec:results-equal-core}

The tandem of Table~\ref{tab:main-corpus} runs on $N$ cores and the single thread on
one, so part of what it buys is capacity: its throughput is $N/T$ against $1/T$, and more
capacity shortens queues under any arrival law. The design question is whether a chain
of stages is the right way to spend the same $N$ cores, and the two alternatives are the
ones HFT systems use. \emph{Sharding} runs $N$ independent single-thread loops,
each owning a disjoint slice of the stream, by instrument or by channel: full service $T$
per message and no hops, so it has the lower median, and its tail depends on whether
a burst lands on one shard or spreads over several. \emph{Dispatch} puts one queue in front of $N$ identical servers, each taking a
whole packet and doing the full work $T$ on it (an M/D/$N$ queue). It has the same
capacity as the tandem and no hops, so a burst drains in about the same time, while each
packet's own latency stays at $T$ instead of paying $(N-1)$ hops.

Sharding is not available on this corpus: a CME burst is one instrument's order flow, so
sharding by instrument sends the whole burst to one core at service $T$ and relieves
nothing, and the corpus is a single instrument. Dispatch is available for any stage that
holds no state from one packet to the next --- SBE decode is such a stage, order-book
application is not --- provided its output is put back in order. Two comparisons follow.
Table~\ref{tab:equal-core} puts dispatch beside the tandem with no ingress hop, no egress
hop and no resequencing charged, which is the lowest latency dispatch could reach.
Table~\ref{tab:dispatch-costed} charges all three.

\begin{table}[H]
\centering
\small
\begin{tabular}{@{}rr rr rr@{}}
\toprule
 & & \multicolumn{2}{c}{tandem} & \multicolumn{2}{c}{M/D/$N$ dispatch} \\
\cmidrule(lr){3-4} \cmidrule(lr){5-6}
$T$ & $N$ & $p_{50}$ & $p_{99}$ & $p_{50}$ & $p_{99}$ \\
\midrule
 2 & 1 &    2.00 &    2.00 &    2.00 &    2.00 \\
 2 & 2 &    3.70 &    3.70 &    2.00 &    2.00 \\
 2 & 4 &    7.10 &    7.10 &    2.00 &    2.00 \\
 2 & 8 &   13.90 &   13.90 &    2.00 &    2.00 \\
\midrule
 4 & 1 &    4.00 &    4.00 &    4.00 &    4.00 \\
 4 & 2 &    5.70 &    5.70 &    4.00 &    4.00 \\
 4 & 4 &    9.10 &    9.10 &    4.00 &    4.00 \\
 4 & 8 &   15.90 &   15.90 &    4.00 &    4.00 \\
\midrule
 8 & 1 &    8.00 &    8.61 &    8.00 &    8.61 \\
 8 & 2 &    9.70 &    9.70 &    8.00 &    8.00 \\
 8 & 4 &   13.10 &   13.10 &    8.00 &    8.00 \\
 8 & 8 &   19.90 &   19.90 &    8.00 &    8.00 \\
\midrule
16 & 1 &   16.00 &   34.59 &   16.00 &   34.59 \\
16 & 2 &   17.70 &   18.31 &   16.00 &   16.00 \\
16 & 4 &   21.10 &   21.10 &   16.00 &   16.00 \\
16 & 8 &   27.90 &   27.90 &   16.00 &   16.00 \\
\midrule
32 & 1 &   32.00 &  146.14 &   32.00 &  146.14 \\
32 & 2 &   33.70 &   52.29 &   32.00 &   48.04 \\
32 & 4 &   37.10 &   37.71 &   32.00 &   32.00 \\
32 & 8 &   43.90 &   43.90 &   32.00 &   32.00 \\
\midrule
64 & 1 &   98.88 &  820.70 &   98.88 &  820.70 \\
64 & 2 &   65.70 &  179.84 &   64.00 &  171.38 \\
64 & 4 &   69.10 &   87.69 &   64.00 &   72.50 \\
64 & 8 &   75.90 &   76.51 &   64.00 &   64.00 \\
\midrule
128 & 1 &  364.45 & 5316.85 &  364.45 & 5316.85 \\
128 & 2 &  164.58 &  886.40 &  128.00 &  868.98 \\
128 & 4 &  133.10 &  247.24 &  128.00 &  223.06 \\
128 & 8 &  139.90 &  158.49 &  128.00 &  128.00 \\
\bottomrule
\end{tabular}
\caption{Real arrivals, corpus medians in \si{\micro\second}: $N$ cores as an $N$-stage tandem
against the same $N$ cores as one queue feeding $N$ whole-packet servers, with no hop and
no resequencing charged to the latter. What the table shows: uncharged, the pool is at or
below the tandem everywhere, because it adds no hops. This is the most dispatch could
achieve; Table~\ref{tab:dispatch-costed} charges its costs. Figure~\ref{fig:bars-equal-core} shows the same numbers as bars.}
\label{tab:equal-core}
\end{table}

\begin{figure}[H]
\centering
\includegraphics[width=1.0\textwidth]{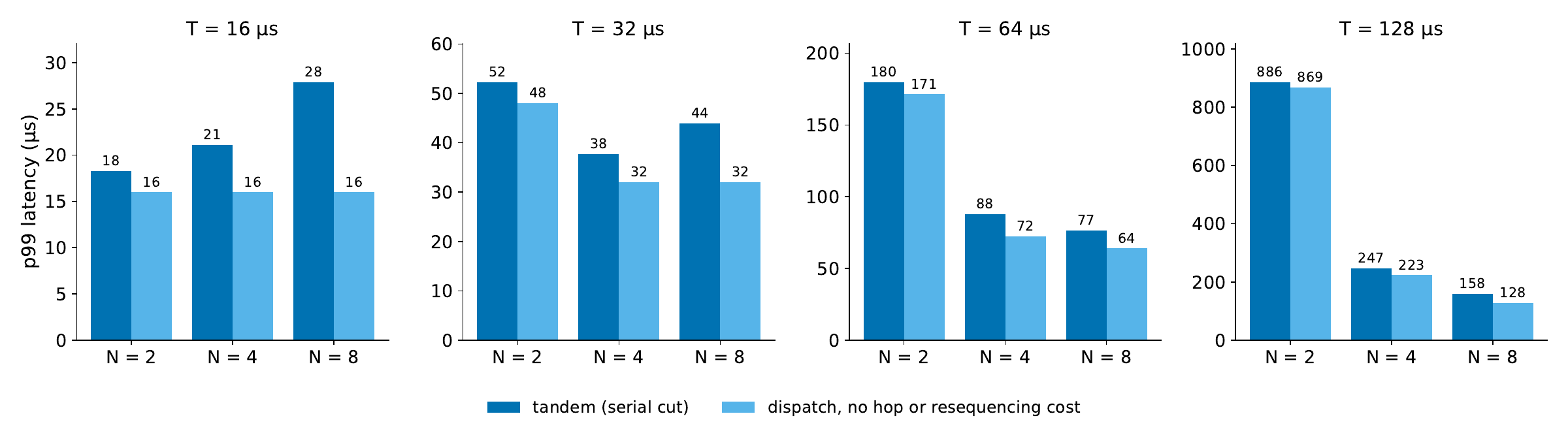}
\caption{Table~\ref{tab:equal-core} as bars: $p_{99}$ of a serial cut into $N$ stages (dark blue)
against a pool of $N$ cores that each take a whole packet, with no hop and no resequencing
charged (light blue). Uncharged, the pool is at or below the tandem everywhere, and at
$16\,\mu\mathrm{s}$ it has no tail at all. This is the most dispatch could achieve;
Figure~\ref{fig:bars-dispatch} charges its costs.}
\label{fig:bars-equal-core}
\end{figure}

Without its costs, dispatch's $p_{99}$ is at or below the tandem's in every cell from
$16\,\mu\mathrm{s}$ (Table~\ref{tab:equal-core}), and its median is $T$ against the
tandem's $T + (N-1)h$.

\begin{table}[tbp]
\centering
\small
\begin{tabular}{@{}rr rr rr r@{}}
\toprule
 & & \multicolumn{2}{c}{tandem} & \multicolumn{2}{c}{dispatch, costs charged} & reordered \\
\cmidrule(lr){3-4} \cmidrule(lr){5-6}
$T$ & $N$ & $p_{50}$ & $p_{99}$ & $p_{50}$ & $p_{99}$ & packets \\
\midrule
 16 & 2 &  17.70 &  19.59 &  19.40 &  20.78 & $0.00\%$ \\
 16 & 4 &  21.10 &  22.48 &  19.40 &  20.78 & $0.00\%$ \\
 16 & 8 &  27.90 &  29.28 &  19.40 &  20.78 & $0.00\%$ \\
\midrule
 32 & 2 &  33.70 &  52.84 &  35.40 &  51.61 & $0.00\%$ \\
 32 & 4 &  37.10 &  39.86 &  35.40 &  38.16 & $0.01\%$ \\
 32 & 8 &  43.90 &  46.66 &  35.40 &  38.16 & $0.01\%$ \\
\midrule
 64 & 2 &  65.70 & 181.08 &  67.40 & 175.74 & $0.01\%$ \\
 64 & 4 &  69.10 &  89.10 &  67.40 &  79.99 & $0.03\%$ \\
 64 & 8 &  75.90 &  81.42 &  67.40 &  72.92 & $0.03\%$ \\
\midrule
128 & 2 & 166.37 & 902.82 & 131.40 & 883.83 & $0.08\%$ \\
128 & 4 & 133.10 & 250.40 & 131.40 & 227.94 & $0.14\%$ \\
128 & 8 & 139.90 & 162.28 & 131.40 & 142.45 & $0.14\%$ \\
\bottomrule
\end{tabular}
\caption{The same $N$ cores as a tandem and as a dispatch pool with every cost charged: a hop
into the servers, a hop out of them, and an in-order resequencer that holds each packet
until its predecessors have left. Corpus medians in \si{\micro\second} at
$h = 1.7\,\mu\mathrm{s}$, service growing with message count at NQ's cost per message;
the last column is the share of packets that wait in the resequencer. What the table
shows: dispatch's median is $T + 2h$ at every $N$, against the tandem's $T + (N-1)h$, so
dispatch is one hop slower at $N = 2$, level at three stages and faster from four. At
$N = 2$ the two have about the same $p_{99}$; from $N = 4$ dispatch is lower on both
quantiles in every cell. Resequencing costs almost nothing. Constant service gives the
same picture. What follows: for a stage that carries no state, a pool of four or more
cores is at least as good as cutting it into as many stages (three cores were not
simulated); at two cores the two are level. For a stage that carries state, the order
book above all, dispatch is not available and cutting the chain is how extra cores are
spent.}
\label{tab:dispatch-costed}
\end{table}

\begin{figure}[H]
\centering
\includegraphics[width=1.0\textwidth]{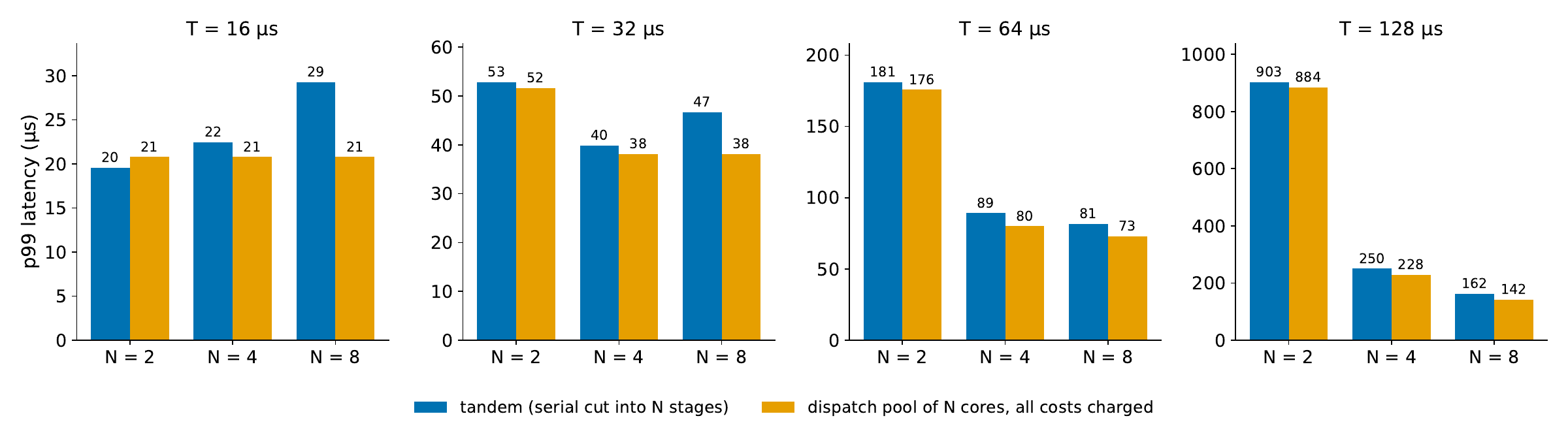}
\caption{Table~\ref{tab:dispatch-costed} as bars: $p_{99}$ of a serial cut into $N$ stages (blue)
against a pool of $N$ cores with the ingress hop, egress hop and resequencing charged
(orange). At $N = 2$ the two are level; from $N = 4$ the pool is lower at every service
time, and at $16\,\mu\mathrm{s}$ its $p_{99}$ does not grow with $N$ while the tandem's
does.}
\label{fig:bars-dispatch}
\end{figure}

With every cost charged (Table~\ref{tab:dispatch-costed}), dispatch pays two hops on
every packet, in and out, and nothing else measurable: the resequencer holds almost no
packets, because a packet that waited behind a burst in the shared queue leaves in the
order it arrived. Its median is therefore $T + 2h$ whatever $N$ is: one hop more than a
two-stage tandem, the same as three stages, and less than four or more. Its tail is lower
than the tandem's from $N = 4$ in every cell. The design consequence is a division by
state. A stateless stage above the publisher period, decode being the natural candidate,
is best served by a dispatch pool of four or more cores (three were not simulated); at two
cores a cut and a pool are level. A stateful stage above the publisher period is served
by cutting the chain, which is the case the rest of the paper analyses.

\findings{
    \item \emph{Sharding by instrument does not help here}: a burst is one instrument's order flow and lands on one shard.
    \item \emph{Dispatch with every cost charged} (hop in, hop out, resequencing) has median $T + 2h$ at any $N$, and almost no packet waits to be put back in order.
    \item \emph{Against a serial cut on the same cores} the pool is level at two cores and has the lower $p_{99}$ from four; three was not simulated.
    \item \emph{Only a stage without state can be dispatched.} Decode can; the order book cannot.
}

\subsection{Median cost of an $N$-stage tandem}
\label{sec:results-median}

This section looks at the typical message rather than the slow one. Its point is that
splitting a task into more stages makes the typical latency worse by a fixed, predictable
amount, exactly one mailbox hop per added stage, and nothing more. The reason is that the
typical message arrives when the queue in front of it is empty, so it pays only the
service time plus the hops. The simulation confirms this to the reported precision in
almost every cell of the sweep. The few exceptions are at the longest service times, where
the queue never empties and even the typical message waits; those sit outside the range
the paper recommends, and splitting cures them anyway.

The median columns of Table~\ref{tab:main-corpus} match the single-message formula of
Proposition~\ref{thm:burst}
\[
p_{50}(N) = T + (N-1)\,h \qquad \text{with } h = 1.7\,\mu\mathrm{s},
\]
in every cell of the Poisson stream and in all but three cells of the real stream. Under
Poisson the identity holds because the queue is empty; on the real stream it holds because
the median message arrives in a quiet stretch between clusters, where the queue is also
empty. Each boundary therefore costs exactly one hop on the median, with no hidden
component, as Theorem~\ref{thm:reduction} says it must at every quantile.

The three exceptions are the longest service times: one stage at $64$ and
$128\,\mu\mathrm{s}$, and two stages at $128\,\mu\mathrm{s}$. There clusters merge into a
persistent backlog and even the typical message waits. This is the regime boundary showing
in the median rather than in the tail, and it lies outside the range the paper's
recommendations address. Splitting restores the identity once each stage is short enough;
at $128\,\mu\mathrm{s}$ the median is back at the no-queue latency by four stages.

\findings{
    \item \emph{Each stage costs exactly one hop at the median}: $p_{50} = T + (N-1)h$ in almost every cell, on real and Poisson arrivals.
    \item \emph{The three exceptions} are the longest service times, where even the typical message waits; splitting restores the identity.
}

\subsection{The design equation: hop cost against tail excess}
\label{sec:results-design}

This section turns the measurements above into a rule for choosing how many stages to
split a task into. The idea is simple. Splitting a task across more threads has a fixed
price: every added stage costs one mailbox hop, and that hop is paid on every message,
fast or slow. It also has a benefit: when packets arrive in bursts, messages queue behind
one another, and shorter stages drain those queues faster, so the slow tail shrinks. The
question is when the benefit outweighs the price.

The section does four things. It writes the median and the tail latency as two short
formulas, a hop-cost term plus a tail term. It tabulates the tail term from the corpus and
shows that it depends only on how long the slowest stage is, and that it disappears
entirely once every stage is short enough. It combines cost and benefit into one
expression the operator can minimise, with a table of where the optimum falls. And it
draws the practical conclusion: the goal is not many threads but a short slowest stage,
and a task already shorter than the publisher period, about seven and a half
microseconds here, should not be split at all.

From the corpus:
\begin{align}
p_{99}^{\mathrm{Hawkes}}(N) &\;=\; T + (N-1)h + \Delta(N), \label{eq:design99} \\
p_{50}^{\mathrm{Hawkes}}(N) &\;=\; T + (N-1)h, \label{eq:design50}
\end{align}
where $\Delta(N) \equiv p_{99}^{\mathrm{Hawkes}}(N) - p_{50}^{\mathrm{Hawkes}}(N)$ is the
``clustering tail excess''. Table~\ref{tab:delta} lists it, together with the ratio
$\Delta(N)/\Delta(1)$ that Theorem~\ref{thm:reduction} bounds by $1/N$.

\begin{table}[H]
\centering
\begin{tabular}{@{}r rrrr r rrr r@{}}
\toprule
 & \multicolumn{4}{c}{$\Delta(N)$ (\si{\micro\second})} & & \multicolumn{3}{c}{$\Delta(N)/\Delta(1)$} & \\
\cmidrule(lr){2-5} \cmidrule(lr){7-9}
$T$ & $N{=}1$ & $N{=}2$ & $N{=}4$ & $N{=}8$ & & $N{=}2$ & $N{=}4$ & $N{=}8$ & $\gamma$ \\
\midrule
 2 &    0.00 &    0.00 &    0.00 &    0.00 & & --- & --- & --- & --- \\
 4 &    0.00 &    0.00 &    0.00 &    0.00 & & --- & --- & --- & --- \\
 8 &    0.61 &    0.00 &    0.00 &    0.00 & & 0.000 & 0.000 & 0.000 & --- \\
16 &   18.59 &    0.61 &    0.00 &    0.00 & & 0.033 & 0.000 & 0.000 & --- \\
32 &  114.14 &   18.59 &    0.61 &    0.00 & & 0.163 & 0.005 & 0.000 & --- \\
64 &  721.82 &  114.14 &   18.59 &    0.61 & & 0.158 & 0.026 & 0.001 & 3.32 \\
128 & 4952.41 &  721.82 &  114.14 &   18.59 & & 0.146 & 0.023 & 0.004 & 2.68 \\
\midrule
bound & & & & & & 0.500 & 0.250 & 0.125 & $\ge 1$ \\
\bottomrule
\end{tabular}
\caption{Corpus-median tail excess $\Delta(N) = p_{99}(N) - p_{50}(N)$ on the real arrivals at
$h = 1.7\,\mu\mathrm{s}$, its ratio to the single-stage value, and the fitted exponent
$\gamma$ in $\Delta(N) \approx \Delta(1) N^{-\gamma}$; the last row is the
Theorem~\ref{thm:reduction} bound. What the table shows: each extra halving of the stage
length moves the excess to the single-thread value at half the service time, so the
values repeat along the diagonals; once each stage is below the publisher period the
excess is zero. Every ratio is far below the bound. Figure~\ref{fig:bars-delta} shows the same numbers as bars.}
\label{tab:delta}
\end{table}

\begin{figure}[H]
\centering
\includegraphics[width=0.85\textwidth]{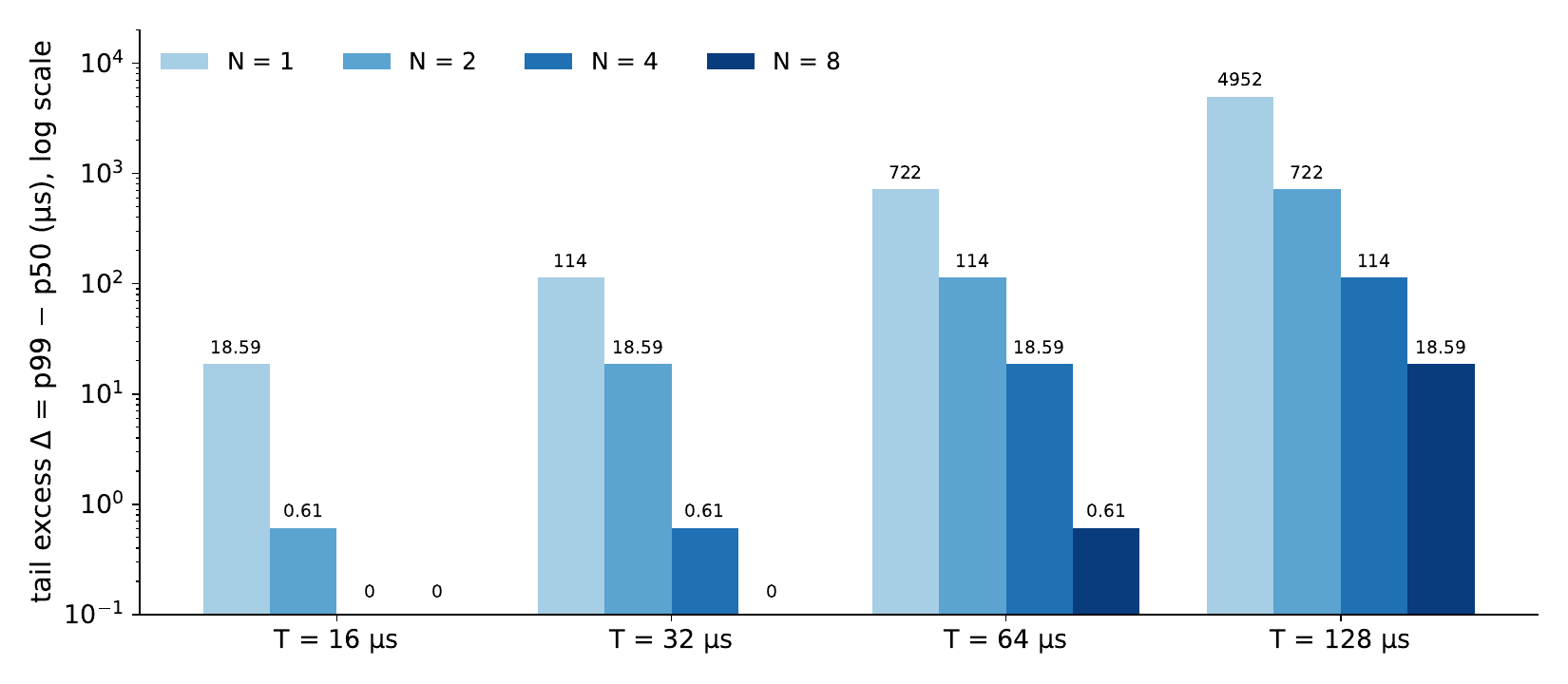}
\caption{Table~\ref{tab:delta} as bars, log scale: the tail excess $\Delta = p_{99} - p_{50}$
for one, two, four and eight stages at each service time. Each extra halving of the stage
length moves the excess down to the single-thread value at half the service time, because
only the largest stage matters. Once each stage is below the publisher period the excess
is zero (bars marked $0$).}
\label{fig:bars-delta}
\end{figure}

Two features of Table~\ref{tab:delta} follow from the theorems. First, $\Delta$ depends on
$(T, N)$ only through the per-stage service $T/N$: every diagonal of constant $T/N$
carries one value, to the reported precision, without fitting. This is the reduction
identity, Equation~\eqref{eq:reduction}, on real arrivals: the tandem's tail is the
single-server tail at $s_{\max}$, and the diagonals of the table are
$\Delta_{\mathrm{single}}(s)$ on a power-of-two grid of $s$.

Second, every ratio is at or below the bound $1/N$, and the pathwise check of
Section~\ref{sec:evaluation} finds no message on any window above it. The margin is large,
because of the threshold. Since $\Delta_{\mathrm{single}}(s)$ is zero for $s$ below the
publisher period (Section~\ref{sec:results-threshold}), a split that pushes the per-stage
service below the publisher period does not merely divide the tail, it removes it. At
$32\,\mu\mathrm{s}$, for example, four stages leave about half a percent of the
single-thread excess, where the bound allows a quarter, and eight leave none. Where a
power law $\Delta(N) \approx \Delta(1) N^{-\gamma}$ can be fitted, its exponent is about
three, far steeper than the $\gamma = 1$ the bound alone would permit. The decay of the
excess is therefore not a $1/N$ division on this corpus; it ends at zero once each stage is
below the publisher period. The design consequence is that the objective is not a large
$N$ but an $s_{\max}$ below the publisher period. The utilisation slack described after
Theorem~\ref{thm:reduction} contributes as well, and Appendix~\ref{app:check} reproduces it
on synthetic clusters.

The operator's stage-count optimum is
\[
N^\star \;=\; \arg\min_N \Big[ (N-1) h \;+\; \kappa \, \Delta(N) \Big]
\]
for the operator's tail weight $\kappa$ (the number of microseconds of median the operator
will spend to remove one microsecond of $p_{99}$ excess). Over the simulated grid $N \in
\{1, 2, 4, 8\}$ at $h = 1.7\,\mu\mathrm{s}$ the breakpoints are:
\begin{center}
\begin{tabular}{@{}r llll@{}}
\toprule
$T$ (\si{\micro\second}) & $N^\star = 1$ & $N^\star = 2$ & $N^\star = 4$ & $N^\star = 8$ \\
\midrule
  2 & always             & ---                        & ---                        & ---                 \\
  4 & always             & ---                        & ---                        & ---                 \\
  8 & $\kappa < 2.8$     & $\kappa \ge 2.8$           & ---                        & ---                 \\
 16 & $\kappa < 0.095$   & $0.095 \le \kappa < 5.6$   & $\kappa \ge 5.6$           & ---                 \\
 32 & $\kappa < 0.018$   & $0.018 \le \kappa < 0.19$  & $0.19 \le \kappa < 11$     & $\kappa \ge 11$     \\
 64 & $\kappa < 0.0028$  & $0.0028 \le \kappa < 0.036$& $0.036 \le \kappa < 0.38$  & $\kappa \ge 0.38$   \\
128 & $\kappa < 0.0004$  & $0.0004 \le \kappa < 0.0056$& $0.0056 \le \kappa < 0.071$& $\kappa \ge 0.071$ \\
\bottomrule
\end{tabular}
\end{center}
A dash means that stage count is never optimal at any $\kappa$: at $T \le 8\,\mu\mathrm{s}$
the tail excess is already zero by $N = 2$, so further splitting only adds hop cost. An
operator who weights tail and median equally ($\kappa = 1$) keeps one stage up to
$8\,\mu\mathrm{s}$ and uses more stages as $T$ grows, which is also the ordering of the raw
$p_{99}$ minima in Table~\ref{tab:main-corpus}. Over the service times an HFT hot path
occupies, an operator must weight the tail several times more than the median before
splitting pays; above the threshold the required weight is small.

At real cut points the stages are unequal, and by Remark~\ref{rem:unequal} the tail excess
depends on the partition $P$ of the servicing chain only through its largest stage,
$\Delta(P) = \Delta_{\mathrm{single}}(s_{\max}(P))$, i.e.\ the single-stage tail excess
measured at service $s_{\max}$. The optimum over partitions is then
\[
P^\star \;=\; \arg\min_P \Big[ \mathrm{hops}(P)\, h \;+\; \kappa\,
\Delta_{\mathrm{single}}\big(s_{\max}(P)\big) \Big],
\]
which Table~\ref{tab:delta} lets the operator evaluate directly, since its diagonals are
$\Delta_{\mathrm{single}}(s)$ on a power-of-two grid of $s$. Two rules follow.
A cut is worth its hop only if it lowers $s_{\max}$ by enough that
$\kappa\,[\Delta_{\mathrm{single}}(s_{\max}) - \Delta_{\mathrm{single}}(s_{\max}')] > h$;
and among partitions with the same $s_{\max}$ the one with the fewest hops is preferred,
so sub-tasks downstream of the bottleneck should be merged onto one thread until their sum
reaches $s_{\max}$.

\begin{center}
\fbox{\begin{minipage}{0.92\linewidth}
\textbf{Design principle (only $s_{\max}$ matters).} For a servicing chain cut at fixed
boundaries, the tail-latency question is not ``how many threads?'' but ``what is the
smallest achievable largest stage $s_{\max}$?'' By Theorem~\ref{thm:reduction} the
clustering tail excess of any partition is $\Delta_{\mathrm{single}}(s_{\max})$,
independent of the stage count; a cut that does not lower $s_{\max}$ compresses no tail
and still costs one hop $h$. Choose the partition with the smallest reachable $s_{\max}$,
then the fewest hops that realise it.
\end{minipage}}
\end{center}

\paragraph{Minimum service time for splitting.}
Equations~\eqref{eq:design99}--\eqref{eq:design50} imply a lower bound on the tasks the
recommendation applies to. Splitting a task of total service $T$ into $N$ stages spends
$(N-1)h$ at every quantile to compress a tail excess $\Delta(1)$ that is zero at
$4\,\mu\mathrm{s}$ and below and grows fast above (Table~\ref{tab:delta}). When $T$ is
below the publisher period there is nothing to remove, and the design equation
prescribes one stage; a hop there adds cost only. The sweep confirms this: at $2$ and
$4\,\mu\mathrm{s}$ a split removes nothing and raises every quantile by the full hop, and at
$16\,\mu\mathrm{s}$ it removes almost the entire excess for one hop, on every window.

Two separate necessary conditions apply. The first is a property of the
\emph{framework}: a split cannot pay unless the tail excess it
removes exceeds the hop cost it adds, which on this framework's shipping BQueue path ($h
\approx 1.7\,\mu\mathrm{s}$) sets a minimum at roughly $T \gtrsim 2h \approx
3.4\,\mu\mathrm{s}$. The second is a property of the \emph{feed}: there is no tail excess
to remove at all until $T$ exceeds the publisher period, about
$\spub$ on this corpus (Section~\ref{sec:results-conditioning}). The binding
constraint is whichever is larger, and on NQ with a BQueue mailbox it is
the feed's, by roughly a factor of two; this is why $T = 4\,\mu\mathrm{s}$, although
above $2h$, has $\Delta(1) = 0$ and should not be split. A framework with a
much more expensive mailbox would find its own minimum binding instead. On CME the binding minimum is
the publisher period, and the rest of the paper uses that term. Sub-microsecond tasks, for
instance a single SBE integer decode or a small handler that fires from a strategy fast
path, belong on a synchronous send rather than on a cross-thread mailbox, regardless of
feed statistics.

\findings{
    \item \emph{Only the largest stage matters.} The tail excess depends on $(T, N)$ only through $T/N$, as Theorem~\ref{thm:reduction} predicts.
    \item \emph{A cut that takes every stage below the publisher period removes the tail rather than dividing it}; the decay is far steeper than the $1/N$ bound.
    \item \emph{Weighting tail and median equally} gives one stage up to $8\,\mu\mathrm{s}$ and more as $T$ grows, up to eight at $64\,\mu\mathrm{s}$ and above.
    \item \emph{The binding minimum is the publisher period $\spub$}, not the framework's $2h$.
}

\subsection{Per-window results}
\label{sec:results-windows}

The results so far are for a typical half-hour window. This section asks whether they
hold on every window, including the busiest ones, or only on average. The answer matters
because a busy window has a far worse single-thread tail than the typical one, so an
operator cannot size a pipeline from the average alone. The finding is that the split wins
on essentially every window from about ten microseconds up, and on no window at eight
microseconds or below, where there is no tail to remove and the split only adds hop
cost, with a sharp crossover between. The
recommendation is therefore not an artefact of averaging.

The corpus medians understate the tail on the busiest windows. Table~\ref{tab:windows}
gives the median and the $95$th percentile over windows of the single-thread $p_{99}/T$.
At short service times the busiest windows are close to the median one; at long service
times one window in twenty has a tail up to twice the median window's. The tandem's win is
not confined to the median window: two stages and four stages both beat one on every
window from $16\,\mu\mathrm{s}$, and on no window at $8\,\mu\mathrm{s}$ or below, where there
is no tail to compress and the split pays only hop cost. The crossover is sharp and lies
between $8$ and $10\,\mu\mathrm{s}$ (Section~\ref{sec:results-threshold}). The fitted
branching ratio is about $0.8$ in every window (Section~\ref{sec:setup-hawkes}); how the
tail depends on it, and on the window's packet rate, is examined in
Section~\ref{sec:results-conditioning}. The $p_{95}$ is recorded per window and omitted
from the tables only for brevity; it lies between $p_{50}$ and $p_{99}$ by construction.
The windows are nested in 276 sessions and are not independent, so uncertainty on the
aggregates is clustered at the session level.

\begin{table}[H]
\centering
\small
\begin{tabular}{@{}lrrrrrrr@{}}
\toprule
$T$ (\si{\micro\second}) & 2 & 4 & 8 & 16 & 32 & 64 & 128 \\
\midrule
$p_{99}/T$, median window        & 1.0 & 1.0 & 1.1 & 2.2 & 4.6 & 12.8 & 41.5 \\
$p_{99}/T$, 95th-percentile window & 1.0 & 1.0 & 1.1 & 2.5 & 6.1 & 19.9 & 81.7 \\
\bottomrule
\end{tabular}
\caption{Single-thread $p_{99}$ relative to the service time, in the median window and in
the window at the $95$th percentile, over the corpus. What the table shows: below the
publisher period no window has a tail; above it the busiest windows are worse than the
median one, by up to a factor of two at long service times. What follows: a pipeline sized
from the median window underestimates the tail it will meet in a busy one.}
\label{tab:windows}
\end{table}

\findings{
    \item \emph{The split wins on every window from $16\,\mu\mathrm{s}$ and on none at $8\,\mu\mathrm{s}$ or below.}
    \item \emph{The crossover is sharp}, between $8$ and $10\,\mu\mathrm{s}$.
    \item \emph{Busy windows are worse than the median one}, by up to a factor of two at long service times (Table~\ref{tab:windows}).
}

\subsection{Location of the tail onset}
\label{sec:results-threshold}

The coarse sweep brackets the onset of the clustering tail between $T = 8$ and $T =
16\,\mu\mathrm{s}$, too wide for a design rule. We reran the single- and two-stage arms
on the same 3512 windows at $T \in \{5, 6, 7, 8, 9, 10, 12, 14, 16\}\,\mu\mathrm{s}$:

\begin{center}
\begin{tabular}{rrrrr}
\toprule
$T$ ($\mu$s) & $p_{99}$ Hawkes & excess over $T$ & windows with a tail & $p_{99}(2) < p_{99}(1)$ \\
\midrule
 5 &  5.00 &  0.00 &   0.0\% &   0.0\% \\
 6 &  6.00 &  0.00 &   0.0\% &   0.0\% \\
 7 &  7.00 &  0.00 &   4.5\% &   0.0\% \\
 8 &  8.61 &  0.61 &  98.6\% &   0.0\% \\
 9 & 10.85 &  1.85 & 100.0\% &  67.0\% \\
10 & 13.52 &  3.52 & 100.0\% &  99.7\% \\
12 & 19.83 &  7.83 & 100.0\% & 100.0\% \\
14 & 26.68 & 12.68 & 100.0\% & 100.0\% \\
16 & 34.59 & 18.59 & 100.0\% & 100.0\% \\
\bottomrule
\end{tabular}
\end{center}

\begin{figure}[H]
\centering
\includegraphics[width=1.0\textwidth]{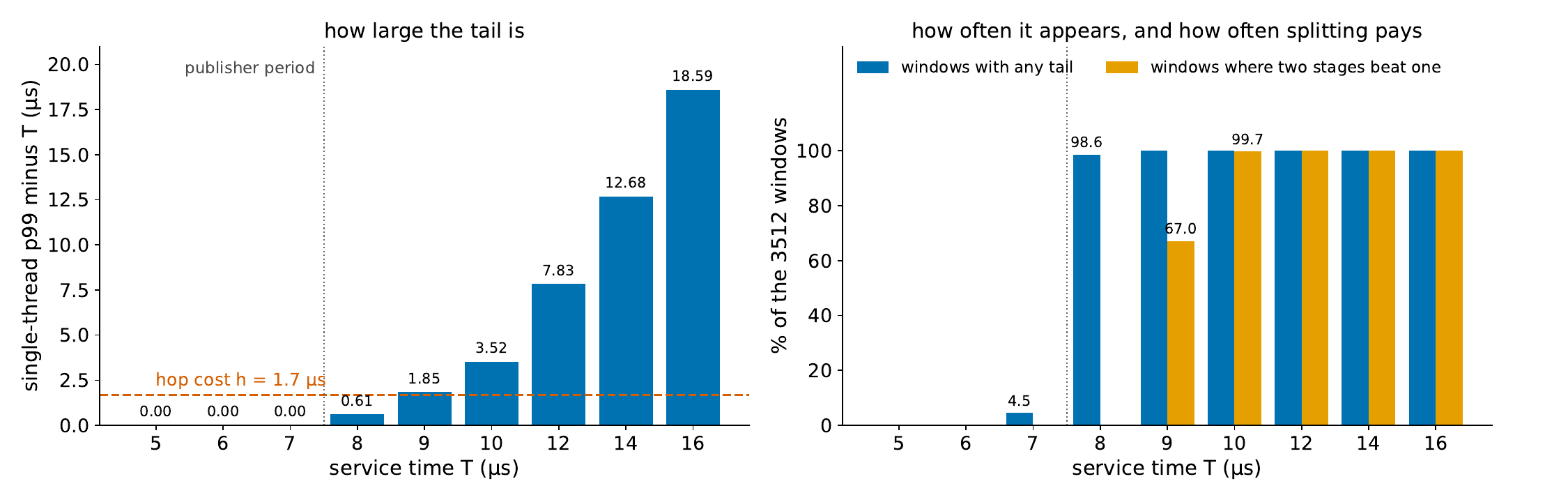}
\caption{The table above as bars; the dotted line marks the publisher period $\spub$. Left: how far the single-thread $p_{99}$ sits above the service time.
Below the publisher period there is no tail at all; above it the tail grows fast and passes the
hop cost of one split (dashed) between $8$ and $9\,\mu\mathrm{s}$. Right: the share of
windows that have any tail (blue) jumps from almost none to almost all between $7$ and
$8\,\mu\mathrm{s}$; the share where two stages beat one (orange) follows about
$1.5\,\mu\mathrm{s}$ later, once the tail is larger than the hop cost.}
\label{fig:bars-onset}
\end{figure}

Two distinct thresholds are visible (Figure~\ref{fig:bars-onset}). The first is the \emph{tail} threshold: between $T =
7$ and $8\,\mu\mathrm{s}$ the share of windows with any tail jumps from almost none to
almost all. The second is the \emph{split} threshold, about $1.5\,\mu\mathrm{s}$ later,
between $8$ and $10\,\mu\mathrm{s}$: a tail must exist \emph{and} exceed the hop cost before
a split repays it. At $8\,\mu\mathrm{s}$ nearly every window has a tail, but it is well
under the $1.7\,\mu\mathrm{s}$ hop, and splitting loses on every window. By
$10\,\mu\mathrm{s}$ the tail is about twice the hop, and splitting wins almost
everywhere.

The first threshold is set by the feed's tight-gap structure rather than by its mean
rate; the next section establishes this.

\findings{
    \item \emph{The tail appears between $7$ and $8\,\mu\mathrm{s}$}: the share of windows with any tail jumps from almost none to almost all.
    \item \emph{Splitting starts to pay about $1.5\,\mu\mathrm{s}$ later}, once the tail exceeds the hop cost, and pays on almost every window from $10\,\mu\mathrm{s}$.
}

\subsection{Dependence of the tail on packet rate and branching ratio}
\label{sec:results-conditioning}

One would expect the slow tail, and with it the case for splitting, to depend on how
busy the market is and on how strongly one transaction triggers the next. If it did, the rule
for splitting would have to change with market conditions. This section tests that
expectation and finds that, over the task lengths a trading hot path uses, it does not
hold: a window with four to five times the packet rate has the same tail, relative to
the task length, as a quiet one. The reason is that a busy market produces more bursts,
not tighter bursts. The tightest spacing between two packets on this feed is the publisher period, the
spacing at which the exchange's publisher sends packets
(Section~\ref{sec:exchange-publisher}). The two packets on either side of such a gap are
almost always two different transactions (Section~\ref{sec:results-transactions}). Neither
the publisher period nor the spacing of transactions inside a burst changes with the
market state. Since it is the
tightest spacing that starts a queue, the point at which a tail appears is a constant of
the feed. Market load
only matters for very long tasks, where the ordinary effect of a busy server takes over.
The practical consequence is that the rule for splitting needs no live estimate of market
conditions; the section closes by confirming the result on five-minute windows, and by
noting that a rule based on the average gap between packets, rather than the tightest,
would be wrong by two orders of magnitude.

The results above are conditioned on the service time $T$ alone. A Hawkes tail depends
on two further quantities that the corpus median averages over: the window's mean packet
rate $\bar\lambda$, and its branching ratio $n = \alpha/\beta$.

We sort the 3512 windows into quintiles, first by $\bar\lambda$ and then by the fitted
$n$, and report the corpus median of the normalised single-stage tail $p_{99}/T$ within
each quintile. Normalising by $T$ removes the service time, so a value of $1.0$ means no
tail at all and the quintile-to-quintile ratio isolates the conditioning effect.

\begin{center}
\begin{tabular}{lrrrrrr}
\toprule
& & \multicolumn{5}{c}{$p_{99}/T$, corpus median} \\
\cmidrule(l){3-7}
quintile & median & $T{=}8$ & $T{=}16$ & $T{=}32$ & $T{=}64$ & $T{=}128$ \\
\midrule
\multicolumn{7}{l}{\emph{by packet rate $\bar\lambda$ (pkt/s)}} \\
Q1 & 202  & 1.08 & 2.29 & 4.74 & 12.25 & 31.76 \\
Q3 & 451  & 1.07 & 2.12 & 4.48 & 12.68 & 42.81 \\
Q5 & 942  & 1.08 & 2.11 & 4.51 & 13.92 & 55.67 \\
\addlinespace
\multicolumn{2}{l}{\textbf{Q5/Q1}} & \textbf{1.00} & \textbf{0.92} & \textbf{0.95} & \textbf{1.14} & \textbf{1.75} \\
\midrule
\multicolumn{7}{l}{\emph{by branching ratio $n$}} \\
Q1 & 0.769 & 1.08 & 2.18 & 4.51 & 11.61 & 30.73 \\
Q3 & 0.798 & 1.08 & 2.16 & 4.61 & 13.17 & 42.83 \\
Q5 & 0.825 & 1.07 & 2.15 & 4.61 & 14.65 & 58.73 \\
\addlinespace
\multicolumn{2}{l}{\textbf{Q5/Q1}} & \textbf{1.00} & \textbf{0.98} & \textbf{1.02} & \textbf{1.26} & \textbf{1.91} \\
\bottomrule
\end{tabular}
\end{center}

\begin{figure}[H]
\centering
\includegraphics[width=1.0\textwidth]{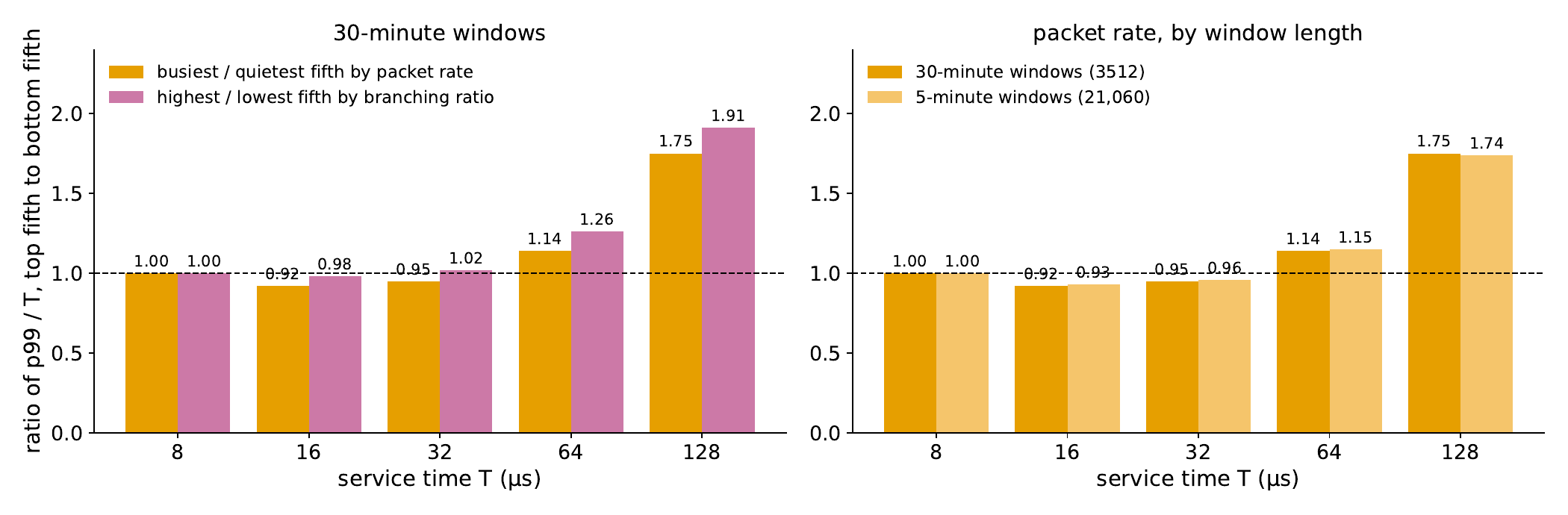}
\caption{The quintile table above and the five-minute table below as bars: how much larger the
normalised tail $p_{99}/T$ is in the top fifth of windows than in the bottom fifth. Left: by
packet rate (orange) and by branching ratio (pink). Up to $32\,\mu\mathrm{s}$ every bar is
within $8\%$ of $1$ (dashed): a busier or more self-exciting market does not change the tail
relative to the service time. From $64\,\mu\mathrm{s}$ both matter, as ordinary
utilisation queueing sets in. Right: the packet-rate ratio on 30-minute and on five-minute
windows; the two agree to $0.01$.}
\label{fig:bars-conditioning}
\end{figure}

The result divides between $32$ and $64\,\mu\mathrm{s}$ (Figure~\ref{fig:bars-conditioning}).
Up to $32\,\mu\mathrm{s}$ the normalised tail is flat in both variables: a more than
fourfold range of packet rate moves it by less than a tenth, and at $16$ and
$32\,\mu\mathrm{s}$ the busiest fifth of windows even has the slightly smaller tail; the
range of branching ratios moves it by at most a few percent. From $64\,\mu\mathrm{s}$ both
variables matter in the expected direction, and at $128\,\mu\mathrm{s}$ the busiest and the
most self-exciting fifths have roughly twice the tail of the quietest. This is the
boundary between the clustering regime and the utilisation regime found elsewhere in the
paper: from $64\,\mu\mathrm{s}$ the Poisson stream has a tail too
(Section~\ref{sec:results-tail}), utilisation $\rho = \bar\lambda T$ is the governing
quantity, and $\rho$ is proportional to the packet rate. The rate matters where
utilisation queueing applies and not below it.

\paragraph{Simulator check of the invariance.} The Poisson stream, run on the same windows
through the same simulator, does show the dependence on rate that utilisation queueing
requires: at $16$ and $32\,\mu\mathrm{s}$ its normalised tail is clearly larger in the
busiest fifth of windows than in the quietest, because utilisation crosses the $1\%$
switch-on of Corollary~\ref{cor:poisson} within the range. The real stream on the same
windows shows none. The simulator therefore resolves the rate; the clustered tail does not
depend on it at these service times. The result is the same without binning into fifths,
by rank correlation over all windows, and between the top and bottom $2\%$ of windows by
rate, a range of more than twentyfold. The mechanism is that at low utilisation clusters do
not overlap. Raising the rate raises the number of clusters per window, each with the same
internal spacing and, at a fixed branching ratio, the same size distribution, so where the
$p_{99}$ packet sits within its cluster, and hence $p_{99}/T$, is unchanged. The rate
enters once clusters begin to overlap, and the branching ratio once $T$ is long enough for
many members of one cluster to fall within one service time; on this corpus both happen
from $64\,\mu\mathrm{s}$.

\paragraph{Packet rate and the tight gaps below $64\,\mu\mathrm{s}$.} The null result has a
mechanical explanation in the arrival stream. Taking the same $\bar\lambda$ quintiles and reporting the corpus median of
the per-window interarrival-gap quantiles:

\begin{center}
\begin{tabular}{lrrrrrr}
\toprule
quintile & $\bar\lambda$ (pkt/s) & $p_1$ gap & $p_5$ gap & $p_{25}$ gap & $p_{50}$ gap & mean gap \\
\midrule
Q1 & 202 & 7.45 & 8.85 & 21.30 & 87.10 & 4944 \\
Q2 & 320 & 7.52 & 8.75 & 21.87 & 84.52 & 3128 \\
Q3 & 451 & 7.53 & 8.70 & 22.70 & 82.16 & 2215 \\
Q4 & 620 & 7.51 & 8.66 & 23.36 & 81.98 & 1612 \\
Q5 & 942 & 7.43 & 8.58 & 22.55 & 76.52 & 1062 \\
\addlinespace
\textbf{Q5/Q1} & \textbf{4.66}$\times$ & \textbf{0.997}$\times$ & \textbf{0.969}$\times$ & \textbf{1.059}$\times$ & \textbf{0.878}$\times$ & \textbf{0.215}$\times$ \\
\bottomrule
\end{tabular}
\end{center}

\begin{figure}[H]
\centering
\includegraphics[width=0.9\textwidth]{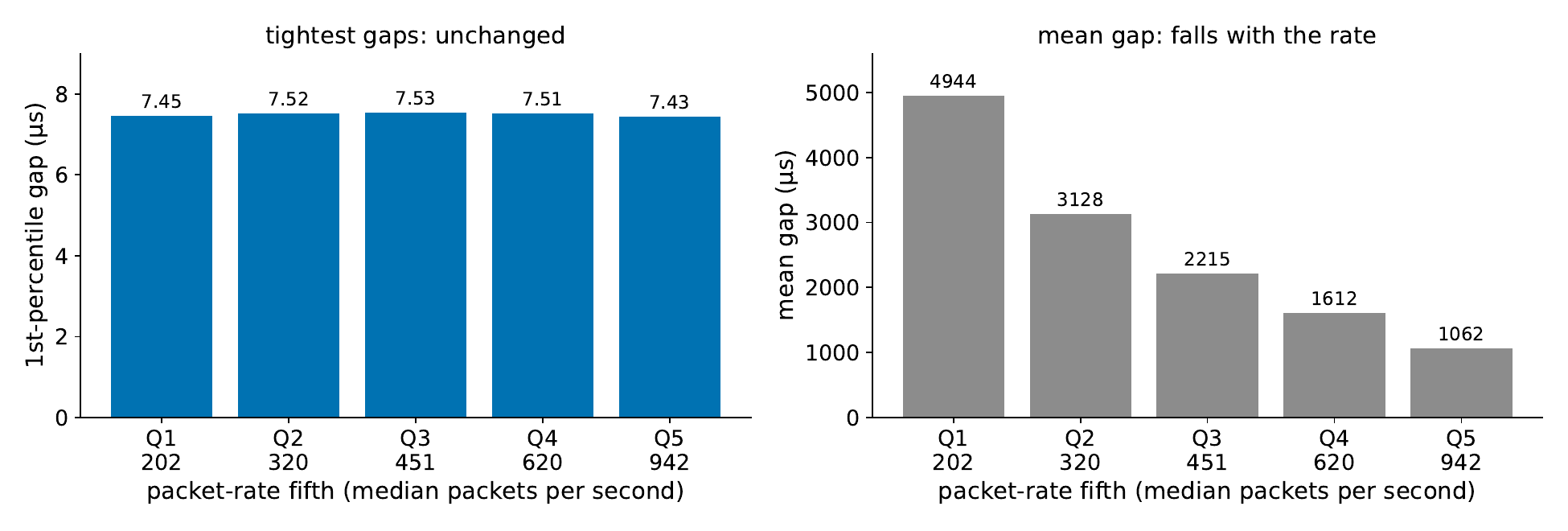}
\caption{The gap-quantile table above as bars, by packet-rate fifth. Left: the tightest gaps
($1$st percentile) stay at $7.43$--$7.53\,\mu\mathrm{s}$ while the rate more than
quadruples. Right: the mean gap falls in proportion to the rate. A busy market has more
bursts, not tighter ones, and the tightest gap is what starts a queue.}
\label{fig:bars-gap-quantiles}
\end{figure}

All gaps in microseconds. The mean gap falls in proportion to the rate, which is
definitional. Every other column is nearly constant, and the tightest gap is the most
constant of all: it stays at the publisher period while the rate more than quadruples
(Figure~\ref{fig:bars-gap-quantiles}).

Raising the intensity of a self-exciting process raises the rate at which clusters are
generated; it does not change the spacing of events inside a cluster. The tight end of that spacing is bounded by the publisher period, a property of the
exchange's publisher rather than of the market state
(Section~\ref{sec:exchange-publisher}). The events on either side of such a gap are
distinct transactions (Section~\ref{sec:results-transactions}, Table~\ref{tab:tx-gaps}),
and how closely successive transactions within a burst follow one another is likewise
invariant to the window's rate. A busy window contains more bursts, not tighter bursts.

The gap statistics on the publisher's clock are censored at the publisher period: two
transactions that the engine processes closer together than $\spub$ reach the receiver one
period apart or, less often, in one packet (Section~\ref{sec:exchange-engine}). A faster self-excitation would therefore show in
the packets rather than the gaps. It does not scale with the branching ratio: from the
lowest to the highest branching-ratio fifth, the share of multi-message packets, the mean
message count per packet and the share of packets carrying more than one transaction do
not rise, and across packet-rate fifths they barely move.

The queueing consequence: a deterministic server of service time $T$
begins to build a queue when it is handed two arrivals separated by less than $T$, so the
service time at which a tail appears is set by the \emph{tight} end of the gap
distribution, not by its mean. The tight end is invariant to $\bar\lambda$; therefore the threshold is invariant to
$\bar\lambda$. The two numbers agree: the tightest gap, $\spub$, the same in every
rate fifth, matches the tail onset of Section~\ref{sec:results-threshold}, between $7$ and
$8\,\mu\mathrm{s}$.

\paragraph{Robustness to window length.} A 30-minute window averages the rate over a long
interval, and the comparison above may have smoothed away the variation that matters. We
therefore reran the whole pipeline, from packet extraction and Hawkes fits to the Poisson
null and the full $(T, N)$ grid, on five-minute windows, six times as many. Shorter windows
do resolve more variation in rate and branching ratio, as expected. The conditioning
result is unchanged to two decimal places:

\begin{center}
\begin{tabular}{lrrrrrr}
\toprule
& $\bar\lambda$ Q5/Q1 & $T{=}8$ & $T{=}16$ & $T{=}32$ & $T{=}64$ & $T{=}128$ \\
\midrule
30-min windows ($n = 3512$)     & $4.66\times$ & 1.00 & 0.92 & 0.95 & 1.14 & 1.75 \\
\phantom{0}5-min windows ($n = 21{,}060$) & $4.98\times$ & 1.00 & 0.93 & 0.96 & 1.15 & 1.74 \\
\bottomrule
\end{tabular}
\end{center}

Entries are the ratio of the normalised single-thread tail $p_{99}/T$ in the top fifth of
windows to that in the bottom fifth. Resolving the rate six times more finely moves no
entry by more than $0.01$. The result is a property of the arrival process, not of the window
length.

\paragraph{Consequence for the design rule.} On NQ the rule can therefore be stated as a
constant of the feed rather than as a function of market state.

\begin{center}
\fbox{\begin{minipage}{0.92\linewidth}
\textbf{Design principle (on NQ, the threshold is a constant of the feed).} Split a stage whose
service time exceeds the publisher period; do not split one below it. Compare the
publisher period with the upper quantiles of the stage's per-packet service under load, not with its
median (Section~\ref{sec:concl-summary}). The publisher period is
the tight end of the feed's interarrival-gap distribution (the $p_1$ gap, about
$\spub$), not its mean, which is two orders of magnitude
larger. It is measured once per feed: on this corpus it is the same in every
packet-rate fifth and every branching-ratio fifth, so the rule requires no runtime estimate of market intensity or
clustering and does not change between a quiet and a volatile session. Both the publisher period and
the spacing of transactions inside a burst were measured on one instrument, NQ front-month on CME; recheck them on any other instrument or feed rather than
assume them.
\end{minipage}}
\end{center}

Appendix~\ref{app:synthetic} constructs a feed on which both dependences appear
at $T = 8\,\mu\mathrm{s}$, and identifies the packed-gap fraction as the quantity that
separates it from CME.

The publisher period is not the mean interarrival gap. On this corpus the mean gap is
about $2\,\mathrm{ms}$ in the median window, some three hundred times the publisher period. A rule
stated on the mean gap would recommend never splitting an HFT stage, and would be wrong
by two orders of magnitude.

\findings{
    \item \emph{Up to $32\,\mu\mathrm{s}$ market conditions do not change the tail relative to $T$}: a more than fourfold range of packet rate moves it by less than a tenth, the range of branching ratios by a few percent.
    \item \emph{From $64\,\mu\mathrm{s}$ both matter}, as utilisation queueing sets in: at $128\,\mu\mathrm{s}$ the busiest windows have about twice the tail of the quietest.
    \item \emph{A busy market has more bursts, not tighter ones}: the tightest gap stays at the publisher period while the rate more than quadruples.
    \item \emph{The same holds on five-minute windows}, and packet size does not rise with the branching ratio.
    \item \emph{A rule based on the mean gap} instead of the tightest would be wrong by two orders of magnitude.
}

\subsection{Sensitivity to the hop cost $h$}
\label{sec:results-h}

The rule depends on the cost of a mailbox hop, and this section asks what changes if that
cost is different. A cheaper hop makes splitting worthwhile over a wider range of task
lengths, but it does not move the point below which there is no tail to remove, since that
point is set by the feed and not by the framework. Any other framework can therefore
measure its own hop cost, put it into the same equation, and read off its own answer; a
framework with a very expensive hop should not split at all. The section also explains
why the hop cost used throughout is a cautious overestimate. It was measured with the
receiving thread asleep, so it includes the cost of waking it. Inside a burst the receiver
never sleeps, because messages arrive faster than it can process them, and the true hop
cost there is a few hundred nanoseconds rather than nearly two microseconds. The benefit
of splitting during a burst is therefore, if anything, understated.

The design equation is monotone in the hop cost: every breakpoint in the table above
scales in proportion to $h$, so a cheaper hop can only raise the best stage count. At a hypothetical $0.5\,\mu\mathrm{s}$ busy-poll mailbox the best stage
count changes at one service time only, $8\,\mu\mathrm{s}$, where it rises from one to
two; below the publisher period no reduction in $h$ can make a split worth taking when
there is no tail to remove. A cheaper hop widens the band of service times over which
splitting pays; it does not move the threshold, which is set by the feed. For any other
framework the procedure is the same: measure the one-way asynchronous hop cost $h$ on its
own dispatch path, put it into Equation~\eqref{eq:design99} with the corpus $\Delta$ of
Table~\ref{tab:delta}, and read off $N^\star$; a framework whose $h$ exceeds the tail excess
a split could remove gets $N^\star = 1$, and the recommendation reverses.

\findings{
    \item \emph{A cheaper hop widens the range where splitting pays but does not move the threshold}: at a third of the hop cost the best stage count changes only at $8\,\mu\mathrm{s}$.
    \item \emph{The measured hop cost is an upper bound}: inside a burst the receiver never sleeps and a hop costs a few hundred nanoseconds.
}

\section{The feed behind the packets: matching engine and publisher}
\label{sec:exchange}

The receiver's arrivals are the output of the exchange. This section looks one step
upstream, at the two CME systems that produce them, through the only window the public feed
opens on them: the timestamps it carries. It treats the exchange as a queueing system
observed from outside and makes no claim about CME's internal design beyond what the MDP3
specification states. All figures are for NQ front-month.

\subsection{Three timestamps}
\label{sec:exchange-clocks}

Every message on the feed carries two exchange timestamps, and the receiver adds a third.
\texttt{transactTime} is the start of event processing at the matching engine: when the
engine began to match, add or cancel the order behind the message. \texttt{sendingTime}
is written into each packet by CME's market-data publisher when it sends the packet.
Both come from CME's synchronised clocks. The receiver's own timestamp, taken at the
socket read, is the third. The path is
\[
\text{order entry} \;\to\; \underbrace{\text{matching engine}}_{\texttt{transactTime}}
\;\to\; \underbrace{\text{publisher}}_{\texttt{sendingTime}} \;\to\; \text{network}
\;\to\; \underbrace{\text{receiver}}_{\text{socket read}}.
\]
The simulator of Part~II uses \texttt{sendingTime} as each packet's arrival time, since the
receiver sees the publisher period plus a network delay that is close to constant. The
feed does not carry the time an order entered the engine, and nothing in it shows the
engine's or the publisher's internal queues directly. The MDP3 end-of-event flag marks the
last message of each matching event.

\subsection{Transactions and events}
\label{sec:exchange-events}

Section~\ref{sec:results-transactions} grouped packets into blocks whose
\texttt{transactTime} ranges overlap. Regrouping by the exact \texttt{transactTime} value,
and checking against the end-of-event flag, on 40 sessions, confirms the grouping. A
\texttt{transactTime} group almost never contains or straddles more than one end-of-event
flag, so one \texttt{transactTime} is one matching event. The flag itself cannot serve as
the grouping key on a single-instrument tape, since most NQ events end on a message for
another security on the channel. Under exact grouping the share of transactions that fit
in one packet is the same to within a tenth of a percent, and there are a few percent more
transactions than blocks, the difference being transactions that share a packet. The
results of Section~\ref{sec:results-transactions} are unchanged at the precision
reported.

\findings{
    \item \emph{One \texttt{transactTime} is one matching event}: groups almost never conflict with the end-of-event flag.
    \item \emph{Regrouping changes nothing reported}: the share of transactions that fit in one packet is unchanged.
}

\subsection{Gaps between transactions on the engine and publisher clocks}
\label{sec:exchange-engine}

\begin{table}[H]
\centering
\small
\begin{tabular}{@{}l rrrr@{}}
\toprule
gap between consecutive transactions & $< 7.5\,\mu\mathrm{s}$ & $< 16\,\mu\mathrm{s}$ & $< 32\,\mu\mathrm{s}$ & $p_1$ gap \\
\midrule
engine clock (\texttt{transactTime})    & $17.1\%$ & $25.4\%$ & $35.4\%$ & $0.18\,\mu\mathrm{s}$ \\
\quad range across sessions             & $14.4$--$22.9$ & $22.3$--$31.8$ & $31.5$--$40.7$ & $0.15$--$0.21$ \\
publisher clock (\texttt{sendingTime})  & $1.04\%$ & $17.8\%$ & $32.1\%$ & $7.46\,\mu\mathrm{s}$ \\
\quad range across sessions             & $0.44$--$1.52$ & $15.5$--$22.2$ & $28.6$--$36.9$ & $7.23$--$7.88$ \\
\bottomrule
\end{tabular}
\caption{The same NQ transactions timed by the matching engine and by the publisher;
median session and full range over 40 sessions. What the table shows: one consecutive
transaction in six leaves the engine less than one publisher period after the previous
one, and the tightest one percent are $0.18\,\mu\mathrm{s}$ apart; almost none leave the
publisher that close, and the publisher's tightest one percent are $7.46\,\mu\mathrm{s}$
apart. The difference holds in every session. What follows: the
shortest gap the receiver sees is the publisher period, not a
property of the market's event process; events the engine produces faster than that are
either held back or packed into one packet. Figure~\ref{fig:bars-engine-clock} shows the same numbers as bars.}
\label{tab:engine-clock}
\end{table}

\begin{figure}[H]
\centering
\includegraphics[width=0.72\textwidth]{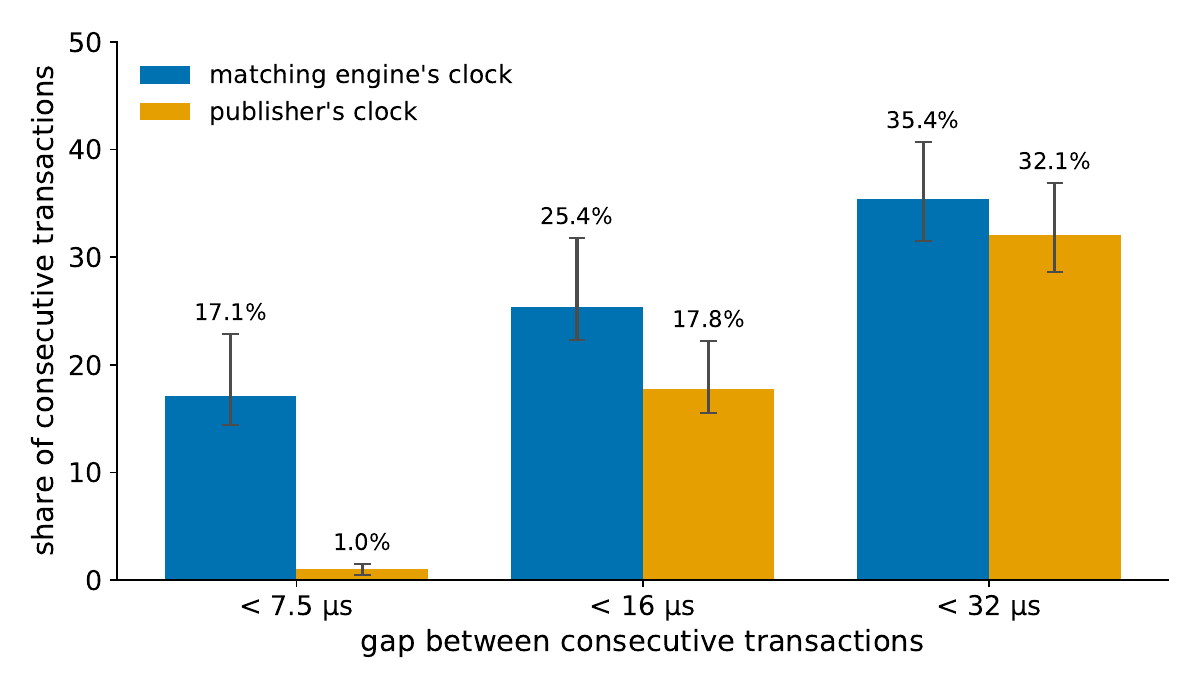}
\caption{Table~\ref{tab:engine-clock} as bars: share of consecutive NQ transactions closer than
$7.5$, $16$ and $32\,\mu\mathrm{s}$, on the engine's clock (blue) and the publisher's
(orange); whiskers span the 40 sessions. Below the publisher period the engine has one
gap in six and the publisher almost none: the publisher does not send faster than its period. By $32\,\mu\mathrm{s}$ the two clocks nearly agree.}
\label{fig:bars-engine-clock}
\end{figure}

\begin{center}
\fbox{\begin{minipage}{0.92\linewidth}
\textbf{Finding.} On NQ, one consecutive transaction in six ($17.1\%$) is processed by the
matching engine less than one publisher period after the previous one, against $1.04\%$
of consecutive packets sent by the publisher. The publisher period of the receiver's
arrival stream is set by the exchange's publisher. The engine's own gaps reach
$0.18\,\mu\mathrm{s}$ at the first percentile, which means the engine went directly from
one event to the next with the next order already waiting: at those moments it was
working through a backlog of orders that reached it within a fraction of a microsecond of
each other.
\end{minipage}}
\end{center}

The receiver does not see these transactions together. Table~\ref{tab:two-clocks} follows
every gap between consecutive NQ transactions from the engine's clock to the
publisher's (the two clocks are defined in Section~\ref{sec:results-transactions}): for the
same two transactions, the table gives the gap at the engine and the gap between the
packets that carry them.

\begin{table}[H]
\centering
\small
\begin{tabular}{@{}l r rrrrrr@{}}
\toprule
 & share of & \multicolumn{6}{c}{gap on the publisher's clock (row \%)} \\
\cmidrule(l){3-8}
gap on the engine's clock & gaps & same packet & $<7.5$ & $7.5$--$10$ & $10$--$16$ & $16$--$100$ & $>100\,\mu\mathrm{s}$ \\
\midrule
$< 1\,\mu\mathrm{s}$            & $5.3\%$  & $15.7$ & $3.7$ & $43.0$ & $16.1$ & $19.9$ & $1.7$ \\
$1$--$7.5\,\mu\mathrm{s}$       & $11.8\%$ & $10.0$ & $3.8$ & $32.1$ & $20.2$ & $29.9$ & $4.1$ \\
$7.5$--$16\,\mu\mathrm{s}$      & $8.5\%$  & $8.8$  & $2.1$ & $21.5$ & $13.6$ & $48.7$ & $5.3$ \\
$16$--$100\,\mu\mathrm{s}$      & $31.4\%$ & $4.9$  & $0.5$ & $4.9$  & $4.0$  & $71.1$ & $14.7$ \\
$> 100\,\mu\mathrm{s}$          & $43.0\%$ & $0.8$  & $0.0$ & $0.1$  & $0.2$  & $15.1$ & $83.7$ \\
\midrule
all gaps                        & $100\%$  & $4.6$  & $1.0$ & $9.5$  & $5.7$  & $37.5$ & $41.6$ \\
\quad of which engine gap $< 7.5\,\mu\mathrm{s}$ & & $43.3$ & $66.2$ & $64.0$ & $56.4$ & $12.2$ & $1.4$ \\
\bottomrule
\end{tabular}
\caption{The gap between consecutive NQ transactions, measured on two clocks for the same two
transactions, over 20 sessions. Rows: the gap on the matching engine's clock. Columns: the
gap between the packets that carry the two transactions, on the publisher's clock, as a
percentage of the row; ``same packet'' means both were sent in one packet. The last row
gives, for each publisher gap, the share of cases in which the engine processed the two
transactions less than $\spub$ apart. What the table shows: when the engine
processes a transaction less than $1\,\mu\mathrm{s}$ after the previous one, the two are
sent in separate packets $84\%$ of the time, most often $7.5$ to $10\,\mu\mathrm{s}$ apart.
Read the other way, most of the short gaps the receiver sees separate transactions that
the engine processed almost together. What follows: transactions do not reach the
receiver together. The publisher spaces them out to its own interval of about
$\spub$ and packs only a few into one packet, so the short gaps that load the
receiver are, for the most part, near-simultaneous engine events delivered back to back at
the publisher period. Figure~\ref{fig:bars-two-clocks} shows the same numbers as bars.}
\label{tab:two-clocks}
\end{table}

\begin{figure}[H]
\centering
\includegraphics[width=0.95\textwidth]{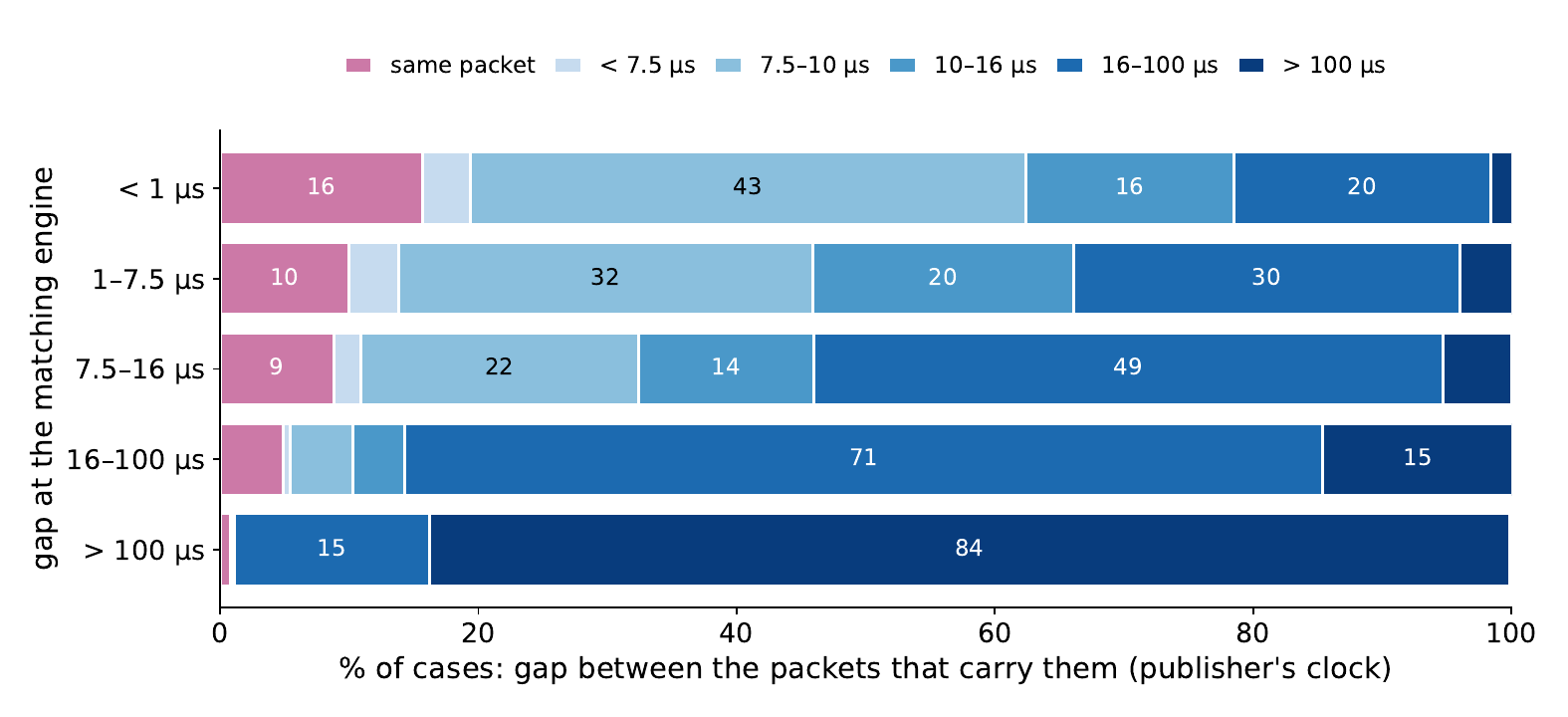}
\caption{Table~\ref{tab:two-clocks} as bars. Each bar is the set of consecutive transactions with
the engine gap on the left, split by the gap between the packets that carry them; the
numbers are percent of the bar. The top bar reads: when the engine processes a
transaction less than $1\,\mu\mathrm{s}$ after the previous one, the two share a packet
$16\%$ of the time (pink) and otherwise arrive in separate packets, most often
$7.5$--$10\,\mu\mathrm{s}$ apart ($43\%$). Near-simultaneous at the engine becomes back to
back at the publisher period at the receiver.}
\label{fig:bars-two-clocks}
\end{figure}

\begin{center}
\fbox{\begin{minipage}{0.92\linewidth}
\textbf{Finding.} Transactions the engine processes a fraction of a microsecond apart do
not reach the receiver together: $84\%$ of them arrive in separate packets, most of them
$7.5$ to $16\,\mu\mathrm{s}$ apart. What the receiver sees is a train of packets at the
publisher period. Its queue grows when its service time exceeds the publisher period, which is
why the threshold of Section~\ref{sec:results-threshold} sits at about
$\spub$. The open market question is why orders reach the engine within a
microsecond of each other; what reaches the receiver is set by the publisher.
\end{minipage}}
\end{center}

\subsection{The publisher as a queue}
\label{sec:exchange-publisher}

\begin{center}
\fbox{\begin{minipage}{0.92\linewidth}
\textbf{Definition (publisher period, $\spub$).} The minimum period between two consecutive
packets from the exchange's market-data publisher: the publisher's service time per
packet, and the inverse of its maximum send rate. On NQ it is about $7.5\,\mu\mathrm{s}$,
measured as the $1$st-percentile gap between packets ($7.43$--$7.53\,\mu\mathrm{s}$,
Section~\ref{sec:results-conditioning}). Two packets never arrive closer together than
$\spub$; during a burst they arrive exactly $\spub$ apart. It is the publisher's period,
not the matching engine's: the engine processes transactions much closer together
(Table~\ref{tab:engine-clock}). The paper refers to it by name throughout.
\end{minipage}}
\end{center}

If the publisher sent events on a fixed schedule, the time from engine processing to
packet send, $\texttt{sendingTime} - \texttt{transactTime}$, would not depend on how busy
the engine is. It does. Its median is about $100\,\mu\mathrm{s}$ and its far tail several
milliseconds, and it rises steeply with the number of transactions the engine processed in
the preceding $50\,\mu\mathrm{s}$ (Table~\ref{tab:publisher}).

\begin{table}[H]
\centering
\small
\begin{tabular}{@{}r r rr r@{}}
\toprule
other NQ transactions in the & share of & \multicolumn{2}{c}{publisher delay} & published sharing a packet \\
\cmidrule(lr){3-4}
previous $50\,\mu\mathrm{s}$ (engine clock) & transactions & $p_{50}$ & $p_{99}$ & with another transaction \\
\midrule
0      & $57.9\%$ & $86\,\mu\mathrm{s}$  & $2.9\,\mathrm{ms}$  & $4.2\%$ \\
1      & $24.1\%$ & $113\,\mu\mathrm{s}$ & $3.5\,\mathrm{ms}$  & $10.0\%$ \\
2--3   & $13.8\%$ & $205\,\mu\mathrm{s}$ & $4.4\,\mathrm{ms}$  & $15.2\%$ \\
4--7   & $3.8\%$  & $406\,\mu\mathrm{s}$ & $6.8\,\mathrm{ms}$  & $20.8\%$ \\
8--15  & $0.4\%$  & $912\,\mu\mathrm{s}$ & $15.0\,\mathrm{ms}$ & $26.6\%$ \\
\bottomrule
\end{tabular}
\caption{Delay from engine processing to packet send, $\texttt{sendingTime} -
\texttt{transactTime}$, by how busy the engine was in the $50\,\mu\mathrm{s}$ before each
transaction, over 20 sessions. What the table shows: the median delay grows about tenfold
from a transaction that follows no other to one that follows eight to fifteen, the
$p_{99}$ grows with it, and the share of transactions sent in a packet with another rises
several times over. What follows: a fixed publishing schedule would keep the delay flat.
Delay and packing growing together is the signature of a queue: the publisher sends about
one packet every $7.5\,\mu\mathrm{s}$, events the engine produces faster than that wait,
and the backlog is packed into fuller packets. The burst measure counts NQ front-month
transactions only, while the publisher serves the whole channel, so the table understates
the load on the publisher. Figure~\ref{fig:bars-publisher} shows the same numbers as bars.}
\label{tab:publisher}
\end{table}

\begin{figure}[H]
\centering
\includegraphics[width=0.9\textwidth]{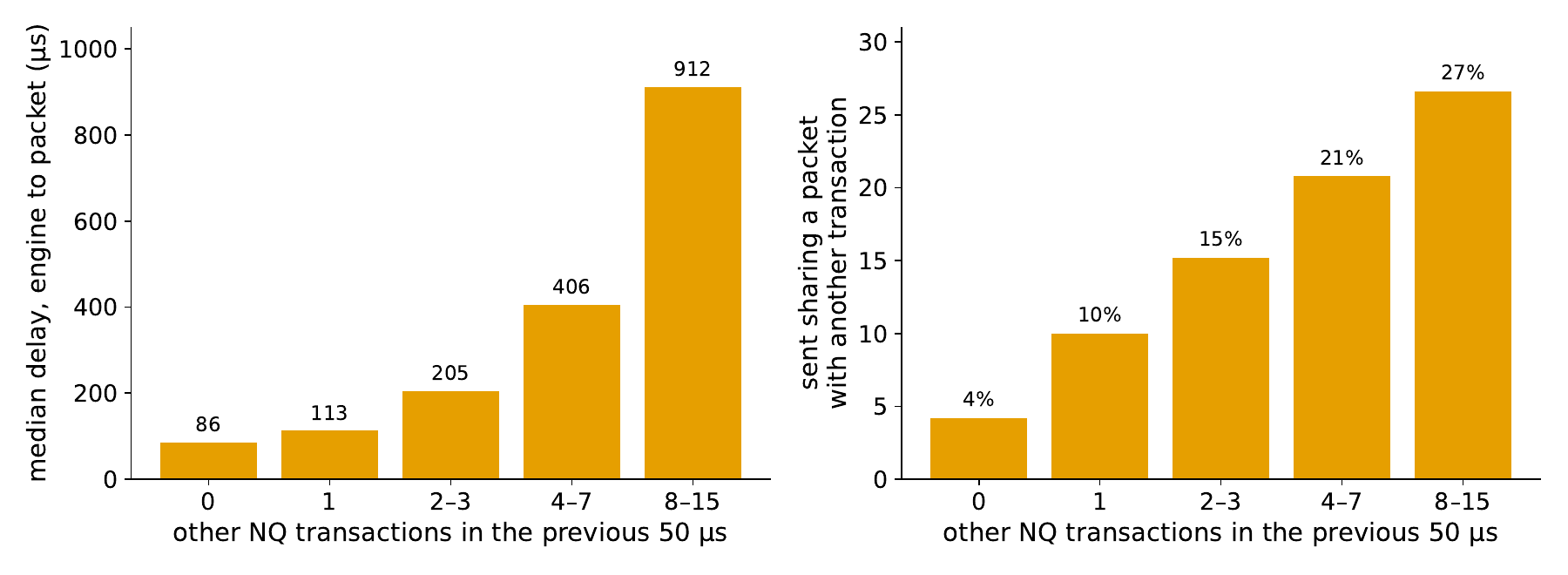}
\caption{Table~\ref{tab:publisher} as bars. Left: median delay from engine to packet grows
tenfold as the engine gets busier. Right: the share of transactions sent in a packet with
another transaction grows from $4\%$ to $27\%$. Both rising together is what a queue at
the publisher looks like.}
\label{fig:bars-publisher}
\end{figure}

A fixed publishing policy would keep the delay flat. Delay and packing growing together
with the engine's load is the signature of a queue with service time $\spub$.

This accounts for the publisher period from first principles. The output of a deterministic server
cannot contain two departures closer together than its service time, and while the server
works through a backlog its departures are spaced by exactly that time. The edge of the
receiver's gap distribution at $\spub$ (Figure~\ref{fig:tx-gaps}) and the
mass of gaps just above it are the publisher draining its backlog.

\findings{
    \item \emph{The delay from engine to packet grows with engine load}, about tenfold from a quiet moment to a burst.
    \item \emph{Packing grows with it}: the share of transactions sent in a shared packet rises several times over.
    \item \emph{The publisher is a queue} with service time $\spub$ per packet; the shortest gap the receiver sees is its period.
}

\subsection{Packet size under publisher backlog}
\label{sec:exchange-drain}

If the publisher's service time is fixed per packet, a backlog can be drained in two ways:
by keeping the period and sending many packets, or by packing more into each packet. The
end-of-event tapes separate the two. On 20 sessions we take four things for each packet: the NQ front-month messages it
carries (its span), the number of distinct transactions in it, the gap to the previous NQ
packet on the publisher's clock, and the publisher delay of its first transaction,
$\texttt{sendingTime} - \texttt{transactTime}$. The last is the measure of how far behind
the publisher is.\footnote{Channel 318 carries other
instruments besides NQ front-month, and the publisher packs their messages into the same
packets. The span here counts NQ front-month messages only, so it is not the packet's
size in bytes, and the publisher's backlog includes work for instruments we do not
decode. Doing this properly means measuring the whole channel and modelling how activity
in one instrument excites activity in the others (cross-excitation, a multivariate
Hawkes process), rather than treating NQ as a stream on its own. On NQ the other
instruments are thin. On ZN they are not: the Treasury futures along the yield curve
trade against each other, a move in one is likely to trigger quotes and trades in the
others within microseconds, and where they share ZN's channel the publisher would carry
those reactions in the same packets. Cross-excitation along the curve is therefore a
candidate cause of the large ZN packets of Section~\ref{sec:crossval-spanrun}. It is
untested here and listed among the open questions of Section~\ref{sec:open-questions}.}

\begin{table}[H]
\centering
\small
\begin{tabular}{@{}l r r r@{}}
\toprule
publisher delay of the & share of & mean span & packets with $\ge 2$ \\
packet's first transaction & packets & (NQ messages) & transactions \\
\midrule
$< 100\,\mu\mathrm{s}$           & $52.3\%$ & $1.000$ & $0.02\%$ \\
$100$--$200\,\mu\mathrm{s}$      & $21.3\%$ & $1.016$ & $1.2\%$ \\
$200$--$500\,\mu\mathrm{s}$      & $14.9\%$ & $1.237$ & $8.9\%$ \\
$0.5$--$1\,\mathrm{ms}$          & $5.9\%$  & $1.278$ & $11.3\%$ \\
$1$--$5\,\mathrm{ms}$            & $5.2\%$  & $1.440$ & $18.2\%$ \\
$> 5\,\mathrm{ms}$               & $0.4\%$  & $1.603$ & $20.6\%$ \\
\bottomrule
\end{tabular}
\caption{Packet size against the publisher's backlog; 20 sessions, $2.43 \times 10^8$
packets. What the table shows: when the publisher is less than $100\,\mu\mathrm{s}$ behind,
packets carry exactly one NQ message; when it is more than $5\,\mathrm{ms}$ behind, they
carry $1.6$ on average and one in five holds two or more transactions. What follows:
packets get bigger when the publisher falls behind. Table~\ref{tab:drain-gap} shows which
packets. Figure~\ref{fig:bars-drain} shows the same numbers as bars.}
\label{tab:drain-delay}
\end{table}

\begin{table}[H]
\centering
\small
\begin{tabular}{@{}l r r r r r@{}}
\toprule
gap to previous & share of & mean span & packets with $\ge 2$ & adjacent on & median other \\
NQ packet & packets & (NQ messages) & transactions & the channel & packets between \\
\midrule
$< 7.5\,\mu\mathrm{s}$             & $1.0\%$  & $1.006$ & $0.08\%$ & $100\%$  & 0 \\
$7.5$--$10\,\mu\mathrm{s}$ (period) & $9.9\%$  & $1.007$ & $0.40\%$ & $100\%$  & 0 \\
$10$--$16\,\mu\mathrm{s}$          & $6.8\%$  & $1.096$ & $2.8\%$  & $95.7\%$ & 0 \\
$16$--$32\,\mu\mathrm{s}$          & $14.2\%$ & $1.141$ & $6.7\%$  & $36.2\%$ & 1 \\
$32$--$100\,\mu\mathrm{s}$         & $25.2\%$ & $1.073$ & $3.4\%$  & $20.6\%$ & 2 \\
$0.1$--$1\,\mathrm{ms}$            & $26.9\%$ & $1.106$ & $4.5\%$  & $12.9\%$ & 4 \\
\bottomrule
\end{tabular}
\caption{Packet size against the gap to the previous NQ packet on the publisher's clock;
same packets as Table~\ref{tab:drain-delay}. The last two columns count the channel's
other packets sent between two consecutive NQ packets, from the packet index of the tape,
which numbers every channel packet that carries order or trade records for any
instrument; packets carrying only other message types are not counted. What the table shows: packets sent at the publisher period carry one NQ message
($1.007$ on average) and in every case follow the previous NQ packet with no other order or
trade packet of the channel between them. Bigger packets appear after gaps of $10$ to $32\,\mu\mathrm{s}$,
which from $16\,\mu\mathrm{s}$ usually contain at least one packet for another instrument.
What follows: the publisher period is the period of the channel's order and
trade packets, not of NQ alone, and the publisher does not answer a backlog with bigger packets sent at the
publisher period.}
\label{tab:drain-gap}
\end{table}

\begin{table}[H]
\centering
\small
\begin{tabular}{@{}l r r@{}}
\toprule
publisher delay (packets at the period only, & mean span & packets with $\ge 2$ \\
gap $7.5$--$10\,\mu\mathrm{s}$) & (NQ messages) & transactions \\
\midrule
$< 100\,\mu\mathrm{s}$          & $1.000$ & $0.04\%$ \\
$100$--$500\,\mu\mathrm{s}$     & $1.006$ & $0.50\%$ \\
$0.5$--$5\,\mathrm{ms}$         & $1.026$ & $1.0\%$ \\
$> 5\,\mathrm{ms}$              & $1.114$ & $2.0\%$ \\
\bottomrule
\end{tabular}
\caption{Packets sent at the publisher period, by the publisher's backlog. What the table shows: even
more than $5\,\mathrm{ms}$ behind, packets sent back-to-back at the publisher period carry $1.1$ NQ
messages on average and $98\%$ of them hold a single transaction. What follows: pacing
holds under load. The publisher keeps its period when it is far behind and does not
switch to fuller packets at the same rate.}
\label{tab:drain-floor}
\end{table}

\begin{figure}[H]
\centering
\includegraphics[width=1.0\textwidth]{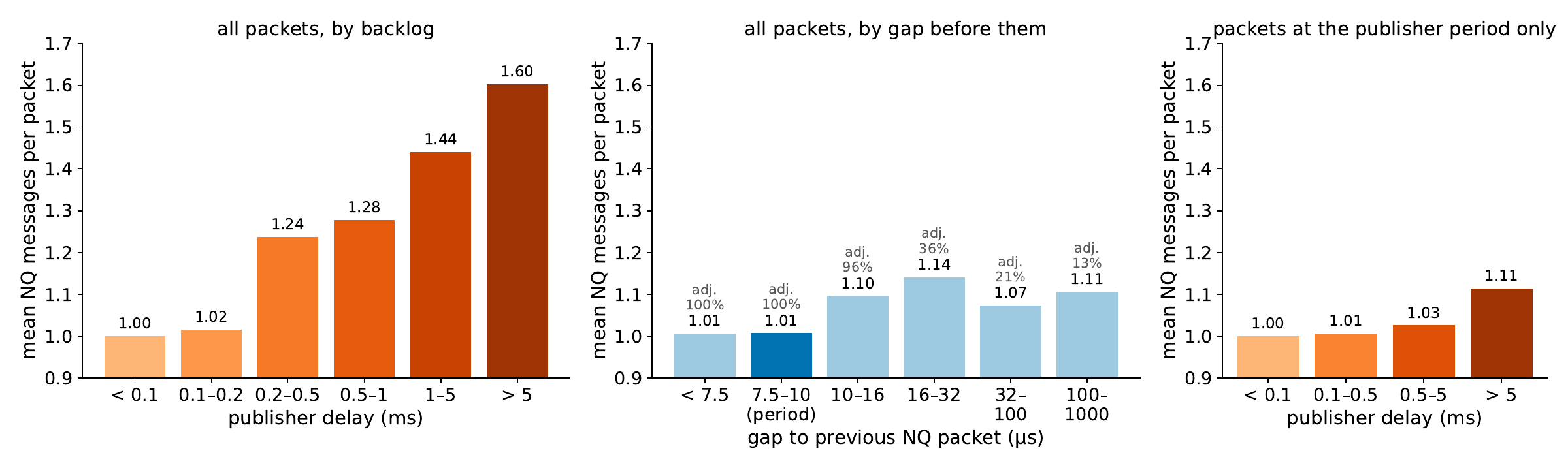}
\caption{Tables~\ref{tab:drain-delay}--\ref{tab:drain-floor} as bars: mean NQ messages per
packet. Left: packets get bigger as the publisher falls further behind. Middle: packets
sent at the publisher period, $7.5$--$10\,\mu\mathrm{s}$ after the previous NQ packet (dark blue),
are single-message packets; the grey ``adj.'' figures are the share with no other order or
trade packet of the channel between the two NQ packets, $100\%$ at the period. Right: packets at the period stay single-message even when the publisher is more than
$5\,\mathrm{ms}$ behind. The publisher drains a backlog with more packets, not bigger
ones.}
\label{fig:bars-drain}
\end{figure}

Three things follow. First, pacing holds under load: even when the publisher is
milliseconds behind, the packets it sends back-to-back are single-message packets at about
publisher period (Table~\ref{tab:drain-floor}), and that period holds across
the channel's order and trade packets (Table~\ref{tab:drain-gap}). Second, coalescing happens elsewhere. Bigger
NQ packets come after gaps of $10$ to $32\,\mu\mathrm{s}$, so when packets are bigger they
also leave more slowly, which fits a service time per packet that grows with its size, the
span-dependent service of Section~\ref{sec:crossval-spanrun}. Third, a receiver sees a
burst at the engine as more packets, not bigger ones: a long run of small packets at the
publisher period, with an occasional larger packet after a longer gap. This is the arrival pattern
that makes the receiver's service time, compared with $\spub$, the quantity
that sets its tail (Section~\ref{sec:exchange-tandem}).

\findings{
    \item \emph{Packets get bigger as the publisher falls behind}, from one NQ message each to about one and a half when it is milliseconds behind.
    \item \emph{Packets sent at the publisher period stay single-message even far behind}: the publisher drains a backlog with more packets, not bigger ones.
    \item \emph{The publisher period holds across the channel's order and trade packets}: packets at the period always follow the previous NQ packet with no other such packet between.
    \item \emph{Only NQ front-month messages are counted}; the full channel's packet size is an open question.
}

\subsection{The receiver as the second stage of a tandem starting at the exchange}
\label{sec:exchange-tandem}

The publisher and the receiver are two servers in series. When a burst at the engine
builds a backlog at the publisher, the publisher emits packets back-to-back at its own
period $\spub$. A receiver whose service time is shorter
than that keeps up and queues nothing; a receiver whose service time is longer falls
behind by the difference on every packet of the drain, and its backlog grows for as long
as the drain lasts. This is the threshold of Section~\ref{sec:results-threshold}, now
explained: the service time at which a receiver's tail appears equals the service time of
the exchange stage upstream of it. It is the tandem of Theorem~\ref{thm:reduction} read in
the other direction: a downstream stage slower than the upstream one is the bottleneck,
and it is the receiver's stage, not the publisher's, that sets the tail the receiver
pays. The two regimes of Section~\ref{sec:results-nulls} follow. At service times just
above the publisher period, a short drain is enough to build a queue, and what matters is how often
the publisher drains, which the gap shuffle preserves. At long service times the queue
needs a sustained drain, which is a run of engine bursts, and that is what the shuffle
destroys.

\findings{
    \item \emph{The receiver's tail threshold is the publisher's service time}: a receiver faster than $\spub$ per packet keeps up with any train the publisher sends; a slower one falls behind on every packet of the train.
    \item \emph{This explains the two regimes of the nulls}: short service times queue on each short drain, long ones on sustained runs of engine bursts.
}

\subsection{Bursts at the matching engine}
\label{sec:exchange-bursts}

Section~\ref{sec:exchange-engine} shows that the engine often processes events a fraction
of a microsecond apart, which means orders were waiting for it; the publisher then
delivers them to the receiver one publisher period apart
(Table~\ref{tab:two-clocks}). Why orders reach the engine within a fraction of a
microsecond of each other is the market question of Section~\ref{sec:concl-cannot}. The $0.18\,\mu\mathrm{s}$
spacing fits many participants sending orders in response to the same public signal, the
\emph{races} that \citet{aquilina2022arms} measured on the London Stock Exchange with
order-entry times the CME public feed does not carry. It is consistent with that
mechanism and does not prove it: a reaction to the previous transaction that is faster than
the engine's backlog would look the same from here. The same observation has been made on
Nasdaq equities from event times alone: \citet{noble2026realitygap} find a sharp mode in
inter-event times at the exchange round trip and about a quarter of events too fast to be
reactions to the previous one, which they read as correlated responses to a shared
signal.

Table~\ref{tab:tight-pairs} tests what transactions separated by a tight gap are. It
takes every gap between consecutive transactions on the engine clock in 40 sessions
($512$ million gaps) and compares the two transactions on either side of a tight gap
(below $16\,\mu\mathrm{s}$) with those on either side of a loose gap ($0.1$ to
$1\,\mathrm{ms}$). A transaction is labelled by the first matching kind among trade,
delete, new and change. Transactions across a tight gap are ordinary book traffic: their
mix is close to that across a loose gap, and only one tight gap in fifty follows a trade.
After a trade, the two differ: across a tight gap the next transaction is another trade
about five times as often as across a loose gap, and a delete about half again as often. A trade followed within microseconds by a further trade
or a cancellation is the pattern a race would leave, and it is more frequent across tight
gaps.
The table cannot measure how much of the engine's activity is racing: the losers of a race
to take liquidity are rejected and never reach the public feed, and a race to cancel
appears only as deletes, which dominate across tight and loose gaps alike. What the feed
shows is that transactions separated by tight gaps are mostly quote placements and
cancellations, and that those following a trade carry the race signature.

\begin{table}[ht]
\centering
\small
\begin{tabular}{lrr}
\hline
Consecutive transactions, engine clock & gap $<16\,\mu\mathrm{s}$ & gap $0.1$--$1\,\mathrm{ms}$ \\
\hline
share of all gaps & 25.8\% & 27.2\% \\
first is a trade & 2.0\% & 3.5\% \\
first is a delete & 41.2\% & 39.8\% \\
first is a new order & 46.0\% & 39.5\% \\
first is a change & 10.8\% & 17.2\% \\
\hline
after a trade: next is a trade & 11.7\% & 2.2\% \\
after a trade: next is a delete & 37.4\% & 23.5\% \\
after a trade: next is a new order & 45.5\% & 68.2\% \\
\hline
same side & 52.5\% & 53.9\% \\
same price & 20.5\% & 17.3\% \\
\hline
\end{tabular}
\caption{The two consecutive transactions on either side of a gap, by gap length on the
engine clock. Forty sessions, $512$ million gaps; the middle bin
($16$--$100\,\mu\mathrm{s}$, $31.5\%$ of gaps) is omitted. ``First'' is the earlier
transaction, ``next'' the later. Transactions across a tight gap are mostly quote
placements and cancellations, in about the same mix as across a loose gap; only $2\%$ of
tight gaps follow a trade. After a trade, a tight follow-up
is five times as likely to be another trade and more likely to be a cancellation, the
signature of a race. Rejected race orders do not appear on the public feed, so the table
does not measure the share of racing. Figure~\ref{fig:bars-tight} shows the same numbers as bars.}
\label{tab:tight-pairs}
\end{table}

\begin{figure}[H]
\centering
\includegraphics[width=1.0\textwidth]{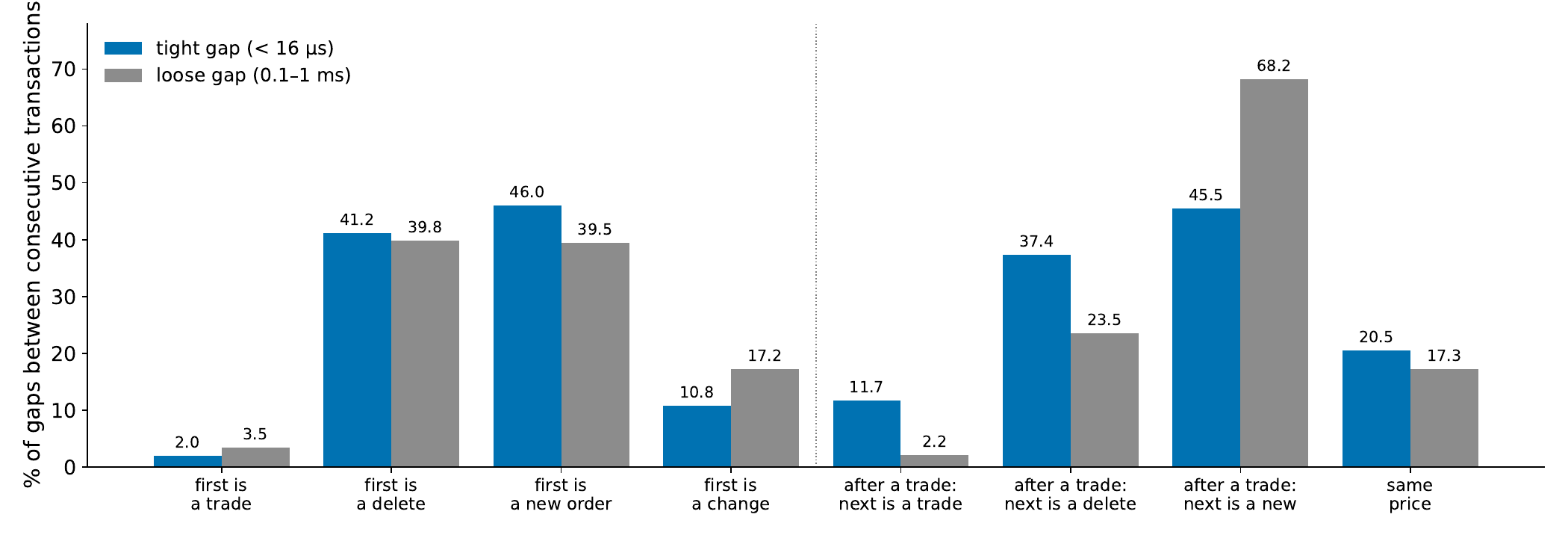}
\caption{Table~\ref{tab:tight-pairs} as bars: the kind of transaction on either side of a tight
gap (blue, under $16\,\mu\mathrm{s}$) and of a loose gap (grey, $0.1$--$1\,\mathrm{ms}$).
Left of the dotted line the two mixes are alike: tight gaps are ordinary quote traffic.
Right of it, after a trade the next transaction is far more often another trade or a
delete across a tight gap, the pattern a race would leave.}
\label{fig:bars-tight}
\end{figure}

The fitted self-excitation is a different phenomenon at a different scale. The Hawkes fit
on the engine clock gives the same branching ratio as on the publisher's clock, about
$0.8$, and a kernel that decays over about $160\,\mu\mathrm{s}$, so it places only a few
percent of its triggering within one publisher period and nearly all of it within a
millisecond. That is the scale of the mode near $100\,\mu\mathrm{s}$ in the gap distribution
(Figure~\ref{fig:tx-gaps}) and of the runs that carry the tail from $64\,\mu\mathrm{s}$. The
self-excitation the fit measures is one transaction making another more likely over the
following hundred microseconds or so; it accounts for the runs and does not account for
transactions a fraction of a microsecond apart, which reach the engine together rather
than one triggering the next.

\findings{
    \item \emph{Transactions across a short engine gap are ordinary quote traffic}: the mix of new orders and deletes matches that across long gaps, and few short gaps follow a trade.
    \item \emph{After a trade, the race pattern shows}: the next transaction is far more often another trade or a delete. The share of racing is not measurable on the public feed.
    \item \emph{The fitted self-excitation acts over about $160\,\mu\mathrm{s}$} and explains the runs, not transactions a fraction of a microsecond apart.
}

\subsection{Conjecture: the publisher's capacity as an information-side speed limit}
\label{sec:exchange-conjecture}

\begin{center}
\fbox{\begin{minipage}{0.92\linewidth}
\textbf{Conjecture.} Because the publisher sends at most one packet per
$\spub$ and queues whatever the engine produces faster than that, it bounds
the rate at which market information reaches participants, and during bursts it delivers that information late: typically by hundreds of
microseconds, and by milliseconds in the tail (Table~\ref{tab:publisher}). Any feedback
loop in which participants react to published data therefore cannot run faster than the
publisher, and a cascade driven by public information unfolds more slowly than the
engine could process it. In this sense the publisher acts as a speed limit on the
information side of the market, whether or not it was designed as one, and protects
downstream receivers from the engine's peak rate.
\end{minipage}}
\end{center}

Three qualifications bound the conjecture. First, delaying feedback does not reduce the
branching ratio of a self-exciting process: each event still triggers the same number of
follow-on events on average, later, so the publisher cannot make a critical cascade
subcritical; it bounds the cascade's speed, not its size. Second, it does not bound
simultaneous reactions to a common signal that reach the engine through order entry, which
are what the engine-clock gaps of Table~\ref{tab:engine-clock} suggest; those do not pass
through the publisher before they are submitted. Third, whether the publisher's capacity
is a design choice or an engineering limit cannot be read from the feed; the queueing
signature of Table~\ref{tab:publisher} points to a limit. The conjecture is related to
proposals that slow markets deliberately to blunt speed races, such as frequent batch
auctions \citep{budish2015fba}; the difference is that here the delay is on the
information leaving the exchange, not on orders entering it.

\subsection{Summary of the findings on the exchange side}
\label{sec:exchange-summary}

The section follows each transaction from the matching engine, through the publisher, to
the receiver. What it finds, in order:

\findings{
    \item \emph{The matching engine processes transactions far closer together than the receiver sees them.} On the engine's clock one consecutive transaction in six is under $\spub$ apart, the tightest a fraction of a microsecond; on the publisher's clock $1\%$ are (Table~\ref{tab:engine-clock}).
    \item \emph{The publisher is a queue with service time $\spub$ per packet.} Its delay grows about tenfold with engine load (Table~\ref{tab:publisher}).
    \item \emph{Near-simultaneous transactions reach the receiver as a train of packets one period apart, not together.} Of transactions processed within $1\,\mu\mathrm{s}$ of the previous one, $84\%$ are sent in a separate packet, most $7.5$--$16\,\mu\mathrm{s}$ later (Table~\ref{tab:two-clocks}).
    \item \emph{Under backlog the publisher sends more packets, not bigger ones.} Packets grow when it is far behind, but those sent at the publisher period stay single-message (Tables~\ref{tab:drain-delay}--\ref{tab:drain-floor}).
    \item \emph{The receiver is the second queue in a tandem that starts at the exchange.} Its tail threshold is the publisher period, which is why the threshold sits at $\spub$ and does not move with market load.
    \item \emph{Transactions across a tight engine gap are ordinary quote traffic; after a trade, the race pattern shows} (Table~\ref{tab:tight-pairs}). Why orders reach the engine within a microsecond of each other is left open.
    \item \emph{The fitted self-excitation acts over about $160\,\mu\mathrm{s}$.} It explains the runs that carry the tail at long service times, not the sub-microsecond engine gaps.
    \item \emph{Conjecture}: the publisher's capacity acts as a speed limit on information leaving the exchange (Section~\ref{sec:exchange-conjecture}).
}

\section{The analytical framework checked against the simulations}
\label{sec:evaluation}

This section checks the closed-form predictions of the analytical framework
(Section~\ref{sec:theory}) against the results of the simulations on the recorded corpus
(Section~\ref{sec:results}): each proposition or theorem is stated with the prediction it
makes and the simulated numbers that test it.

\paragraph{Proposition~\ref{thm:burst} (burst-limit throughput).} The single-message
identity $p_{50}(N) = T + (N-1) h$ is predicted exactly by Proposition~\ref{thm:burst} for
any arrival law when the median message finds the queue empty. Table~\ref{tab:main-corpus}
confirms it in every cell of the Poisson stream and in all but three cells of the real
stream, to the reporting precision; the three exceptions are the longest service times,
where the median message itself waits (Section~\ref{sec:results-median}).

\paragraph{Theorem~\ref{thm:reduction} (reduction and scaling bound).} The prediction is
$\Delta(N) \le \Delta(1)/N$ on every window, and more strongly $W_j^{(N)} \le W_j^{(1)}/N +
(N-1)h$ for every individual message. The corpus-median ratios in Table~\ref{tab:delta}
are at or below $1/N$ in every cell, with a wide margin: the closest any cell comes is about
a third of the bound. The margin is largest where a split drives the per-stage service
below the publisher period and removes the tail outright rather than dividing it, which is
why the fitted decay exponent is about three rather than one. The per-message form was
checked directly by the simulator on every window, every service time, all four
packet-level streams and every stage count: no message violates it. The reduction identity,
Equation~\eqref{eq:reduction}, is visible in Table~\ref{tab:delta} without any fitting:
$\Delta$ is constant along every diagonal of constant $T/N$, so the tandem's tail is the
single-server tail at $s_{\max}$ on real arrivals exactly as on synthetic ones
(Appendices~\ref{app:check} and~\ref{app:unequal}).

\paragraph{Corollary~\ref{cor:poisson} (Poisson null).} The prediction is
$p_{q}^{\mathrm{P}}(N) = T + (N-1)h$ wherever the per-stage utilisation satisfies
$\rho / N < 1 - q$, and a visible wait wherever it does not. Table~\ref{tab:main-corpus}
matches both halves: the Poisson $p_{99}$ exceeds the no-queue latency in exactly the six
cells whose per-stage utilisation is above $1\%$, all at long service times, and in no other
cell. The Poisson $p_{99.9}$ likewise exceeds it exactly where the per-stage utilisation is
above $0.1\%$. The Poisson stream is therefore a working queue whose small wait appears
where theory says it must, and splitting under it costs $(N-1)h$ at every reported
quantile.

\paragraph{The renewal null against the cluster picture.} The gap shuffle of
Section~\ref{sec:results-nulls} is a renewal stream with the real gap distribution, so
under it the chance that a short gap is followed by another equals the share of short gaps
$F(T)$, and runs of short gaps are as long as independent draws make them. The cluster
picture of Section~\ref{sec:theory} predicts that the real stream has longer runs, and more
so at longer service times, where many members of one cluster fall within one service
time. Table~\ref{tab:renewal} compares the two on 254 windows from 20 sessions. The real
runs are longer than the renewal ones at every service time, and the excess grows with $T$
from $16\,\mu\mathrm{s}$: at $16\,\mu\mathrm{s}$ the runs are within a few percent of
renewal and the shuffle keeps most of the tail, while at $128\,\mu\mathrm{s}$ they are a
quarter longer and the shuffle keeps little. At $8\,\mu\mathrm{s}$ short gaps are strongly
bunched, six times the renewal chance, yet the shuffle keeps most of the tail there too,
because the tail at that service time is a fraction of a microsecond and consists of
single waits that need no run.

\begin{table}[H]
\centering
\small
\begin{tabular}{@{}lrrrrr@{}}
\toprule
$T$ (\si{\micro\second}) & 8 & 16 & 32 & 64 & 128 \\
\midrule
chance a short gap is followed by another, real      & 0.120 & 0.209 & 0.363 & 0.534 & 0.725 \\
same, renewal stream (share of short gaps $F(T)$)  & 0.020 & 0.183 & 0.325 & 0.456 & 0.655 \\
mean run of short gaps, real                          & 1.14 & 1.26 & 1.57 & 2.14 & 3.63 \\
mean run of short gaps, renewal                       & 1.02 & 1.22 & 1.48 & 1.84 & 2.90 \\
\bottomrule
\end{tabular}
\caption{Runs of short gaps (gaps shorter than $T$) on the real stream against a renewal
stream with the same gaps in random order; per-window medians over 254 windows from 20
sessions. What the table shows: short gaps follow one another more often than chance at
every service time, and the real runs grow longer than the renewal ones as $T$ grows. What
follows: the runs are what the gap shuffle destroys, which is why it removes more of the
tail at long service times (Table~\ref{tab:nulls}).}
\label{tab:renewal}
\end{table}

\paragraph{The cluster representation against Section~\ref{sec:results-transactions}.}
The Hawkes--Oakes picture of Section~\ref{sec:theory-cluster} is one of immigrant events
each generating a cluster of offspring; Section~\ref{sec:results-transactions} identifies
the events as matching-engine transactions. Two of its findings are what the picture
predicts if the feed relays events one packet each: the branching ratio fitted on
transaction starts equals the one fitted on packets ($0.795$ against $0.798$), and
merging each transaction into one packet or decoupling transaction sizes from their
timing leaves the tail unchanged (Table~\ref{tab:tx-arms}). The corpus also shows that
the exchange does not add clustering of its own, since the pieces of a split transaction
are paced at $15\,\mu\mathrm{s}$, above the publisher period. What the corpus does not test is the
cluster-size law itself: the Borel distribution of Appendix~\ref{app:borel} is a
statement about offspring counts, which the data do not label, and its check in
Appendix~\ref{app:borel} remains indirect.

\paragraph{Summary.} All three analytical statements hold on the full corpus. The
per-message bound of Theorem~\ref{thm:reduction} holds on every message of every window
under all four packet-level streams, and the tail falls far faster than the bound requires:
where a split takes each stage below the publisher period, it removes the tail rather
than dividing it. The renewal stream behaves as the cluster picture says it should,
keeping the one-gap tail at short service times and losing the run-driven tail at long
ones, and the transaction analysis places the clusters in the matching engine's event
stream.

\findings{
    \item \emph{Proposition~\ref{thm:burst} holds}: median $= T + (N-1)h$ in all but three cells, the exceptions at the longest service times.
    \item \emph{Theorem~\ref{thm:reduction} holds on every message of every window}, and the tail falls far faster than its $1/N$ bound.
    \item \emph{Corollary~\ref{cor:poisson} holds}: the Poisson wait appears exactly where utilisation per stage exceeds $1 - q$.
    \item \emph{Runs of short gaps are longer than a renewal stream's}, increasingly so with $T$, which is why shuffling the gaps removes more of the tail at long service times.
}

\clearpage

\section{Cross-validation against a live production receiver}
\label{sec:crossval}

The results of Part~II so far are measurements on recorded data: pcap-derived message
tapes replayed offline. That design gives scale and repeatability at the cost of external
validity. A replayed tape cannot show whether the fitted clustering survives the network
path to a receiver, whether it is specific to the 2025--2026 sample, or whether the
service times the sweep assumes correspond to those of a production receiver.

This section addresses all three with instrumentation of the \kasparhft{} production system. The two
measurements were built for different purposes and share only the exchange they observe.
(The live measurement was taken seven months after the corpus ends, so it also tests
whether the clustering is specific to the period the corpus covers.)
One limit is stated at the outset. The live per-message record carries the receive and
publish timestamps, the ring depth and the position in the datagram, but not the
message's \texttt{transactTime}. The transaction-level identification of
Section~\ref{sec:results-transactions} is therefore a corpus result that the live data
cannot repeat. What the live data cross-validate is the packet process; the transaction
result rests on the corpus alone.

\subsection{The recorded corpus and the live measurement compared}
\label{sec:crossval-setup}

The live measurement instruments the \kasparhft{} market-data path in production: a
timestamp \texttt{t0} when a UDP datagram is handed up from the socket, a timestamp
\texttt{t1} when the resulting book update is published, and one 16-byte record per
admitted message carrying the pair together with the ring-buffer depth and the message's
position inside its datagram. It covers 6.81 million messages on six streams --- ES, NQ
and ZN futures, book and trade --- over 53 minutes of a single afternoon. Its latency
results are reported in \citet{mayeski_fastsend}; we draw on them here, and on the
underlying per-message records, as a control on Part~II.

The live receiver was measured in two configurations. In the \emph{two-stage receiver}, the
one these records come from, a socket-reader actor drains the socket and hands each
datagram through an asynchronous mailbox to a decoder actor on another thread, which
decodes it and updates the book (the wiring is described in Section~\ref{sec:kaspar}). Its
socket-to-book time therefore covers the decode-and-book work and the cross-thread
hand-off. In the \emph{one-thread receiver} the socket-reader thread decodes each packet and
updates the book itself, with no hand-off. The cross-validation of this section and the
analyses built on it use the two-stage receiver, and their conclusions are stated for it;
the one-thread receiver is reported in Section~\ref{sec:crossval-onethread}.

\begin{table}[H]
\centering
\small
\begin{tabular}{@{}lll@{}}
\toprule
 & Part~II corpus & production measurement \\
\midrule
source        & pcap tapes, replayed offline  & live multicast, production recorder \\
period        & 2025-01 to 2026-02            & 2026-09-16 \\
instruments   & NQ front month                & ES, NQ, ZN $\times$ book, trade \\
arrival clock & CME publisher \texttt{sendingTime} & receiver \texttt{t0} at socket read \\
estimator     & per-window, median of 3512 windows & pooled over one 53-min capture \\
sample        & 3512 windows, 276 sessions    & 6.81 M messages, 6 streams \\
\bottomrule
\end{tabular}
\caption{The two measurements compared in this section: different instruments, a period
that does not overlap, a different arrival clock, and a different estimator.}
\label{tab:crossval-setup}
\end{table}

The arrival-clock difference is the most significant. Part~II keys arrivals on
\texttt{sendingTime}, the timestamp CME's publisher writes into the packet header as the
packet leaves the exchange. The live measurement keys them on \texttt{t0}, taken by our
own process after the packet has crossed the network, the NIC and the kernel. If the
clustering reported in Section~\ref{sec:setup-hawkes} were an artefact of exchange-side
emission scheduling, the intervening path, with its own queueing, coalescing and
interrupt moderation, would be expected to attenuate it. It does not; and
Section~\ref{sec:results-transactions} rules out the exchange's packetisation as the
source from the other side, by showing that the clustering is already present in the
matching engine's transaction stream.

\subsection{Clustering statistics in the corpus and on the live receiver}
\label{sec:crossval-agree}

This section checks that the bursty arrival pattern the paper is built on is real, and
not something produced by replaying recordings. It takes four numbers that each measure
how far packet arrivals are from evenly spread. The first is the share of gaps between
packets that are tiny, under a tenth of the average gap. The second is how much the
variability of packet counts drops when the order of the packets is scrambled, which
removes any bunching. The third is a measure of whether busy stretches persist rather
than fading at once. The fourth is how strongly one event tends to trigger the next.
Each of the four is computed twice, independently: once on the recorded corpus that the
rest of the paper uses, and once on the live production receiver, on different instruments, with a different clock. Three of the four match directly, to
within the spread seen across the live streams. The fourth is measured in two different
ways, one an estimate and one an upper bound, and the estimate falls inside the bound.
The conclusion is that the bursts are a property of the feed as received, not of
replay: they show up at the same strength whether one looks at a recording or at the
live wire. Where they come from is the question Section~\ref{sec:results-transactions}
answers on the corpus: the matching engine's transaction stream.

Table~\ref{tab:crossval-stats} puts the four statistics that both measurements compute
side by side. The NQ column is the like-for-like comparison, since Part~II's corpus is NQ
front month; the third column gives the range across the three live book streams.

\begin{table}[H]
\centering
\small
\begin{tabular}{@{}lrrr@{}}
\toprule
 & Part~II (NQ, 3512 win) & live NQ book & live, three books \\
\midrule
$P(\text{gap} < \text{mean}/10)$      & $73.2\%$ & $63.1\%$ & $63.1$--$84.8\%$ \\
Fano(5\,s) collapse under shuffle     & $14.8\times$ & $12.2\times$ & $12.2$--$18.5\times$ \\
Hurst exponent $H$                    & $0.740$ & $0.702$ & $0.702$--$0.743$ \\
branching ratio $n$                   & $0.798$ (MLE) & \multicolumn{2}{c}{$\le 0.850$--$0.972$ (Fano ceiling)} \\
\bottomrule
\end{tabular}
\caption{The same four Poisson-rejection statistics, computed independently on recorded
and on live data. The Poisson values are $9.52\%$, $1\times$, $0.5$ and $0$ respectively. Figure~\ref{fig:bars-crossval-stats} shows the same numbers as bars.}
\label{tab:crossval-stats}
\end{table}

\begin{figure}[H]
\centering
\includegraphics[width=1.0\textwidth]{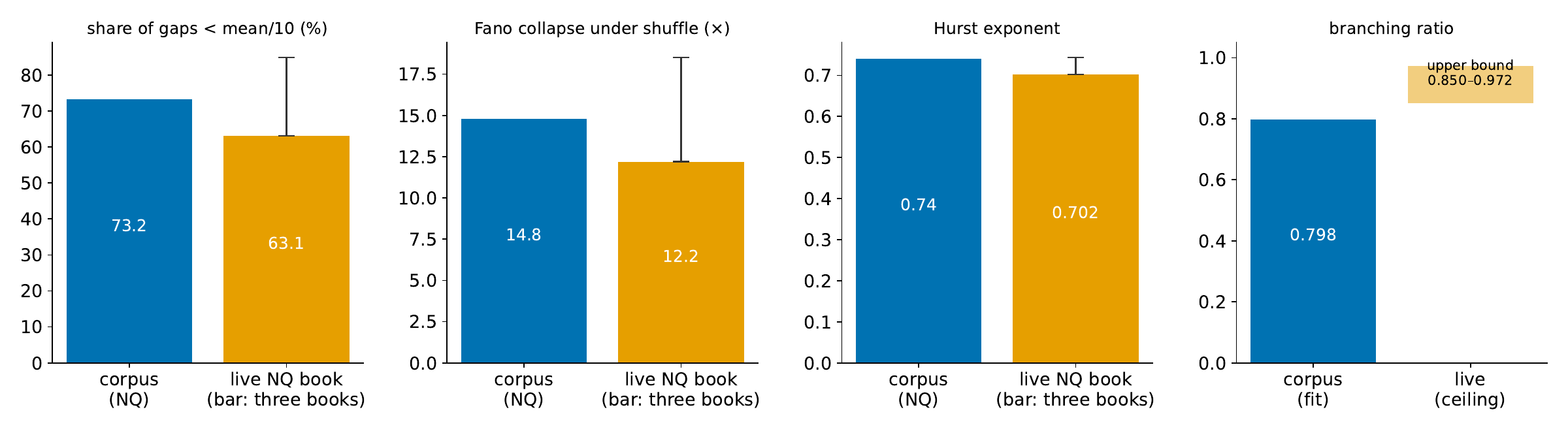}
\caption{Table~\ref{tab:crossval-stats} as bars: the clustering statistics on the recorded corpus
(blue) and on the live receiver's NQ book stream (orange; whiskers span the three live book
streams). The share of tiny gaps, the collapse of the Fano factor under shuffling and the
Hurst exponent agree within the spread of the live streams. The corpus branching ratio
lies below every upper bound the live streams give (shaded).}
\label{fig:bars-crossval-stats}
\end{figure}

Three of the four are direct comparisons and agree to within the spread across live
streams. The fourth compares different estimators. Part~II fits the branching ratio by
maximum likelihood, per window, using the Ogata recursion
(Section~\ref{sec:setup-hawkes}); the live measurement does not fit a Hawkes process, and
bounds $n$ from above by inverting the Fano ceiling, $n \le 1 - 1/\sqrt{F}$. The fitted value of about $0.8$ lies within every bound the live streams imply. The bound is loose, and the agreement is weak evidence on
its own; its value is that it is independent evidence and is consistent.

The arrival process modelled in Part~II is therefore not an artefact of replay, of the
sample period, or of the exchange-side clock; it is present at the same strength in a live receiver on other instruments.

\findings{
    \item \emph{The clustering is not an artefact of replay}: the share of tiny gaps, the collapse under shuffling and the Hurst exponent agree between the corpus and the live receiver, on other instruments and a different clock.
    \item \emph{The branching ratio} fitted on the corpus, about $0.8$, lies inside every upper bound the live streams give.
}

\subsection{The two-stage receiver's service floor and the tail threshold}
\label{sec:crossval-floor}

The live system supplies one number that a recording cannot: how long the two-stage
receiver takes to carry a message from the socket read to the book publish when nothing is
queued in front of it. That time covers the decode-and-book work and the hand-off from the
socket-reader thread to the decoder thread; we call it the receiver's \emph{service floor}.
This section measures it three independent ways and finds they agree, at about seven
microseconds. The notable result is where that number falls. Part~II found, from the
recorded corpus, that a tail appears once a stage takes longer than the tightest gap
between packets, about seven and a half microseconds. The two-stage receiver sits just
under that, within a few hundred nanoseconds, by two measurements that share nothing.
It therefore operates right at the edge of the tail. The section then
shows what pushes it over: a packet carrying a single message is handled below the
threshold, and a packet carrying two or more at or above it, on every stream measured.
Whether a packet's service crosses the threshold is decided by whether it carries one
message or more.

Isolating messages that arrive to an empty ring and first in their datagram,
so that neither queueing nor in-packet decode contributes, gives the service floor
directly; the paper calls these the \emph{empty-queue messages}. Three estimators agree (Table~\ref{tab:crossval-floor}):

\begin{table}[H]
\centering
\small
\begin{tabular}{@{}lrrr@{}}
\toprule
method & ES & NQ & ZN \\
\midrule
direct median, empty ring and first in packet & $7.01$ & $7.24$ & $6.89$ \\
fitted intercept of the median position ladder & $6.98$ & $7.23$ & $6.83$ \\
single-message packets during an FOMC release & $7.1$ & $7.5$ & $7.0$ \\
\bottomrule
\end{tabular}
\caption{Three independent estimates of the service floor of the two-stage \kasparhft{} receiver (socket to book), in
$\mu$s, on the book streams. The three differ in sample, estimator and load regime. The
third was taken while the packet rate rose $4.6$ to $9.8\times$ within a second; the service floor
rose by at most $0.27\,\mu\mathrm{s}$ against either other estimator. Figure~\ref{fig:bars-crossval-floor} shows the same numbers as bars.}
\label{tab:crossval-floor}
\end{table}

\begin{figure}[H]
\centering
\includegraphics[width=0.8\textwidth]{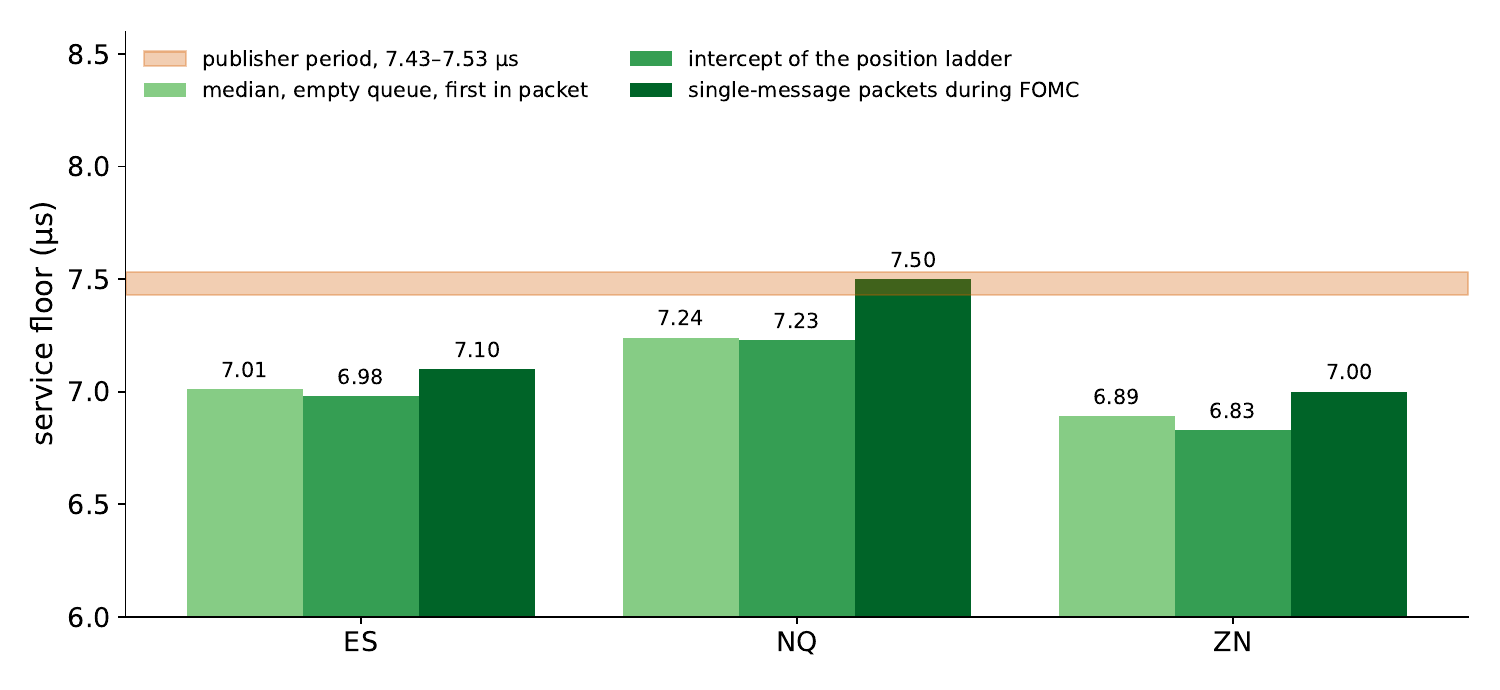}
\caption{Table~\ref{tab:crossval-floor} as bars: the two-stage receiver's service floor by three
independent estimates on each instrument, against the tail threshold measured on the
corpus (shaded band). All nine estimates lie in or just below the band: the two-stage
receiver runs right at the edge of the tail.}
\label{fig:bars-crossval-floor}
\end{figure}

Across all six streams the fitted service floor lies within half a microsecond of
$7\,\mu\mathrm{s}$, on six multicast channels on one host and one afternoon.

Compare this with Section~\ref{sec:results-conditioning}. The service time at which a
stage on this feed begins to build a tail is the publisher period, the tightest gap
between packets, $\spub$, the same at every level of market activity;
the refined sweep of Section~\ref{sec:results-threshold} puts the switch-on between $7$ and
$8\,\mu\mathrm{s}$. The two-stage receiver's service floor is about $7\,\mu\mathrm{s}$.

The two-stage receiver operates at the threshold, to within a few hundred nanoseconds,
by two measurements made on different data with different estimators. The agreement is
close to definitional for a $p_{99}$ target: a server that never varies its service time
leaves the no-queue latency at $p_{99}$ once about one packet in a hundred arrives within
one service time of the previous one, and the publisher period is by definition the gap
that one packet in a hundred sits at. For a $p_{99.9}$ target the relevant gap is the one
that one packet in a thousand sits at, which is shorter, and Table~\ref{tab:main-corpus}
already shows a $p_{99.9}$ excess at the sweep point just above the publisher period.

The live measurement locates the crossing. Per-packet service is the service floor plus one slope for each
message beyond the first (Section~\ref{sec:crossval-span}), so the service time of a
datagram carrying $\text{span}$ messages is $\text{service floor} + (\text{span}-1) \times
\text{slope}$. Evaluating that on each stream's own fitted pair:

\begin{table}[H]
\centering
\small
\begin{tabular}{@{}lrrr@{}}
\toprule
 & span 1 & span 2 & span 3 \\
\midrule
ES book  & 6.98 & 7.55 & 8.11 \\
NQ book  & 7.23 & 7.54 & 7.85 \\
ZN book  & 6.83 & 7.80 & 8.76 \\
ES trade & 6.81 & 7.52 & 8.24 \\
NQ trade & 7.14 & 7.67 & 8.19 \\
ZN trade & 6.81 & 7.77 & 8.74 \\
\midrule
range    & $6.81$--$7.23$ & $7.52$--$7.80$ & $7.85$--$8.76$ \\
\bottomrule
\end{tabular}
\caption{Per-packet service time in $\mu$s by datagram span, from each stream's measured
service floor and slope, against the publisher period $\spub$
(Section~\ref{sec:results-conditioning}). Every stream is below the threshold at one message per packet and at or above it at
two; one stream, ES trade, sits just inside the measured band at two. Figure~\ref{fig:bars-crossval-crossing} shows the same numbers as bars.}
\label{tab:crossval-crossing}
\end{table}

\begin{figure}[H]
\centering
\includegraphics[width=1.0\textwidth]{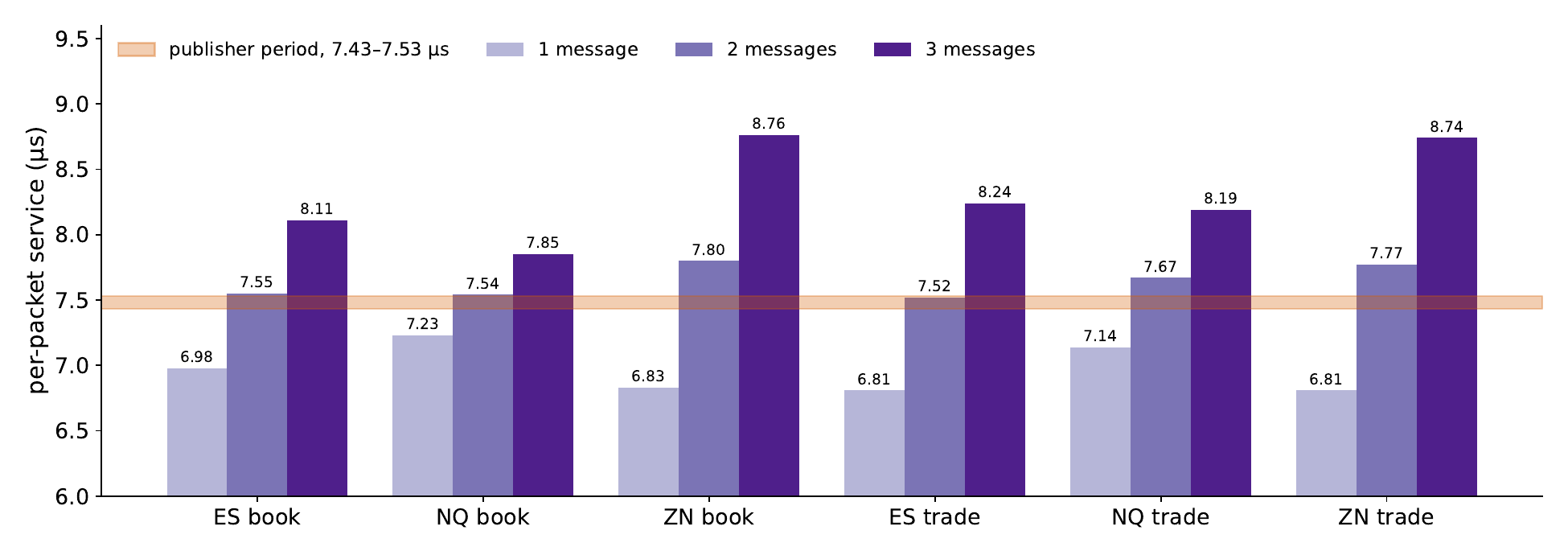}
\caption{Table~\ref{tab:crossval-crossing} as bars: per-packet service for one, two and three
messages on each live stream, against the tail threshold (shaded band). A one-message packet is served below the threshold on every stream; a two-message packet
is at or above it on every stream, one of them just inside the band. Whether a packet can start a queue is decided by whether it carries
more than one message.}
\label{fig:bars-crossval-crossing}
\end{figure}

All six streams are below the threshold when a datagram carries one message and at or
above it when it carries two (Table~\ref{tab:crossval-crossing}). All six crossings fall in the same unit interval of span,
from service floors and slopes fitted separately on six multicast channels against a threshold
measured on a different corpus. For this receiver, whether a packet's service is above or
below the threshold is determined by whether it carries one message or more.

\findings{
    \item \emph{The two-stage receiver's service floor is about $7\,\mu\mathrm{s}$} by three independent estimates, just under the publisher period.
    \item \emph{A packet with one message is served below the publisher period, one with two or more at or above it}, on all six live streams.
}

\subsection{The one-thread receiver}
\label{sec:crossval-onethread}

The service floor of the previous section belongs to the two-stage receiver; it is not the
cost of decoding. The one-thread receiver, in which the socket-reader thread decodes each
packet and updates the book itself, was measured on live CME ES, NQ and ZN futures in two
five-minute runs, the first minute of each discarded, on one feed, with kernel-bypass
networking (Onload) and threads not pinned. Latency is taken per message from the socket
read to the book publish; as in the two-stage measurement, time spent in the NIC and the
socket buffer before the read is excluded. Table~\ref{tab:crossval-onethread} gives the
quantiles.

\begin{table}[H]
\centering
\small
\begin{tabular}{@{}lrrrrrrr@{}}
\toprule
stream & messages & $p_{1}$ & $p_{10}$ & $p_{50}$ & $p_{90}$ & $p_{99}$ & $p_{99.9}$ \\
\midrule
ES book  & 401{,}462 & 0.5 & 0.7 & 1.1 & 2.7 &  9.9 &  24.1 \\
NQ book  & 581{,}010 & 0.4 & 0.5 & 0.8 & 1.9 &  4.7 &   8.8 \\
ZN book  & 141{,}972 & 0.4 & 0.6 & 1.0 & 2.5 & 20.6 &  69.3 \\
ES trade &  45{,}733 & 0.5 & 0.7 & 2.5 & 8.8 & 36.2 &  70.1 \\
NQ trade &  21{,}093 & 0.4 & 0.6 & 0.9 & 3.3 & 17.8 &  31.5 \\
ZN trade &  13{,}282 & 0.4 & 0.7 & 6.6 & 41.5 & 90.5 & 124.0 \\
\bottomrule
\end{tabular}
\caption{Socket-to-book latency of the one-thread receiver on the six live streams, in
\si{\micro\second}.}
\label{tab:crossval-onethread}
\end{table}

On one thread the median book update takes $0.8$--$1.1\,\mu\mathrm{s}$, far below the
publisher period of about $7.5\,\mu\mathrm{s}$ and the two-stage receiver's service floor of
about $7\,\mu\mathrm{s}$. The hand-off from the reader to the decoder thread is the main
structural difference between the two configurations, but the runs were taken separately
and differ in other respects as well, so the gap of about six microseconds is not all
hand-off.

The result is the paper's own prescription at work. The two-stage receiver runs at the
edge of the tail, where a packet of two or more messages, or a service slower than its
median, crosses the publisher period and opens a queue
(Sections~\ref{sec:crossval-spanrun} and~\ref{sec:fungibility-measured}). The one-thread
receiver runs far below the period at the median, in the regime where a train of packets
at the publisher period builds no queue behind a typical packet and a single thread is the
right design (Section~\ref{sec:results-tail}). It still has a tail, which is largest on the
trade streams, whose packets carry the most messages (Section~\ref{sec:crossval-span}): on
ZN trade the $p_{99}$ is $90.5\,\mu\mathrm{s}$ and the $p_{99.9}$ $124\,\mu\mathrm{s}$.

\findings{
    \item \emph{On one thread, with no hand-off, the median book update takes $0.8$--$1.1\,\mu\mathrm{s}$}, far below the publisher period; the trade streams reach $17.8$--$90.5\,\mu\mathrm{s}$ at $p_{99}$.
    \item \emph{The two-stage receiver's service floor is not the cost of decoding}: it includes the cross-thread hand-off, the main structural difference between the two runs.
}

\subsection{Statistics that differ between the corpus and the live receiver}
\label{sec:crossval-disagree}

Two numbers do not match between the recorded and the live measurements, and this
section explains why neither is a cause for concern. In both cases the disagreement is
about how the number was computed, not about the feed, and in both cases one of the two
measurements can diagnose the other. The first is the raw level of burstiness, which
reads about four times higher live. That is because the live figure was taken over one uncut stretch that mixes a quiet
period and a busy one, and a rate that drifts across the stretch inflates the measure
whether or not the arrivals are genuinely bursty. The recorded corpus computes the same
quantity inside each half-hour window, which removes that drift. The change in burstiness when the packets are scrambled, which is the effect
the paper uses, agrees closely between the two. The second is a control value
that should come out at exactly one half after scrambling, and reads a little above it
live. That is a known small-sample bias of the estimator on a short capture; on the much
larger recorded corpus the same control comes out at one half, as it should. Neither
discrepancy touches any quantity the paper's argument rests on, and for the level of
burstiness the recorded corpus gives the more conservative value, which is the one
reported.

\paragraph{The absolute Fano factor and the length of the stretch.}
Two estimators are in play. The \emph{windowed} estimator cuts the data into 30-minute
windows, computes the Fano factor of the bin counts inside each window, and reports the
median across windows; the \emph{pooled} estimator treats one long uncut stretch as a
single series and computes the Fano factor over all of it. (This is a different sense
of ``pooled'' from the histograms of Section~\ref{sec:results-transactions}, which pool
gap \emph{lengths} across windows into one distribution; a distribution of lengths is
indifferent to slow rate drift, a count-variance statistic is not.) Part~II, windowed,
measures a Fano factor at a five-second bin of a few hundred; the live streams, pooled
over the whole capture, give about four times that. The ratio under the gap shuffle agrees closely (Table~\ref{tab:crossval-stats}),
so the discrepancy is in the level, not in the effect. The cause is the estimator row of
Table~\ref{tab:crossval-setup}. The live figure is computed over one uncut 53-minute
capture that spans a quiet stretch and a busy one; a rate that drifts between the two
inflates the count variance and reads as clustering whether or not the process is
self-exciting. Part~II computes the Fano factor \emph{within} each 30-minute window and
then takes the corpus median, which removes between-window rate variation by
construction: a slow drift from a quiet stretch to a busy one adds count variance that
the estimator cannot distinguish from bursts, so the pooled level reads roughly four times
higher. The gap-shuffle ratio cancels the drift, since the drift is present in the
numerator and the denominator alike, which is why the ratio agrees between the two
measurements and why the paper uses the ratio, not the level. The windowed estimate is
the conservative one and is the value reported.

\paragraph{The shuffled Hurst null and finite-sample bias.} A shuffle of a gap sequence
destroys the ordering while keeping every gap, so the shuffled series is a renewal process
and its Hurst exponent must be $0.5$. The live measurement recovers somewhat more,
attributes the excess to finite-sample bias in the Hurst estimator on a short capture, and
states its result as a difference from the unshuffled value rather than a distance from
$0.5$. Part~II confirms the attribution: over the whole corpus the same shuffle, and the
Poisson null, both return $0.50$ (Section~\ref{sec:setup-hawkes}). The bias vanishes with
more, shorter series.

\findings{
    \item \emph{The live Fano level reads four times higher} because it is computed over one long stretch whose rate drifts; the shuffle ratio the paper relies on agrees.
    \item \emph{The live shuffled Hurst exponent sits above $0.5$} from small-sample bias; on the corpus it is $0.50$.
}

\subsection{Per-packet service time, message count and coalesced transactions}
\label{sec:crossval-span}

The simulator in Part~II assumed that every packet takes the same time to process. The
live measurement shows that it does not, and this section sets out what that means. A
packet on this feed carries a variable number of messages, usually one but sometimes
dozens, and they are decoded one after another, so a packet takes longer the more it
carries. Large packets are rare but expensive: on the most fragmented stream, ZN trade,
one packet in nine
carries two thirds of the messages and consumes six sevenths of the decode time. The
task length used throughout the paper is therefore the work for a typical packet, not a
fixed constant of the system. Two points follow, set out at the end of the section. The results survive, because the
cost per message is constant and only the bundling varies. And the number of messages per
packet is a second place where the same clustering is paid: transactions that reach the
publisher while it is busy can be packed into one packet rather than sent separately
(Section~\ref{sec:exchange-drain}). It acts
mainly where bursts across packets do not, on tasks at or just below the tightest gap
between packets, which is where the two-stage receiver runs, and adds a few percent above
$16\,\mu\mathrm{s}$. The next section reruns the sweep with the message
count included and measures the share of each channel.

The sweep of Section~\ref{sec:setup-sweep} treats each packet as one arrival with a
constant service time $T$. Per-packet service is not constant. A UDP datagram on this feed carries a variable number of SBE messages ---
call it the packet's \emph{span} --- which are decoded in order, so a message at position
$k$ waits for $k$ decodes before its own. Median latency follows
\[
\text{median} \;=\; \text{service floor} \;+\; \text{slope} \times \text{position},
\]
with the service floor of Section~\ref{sec:crossval-floor} and a slope of $0.31$ to
$0.97\,\mu\mathrm{s}$ per message depending on the stream. Span is not close to constant
(Tables~\ref{tab:crossval-span} and~\ref{tab:crossval-conc}):

\begin{table}[H]
\centering
\small
\begin{tabular}{@{}lrrrrrr@{}}
\toprule
stream & packets & messages & $p_{50}$ & $p_{90}$ & $p_{99}$ & max \\
\midrule
ES book  & 2{,}123{,}042 & 2{,}415{,}610 & 1 & 1 & 7 & 45 \\
NQ book  & 2{,}506{,}205 & 3{,}010{,}036 & 1 & 2 & 11 & 36 \\
ZN book  & \phantom{0}824{,}536 & \phantom{0,}965{,}334 & 1 & 1 & 6 & 45 \\
ES trade & \phantom{0,0}83{,}140 & \phantom{0,}235{,}601 & 3 & 8 & 17 & 83 \\
NQ trade & \phantom{0,0}50{,}279 & \phantom{0,0}83{,}233 & 2 & 7 & 16 & 85 \\
ZN trade & \phantom{0,0}22{,}864 & \phantom{0,}100{,}659 & 2 & 11 & 71 & 85 \\
\bottomrule
\end{tabular}
\caption{Messages per UDP datagram (span), packet-weighted quantiles. The median book packet
carries one message; one packet in a hundred carries about ten on the book streams and
many times that on ZN trade. Figure~\ref{fig:bars-crossval-span} shows the same numbers as bars.}
\label{tab:crossval-span}
\end{table}

\begin{figure}[H]
\centering
\includegraphics[width=1.0\textwidth]{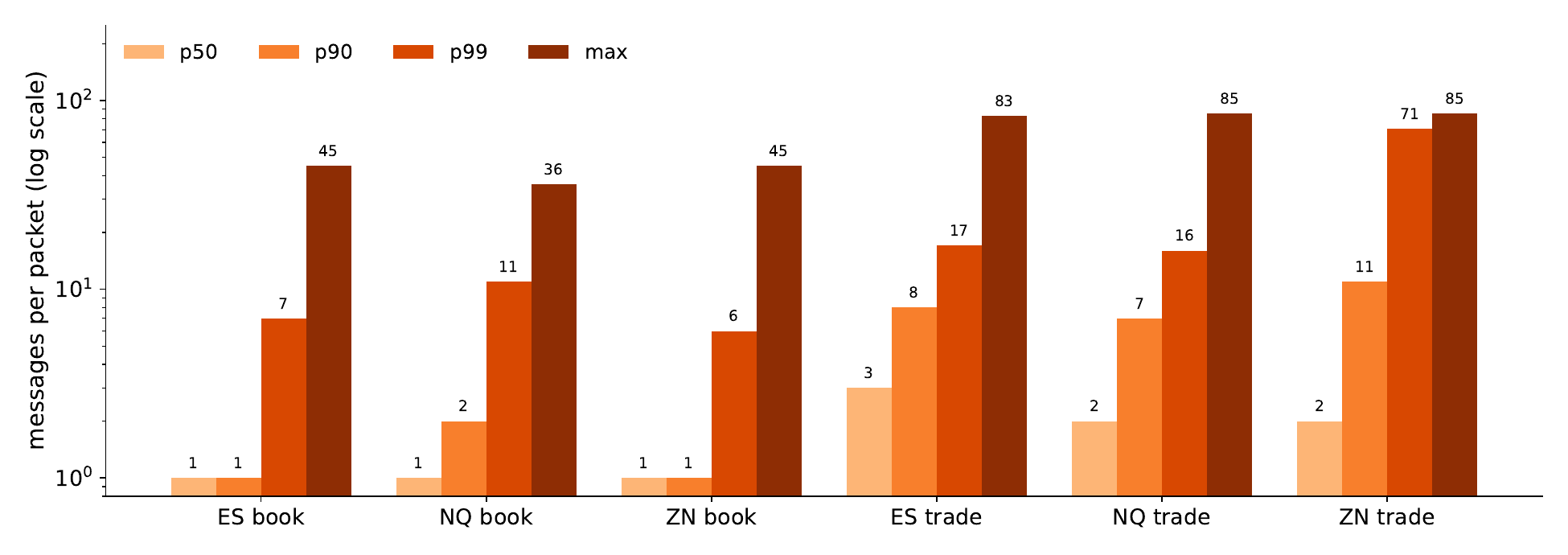}
\caption{Table~\ref{tab:crossval-span} as bars, log scale: messages per packet at the median,
$90$th and $99$th percentile and the maximum, per live stream. Book streams carry one
message in most packets and up to $45$; trade streams carry more, and ZN trade reaches $71$
at $p_{99}$.}
\label{fig:bars-crossval-span}
\end{figure}

Because service is charged per message, large packets account for a disproportionate
share of processing time:

\begin{table}[H]
\centering
\small
\begin{tabular}{@{}lrrr@{}}
\toprule
span $\ge 10$ & \% of packets & \% of messages & \% of decode time \\
\midrule
ES book  & 0.52 & 4.55 & 5.34 \\
NQ book  & 1.68 & 4.87 & 2.57 \\
ZN book  & 0.47 & 7.29 & 7.29 \\
NQ trade & 5.34 & 19.31 & 16.16 \\
ES trade & 6.69 & 26.53 & 40.10 \\
ZN trade & 11.31 & 66.18 & 85.54 \\
\bottomrule
\end{tabular}
\caption{Concentration of work in large datagrams. On ZN trade, $11.3\%$ of packets carry
$66\%$ of the messages and consume $86\%$ of decode time. Figure~\ref{fig:bars-crossval-conc} shows the same numbers as bars.}
\label{tab:crossval-conc}
\end{table}

\begin{figure}[H]
\centering
\includegraphics[width=1.0\textwidth]{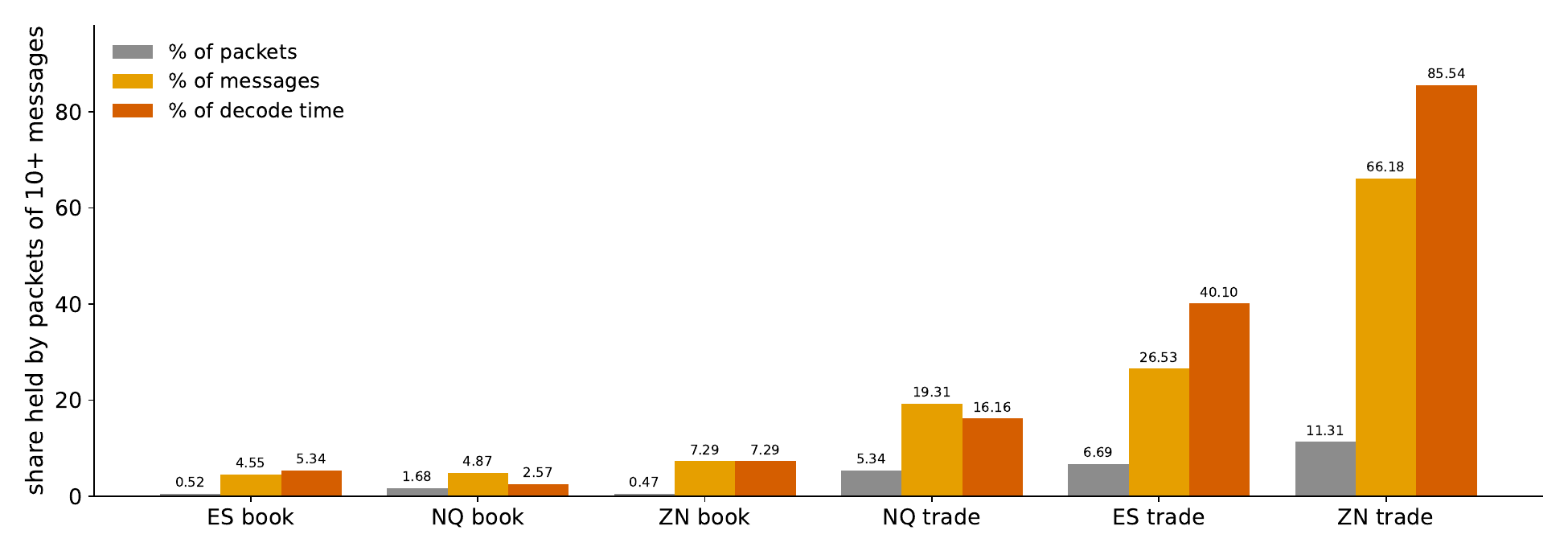}
\caption{Table~\ref{tab:crossval-conc} as bars: the share of packets, of messages and of decode
time held by packets of ten or more messages. On ZN trade $11\%$ of packets hold $66\%$ of
the messages and $86\%$ of the decode time; on the book streams long packets are under
$2\%$ of packets but several times that share of messages.}
\label{fig:bars-crossval-conc}
\end{figure}

The constant-service assumption is therefore violated at the packet level, and Part~II's
$T$ is the per-packet work at a typical span rather than a constant of the system. Two
points follow.

First, at the message level the assumption approximately holds: the slope is fixed
within a stream to the precision of the fit, and it is the composition of packets, not
randomness in per-message work, that makes per-packet service vary. Theorem~\ref{thm:reduction}
is pathwise and arrival-law independent, and Remark~\ref{rem:unequal} already handles
unequal stages through $s_{\max}$; what the live data adds is a measured magnitude for the
variability that Section~\ref{sec:theory-summary} lists as an assumption.

Second, span is a second channel for the same clustering on the same stream, active
where the queueing channel is not. Section~\ref{sec:crossval-spanrun} reruns the sweep
with span-dependent service. The in-packet decode supplies the tail below the
publisher period, where queueing across packets supplies none, and adds a few percent
above $16\,\mu\mathrm{s}$, where queueing across packets supplies nearly all of it.

\paragraph{Spans and transactions.} On the corpus, where every message carries its
\texttt{transactTime}, a multi-message packet is usually not one large transaction but
several transactions delivered together. About $4\%$ of packets carry two or more
messages, and about three quarters of those carry more than one transaction. The publisher,
that is, packs transactions that reach it while it is busy into one packet
(Section~\ref{sec:exchange-drain}). This is consistent with the publisher period: two
transactions that the engine processes closer together than $\spub$
never appear as two packets closer together than that; they appear either as one packet
with more messages or, more often, as two packets at the publisher period
(Table~\ref{tab:two-clocks}), which is why so few packet gaps fall below the publisher period (Table~\ref{tab:tx-gaps}). Two consequences follow for the rest of this section.
The in-packet tail that a large packet produces is, for the most part, the same clustering
of transactions as the queueing tail, delivered inside a packet instead of across packets;
and the message count per packet is set by how closely transactions follow one another,
not by how large individual transactions are. The two-mechanism description below,
queueing across packets and decoding within them, stands as a description of where the
wait is paid; both originate in the transaction stream.

\findings{
    \item \emph{Service is not constant per packet}: each further message adds a fixed cost, from about a third of a microsecond to about one depending on the stream.
    \item \emph{Long packets are rare but expensive}: on ZN trade one packet in nine carries two thirds of the messages and most of the decode time.
    \item \emph{A long packet is usually several transactions packed together}: about three quarters of multi-message packets hold more than one transaction.
}

\subsection{The sweep rerun with span-dependent service}
\label{sec:crossval-spanrun}

The previous section found that a packet takes longer to process the more messages it
carries. This section reruns the whole simulation with that fact built in, to find out
what it changes. The processing time is modelled three ways: fixed per packet, as in
Part~II; growing with message count at the rate measured on the live system; and, as a
deliberately exaggerated upper bound, fully proportional to message count. All three do the same total work on average, so any difference in the slow tail comes
from how the work is spread across packets. The findings are collected at the end of the
section.

We re-extracted the per-packet span for all 3512 windows --- aligned element-wise against the
cached arrival arrays, with every window's span count checked against the packet count
already in the cache --- and re-ran the entire sweep with service that depends on it.

Three service models are compared. All three hold the \emph{mean} service per window at
$T$, so the utilisation $\rho$ is identical across them and only the distribution of work
across packets differs; any difference in the tail is attributable to the span
dependence.
\begin{align*}
\text{const} &: \quad S_i = T && \text{(the Part~II assumption)} \\
\text{span}  &: \quad S_i = A\big(1 + r\,(\sigma_i - 1)\big), \quad r = 0.0432
  && \text{(measured slope / service floor)} \\
\text{prop}  &: \quad S_i = T\,\sigma_i / \bar{\sigma} && \text{(upper bound)}
\end{align*}
with $A$ set so that the window mean of $S_i$ is $T$. The \emph{span} model uses the ratio of the cost per extra message to the service floor,
measured on NQ book in Section~\ref{sec:crossval-floor}.
The \emph{prop} model removes the per-packet fixed cost, making service proportional
to the message count; this overstates the dependence, since the measured service floor exceeds
the slope by more than twenty to one, and it is included as an upper bound. As a correctness check the \emph{const} arm reproduces the published
constant-service grid of Table~\ref{tab:main-corpus} to zero difference in all 28 cells.

\paragraph{Message-weighted quantiles.} On this corpus almost every packet carries one
message, so a packet-weighted quantile is dominated by single-message packets and
changes little under any service model. Weighting by message, which counts
a $20$-message datagram twenty times, the in-packet position distribution is concentrated
at zero with a long upper tail:

\begin{center}
\begin{tabular}{@{}lrrrrrr@{}}
\toprule
message quantile & $p_{50}$ & $p_{90}$ & $p_{95}$ & $p_{99}$ & $p_{99.9}$ & max \\
\midrule
in-packet position & 0 & 0 & 1 & 4 & 19 & 84 \\
\bottomrule
\end{tabular}
\end{center}

Most messages, nearly nine in ten, are alone in their packet and pay no in-packet decode;
the rest make up the tail. The live measurement reports the same shape. The comparison
below is message-weighted.

\begin{table}[H]
\centering
\small
\begin{tabular}{@{}r rrr rrr rr rr@{}}
\toprule
 & \multicolumn{3}{c}{const} & \multicolumn{3}{c}{span} &
   \multicolumn{2}{c}{span/const} & \multicolumn{2}{c}{prop/const} \\
\cmidrule(lr){2-4} \cmidrule(lr){5-7} \cmidrule(lr){8-9} \cmidrule(lr){10-11}
$T$ & $p_{50}$ & $p_{99}$ & $p_{99.9}$ & $p_{50}$ & $p_{99}$ & $p_{99.9}$ &
 $p_{99}$ & $p_{99.9}$ & $p_{99}$ & $p_{99.9}$ \\
\midrule
  2 &    2.00 &     2.00 &      2.00 &    1.99 &     2.34 &      3.54 & 1.17 & 1.77 &  4.64 & 23.35 \\
  4 &    4.00 &     4.00 &      4.00 &    3.99 &     4.67 &      7.08 & 1.17 & 1.77 &  4.69 & 26.73 \\
  8 &    8.00 &     8.56 &     10.10 &    7.97 &     9.63 &     14.87 & 1.12 & 1.47 &  5.21 & 24.64 \\
 16 &   16.00 &    34.00 &     53.55 &   15.95 &    35.08 &     56.62 & 1.03 & 1.06 &  3.36 & 11.53 \\
 32 &   32.00 &   143.74 &    256.51 &   31.89 &   144.41 &    260.14 & 1.00 & 1.01 &  2.37 &  5.88 \\
 64 &   98.80 &   815.22 &   1557.72 &  100.02 &   820.05 &   1573.97 & 1.01 & 1.01 &  2.37 &  3.94 \\
128 &  366.78 &  5440.22 &  12844.51 &  371.91 &  5550.42 &  13159.71 & 1.02 & 1.02 &  2.51 &  3.10 \\
\bottomrule
\end{tabular}
\caption{Effect of the number of messages per packet (the packet's \emph{span}, not its byte
length) on the single-stage tail. Each row is one service time
$T$; every entry is the corpus median over 3512 windows of a message-weighted latency
quantile, in $\mu$s, under real (Hawkes) arrivals with a single stage ($N = 1$). The
\emph{const} block charges every packet the same service $T$, the Part~II assumption.
The \emph{span} block charges service that grows with the number of messages in the
packet, at the slope measured on the live system. Both do the same mean work per window,
so any difference between the two blocks is due to message count alone. The
\emph{span/const} columns give that difference directly as a ratio; a value near $1$
means message count adds nothing. The \emph{prop/const} columns repeat the ratio for a
deliberately exaggerated model in which service is fully proportional to message count,
as an upper bound. Read across: message count inflates the tail only at $T \le
8\,\mu\mathrm{s}$, at or just above the publisher period, where clustering builds little
or no queue, and is negligible above it. Figure~\ref{fig:bars-spancmp} shows the same numbers as bars.}
\label{tab:spancmp}
\end{table}

\begin{figure}[H]
\centering
\includegraphics[width=0.85\textwidth]{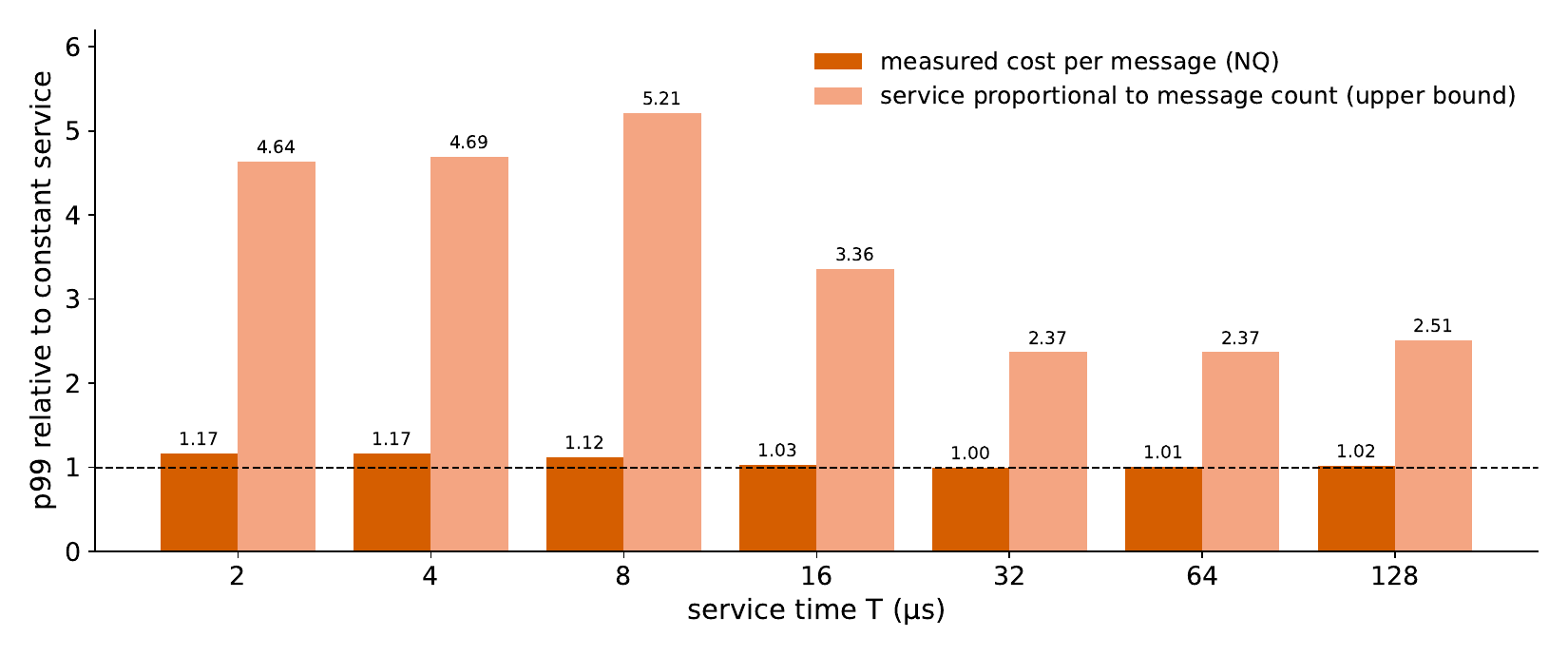}
\caption{Table~\ref{tab:spancmp} as bars: $p_{99}$ under service that grows with the message count,
relative to constant service with the same mean. With the measured cost per message
(vermillion) the tail is $12$--$17\%$ larger at $8\,\mu\mathrm{s}$ and below and within
$3\%$ from $16\,\mu\mathrm{s}$. The exaggerated model in which service is proportional to
message count (light) bounds the effect from above.}
\label{fig:bars-spancmp}
\end{figure}

\paragraph{Regimes of the two tail mechanisms.} The span/const ratio of
Table~\ref{tab:spancmp} is largest where the queueing tail is absent and smallest where it
dominates. At $2$ and $4\,\mu\mathrm{s}$ the constant-service model has no tail at all,
while the span-dependent model has a small one: below the publisher period there is a
tail, but it is not a queueing tail. From $32\,\mu\mathrm{s}$, where queueing across packets
makes the $p_{99}$ several times $T$, span adds at most a few percent.

\paragraph{Decomposition at the production operating point.} The two-stage receiver's service floor is about $7\,\mu\mathrm{s}$ (Section~\ref{sec:crossval-floor}), so the comparable row is the sweep point just above the publisher period, $T = 8\,\mu\mathrm{s}$. Decomposing the message-weighted $p_{99}$ excess over the median ($p_{50}$) into the part the
constant-service model already produced (queueing) and the remainder introduced by span-aware service (in-packet decode):

\begin{center}
\begin{tabular}{@{}rrrrr@{}}
\toprule
$T$ & $p_{50}$ & queueing excess & in-packet decode excess & its share \\
\midrule
  2 &   1.99 &    0.00 &   0.34 & $100\%$ \\
  4 &   3.99 &    0.00 &   0.69 & $100\%$ \\
  8 &   7.97 &    0.56 &   1.09 & $66\%$ \\
 16 &  15.95 &   18.00 &   1.14 & $6\%$ \\
 32 &  31.89 &  111.74 &   0.78 & $1\%$ \\
 64 & 100.02 &  716.42 &   3.61 & $1\%$ \\
128 & 371.91 & 5073.43 & 105.07 & $2\%$ \\
\bottomrule
\end{tabular}
\end{center}

At the sweep point just above the publisher period, decoding the messages ahead in the
same packet supplies about two thirds of the $p_{99}$ excess and queueing across packets
the rest. The live measurement reports the same ordering: the service floor plus the cost
of the messages ahead in the packet reproduces about three quarters of the observed
$p_{99}$ on most streams, with the remainder left to queueing and to stalls. Applying that
definition to the simulator gives a higher share, because the only residual the model
contains is queueing; it has no stall population. Both measurements attribute the larger
share of the tail at the production service time to in-packet decode. Figure~\ref{fig:bars-span-share} shows the split at every service time.

\begin{figure}[H]
\centering
\includegraphics[width=0.78\textwidth]{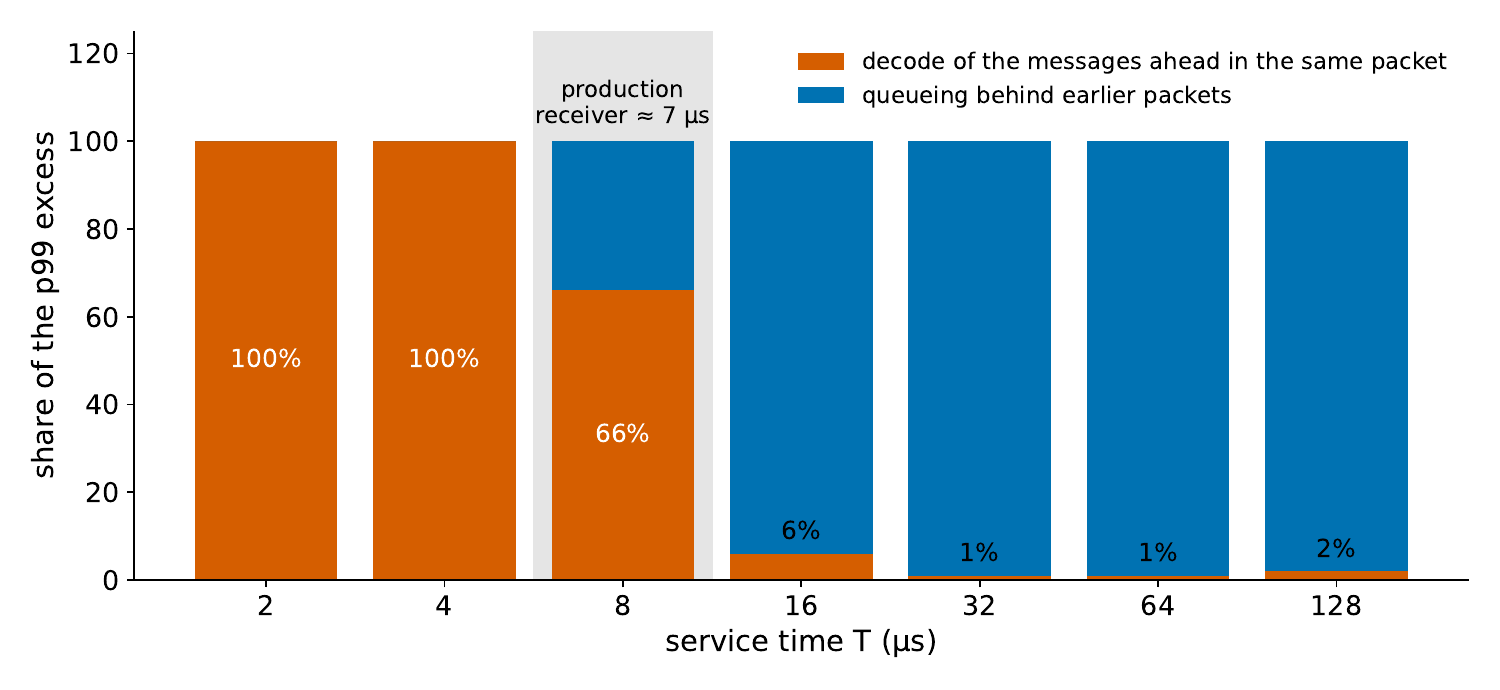}
\caption{The decomposition above as bars: of the simulated $p_{99}$ excess at each service time,
the part due to decoding the messages ahead in the same packet (vermillion) and the part
due to queueing behind earlier packets (blue). At and below the publisher period the decode of long
packets is the tail; the shaded bar, $T = 8\,\mu\mathrm{s}$, is the nearest simulated
service time to the two-stage receiver's service floor, and there it is two thirds.
From $16\,\mu\mathrm{s}$ queueing is almost all of it.}
\label{fig:bars-span-share}
\end{figure}

This is consistent with Part~II. The queueing claim stands: below about
$\spub$ of service, arrival clustering does not build a queue and splitting
the servicing chain removes nothing. The span-dependent run adds that a second mechanism, in-packet decode, is active in the
same regime. It is deterministic within the packet, and a stage cut does not shorten it.
The queue that a large packet opens behind itself is taken up in
Section~\ref{sec:fungibility-measured}. $N^\star$, the $p_{99}$-minimising stage count, is
unchanged at all seven service times when service is made span-dependent. The design rule of Section~\ref{sec:recommendations} is
unaffected.

\paragraph{The in-packet ladder.} Median latency by in-packet position at $N = 1$
reproduces the linear form of the live measurement. At the sweep point just above the publisher period the corpus-median latency grows
linearly with position in the packet, and the fitted cost per position is the cost per
message that was put into the model, which confirms that queueing contributes nothing at
the median at this service time. The live NQ-book ladder has the same form; the levels
differ only by the difference between the measured service floor and the simulated service
time.

\paragraph{Compounding of the two channels on NQ.} There are two places
a message can be made slow. One is queueing: it waits behind a run of earlier packets.
The other is in-packet decode: it sits deep inside a large packet and waits for the
messages ahead of it in that same packet to be decoded first. Section~\ref{sec:crossval-span}
showed that a large packet is usually several transactions packed together by the publisher, so
both are deliveries of the same clustering. The question here is narrower: whether a
large packet also tends to arrive behind a tight gap, so that one message pays both
delays at once. If it did, the two would compound and the Part~II tables, which model
queueing alone, would understate the tail above the publisher period. If a large packet arrives
after a tight gap no more often than after a loose one, the tables stand. Two tests find
the second case. Both compare a packet's span with the gap that precedes it, so they
cannot see the coalescing itself, which removes a gap rather than tightening one.

The first is a control arm in which each window's span sequence is randomly permuted
across packets, preserving the span marginal and the arrival times exactly and destroying
only the association between them. If large packets arrived preferentially inside
clusters, breaking that association would lower the tail. It does not: the real and the
shuffled $p_{99}$ agree within a couple of percent at every service time, with no
systematic sign. Second, measuring the association directly, a packet's size is almost
uncorrelated with the gap before it, and packets arriving after the shortest gaps are, if
anything, slightly smaller than those after the longest. The association is weakly
negative.

The live measurement reports the same instrument dependence: mean span rises with queue
depth on ZN, where a quiet book is filled
only by genuine bursts, but \emph{falls} with queue depth on ES and NQ, whose rings back up
with many small packets. Part~II's corpus is NQ, and it behaves like the live NQ streams.
The transaction-level arms of Section~\ref{sec:results-transactions} give the same
answer a third way: decoupling transaction sizes from their timing (TS) and merging each
transaction into one packet (TM) leave the tail about as it was at every service time from $16\,\mu\mathrm{s}$ up (Table~\ref{tab:tx-arms}). Above the publisher period, then, in-packet decode adds a few percent to a tail that
queueing across packets produces, and the constant-service tables of Part~II measure that
queueing. At the publisher period, where the two-stage receiver runs, in-packet decode
carries most of the tail. Since a long packet is several transactions packed together
(Section~\ref{sec:crossval-span}), the two channels are two deliveries of the same
clustering. On an instrument like ZN, where the far tail is a few
large packets, the second delivery dominates; Section~\ref{sec:fungibility-measured}
takes this up.

\findings{
    \item \emph{At the two-stage receiver's service time, long packets carry most of the tail}: in the simulation at the sweep point just above the publisher period about two thirds of it, and on the live receiver in-packet position reproduces about three quarters of the $p_{99}$ on most streams.
    \item \emph{From $16\,\mu\mathrm{s}$ up, message count adds a few percent} to a tail set by queueing across packets.
    \item \emph{The best stage count does not change}: a stage cut does not shorten the decode of a long packet.
    \item \emph{On NQ long packets do not arrive preferentially after short gaps}; on ZN this is untested.
}

\subsection{The production handler's tail}
\label{sec:fungibility-measured}
\label{sec:crossval-production}

The live records also give the production handler's own tail, on the two-stage receiver
with serial decode: one decoder thread, behind the reader's hand-off, decoding every packet
inline. The question is the one the
paper opened with, put to a real handler. On the live feed, why is the slowest message
in a thousand so much slower than the typical one? Read with Sections~\ref{sec:results-transactions}
and~\ref{sec:crossval-span}, the answer is that the clustering of transactions reaches
this handler in two forms. Most transactions that the engine processes close together reach it as a train of
packets one publisher period apart (Table~\ref{tab:two-clocks}). Some are packed by the
publisher into one packet carrying several messages. The handler decodes those one after
another, so the packet holds the thread for many times the service floor and everything
behind it queues. On the corpus, a multi-message packet is several transactions delivered together
far more often than it is one large transaction ($77\%$ of packets with two or more
messages carry more than one transaction, Section~\ref{sec:crossval-span}), so a large
packet is a transaction cluster in its own right. On the ZN channel the second form
dominates the far tail. What follows is argued from the fitted service floor and slope of
Section~\ref{sec:crossval-floor} and the reported percentiles, not from a
message-by-message decomposition of the window.

\paragraph{Setup.} ZN, book stream, one \SI{900}{\second} window, recovery
windows excluded, the shipping serial configuration (\texttt{cme\_decode\_workers 0}).
Latency is the leg-1 socket-read-to-book-publish time of Section~\ref{sec:crossval-floor},
per message. ZN is the instrument on which Section~\ref{sec:crossval-spanrun} left the
coupling of span and bursts open. Two populations are reported: every message, and the
empty-queue messages of Section~\ref{sec:crossval-floor}, which arrived to an empty queue
and were first in their packet, so that neither queueing nor in-packet decode is in their
latency.

\begin{figure}[H]
\centering
\includegraphics[width=0.72\textwidth]{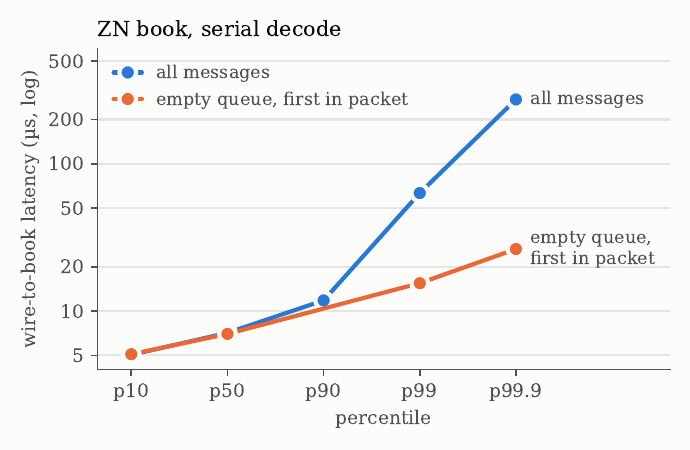}
\caption{Wire-to-book latency by percentile on the ZN book stream, serial decode, for every message and for the empty-queue messages alone (empty queue, first in packet; its $p_{90}$ was
not reported). Log latency axis. The two populations share a minimum near
$5$--$7\,\mu\mathrm{s}$ and separate only in the tail: the empty-queue messages reach
$26\,\mu\mathrm{s}$ at $p_{99.9}$, all messages $274\,\mu\mathrm{s}$.}
\label{fig:zn-serial-tail}
\end{figure}

\begin{table}[H]
\centering
\small
\begin{tabular}{@{}lrrrrrrr@{}}
\toprule
population & messages & $p_{10}$ & $p_{50}$ & $p_{90}$ & $p_{99}$ & $p_{99.9}$ & mean \\
\midrule
all messages                    & 157{,}699 & 5.04 & 7.11 & 11.81 & 63.30 & 273.69 & 9.95 \\
empty queue, first in packet    & 122{,}500 & 5.09 & 6.99 & ---   & 15.46 &  26.37 & --- \\
\bottomrule
\end{tabular}
\caption{The data behind Figure~\ref{fig:zn-serial-tail}, in \si{\micro\second}. The empty-queue messages are $78\%$ of all messages. The window maximum for all messages was
$3046\,\mu\mathrm{s}$; the report that records this measurement
(\texttt{tech\_reports/\allowbreak serial\_vs\_parallel\_decode.md} in the \kasparhft{} repository)
flags a few such isolated spikes as possibly residual recovery, so it is not used here.}
\label{tab:zn-serial-tail}
\end{table}

\paragraph{The spread from the median to $p_{99.9}$.} The slowest message in a thousand
takes $274\,\mu\mathrm{s}$, against $7\,\mu\mathrm{s}$ for the median, some forty times as
long. Restricting to messages that met an empty queue and were first in their packet
removes nine tenths of that far tail, down to $26\,\mu\mathrm{s}$, without changing the
fastest messages. The difference between the two is the contribution of queueing and of
position inside the packet. The restricted messages are not at the service floor in their
own tail either: about twice the service floor at $p_{99}$ and nearly four times at
$p_{99.9}$. That residual, present with neither queueing nor in-packet decode, is not
resolved by these records (Section~\ref{sec:open-questions}).

\paragraph{How the live tail forms.} The three findings of this paper meet here. Packets
reach the handler as trains one publisher period apart
(Section~\ref{sec:exchange-drain}). The handler's median service, about $7\,\mu\mathrm{s}$
on the book streams (Table~\ref{tab:tails-six}), is just under the publisher period, so a train of single-message packets served at the median builds no queue; this
is why its $p_{10}$ and $p_{50}$ sit at the service floor. Two things push a service over the period. A packet with two or more messages costs the service floor plus one per-message slope for
each further message, which crosses $\spub$ on every stream
(Section~\ref{sec:crossval-floor}). And a single-message packet sometimes takes far longer
than the median for reasons these records do not show: on ZN, messages with nothing queued
reach about twice the service floor at $p_{99}$. Either way the packets still arriving behind it queue,
and the queue lasts until the handler catches up with the train. The large-packet
mechanism is a matter of arithmetic on the measured service floor and slope; the slow-service
mechanism is measured only as the empty-queue percentiles above. How the ZN $p_{99.9}$
divides between the decode inside long packets, the queue behind them and
slow services was not decomposed, because the per-message records were not retained.
What the records do show is that the part of the tail not explained by queueing across
packets is large at this service time: the service floor plus the decode cost of the messages
ahead in the same packet reproduces about three quarters of the observed $p_{99}$ on four of
the six streams (Section~\ref{sec:crossval-spanrun}).

\paragraph{Large packets and the queue behind them.} The message count of a packet is set
by how closely transactions followed one another, not by how large any one transaction
was: few transactions carry more than one message (Table~\ref{tab:tx-counts}), while a
packet of ten or more messages is usually ten or more transactions that the publisher
packed into one packet. Section~\ref{sec:crossval-floor} gives what such a packet costs this
handler on ZN: the service floor of about $7\,\mu\mathrm{s}$ plus about $1\,\mu\mathrm{s}$ for
each further message. The largest packet seen, forty-five messages
(Section~\ref{sec:crossval-span}), therefore holds the thread for about $50\,\mu\mathrm{s}$,
seven times the service floor. That is the first cost: the last message of that packet waits
over $40\,\mu\mathrm{s}$ before its own decode. Yet the in-packet decode accounts for only a
fifth of the $p_{99.9}$. The rest is the queue the packet opens behind it, the second cost.
At one message per packet the service is about $7\,\mu\mathrm{s}$, just under the publisher period measured on NQ (ZN's own period was not measured), and a queue almost never forms;
that is the empty-queue population of Table~\ref{tab:zn-serial-tail}. While a forty-five-message packet holds the thread, every packet that arrives waits
behind it, and each adds its own service to the wait of the ones behind it. The wait of
the $k$-th packet in the queue is the sum of the service times ahead of it, so it grows
with the number queued, not with any one packet's size. ZN packets of ten or more messages are fewer than
one in two hundred, but each is a burst of transactions that this handler pays for twice,
once inside the packet and once in the queue behind it, and the hundreds of microseconds at
$p_{99.9}$ are those queues.

\paragraph{Across the six streams.} The same measurement covers three instruments and,
on each, a book stream and a trade stream. Figure~\ref{fig:tails-six} puts the six
serial tails side by side and Table~\ref{tab:tails-six} adds, for each stream, its
packet rate, the size of its packets and the cost of each extra message. Three things
are visible at once. First, the tail does not follow the packet rate. NQ book is by far the busiest stream
and has the smallest tail; ZN book carries a third of its rate and a tail an order of
magnitude larger; and the trade streams, far quieter than the books, have larger tails
than the books on every instrument. Second,
on every instrument the trade stream carries more messages per packet than the book
stream and has the larger tail: a trade is a transaction that touches several resting
orders and is followed by the reaction to it, so it is where transactions arrive closest
together and are coalesced most. Third, ZN, with the largest trade packets and the highest cost per extra message, has
the largest tails on both streams. Its book packets are not large at the $99$th
percentile, but at three times NQ's cost per message the largest of them produce the
longest excursions of any book stream.
The clustering statistics are similar across the three instruments (the Fano-implied
branching ceilings of Section~\ref{sec:crossval-agree} span $0.85$--$0.97$ across all six
streams, and the NQ publisher period is assumed to hold on the others). What separates a twenty-fold spread of tails across streams with similar clustering
statistics is not the arrival rate and not those statistics. It is how much of the
clustering the publisher delivers inside a packet, that is the number of transactions
packed together, times the cost this handler pays for each, plus the queue that opens
behind the packet. On the live system
the far tail is the clustering of transactions delivered as span.

\begin{figure}[H]
\centering
\includegraphics[width=\textwidth]{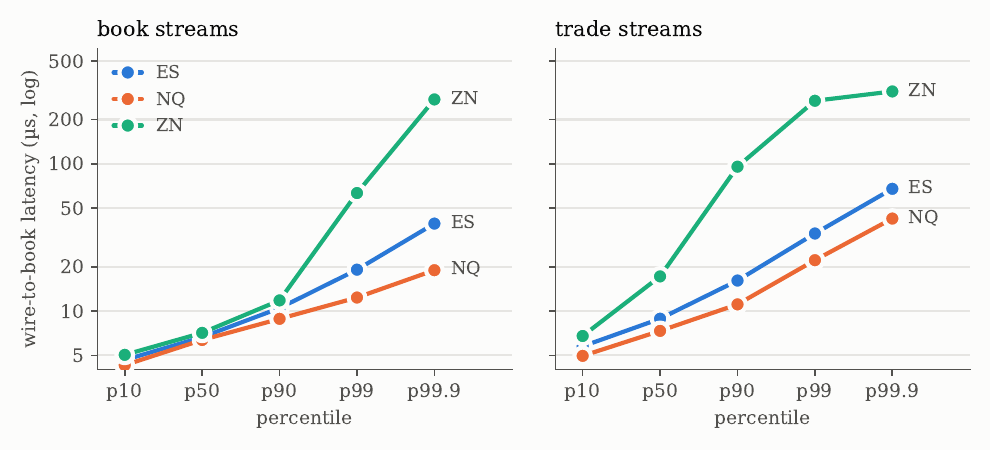}
\caption{Serial wire-to-book latency by percentile on the six live streams, one
\SI{900}{\second} window each, recovery excluded (book streams left, trade streams
right, log latency axis). The tails order by messages per packet and cost per
message, not by packet rate: the busiest stream (NQ book) has the smallest tail.}
\label{fig:tails-six}
\end{figure}

\begin{table}[H]
\centering
\small
\begin{tabular}{@{}lrrrrrr@{}}
\toprule
stream & packets & span $p_{99}$ & \si{\micro\second}/msg & $p_{50}$ & $p_{99}$ & $p_{99.9}$ \\
\midrule
NQ book  & 2{,}506{,}205 &  11 & 0.31 &  6.36 &  12.35 &  18.94 \\
ES book  & 2{,}123{,}042 &   7 & 0.57 &  6.65 &  19.10 &  39.25 \\
ZN book  &   824{,}536   &   6 & 0.97 &  7.11 &  63.30 & 273.69 \\
ES trade &    83{,}140   &  17 & 0.71 &  8.87 &  33.62 &  67.65 \\
NQ trade &    50{,}279   &  16 & 0.53 &  7.33 &  22.15 &  42.47 \\
ZN trade &    22{,}864   &  71 & 0.96 & 17.19 & 268.04 & 310.13 \\
\bottomrule
\end{tabular}
\caption{The six streams of Figure~\ref{fig:tails-six}, ordered by packet rate.
Packets and span $p_{99}$ are from the 53-minute capture of
Section~\ref{sec:crossval-span}; the cost per extra message is the fitted slope of
Section~\ref{sec:crossval-floor}; the latency quantiles, in \si{\micro\second}, are
from the \SI{900}{\second} serial windows. Rate falls down the table while the tail
rises; span $p_{99}$ and cost per message, taken together, track the tail.}
\label{tab:tails-six}
\end{table}

\paragraph{Consequences for the design rule.} The live handler's tail is the mechanism of
Sections~\ref{sec:results-transactions} and~\ref{sec:crossval-span}, seen from the
receiver's side. Section~\ref{sec:crossval-span} measured that a packet's decode time is
the service floor plus one slope for each message beyond the first, with packets of dozens of messages on both book and trade streams, and that such packets
are usually several transactions packed together. Section~\ref{sec:crossval-spanrun} re-ran the sweep with that span-dependent service.
At service times at or just above the publisher period the message count supplies about
two thirds of the $p_{99}$ excess; well above it, where the queue across packets takes
over, it adds a few percent; and the best stage count is unchanged. The two are one
clustering delivered two ways, inside and across packets. The ZN records put the two-stage handler in the first regime. At its service floor the
typical packet queues behind nothing, and the verdict of Section~\ref{sec:kaspar}, not to
split at the publisher period, stands. The far tail is paid on the packed packets, whose
effective service is several times the service floor. That is the regime in which the design rule of Section~\ref{sec:results-design} admits
a cut. A cut would shorten the queue such a packet opens behind it but not the decode
inside it. The span-aware sweep of Section~\ref{sec:crossval-spanrun}, calibrated on NQ's
cost per message, a third of ZN's, found the best stage count unchanged, so whether a cut
pays on ZN is for the measurement to decide.
The other lever is the per-message cost itself (Section~\ref{sec:recommendations}), which
scales the whole excursion. Which lever to pull is decided by the end-of-build
measurement of Section~\ref{sec:fungibility-workflow}.

\findings{
    \item \emph{On ZN the far tail is large}: about $270\,\mu\mathrm{s}$ at $p_{99.9}$, against about $26\,\mu\mathrm{s}$ for messages with nothing queued and first in their packet.
    \item \emph{Across the six live streams the tail follows packet size times cost per message, not packet rate}: the busiest stream, NQ book, has the smallest tail.
    \item \emph{Long packets and occasional slow services each start a queue in the publisher's train}; how the ZN tail divides between them was not measured.
    \item \emph{Messages with no queue still reach twice the service floor at $p_{99}$}, for reasons the records do not show.
}

\section{Summary of Part II: findings and what they mean for HFT system design}
\label{sec:part2-summary}

Part II recorded the packet arrivals of NQ front-month on CME, ran a simulated receiver
over them at a range of service times and stage counts, followed the same transactions
upstream through the exchange's matching engine and publisher, and checked the results
against a live production receiver. This section collects what it found and states, for
each finding, what it means for the design of an HFT trading system. Part III turns these
consequences into a design procedure for a concrete system. All results are for one
instrument.

\paragraph{The arrival stream.} Packets do not arrive at random. They come in clusters
that a Hawkes process describes as near-critical, with a branching ratio of about $0.8$
(Section~\ref{sec:setup-hawkes}). The clustering belongs to the matching engine's
transactions, not to how the exchange packs them into packets
(Section~\ref{sec:results-transactions}). At the engine, one transaction in six follows the
previous one by less than $\spub$, often by a fraction of a microsecond. The
publisher does not pass that on: it is a queue that sends about one packet every
$\spub$, so a burst at the engine reaches the receiver as a train of packets one
period apart, mostly one message each (Section~\ref{sec:exchange}).
\emph{For design:} a receiver cannot be sized from the average packet rate, which is
hundreds of times too slow to matter. What it must survive is a train of packets at the
publisher period, and that period, not the rate, is the number to measure on a feed.

\paragraph{What makes the receiver's tail.} A single-threaded receiver queues only when
two packets arrive closer together than its service time. Below the publisher period it
has no queueing tail at all; above it, the tail grows much faster than the service time
(Section~\ref{sec:results-tail}). Up to about $32\,\mu\mathrm{s}$ what sets it is how many
gaps are shorter than the service time; beyond that it is how long the runs of such gaps
last (Sections~\ref{sec:results-nulls} and~\ref{sec:results-transactions}). The tail
relative to the service time does not change with market load or branching ratio in that
range, because a busier market brings more bursts, not tighter ones
(Section~\ref{sec:results-conditioning}). \emph{For design:} the question to ask of each
stage is whether its service time is above or below the publisher period. The answer does
not depend on the session being quiet or volatile, so it can be settled once per feed and
does not need a live estimate of market conditions.

\paragraph{The production operating point.} The two-stage production receiver, whose
socket reader hands each packet to a decoder thread, runs at about $7\,\mu\mathrm{s}$, just
under the publisher period. There the queue across packets is
small, and most of the tail comes instead from long packets, whose messages are decoded one
after another while the packets behind them queue, and from service times that vary for
reasons the records do not show (Sections~\ref{sec:crossval-spanrun}
and~\ref{sec:fungibility-measured}). A one-message packet is served below the period and a
packet of two or more messages above it (Section~\ref{sec:crossval-floor}).
\emph{For design:} at this operating point the levers are the cost of each extra message
and the spread of service times, not the number of threads. A stage cut does not shorten
the decode of a long packet. What matters is the upper quantiles of per-packet service
under load, measured on the real feed, compared with the publisher period. The same
receiver decoding on its socket-reader thread updates the book in about $1\,\mu\mathrm{s}$
at the median (Section~\ref{sec:crossval-onethread}), far below the period, where a train
of packets at the period builds no queue behind a typical packet and one thread is the
right design; moving the receiver well under the period is the remedy this analysis
prescribes.

\paragraph{What reduces the tail.} Cutting the servicing chain into stages on separate
threads costs one hop per stage on every message and shrinks the queueing tail, because
only the slowest stage queues (Section~\ref{sec:results-design}). The tail appears between
$7$ and $8\,\mu\mathrm{s}$; a cut starts to pay between $8$ and $10\,\mu\mathrm{s}$, once the
tail is larger than the hop, and from $16\,\mu\mathrm{s}$ it pays on every window
(Sections~\ref{sec:results-threshold} and~\ref{sec:results-windows}). The best number of
stages grows with the service time. For a stage without state, such as decode, a pool of
cores that each take whole packets does as well as a cut at two cores and better from four,
all costs charged (Section~\ref{sec:results-equal-core}). A cheaper hop widens the range
where cutting pays but does not move the threshold (Section~\ref{sec:results-h}).
\emph{For design:} cut a chain whose service time is above the split threshold, at the
natural boundary that makes the slowest stage shortest; a cut that leaves the slowest
stage as it was only adds a hop. Merge short stages onto one thread. Give a stateless
stage a pool of cores instead of a cut when four or more cores are available; the order
book, which carries state, can only be cut.

\paragraph{The theory and the live system.} The two analytical results of Part~I hold on
every message of every window, and the Poisson stream behaves exactly as the theory
predicts (Section~\ref{sec:evaluation}). A live production receiver, on three instruments
and a different clock, shows the same clustering (Section~\ref{sec:crossval-agree}).
\emph{For design:} the design equation of Section~\ref{sec:results-design} can be applied
with confidence to a real pipeline: its terms are exact, and its inputs, the publisher period and the service time of each stage, can be measured directly.

\findings{
    \item \emph{The receiver's tail threshold is the publisher period $\spub$.} \emph{Design:} measure it once per feed and compare every stage with it.
    \item \emph{Above it, transaction timing sets the tail and only the slowest stage matters.} \emph{Design:} cut chains above the split threshold of about $10\,\mu\mathrm{s}$ at the boundary that shortens the slowest stage.
    \item \emph{Just under the publisher period, where the two-stage receiver runs, long packets and variable service carry most of the tail.} \emph{Design:} reduce per-message cost and service spread; do not add threads.
    \item \emph{A pool of whole-packet servers beats a cut from four cores for a stage without state.} \emph{Design:} pool decode, cut the order book.
    \item \emph{The threshold does not move with market load.} \emph{Design:} no runtime estimate of market conditions is needed.
}

\paragraph{What Part II leaves open.} Only NQ front-month was measured, and only its
messages were counted in each packet; the full channel, ZN and other instruments are open
(Section~\ref{sec:open-questions}). Service time in the simulation is constant or grows
with message count; its spread on a live receiver is not modelled. Why orders reach the
matching engine within a microsecond of each other is a question about the market and is
not answered here (Section~\ref{sec:concl-cannot}). A designer applying these results to
another feed should measure the publisher period and the per-packet service there rather
than assume them.

\part{Best practices: implications for HFT systems design}

\section{The reference instance: \kasparhft{} and its \texttt{fast\_send} pipeline}
\label{sec:kaspar}

Part~III applies the design consequences of Part~II, collected in
Section~\ref{sec:part2-summary}, to a concrete system. This section introduces the system that the rest of Part~III works on, and applies the
design rule of Section~\ref{sec:results-design} to it. The rule is to cut a chain only
where the cut shortens its slowest stage, and never once that stage is already below the
publisher period, because below it there is no tail left to remove. It first says what \kasparhft{} is and names the two ways its
actors hand messages to one another, one synchronous and one queued, because the choice
between those two at each boundary is exactly what the design rule decides. It then lays
out how the production pipeline is actually wired: which boundaries are queued and why,
and which stages run as one unbroken chain on a single thread. Finally it runs the design
equation on that chain. On the chain as it stands, which takes about seven microseconds
from socket read to book update behind the reader's hand-off and about one at the median
when the reader runs it itself, the rule says do not split it. On a longer chain, sixteen microseconds, it says make one balanced
cut and no more. The section after this one turns those verdicts into a procedure for
changing the real system one boundary at a time and measuring the result.

The recommendations in this Part are stated for HFT systems in general and worked out
on one system.
\kasparhft{} \citep{kaspar_repo,mayeski_fastsend} is an open-source C++20
trading system for CME futures built on a shared-nothing actor framework, available at
\url{https://github.com/vincent212/kaspar-hft}: every component (socket reader, feed
arbitrator, SBE decoder, order book, strategy, order manager, iLink session) is an actor
with private state and a mailbox, and components interact only by passing messages. It is
at once a production system (MDP3 multicast in, full order books reconstructed order by
order, authenticated iLink~3 sessions out) and, with the same actor code driven from recorded
packet captures, a deterministic replay simulator in its own right. Part~II does not run
on that simulator: its measurements are made on a standalone tandem-Lindley model of the
pipeline (Section~\ref{sec:setup-simulator}), driven by packet-arrival timestamps
extracted from the same class of capture. The framework supplies the architecture under
study and the calibration of $h$; it is not in the measurement loop. The earlier paper
\citep{mayeski_fastsend} documents the four extensions that make the actor model usable
at HFT latencies. \texttt{fast\_send}, the synchronous send, is a delivery in which the
sending thread runs the receiver's handler inline and takes the reply as a value. Actor
\texttt{Group}s co-schedule a set of actors on one thread behind a single mailbox. Mailbox
queue implementations are selectable per actor. A memory pool removes per-message
allocation and its tail. The defining property of \texttt{fast\_send} is
\emph{receiver transparency}: a handler is written once and has no way to tell whether it
was invoked synchronously or asynchronously, on which thread, or whether its sender is
blocked. Section~\ref{sec:implementation} builds on that property.

The same paper built and measured a production trading pipeline; that pipeline is the design this paper re-examines. It has three
sections joined by exactly two cross-thread hops. Two socket-reader actors, one per
redundant multicast channel, do nothing but drain their UDP sockets and hand datagrams to
a buffer actor by an asynchronous \texttt{send}; that hop is asynchronous because a reader
that blocked on its callee would drop packets under a burst. From the buffer, arbitration
between the A and B feeds, SBE decode, order-book construction, and signal generation run
as a single chain of \texttt{fast\_send}s on one thread, with no queue and no context
switch between stages. When the signal decides to act it hands the order to an
order-management actor by a second asynchronous \texttt{send}, because the iLink write
takes microseconds and blocking the decode thread on it would stall the market-data queue.
The measured cost of each delivery path is given in Section~\ref{sec:implementation}.
Measured through this pipeline on live ES, NQ, and ZN feeds, socket-to-book time is
$\sim 7\,\mu\mathrm{s}$; with decoding moved onto the socket-reader thread it is about
$1\,\mu\mathrm{s}$ at the median, of which a synchronous \texttt{fast\_send} hop is about $1\%$.
The same measurement found the arrival stream strongly clustered, with a Fano-implied
branching-ratio ceiling of $0.85$--$0.97$. That
observation motivated the present paper.

Read against Part~II, this pipeline is the single-thread servicing chain of
Section~\ref{sec:intro-singlethread-priors} realised with actors: the whole $T$ sits on
one thread between two hops that were placed where blocking was unacceptable. The present paper adds that the same logic applies inside the chain once the chain is
long enough for the feed to queue behind it. The recommendations that follow are stated as
changes to this pipeline: where to cut it, in what order to place the stages, which
boundaries to leave on \texttt{fast\_send}, and how to decide each by measurement.

\paragraph{The design rule applied to this pipeline.} The rule of
Section~\ref{sec:results-design}, cut only where the cut shortens the slowest stage and
never below the publisher period, can be run on the reference chain directly. It needs one
input: the tail excess that a single stage of a given length $s$ would carry,
$\Delta_{\mathrm{single}}(s)$. Table~\ref{tab:delta} supplies it. Section~\ref{sec:results-design}
showed that the excess depends only on the per-stage service $T/N$, so every entry of that
table is the excess of one stage of length $T/N$, and read that way the table is a lookup:
zero for a stage shorter than the publisher period, and growing fast above it.

Whatever the split of the chain across its natural boundaries (arbitration, SBE decode,
book application, signal), every partition has its slowest stage at $7\,\mu\mathrm{s}$ or
less, since the two-stage receiver's service floor of about $7\,\mu\mathrm{s}$ already
contains the whole chain and the hand-off in front of it. That is just under the publisher
period, where Part~II finds no queueing tail in almost every window
(Section~\ref{sec:results-threshold}); run on the reader thread, the chain takes about
$1\,\mu\mathrm{s}$ at the median for a book update (Section~\ref{sec:crossval-onethread}),
far below it. The chain has
no queueing excess to remove, every cut returns exactly its hop tax, and $N^\star = 1$.
The tail it does carry comes from the few packets that carry many messages: their decode
one message after another, which a stage cut does not shorten, and the queue their long
service opens behind them (Section~\ref{sec:fungibility-measured}).

\paragraph{The design rule applied to a longer chain.} The rule applies once the chain
exceeds the publisher period, for instance when a strategy stage is added, or on a feed with a shorter
publisher period. Take a
$16\,\mu\mathrm{s}$ chain with the same four boundaries, decode and book now at
$6\,\mu\mathrm{s}$ each and the two small stages at $2\,\mu\mathrm{s}$:

\begin{center}
\begin{tabular}{@{}lrrrr@{}}
\toprule
partition of the $16\,\mu\mathrm{s}$ chain & $s_{\max}$ (\si{\micro\second}) & hops & $\Delta_{\mathrm{single}}(s_{\max})$ (\si{\micro\second}) & $p_{99}$ excess $+$ hop tax \\
\midrule
single thread ($16$)                        & 16 & 0 & 18.59 & 18.59 \\
$8 + 8$ (arb$+$decode\,$|$\,book$+$signal)  &  8 & 1 &  0.61 &  2.31 \\
$8 + 6 + 2$ (three stages)                  &  8 & 2 &  0.61 &  4.01 \\
$2 + 6 + 6 + 2$ (all four split)            &  6 & 3 &  0.00 &  5.10 \\
\bottomrule
\end{tabular}
\end{center}

Every entry is read off Table~\ref{tab:delta} and the refined sweep of
Section~\ref{sec:results-threshold}. Three points follow. The balanced two-way cut is the
best option: it removes almost all of the tail excess for one hop, about an eightfold
improvement in the objective on these illustrative stage times. Three stages do not lower $s_{\max}$, since the $8\,\mu\mathrm{s}$ stage
remains the bottleneck, so the second hop costs $1.7\,\mu\mathrm{s}$ and removes nothing
(Remark~\ref{rem:unequal}, third consequence). Splitting all four lowers $s_{\max}$ to $6$
and removes the small remaining excess at the cost of two more hops, and is worse. The recommendation for a chain of this length is one cut, placed so that the two
halves are as equal as the natural boundaries allow, with the larger half at ingress
(Remark~\ref{rem:dissipation}). The numbers are a prior for the end-of-build measurement
of Section~\ref{sec:fungibility-workflow}, not a substitute for it. The per-stage service
times are illustrative and are replaced by the operator's profile of the chain.

\section{Implementation on \kasparhft{}: setting the send mode of each boundary}
\label{sec:implementation}
\label{sec:fungibility}

The practical claim of this paper is that on Hawkes-clustered feeds the choice of
\texttt{fast\_send} versus asynchronous \texttt{send} at each actor boundary is a
late-stage optimisation decision, made per boundary from measured end-to-end median and
tail latency of the assembled pipeline on the target feed, rather than an early
architectural commitment. \kasparhft{} supports this workflow. This section states the
property that makes it possible and the procedure it enables.

The design equation maps onto framework choices in a specific way. The paper's model of an
$N$-stage tandem assumes each stage's actor runs on its own thread and sends its output to
the next stage through an asynchronous mailbox; the hop cost $h$ is the wall-clock
latency of that handoff.

\paragraph{Two hop paths in the shipping \kasparhft{} framework.} Of the delivery
mechanisms described in Section~\ref{sec:kaspar}, two matter for the design equation. The
default asynchronous \texttt{send} uses \texttt{BQueue}, a blocking mailbox built on
\texttt{std::mutex} $+$ \texttt{std::condition\_variable} $+$
\texttt{boost::circular\_buffer}; a producer thread enqueues under the mutex and a
consumer thread waits on the condition variable. This is the hop path the paper's
Section~\ref{sec:setup-simulator} model corresponds to, and halving its measured $\sim
3.4\,\mu\mathrm{s}$ round trip gives the one-way hop cost $h \approx 1.7\,\mu\mathrm{s}$
used throughout Part~II.

The \texttt{fast\_send} path is different: it is a \emph{synchronous} dispatch that
acquires a per-actor mutex, calls the receiver's handler inline on the sender's thread,
and returns the reply. There is no queue and no cross-thread wake-up. Its per-call cost is
much smaller than the BQueue path's, but it is not a pipelined hop: the sender blocks
until the receiver finishes, so a chain of \texttt{fast\_send}s serialises through all
stages on one thread and offers none of the throughput speedup of
Proposition~\ref{thm:burst}. The paper's design equation does not apply to
\texttt{fast\_send} chains.

\subsection{Receiver transparency of \texttt{fast\_send} and \texttt{send}}
\label{sec:fungibility-transparency}

Every actor in \kasparhft{} registers its handlers through the
\texttt{MESSAGE\_HANDLER(MsgType,} \texttt{handler\_method)} macro. A handler is a function of one
message argument; nothing in its body can access, observe, or condition on the delivery
mode by which it was invoked. This property is called \emph{receiver transparency} in
\citet{mayeski_fastsend} and is what makes the send-mode choice a purely sender-side
decision.

Concretely:
\begin{itemize}
  \item \textbf{Synchronous delivery:} \texttt{target->fast\_send(m,\ this)}. The sender's
thread acquires \texttt{target}'s per-actor mutex, calls the receiver's handler inline on
the sender's thread, and takes the reply as a return value. No queue, no cross-thread
wake-up. Round-trip cost $\sim 30$ ns per Table~4 of \citet{mayeski_fastsend}.
  \item \textbf{Asynchronous delivery:} \texttt{target->send(m,\ this)}. The message is
enqueued on \texttt{target}'s BQueue mailbox (\texttt{mutex} $+$
\texttt{condition\_variable} $+$ \texttt{boost::circular\_buffer}); \texttt{target}'s own
thread wakes up on the condition variable and runs the handler. Round-trip cost $\sim
3.4\,\mu\mathrm{s}$ across threads, or $\sim 90$ ns same-thread within a \texttt{Group},
per the same reference.
\end{itemize}
Both paths deliver the same handler with the same message. Neither the handler's body nor
any actor downstream of it can tell which path fired. Switching a boundary between the two
is a one-line change at the sender's call site --- \texttt{target->fast\_send(m,\ this)}
$\leftrightarrow$ \texttt{target->send(m,\ this)} --- with no code change on the receiver
side, no header change, no interface change, no test change.

\subsection{End-to-end measurement of the boundary choice}
\label{sec:fungibility-why-measure}

The end-to-end latency profile of an HFT pipeline is a joint property of (i)~which stages
exist and where they are cut, (ii)~which pairs of adjacent stages are joined by
\texttt{fast\_send} versus \texttt{send}, and (iii)~the arrival law of the input feed.
Proposition~\ref{thm:burst} and Theorem~\ref{thm:reduction} both fold across multiple
boundaries: the throughput speedup is $N$-fold and the tail contraction is $N^{-\gamma}$
with $\gamma \ge 1$, so the joint optimum is not deducible from local reasoning at a
single boundary. A designer who reasons one boundary at a time may underestimate the
aggregate tail benefit of adding hops; this is the reasoning behind the single-thread
convention of Section~\ref{sec:intro-singlethread-priors}.

Getting the joint optimum therefore requires end-to-end measurement of the assembled
pipeline against a representative feed, and iteration over per-boundary flips. The
observable objective is a two-tuple $(p_{50}, p_{99})$ of end-to-end latency under the
target feed, with an operator-chosen trade-off between the two
($\kappa$ of Section~\ref{sec:results-design}). The design equation of
Section~\ref{sec:results-design}, $N^\star = \arg\min_N [(N-1)h + \kappa\,\Delta(N)]$,
or in its partition form the cut with the smallest reachable $s_{\max}$ and then the
fewest hops, supplies a prior on where the optimum lies under a Hawkes feed; the
measurement decides.

\subsection{The measure-and-flip procedure for setting each boundary}
\label{sec:fungibility-workflow}

The procedure has five steps. Draft the pipeline as a chain of actors at the natural
cut points. Wire each boundary with whichever of \texttt{fast\_send} or \texttt{send} is
convenient during development. Once the pipeline runs end to end, measure the median and
$p_{99}$ on a representative tape as a baseline. Flip one boundary at a time, re-measure,
and accept the flip if the pair improves under the operator's trade-off. Iterate until no
single flip improves the objective. The full design flow,
including the steps that precede this, is given in Section~\ref{sec:recommendations}.

Under the arrival law of a CME MDP3 feed the design equation of
Section~\ref{sec:results-design} predicts what the iteration converges to. Asynchronous
\texttt{send} goes at every boundary whose removal would lengthen the largest stage.
\texttt{fast\_send} goes everywhere else: between short sub-tasks below the publisher
period, and between sub-tasks downstream of the bottleneck whose combined service stays at
or below the largest stage (Remark~\ref{rem:unequal}). The measurement decides; the prediction indicates where to expect it to converge.

\subsection{Thread mapping and send mode as deployment decisions}
\label{sec:fungibility-decoupling}

A consequence of Section~\ref{sec:fungibility-workflow}: in \kasparhft{} two distinct
topology decisions are not knowable at design time.
The first is the \emph{actor-to-thread mapping} --- which actors share a thread by being
placed in a common \texttt{Group}, and which run on their own thread. The second is the
\emph{per-boundary send-mode choice} --- whether each pair of adjacent actors is joined by
synchronous \texttt{fast\_send} or asynchronous \texttt{send}. Neither decision is
deducible from the source code alone. Both are functions of the target feed's arrival
statistics and the operator's median-versus-tail trade-off, both of which are properties of
the deployment environment. The optimum can only be measured from the assembled system
running against a representative feed.

Under receiver transparency this is not a limitation. Actor code,
message-type headers, and the full suite of unit and integration tests can be written and
run before either the actor-to-thread mapping or the send-mode choice has been decided.
Code construction and testing therefore proceed independently of the latency-optimisation
consideration --- both decisions are deferred to the end of the build, informed by
measurements that could not have been taken any earlier because the system did not yet
exist to measure. Correctness is established first; latency is optimised against the
specific feed the system will run on, second. Two developers building two actors on the
two sides of a boundary do not need to agree on whether their actors will ultimately live
on the same thread or on different threads, nor on whether their handoff will be
synchronous or asynchronous; they need to agree only on the message type. The message-type agreement is checked at compile time; the two topology decisions are
parameters set at deployment.

This decoupling of correctness from latency optimisation is the practical consequence of
receiver transparency, and distinguishes the design cycle from that of frameworks in
which either topology decision enters the receiver's interface
(Section~\ref{sec:fungibility-contrast}). Design proceeds without
a commitment to a topology; the commitment is made once, near deployment, on the basis of
measurements the code did not need to know about while it was being written.

\subsection{Frameworks in which the send mode is part of the receiver's type}
\label{sec:fungibility-contrast}

Most low-latency actor frameworks bake the sender's choice of delivery mode into the
receiver's type. An Akka actor's \texttt{tell} versus \texttt{ask}, a CAF actor's
\texttt{send} versus \texttt{request}, and an Erlang process's cast versus call each
expose different receiver-side surfaces, so flipping a boundary from sync to async is a
small refactor rather than a one-line change. Under those frameworks the design flow of
Section~\ref{sec:fungibility-workflow} is not tractable, and the operator's incentive is
therefore to lock in the send-mode choice early --- reasoning locally --- rather than to
defer it to an end-of-build measurement pass. This incentive is a plausible contributing cause of the single-thread convention: frameworks that penalise deferral push
designers toward the topology that is least sensitive to a wrong early guess, which is the
collapsed single-thread servicing chain.

\kasparhft{}'s joint mailbox API of \texttt{fast\_send} and \texttt{send}, with
receiver transparency, is what makes measurement-driven send-mode optimisation practical.
The design rule depends on this property operationally.

\section{Pre-spinning receivers and pre-warming on self-exciting feeds}
\label{sec:other-arch}

The paper's central claim is that the single-thread hot-path recommendation does not
hold above a service-time threshold on Hawkes-clustered feeds
(Section~\ref{sec:intro-singlethread-priors}, Section~\ref{sec:results}). Two adjacent
conventions, pre-spinning receivers and pre-warming pipelines, are also affected by the
arrival law, and this section examines them. Neither intervention is measured
here; both are future work. The argument is that the arrival law provides much of the
benefit these interventions are intended to deliver, on the tail quantile.

\paragraph{Pre-spinning receivers.} The busy-poll recommendation of
\citet{rigtorp_lowlatency}, \citet{belay2014ix}, and the Aeron/Chronicle lineage
(Section~\ref{sec:intro-singlethread-priors}) pins a spinning thread per receiver so the
wake path from a sleeping condition variable is never taken. On a Poisson feed, any
arrival can find the receiver asleep, and the wake latency is paid on average;
pre-spinning eliminates it. On a Hawkes feed with branching ratio $n \approx 0.8$,
sleep/wake transitions happen only at cluster boundaries. Every $p_{99}$ message is inside a cluster, since at these service times the Poisson
null produces no wait at all (Section~\ref{sec:results-tail}). By the time a tail-setting
message arrives, then, the receiver is busy. At short service times it is still serving
the predecessor the publisher sent one period earlier; at long ones it has been processing
a run for many multiples of the per-message service time
(Section~\ref{sec:results-nulls}). In either case it is effectively spinning
as a natural consequence of the load. Under the model used here, pre-spinning therefore saves little or
no wake latency at the tail; this is a prediction of the queueing model, not a measurement
of a production receiver's state, and should be evaluated directly. It saves wake latency
for the median and for the minimum latency (inter-cluster arrivals which do find the
receiver asleep after seconds-scale gaps), but the saving is at most the wake-path part of one hop per boundary (Section~\ref{sec:results-h}), and is much smaller than the tail win from splitting. There is also an operational cost. In the present author's production experience,
pinning spinning threads across the servicing chain has degraded the system as a whole:
steady core burn, cache thrash between the spin loop and the actor's own working set, and
pressure on shared resources such as the last-level cache and memory bandwidth, which
together offset the median-latency gains. This is anecdotal and not measured here.

\paragraph{Pre-warming.} Pre-warming is a subset of the wake-up cost discussed above. The $h = 1.7\,\mu\mathrm{s}$ one-way hop cost
that this paper calibrates already captures the entire process-wake-up path on the
receiver side: the kernel context switch, the reload of the receiver's instruction cache,
the L1/L2 misses on its actor state and message pool, the TLB refill, and the scheduler
putting the thread back on a core. To the extent the thread hop is dominated by
cluster-wait on Hawkes feeds --- because every $p_{99}$ arrival is intra-cluster and the
receiver is, under the model, already awake --- the cache-warmth components inside that
$h$ should follow the same pattern, since they are the same $h$ decomposed. This is a
consequence of the model rather than a measured result. Pre-warming interventions that
fire synthetic traffic to keep instruction caches, data caches, TLBs, and branch
predictors hot are therefore attacking the same $h$ that Section~\ref{sec:results} finds
is dominated by cluster-wait. What such interventions can meaningfully deliver is confined
to the median and the minimum --- the inter-cluster first-arrival case where the receiver
truly did cool down --- and even there the delivery is bounded by the wake-up cost itself
at the cost of continuous L1/L2 pollution from the warmup traffic. Under Hawkes the
tail is unaffected.

\paragraph{Measurement of both interventions.} Pre-spinning and pre-warming are first-order
tail-latency techniques when arrivals are uncorrelated. On self-exciting exchange feeds the arrival law itself provides the
receiver-awake and cache-warm state on the tail quantile, so both interventions are
second-order there and can be counterproductive once their steady costs are included. Both claims are
corollaries of the paper's main result rather than separate contributions, but each
deserves a dedicated empirical evaluation on Hawkes-clustered feeds, analogous to the
single-thread evaluation this paper carries out. The two open questions are:
\begin{itemize}
    \item What is the actual ROI of pinned busy-poll receivers on the Hawkes-load $p_{50},
p_{99}, \max$ triple, net of the steady core-burn cost and the cache pressure the spin
loop imposes on the actor's own working set?
    \item What is the actual ROI of periodic pre-warming traffic on the same triple, net
of the L1/L2 pollution it creates and the residual advantage it delivers only on
inter-cluster first-arrivals?
\end{itemize}
Both are direct extensions of the present paper's methodology (measure under real CME MDP3
tapes; use a matched Poisson null to isolate the arrival-law dependence) and both are
candidate follow-ons.

\section{Recommendations for HFT system designers}
\label{sec:recommendations}

Section~\ref{sec:part2-summary} states, finding by finding, what Part~II means for the
design of an HFT trading system. This section puts those consequences in the form of recommendations.
The paper's results qualify one architectural default (the single-thread hot path as an
unconditional choice) and question two adjacent ones (busy-poll receivers,
pre-warming). They do not contradict
per-message code optimisation. This section separates what the paper leaves unchanged
from what it argues against, and states the design flow the design equation implies.

\begin{center}
\fbox{\begin{minipage}{0.92\linewidth}
\textbf{Design principle (measure the split threshold on the target feed).} The split
threshold is the service time above which cutting a stage into two lowers its $p_{99}$.
On NQ it lies between $9$ and $10\,\mu\mathrm{s}$, just above the publisher period by the
amount the tail needs to exceed the hop cost (Section~\ref{sec:results-threshold}). It is
set by how many packets arrive closer together than the stage's service time, which on
CME depends on the publisher period and not on the packet rate or the fitted branching
ratio at HFT service times. On a feed that packs a large share of its gaps at its minimum
period, the same stage would be split (Appendix~\ref{app:synthetic}), so the threshold
cannot be read off fitted parameters. Measure it on the target feed, under the load the
system will see; the full rule is the design-rule box of
Section~\ref{sec:concl-summary}.
\end{minipage}}
\end{center}

\subsection{Recommendations left unchanged}

Per-message hot-path optimisation --- branch-free decoders, cache-line-aligned
actor state, atomic-free single-writer queues, SBE decode in place, memory pools,
\texttt{fast\_send} at $\sim 30$ ns, kernel-bypass NICs (Solarflare Onload, DPDK, VMA),
\texttt{isolcpus}, \texttt{SCHED\_FIFO}, huge pages, RTC counters instead of
\texttt{clock\_gettime}, non-blocking IO on the socket, avoidance of allocation on the hot
path, careful use of the branch predictor --- remains necessary. All of it is contained
in the service time $T$, which the paper takes as given.
Proposition~\ref{thm:burst} gives an $N$-factor throughput speedup at any $T$, and the
$1/N$ bound of Theorem~\ref{thm:reduction} applies to a single-stage wait $\Delta(1)$
that itself scales with $T$, so the absolute gain from splitting grows with $T$ while
the hop cost does not. The splitting rule is therefore complementary to per-message optimisation. A two-stage
tandem on a $20\,\mu\mathrm{s}$ servicing chain has lower latency at every quantile than
the same tandem on a $40\,\mu\mathrm{s}$ chain, and the difference is entirely in the
per-message layer. Per-message optimisation that brings a chain below the publisher period
removes the case for splitting it at all (Section~\ref{sec:results-threshold}).

\subsection{Recommendations the paper argues against}

Two conventions in the layer above per-message optimisation. First, the single-thread
servicing chain as the unconditional topology default (Section~\ref{sec:intro-singlethread-priors}): on
Hawkes-clustered feeds with $T$ above the publisher period of Section~\ref{sec:results-design}, the
tail is dominated by cluster wait and adding threads along the chain reduces it at a
median cost of $(N-1)h$. Second, pre-spinning receivers and pre-warming pipelines as
always-on interventions (Section~\ref{sec:other-arch}): on Hawkes-clustered feeds the
receiver is already busy and warm at the arrival of every $p_{99}$ message, so these
interventions affect the median and minimum, not the tail, at a steady CPU and cache cost.
The claim that a faster inner loop implies a single thread is preferable conflates two
claims: that a faster inner loop lowers latency, which the paper does not dispute, and
that a single thread is therefore the best topology, which the paper contradicts above
the publisher period.

\subsection{Recommended design flow}

The design flow that follows from the paper's results, and that we recommend as the
canonical HFT infrastructure workflow on self-exciting exchange feeds:
\begin{enumerate}
    \item \textbf{Retain per-message optimisation.} This is orthogonal to the paper's
recommendation and determines absolute latency. Basis: just under the publisher period,
where the two-stage production receiver runs, the tail that remains is in packets of several messages
and in variable service, so the cost per message and the spread of service are the
levers there (Sections~\ref{sec:crossval-spanrun} and~\ref{sec:fungibility-measured}).
    \item \textbf{Measure the target publisher period, and use it as the split/don't-split
test.} Take the corpus of packet arrivals, form the per-window gap distribution, and read
off a low quantile of it; we use the $1$st percentile. On NQ it is the publisher period
(Section~\ref{sec:results-conditioning}). Compare it with the upper quantiles of each
stage's per-packet service under load, not with the median. A stage below it will not
develop a queue and must not be split. A stage above it should be split, once the tail
exceeds the hop cost, which on this corpus happens a microsecond or two further along
(Section~\ref{sec:results-threshold}). Use the tight end of the gap distribution, never
the mean: the mean gap is some three hundred times longer, and a rule stated on it would
say never to split anything. On NQ the publisher period did not move between quiet and
volatile sessions, so it was measured once per feed; that is a finding about this feed's
publisher, to be checked on another feed rather than assumed.
    \item \textbf{Draft the pipeline as a chain of actors at natural cut points, then
balance it.} On a CME MDP3 handler: \texttt{SocketReader}, SBE decoder, order book,
strategy actor, order manager --- four or five stages, no artificial decomposition
(Section~\ref{sec:implementation}). The tail factor is $s_{\max}/T$
(Remark~\ref{rem:unequal}), so the goal of the decomposition is the smallest largest
stage, and the largest stage belongs at ingress so that all queueing sits in one mailbox
(Remark~\ref{rem:dissipation}). Basis: the tail excess on the corpus depends on the
service time only through the largest stage (Table~\ref{tab:delta}).
    \item \textbf{Divide the stages above the publisher period by state.} A stage that carries no
state from one packet to the next, such as SBE decode, can be served by a dispatch pool of
whole-packet servers; with four or more cores for it, a pool is better than cutting it
into the same number of stages, at two cores the two are level
(Table~\ref{tab:dispatch-costed}). A stage that carries state, the order book above all,
is served by cutting the chain.
    \item \textbf{Wire boundaries with whichever of \texttt{fast\_send} or \texttt{send}
is convenient during development.} Receiver transparency in \kasparhft{} makes this a
one-line-per-boundary decision that can be deferred (Section~\ref{sec:fungibility}).
    \item \textbf{Do not commit to an actor-to-thread mapping at design time.} Grouping
decisions and per-boundary send-mode decisions are both scalar dials at deployment, not
code-structure decisions (Section~\ref{sec:fungibility-decoupling}).
    \item \textbf{Measure the assembled pipeline against a representative feed.}
Corpus-median $p_{50}$ and $p_{99}$ of end-to-end latency, per operator-chosen trade-off
(Section~\ref{sec:fungibility-workflow}).
    \item \textbf{Flip one boundary at a time and re-measure.} Accept the flip if
$(p_{50}, p_{99})$ improves. Iterate until no single-boundary flip helps.
    \item \textbf{Prior on where step 7 converges.} On Hawkes-clustered feeds, async
\texttt{send} at every boundary whose removal would raise the largest stage;
\texttt{fast\_send} at every other boundary, in particular between sub-tasks downstream of
the bottleneck whose combined service stays at or below $s_{\max}$, and between short
sub-tasks below the publisher period. A hop that does not lower $s_{\max}$ costs $h$ at every
quantile and compresses nothing.
    \item \textbf{Do not pre-spin or pre-warm by default.} If pre-spinning or pre-warming
is used, measure the ROI on a Hawkes-representative tape, net of the steady core-burn and
cache-pollution costs. If the measurement improves the tail materially, keep the
intervention. The paper's prior is that on Hawkes feeds it will not
(Section~\ref{sec:other-arch}).
\end{enumerate}
The workflow prioritises measurement over local reasoning because the arrival law of
the target feed is a first-order determinant of the optimum and is not a property of the
code.

\clearpage
\part{Conclusions}

\section{Conclusion}
\label{sec:conclusion}

This paper measured a real exchange feed, over a year of NQ packets and matching-engine
transactions, from the matching engine through the publisher to a receiver, and derived
from what it found how the trading system that consumes it should be designed: what makes a
receiver's latency tail, and how to reduce it. It is not about why the market behaves as
it does. Section~\ref{sec:concl-summary} gives the answer in one place;
the remaining subsections state what the paper establishes about the cause of the tail,
what it does not, how to reduce the tail, what it revises in a prior design context, and
the caveats. All empirical results are for one instrument, NQ front-month on CME.

\subsection{Summary: what causes the receiver's tail and how to reduce it}
\label{sec:concl-summary}

\paragraph{What causes the tail.} The chain starts at the exchange. The matching engine
often processes one transaction less than a microsecond after the previous one
(Table~\ref{tab:engine-clock}). The market-data publisher cannot send that fast. It is a
queue with service time $\spub$ per packet, so a burst at the engine
leaves it as a train of packets one period apart, mostly one message each, with an
occasional larger packet after a longer gap (Tables~\ref{tab:two-clocks}
and~\ref{tab:drain-gap}). The receiver is the next queue in line. If it handles a packet
faster than the publisher sends them, it keeps up with any train and builds no queue from
the trains alone. If it is slower, it falls behind by the difference on every packet of the
train, and the messages at the end of the train wait. At service times a little above the
publisher period the wait is mostly behind the packet or two just ahead; at much longer
service times it is behind runs of transactions spread over a few hundred microseconds,
which is the self-excitation a Hawkes fit measures (Table~\ref{tab:nulls},
Section~\ref{sec:exchange-bursts}).

The two-stage production receiver, whose socket reader hands each packet to a decoder
thread, runs just under the publisher period, at about $7\,\mu\mathrm{s}$, and there two
further mechanisms carry most of its tail. The first is long packets. Decode
cost grows with the number of messages, so a packet with two or more messages takes longer
than the publisher period; its later messages wait for the earlier ones, and the packets
arriving behind it queue while it holds the thread. In the simulation at that service time the decode inside packets supplies about two
thirds of the tail, and on the live receiver the decode cost of a message's position in
its packet reproduces most of the observed $p_{99}$ (Section~\ref{sec:crossval-spanrun}).
On ZN, where each extra message costs more, a few long packets set the far tail
(Section~\ref{sec:fungibility-measured}). The second is
service time that varies from packet to packet: on the live ZN receiver, messages with
nothing queued ahead of them and first in their packet still take about twice the service
floor at $p_{99}$, and any such service in a train can start a queue. The simulation
models the first mechanism only through NQ's cost per message and the second not at all,
so its tails at the production service time are lower than a real receiver's. The packet
rate does not cause the tail at HFT service times: utilisation there is tiny
(Section~\ref{sec:results-conditioning}). Why orders reach the engine within a microsecond
of each other is a question about the market and is left open
(Section~\ref{sec:concl-cannot}).

The remedy the analysis prescribes is to move the receiver well under the publisher
period, and the production receiver bears it out. With its socket-reader thread decoding
each packet and updating the book itself, a book update takes $0.8$--$1.1\,\mu\mathrm{s}$
at the median on ES, NQ and ZN (Section~\ref{sec:crossval-onethread}), far below the
period, where a train of packets at the period builds no queue behind a typical packet and
one thread is the right design. The cross-thread hand-off is the main structural
difference from the two-stage receiver, though not the only difference between the runs,
and the one-thread receiver still has a tail, largest on the trade streams, whose packets
carry the most messages.

\paragraph{How to reduce it: split serially.} Cut the servicing chain into stages on
separate threads, each passing its output to the next through a mailbox hop. Every message
pays one hop per boundary at the median, and the queueing excess shrinks by at least the
number of stages, because only the slowest stage queues: the tail is set by the service
time of the largest stage (Theorem~\ref{thm:reduction}). A stage below the publisher period does not queue at all. On NQ, one cut into two halves removes most of the tail in
every window at every service time well above the publisher period; the best number of stages grows
with the service time, limited by the natural cut points in the chain
(Section~\ref{sec:concl-mitigate}). Serial splitting is the only option for a stage that
carries state from one message to the next, such as the order book.

\paragraph{How to reduce it: split in parallel.} Give whole packets to a pool of cores,
each running the full stage, and put them back in order on the way out. With the hop in,
the hop out and the resequencing all charged, the median costs two hops whatever the pool
size, and almost no packet waits to be put back in order (Table~\ref{tab:dispatch-costed}).
A pool of two cores is level with a two-stage serial cut; a pool of four or more is better
than a serial cut into as many stages (three cores were not simulated). It works only for a
stage without state, such as SBE decode. A design that decodes in a pool and feeds one
order-book thread combines the two and has not been costed here
(Section~\ref{sec:future}).

\begin{center}
\fbox{\begin{minipage}{0.92\linewidth}
\textbf{Design rule (NQ front-month on CME).} Measure the service time of the servicing
chain per packet under real market load, including multi-message packets and bursts, and
compare its upper quantiles, not its median, with the publisher period (about
$7.5\,\mu\mathrm{s}$).
\begin{itemize}
    \item Upper quantiles below the period: keep one thread. Splitting only adds hop
cost.
    \item Median below the period but upper quantiles above it, the regime of the
two-stage production receiver: the tail comes from long packets and slow services. Reduce the
per-message decode cost and the spread of service times; a stage cut does not shorten
the decode of a long packet.
    \item Typical service above the period, stage carries state: split serially, at the
natural boundary that minimises the largest stage, with more stages the longer the
service time.
    \item Typical service above the period, stage carries no state, four or more cores
available: dispatch whole packets to a pool and resequence.
\end{itemize}
On another feed, measure its publisher period first; the threshold is that period.
\end{minipage}}
\end{center}

\subsection{What causes the tail: what the paper establishes}
\label{sec:concl-can}

The questions were what produces the queueing tail behind a single-threaded receiver ---
the packet rate, the clustering of arrivals, or the number of messages per packet --- and
are two queues better than one. The first is answered by controlled simulation on the
real arrival times: each counterfactual stream of Section~\ref{sec:setup-nulls} removes
one property of the stream and keeps the rest, so the change in the tail measures that
property's contribution. The following hold on NQ.

\begin{enumerate}
    \item \emph{The tail is produced by the timing of matching-engine transactions and, at
production service times, by the number of messages per packet.} Putting transactions at
random moments removes most of the tail (Table~\ref{tab:tx-arms}); the message count is
claim~5.
    \item \emph{Not by the rate.} Utilisation is tiny at every service time an HFT hot path
occupies, and the busiest windows have the same tail, relative to the service time, as
the quietest (Section~\ref{sec:results-conditioning}); the rate matters only at service
times far above the publisher period, where ordinary utilisation queueing begins. A
stream that keeps every second's packet count and places packets at random within the
second has no tail at HFT service times (Table~\ref{tab:nulls}).
    \item \emph{Not by the exchange's packetisation.} Almost every transaction fits in one
packet, the pieces of a split one arrive well apart, and the short gaps lie between
different transactions; merging each transaction into one packet or spreading its
packets leaves the tail about as it was (Section~\ref{sec:results-transactions}).
    \item \emph{Not by how sizes are ordered in time.} Decoupling transaction sizes from
their timing, or shuffling packet sizes across packets, leaves the tail about as it was
at service times above the publisher period (Table~\ref{tab:tx-arms},
Section~\ref{sec:crossval-spanrun}). This says the order in which long packets arrive does
not matter; it does not say their length does not matter, which is claim~5.
    \item \emph{Just under the publisher period, where the two-stage production receiver
runs, long packets carry most of the tail.} Decode time grows with the number of messages, so a long packet delays its
own later messages and holds the thread while the packets behind it queue. At that
service time this is about two thirds of the simulated tail, and on the live receiver the
decode cost of a message's position in its packet reproduces most of the observed
$p_{99}$; it sets the live ZN far tail (Sections~\ref{sec:crossval-spanrun}
and~\ref{sec:fungibility-measured}). At service times well above the publisher period it adds a few
percent. Most multi-message packets hold several transactions packed together by the
publisher (Section~\ref{sec:crossval-span}).
    \item \emph{The timing acts through two channels whose weight depends on the service
time.} At short service times it acts through consecutive transactions sent back to back
at the publisher period, most of which the engine processed almost together
(Table~\ref{tab:two-clocks}): a stream with the real gaps in random order keeps most of
the tail there. At long service times it acts through runs of transactions, which the same
shuffle breaks up, keeping little of the tail (Table~\ref{tab:nulls}). The runs are the
self-exciting component of the transaction stream, which a Hawkes fit measures
(Table~\ref{tab:tx-fano}).
    \item \emph{The shortest gap the receiver sees, and the pile-up of gaps just above it,
are the exchange publisher's output.} On the matching engine's clock one consecutive
transaction in six follows the previous one within the publisher period, against about
one packet in a hundred on the publisher's clock (Table~\ref{tab:engine-clock}), and the
delay from engine to packet grows about tenfold during engine bursts
(Table~\ref{tab:publisher}). The publisher is a queue sending about one packet per
$\spub$, and the receiver is the second stage of a tandem that begins at the
exchange (Section~\ref{sec:exchange-tandem}).
    \item \emph{The results do not depend on the transaction grouping.} Regrouping by exact
\texttt{transactTime}, checked against the end-of-event flag, leaves them unchanged
(Section~\ref{sec:exchange-events}), and session-bootstrap intervals on the key shares are
a few points wide (Tables~\ref{tab:nulls} and~\ref{tab:tx-arms}).
    \item \emph{The arrival statistics are not an artefact of replay.} A live production
receiver, on other instruments and a different clock, gives the same clustering
statistics, and the service floor of its two-stage configuration coincides with the
service time at which the simulated tail appears (Section~\ref{sec:crossval}).
\end{enumerate}

\subsection{What causes the tail: what the paper does not establish}
\label{sec:concl-cannot}

\paragraph{Why the matching engine bursts.} The paper traces the receiver's short gaps
to the exchange's publisher draining a backlog, and the backlog to the matching engine
processing events a fraction of a microsecond apart (Section~\ref{sec:exchange}). It does
not show why orders reach the engine so close together. They do not reach the receiver
together: when the engine processes two transactions almost at once, the publisher nearly
always sends them in separate packets one period apart (Table~\ref{tab:two-clocks}). The
open question is about the engine's input. Three mechanisms remain open. (a)~Reaction:
a transaction causes other participants to cancel, replace or trade within microseconds,
each as a separate transaction. (b)~Common reaction: many participants respond to the
same public signal at once, their orders reach the matching engine together, and the
engine processes them one after another. (c)~Pacing or batching inside the exchange's
order entry, with no market information in the spacing. Separating them needs the time
each order \emph{entered} the exchange, not only when it was processed and published: (a)
predicts that the second order entered after the first transaction was published, (b)
that both entered before either was published. The public MDP3 feed carries processing
and publication times only, and the paper's own order-entry sessions record only its own
orders.

The nearest published evidence is \citet{aquilina2022arms}, who had message-level data
from the London Stock Exchange with participants' order-entry times and found that
\emph{races} --- several participants sending messages at the same price level in response
to the same public signal, within microseconds of each other --- are frequent and account
for a material share of trading volume. That is mechanism~(b). The sub-microsecond
engine-clock gaps are what (b) would produce, and \citet{noble2026realitygap} see the same
too-fast events on Nasdaq; it is likely that the same mechanism is at play on CME. The
paper has no data that test it, and states it as a hypothesis.

\paragraph{Whether the publisher's capacity is designed.} The publisher behaves as a queue
with a capacity of one packet per publisher period; whether that capacity is a
design choice or an engineering limit, and whether it acts as a speed limit on the
information side of the market, is stated as a conjecture in
Section~\ref{sec:exchange-conjecture}, with its qualifications.

\paragraph{Self-excitation at HFT service times.} It is not the cause there, and this is
now shown rather than argued. The fitted self-excitation acts over about
$160\,\mu\mathrm{s}$ and places only a few percent of its triggering within the publisher period (Section~\ref{sec:exchange-bursts}); it is the cause of the runs, which carry the
tail at long service times. At short service times the carrier is the publisher draining, at its own period, backlogs built by orders reaching the engine within a microsecond of
each other, whose origin is the open question above.

\paragraph{Events below the publisher period.} Two transactions that the engine processes
closer together than the publisher period reach the receiver one period apart or, less
often, in one packet (Table~\ref{tab:two-clocks}), so every gap statistic on the
publisher's clock is censored there. On the engine's clock they are visible
(Table~\ref{tab:engine-clock}), and on the publisher's clock neither the message count per
packet nor the share of packets carrying several transactions rises with the fitted
branching ratio (Section~\ref{sec:results-conditioning}), so no faster self-excitation
scaling with the branching ratio is hidden in the packets. Faster self-excitation of
constant strength would not be seen by that test.

\paragraph{Why a receiver's service time varies.} On the live ZN receiver, messages that
met an empty queue and were first in their packet, so that neither queueing nor in-packet
decode contributes, still take about twice the service floor at $p_{99}$ and several times
it at $p_{99.9}$ (Table~\ref{tab:zn-serial-tail}). The records do not show why. Because
packets arrive at the publisher period, each such service can start a queue, so this
spread is a source of tail in its own right, and the simulation, whose service is constant
or grows with message count, does not contain it.

\paragraph{How the live far tail divides.} The ZN far tail is argued to be long packets
and the queues behind them, from the measured service floor and cost per message
(Section~\ref{sec:fungibility-measured}). How it divides between the decode inside long
packets, the queue behind them, and slow services was not measured, because the live
per-message records were not retained.

\paragraph{Instruments other than NQ.} One instrument was measured. On ZN the live data
show the message count per packet rising with queue depth, where on NQ it falls, and
whether the two channels compound there has not been tested; no ZN corpus exists. The
paper also treats NQ front-month as a stream on its own, while the publisher serves the
whole channel. On NQ the other instruments are thin; on ZN the Treasury futures along the
yield curve are likely to excite one another, and where they share a channel the publisher
would carry those reactions in the same packets. Cross-excitation between instruments is
a candidate cause of ZN's large packets that the paper does not test
(Section~\ref{sec:open-questions}).

\subsection{How to reduce the tail}
\label{sec:concl-mitigate}

The second question was whether a servicing chain cut into two or more stages on separate
threads is better than one thread. Above a service-time threshold, yes; below it, no; and
the threshold is a property of the feed, measured here on NQ. Above the publisher period
the single-thread tail is many times the hop cost, so the hop is small against what it
lets the pipeline remove; below it the hop is cost only. Cutting a task into two halves
removes most of the tail for one hop on the median, on every window, at every service time well above the publisher period, and it loses at service times at or below it
(Table~\ref{tab:main-corpus}). A tandem of $N$ stages divides the tail by $N$ or better on
every sample path (Theorem~\ref{thm:reduction}), in practice by far more where a cut takes
each stage below the publisher period, while paying one hop per boundary at every
quantile; only its largest stage matters. The best stage count rises with the service
time, up to the number of natural cut points in the decode chain. These results do not
depend on what places transactions close together: the design rule needs the gaps the
receiver sees, not their cause.

The single-thread rule is correct whenever the service time of the chain is below the
feed's threshold, and on this corpus that includes the reference system's socket-to-book
path, at about $7\,\mu\mathrm{s}$ in the two-stage receiver and about $1\,\mu\mathrm{s}$ at the
median on one thread, which should not be split (Section~\ref{sec:kaspar}).
The threshold is the service time of the exchange's publisher upstream of the receiver,
which is the tight end of the gap distribution the receiver sees, and it does not move
with packet rate or fitted branching ratio at HFT service times; at longer service times
both raise the tail (Section~\ref{sec:results-conditioning}).
Appendix~\ref{app:synthetic} constructs a feed on which both dependences appear at a
production service time and on which splitting pays where on CME it does not. On another
feed the threshold, and its dependence on load and on clustering, must be measured.

Spending the same cores on a dispatch pool that sends whole packets to identical servers
is the alternative for a stage without state, such as SBE decode; with every cost
charged it is level with a cut at two cores and better from four
(Table~\ref{tab:dispatch-costed}, Section~\ref{sec:concl-summary}). The order book
carries state and can only be cut.

The recommendation is that HFT designers measure end-to-end latency of the assembled
pipeline under actual market conditions, including peak load, and split the servicing
chain wherever the measurement shows a tail excess larger than the hop cost. Splitting the
hot path is a design option of equal standing with optimising it. For a chain above the threshold, one cut at a natural boundary reduces the tail by more
than any per-message optimisation is likely to. On this corpus a balanced cut of a chain
twice the publisher period removes almost all of its tail excess for one hop
(Section~\ref{sec:kaspar}), whereas per-message work moves the service floor and the cost
per message by tenths of a microsecond. Below the threshold, for the typical packet, per-message optimisation is the
lever; the far tail that a few many-message packets open on a live channel
(Section~\ref{sec:fungibility-measured}) is a separate regime, whose remedy the
end-of-build measurement decides. Framework designers should measure busy-poll dequeue on pinned isolated cores against
the wait/notify mailbox on a representative tape rather than assume it is faster
(Section~\ref{sec:other-arch}), drive the mailbox mutex uncontended, and keep memory-pool
discipline. They should cut the decode pipeline at whichever of its natural boundaries
minimises the largest stage, put that stage at ingress, and merge everything downstream
up to its service time. The tail is set by the largest stage, and a hop that does not
lower it is pure cost. Strategy designers of both percentage-of-volume (POV)
and race-sensitive algorithms benefit for the same reason.

\subsection{Revision of a prior design context}
\label{sec:concl-revision}

An earlier paper by the present author \citep{mayeski_fastsend} did not itself argue for
the single-thread servicing chain; it took the collapsed single-thread servicing chain as
the widely accepted design context of the time and offered \texttt{fast\_send} as the
mechanism that makes actor composition inside it free of context-switch overhead. What the
present paper retracts is the design context, not the mechanism: on Hawkes-clustered feeds
the single-thread servicing chain should no longer be the unconditional default. The
corrected position is that \texttt{fast\_send} (synchronous, inline, sub-microsecond) and
asynchronous \texttt{BQueue} handoffs (a real pipelined hop at a few microseconds) are
\emph{complementary} tools, and the choice between them is a per-boundary design decision
that depends on the service time being chained and the arrival law of the feed the
pipeline sits on.

In shipping \kasparhft{} the \texttt{fast\_send} versus asynchronous-send choice at each
actor boundary is normally left open until the end of the build, and is resolved by
measuring the end-to-end median and tail latency of the assembled pipeline under the
target feed. A convenient property of the \kasparhft{} actor API is that the receiving
actor's \texttt{MESSAGE\_HANDLER} is agnostic to how it was invoked: swapping a synchronous
\texttt{fast\_send} at a call site for an asynchronous \texttt{send} (or vice versa) is a
one-line change on the sender side and requires no code change on the receiver side. The
choice therefore lives entirely at each individual call site and can be flipped per
boundary at end-of-build with minimal engineering cost. The paper's design equation now
supplies the arrival-law-aware default value of that end-of-build measurement: on
Hawkes-clustered feeds, tasks above the threshold should be split into shorter actors
joined by asynchronous sends, on the measured result that splitting reduces the tail by
more than the hop cost raises the median. The rule presupposes that the task being split
has a service time comfortably above the publisher period of
Section~\ref{sec:results-design}; a task no longer than a hop cannot be usefully split.
Short synchronous chains inside one actor's servicing of one message therefore continue
to use \texttt{fast\_send} unchanged. The end-of-build measurement decides any specific
boundary.

\subsection{Caveats}
\label{sec:concl-caveats}

The corpus is a single instrument, NQ front-month; a ZN recording is the first
replication needed, for the reason given in Section~\ref{sec:concl-cannot}, ahead of ES
and BTC. Transactions are grouped by exact \texttt{transactTime}, checked against the
end-of-event flag (Section~\ref{sec:exchange-events}). The two-stage measurements of
Sections~\ref{sec:crossval} and~\ref{sec:fungibility-measured} come from one host on one
afternoon and the one-thread measurement from two five-minute runs, without replication, and the corpus windows are nested in sessions, so
uncertainty on the aggregates is session-clustered; the headline survival shares carry
session-bootstrap intervals. Constant per-packet service is a simplification, and at the
production service time it is the one that matters most.
Section~\ref{sec:crossval-spanrun} removes it for message count, at NQ's cost per
message, which is lower than ZN's, and finds the stage count unchanged. Service that
varies for other reasons is not modelled, although the live receiver shows it at twice
the service floor at $p_{99}$ with no queue. The simulated tails at the production service
time therefore understate a real receiver's. The cache-pollution cost of spinning receivers on many cores is bounded
but not measured. And Theorem~\ref{thm:reduction} is a bound whose slack depends on the
arrival law; the paper measures that slack rather than deriving it, and a closed form for
the slack on Hawkes input is open.

None of these caveats reverses the two answers on this corpus.

\section{Future work}
\label{sec:future}

The section first states the open questions the paper raises, then lists the remaining
work in two groups: questions about the market, which are outside the paper's scope, and
questions about the receiver, which are within it but were not answered with the data at
hand.

\subsection{Open questions this paper raises}
\label{sec:open-questions}

The paper answers its questions on NQ front-month. Measuring it raised others that
need data or scope beyond this paper, and each is better answered in a shorter, focused
study.

\begin{enumerate}
    \item \textbf{Does filtering the channel to NQ front-month change any result?}
Channel 318 carries other instruments, and the publisher packs their messages into the
same packets as NQ. Every packet statistic in the paper counts NQ front-month messages
only: span, transactions per packet, and the packet stream fed to the simulated receiver.
Two things could change on the full channel. First, packets are larger in bytes and
messages than the NQ span shows, which bears on the span-dependent service of
Section~\ref{sec:crossval-spanrun} and the tables of Section~\ref{sec:exchange-drain}.
Second, a receiver subscribed to the channel sees every packet, including those without
NQ, so its arrival stream is denser than the one simulated unless those packets are
discarded at negligible cost. The raw captures hold the full channel; the measurement
is packet bytes and message counts on every packet, joined to the NQ tapes on the
publisher's send time, and the sweep rerun with every channel packet as an arrival. The
publisher period is not in question: Table~\ref{tab:drain-gap} shows it on the channel's order
and trade packets.
    \item \textbf{Do instruments excite each other, above all on ZN?} The Treasury futures
along the yield curve (two-, five-, ten- and thirty-year and the ultra bond) are priced
off one another, so a move in one is likely to trigger quotes and trades in the others
within microseconds. Where they share ZN's channel, those reactions reach the publisher
together and leave in the same packets, which would make ZN's large packets a product of
the curve and not of ZN alone. The test is a multivariate Hawkes fit across the channel's
instruments on the engine clock and packet size conditioned on cross-triggering. It needs
a full-channel ZN recording.
    \item \textbf{Why do orders reach the matching engine within a microsecond of each
other?} The receiver sees them one publisher period apart; the
question is about the engine's input. Reaction,
common reaction to one signal (races), or pacing inside order entry
(Section~\ref{sec:concl-cannot}). It needs order-entry times for all participants.
    \item \textbf{How much of the engine's activity is racing?} Table~\ref{tab:tight-pairs}
shows the race signature in tight gaps that follow a trade, but rejected race orders do not reach
the public feed.
    \item \textbf{Is the publisher's capacity a design choice, and does it act as a speed
limit on information?} The conjecture of Section~\ref{sec:exchange-conjecture}.
    \item \textbf{Where does the spread in a receiver's service time come from?} On the
live ZN receiver the empty-queue messages reach twice the service floor at $p_{99}$ with nothing queued and
nothing ahead in the packet. The measurement is the full distribution of per-packet
service for single-message packets, with the candidate causes (cache misses, interrupts,
rare message templates, book state) recorded alongside.
    \item \textbf{What does a random service time do to the tail and the design rule?} The
sweep rerun with service drawn from the measured live distribution instead of a constant,
and the span-aware sweep rerun at ZN's per-message cost, three times NQ's.
    \item \textbf{How does the live far tail divide?} Retaining the per-message live
records would allow the ZN $p_{99.9}$ to be split into in-packet decode, the queue behind
long packets, and slow services.
    \item \textbf{Can decoding one packet's messages in parallel shorten the long-packet
tail?} A stage cut does not shorten it. An unmerged experiment in the \kasparhft{} repository that dispatched every message to
parallel decoders was slower than serial at every book percentile. A later one that sent
only multi-message packets to parallel decoders lost at the median but improved
$p_{99.9}$ on most streams, on a build with recovery disabled and without validation of
its output. Whether it pays in a
production build is open.
    \item \textbf{Does the result replicate on ZN, ES and BTC?} The design rule is measured
on NQ; the threshold equals the upstream publisher's service time only if each channel
is paced the same way.
\end{enumerate}

The lists below give the remaining work in two groups.

\paragraph{Why orders reach the matching engine within microseconds of each other (questions about the market).}
\begin{itemize}
    \item What places two distinct transactions within a few microseconds of each other:
reaction to the previous transaction, common reaction of many participants to the same
signal, or pacing or batching inside the exchange's order entry (Section~\ref{sec:concl-cannot}).
Separating them needs order-entry times for all participants, the data
\citet{aquilina2022arms} had for the London Stock Exchange. Races of the kind they found
are the leading hypothesis. Their data were exchange message data, which, unlike the
order-book data of a public feed, include the attempts to trade or cancel that failed, so
both the winner and the losers of a race can be seen. On FTSE 100 stocks they found that latency-arbitrage races, in which several firms
react to the same public signal and try to take or cancel the same stale quote, are
frequent, about one per minute per symbol, and extremely fast, the typical race lasting
five to ten microseconds. Races account for about a fifth of trading volume, and the top
six firms account for most of the wins and losses. Races make up roughly a third of the
cost of liquidity as measured by price impact and the effective spread, and a market
design that removed them would lower that cost by about a sixth. The timescale of their races is the timescale of the engine-clock
gaps of Table~\ref{tab:engine-clock}, and the failed attempts that reveal a race are what
the CME public feed does not carry (Section~\ref{sec:exchange-bursts}).
    \item Which instruments' transactions sit on either side of a tight gap on the engine
clock, across the whole channel. Table~\ref{tab:tight-pairs} gives message type, side and price level for NQ
front-month; which other instruments' transactions sit within microseconds of an NQ
transaction is not measured.
    \item Cross-excitation between instruments, above all on ZN
(Section~\ref{sec:open-questions}, question 2).
    \item A power-law Hawkes kernel against the exponential, fitted on the engine clock.
    \item The publisher queue measured channel-wide rather than on NQ front-month only, and
on ES and ZN, so that its service interval and backlog are measured under the full load it
serves.
\end{itemize}

\paragraph{The queue in the receiver.}
\begin{itemize}
    \item A ZN recording (pcap and per-message records) and the full sweep on it, with
span-dependent service. On ZN the live data suggest that large packets arrive during
queues, so queueing across packets and in-packet decode may compound, which would add a span
term to the design rule. This is the first replication; ES and BTC follow.
    \item Dispatch under state: a design that dispatches the stateless decode stage and
feeds a single order-book stage, costed end to end.
    \item Closed form for the slack in Theorem~\ref{thm:reduction} under Hawkes input,
i.e.\ the exponent $\gamma - 1$ as a function of $(\rho, n, \beta T)$. By the reduction
this is a single-server question about $M^{[K]}/D/1$ with Borel batches at two service
levels, which is tractable; the paper measures $\gamma$ instead.
    \item Load-conditioned hop cost $h(\rho)$. This paper calibrates $h =
1.7\,\mu\mathrm{s}$ from the low-load wake-latency figure in \citet{mayeski_fastsend};
under sustained-load conditions inside a cluster the receiver never sleeps and the
effective $h$ collapses toward the mutex + cache-line-transfer cost of a few hundred
nanoseconds (Section~\ref{sec:results-h}). Measuring $h(\rho)$ directly and rerunning the
sweep against it would sharpen the reported tail reduction, which the single-$h$
simulation understates.
    \item Cache-pollution cost per spinning-receiver stage count, measured.
    \item Overnight-session regime.
    \item Own-capture pcap validation of the network sub-stage.
\end{itemize}

\section{Reproducibility}
\label{sec:reproducibility}

The code for this paper is in the \kasparhft{} repository \citep{kaspar_repo},
\url{https://github.com/vincent212/kaspar-hft}.

\bibliographystyle{plainnat}

\clearpage
\part*{Appendices}
\addcontentsline{toc}{part}{Appendices}
\appendix

\section{Point-process background}
\label{app:pp-background}

This appendix defines the terms the paper uses for arrival streams: the Poisson,
renewal and Hawkes processes, the conditional intensity, the branching ratio, the Fano
factor and the Hurst exponent, and what it means to fit a Hawkes process. It assumes no
prior knowledge of point processes. Readers who know what a conditional intensity is can
go straight to Section~\ref{app:pp-canon}. Two notational points: $\tau$ is a bin width
for counting arrivals, and $T$ is reserved, as in the rest of the paper, for the service
time of a receiver.

\subsection{The Poisson process}
\label{app:pp-warmup}

\paragraph{The setup.} Watch packets arrive from a market-data feed and write down the
time of each arrival:
\begin{verbatim}
09:30:00.213    <- 1st arrival
09:30:00.451    <- 2nd
09:30:00.802    <- 3rd
09:30:01.014    <- 4th
09:30:01.317    <- 5th
...
\end{verbatim}
That list of timestamps is the object being modelled. Nothing else --- no prices, no
sides, only when.

\paragraph{The Poisson idea.} Suppose that in every very small slice of time there is a
small, \emph{independent} chance of an arrival. Independent means that whether an arrival
happens now does not depend on when the previous one came, on how many came in the last
minute, or on anything else about the past. The \emph{rate} of the process is a single
number, written $\mu$, in arrivals per second. At a rate of $300$ per second, on average
$300$ arrivals occur in each second. Three consequences follow from the independence.

\paragraph{(1) The gap between consecutive arrivals is exponentially distributed with
mean $1/\mu$.} At $\mu = 300$ per second the mean gap is
$$\text{mean gap} = \frac{1}{\mu} = \frac{1}{300/\mathrm{s}} \approx 3.3\ \mathrm{ms}.$$
The gap distribution is memoryless: however long one has been waiting, the expected time
to the next arrival is still $3.3$\,ms. The probability that a gap exceeds $\Delta$ is
$$\Pr(\text{gap} > \Delta) = e^{-\mu\,\Delta},$$
so at this rate $\Pr(\text{gap} > 3.3\,\mathrm{ms}) \approx 36.8\%$,
$\Pr(\text{gap} > 6.6\,\mathrm{ms}) \approx 13.5\%$ and
$\Pr(\text{gap} > 10\,\mathrm{ms}) \approx 5.0\%$. The probability that a gap is shorter
than a tenth of the mean is
$$\Pr\!\big(\text{gap} < \tfrac{1}{10}\,\text{mean}\big) = 1 - e^{-0.1} \approx 9.52\%,$$
whatever the rate. This is the Poisson benchmark used in Section~\ref{sec:setup-hawkes}
for how often two arrivals come nearly back to back; on the corpus the share is $73\%$.

\begin{figure}[H]
    \centering
    \includegraphics[width=0.75\textwidth]{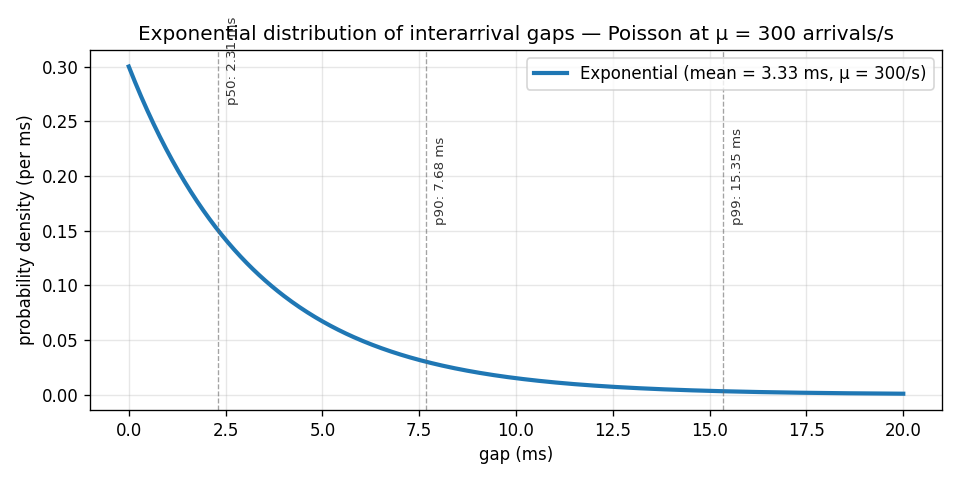}
    \caption{The exponential density of interarrival gaps for a Poisson process at $300$
    arrivals per second. Dashed vertical lines mark the median, $p_{90}$ and $p_{99}$.
    Half the gaps are shorter than the median of $2.3$\,ms and $10\%$ are longer than
    $7.7$\,ms. The CME packet stream has a different shape: far more mass near zero and a
    much heavier right tail (Figure~\ref{fig:tx-gaps}).}
    \label{fig:exponential-gaps}
\end{figure}

\paragraph{(2) The count in any bin of width $\tau$ is Poisson-distributed.} Let $C$ be
the number of arrivals in a bin of width $\tau$. Then
$$\mathrm{E}[C] = \mu \tau, \qquad \mathrm{Var}[C] = \mu \tau,$$
so the ratio of variance to mean is $1$ for every $\tau$. This ratio is the \emph{Fano
factor} (Section~\ref{app:pp-fano}); a value other than $1$ rejects Poisson.

\paragraph{(3) Counts in disjoint bins are independent.} The number of arrivals in
$[0, 1\,\mathrm{s}]$ says nothing about the number in $[1\,\mathrm{s}, 2\,\mathrm{s}]$.

\paragraph{Poisson as the null hypothesis.} Poisson is the stream with no memory, no
clustering and no regularity: arrivals are as uncoordinated as possible. A departure
from Poisson is structure in the arrival stream. Identifying where that structure comes
from is a separate question, which for this feed the paper takes up in
Sections~\ref{sec:results-transactions} and~\ref{sec:exchange}.

\paragraph{Two ways a real feed departs from Poisson.} The first is \emph{bunching}:
many gaps are far shorter than the exponential law predicts. On the corpus $73\%$ of
packet gaps are shorter than a tenth of the mean gap, against $9.52\%$ for Poisson. The
second is \emph{memory in the rate}: after a burst starts, more arrivals are likely to
follow. Modelling either needs a process whose rate can change with its own history,
which is a \emph{conditional-intensity point process}.

\paragraph{Notation.} Some authors write the Poisson rate as $\lambda$. Here $\mu$ is the
\emph{background} rate of a Hawkes process, which reduces to a Poisson process when there
is no excitation, and $\lambda(t)$ is the general, history-dependent intensity. For a
Poisson process $\lambda(t) = \mu$ at every $t$.

\subsection{Points, processes and marks}
\label{app:pp-pointprocess}

\begin{itemize}
    \item A \textbf{point} is a single instant at which something happens: in this paper,
the timestamp of one packet (or, in Section~\ref{sec:results-transactions}, of one
matching-engine transaction).
    \item A \textbf{process} is a random object that unfolds in time; here, a random list
of arrival times $t_1 < t_2 < t_3 < \cdots$.
    \item A \textbf{point process} is a random sequence of points on the time axis,
specified by \emph{when} events happen and nothing else.
\end{itemize}
If each event carries additional information (message count, side, price, size), the
process is a \textbf{marked point process}: the timestamps are the points and the extra
information is the \emph{mark}. The paper's arrival streams are marked --- the span of a
packet (Section~\ref{sec:crossval-span}) is a mark --- but the diagnostics of
Section~\ref{sec:setup-hawkes} use the timestamps alone. Other names in the literature
are \emph{counting process} (emphasising $N(t)$), \emph{temporal point process} (as
against spatial), \emph{renewal process} and \emph{self-exciting process}, the last two
being particular classes (Section~\ref{app:pp-canon}).

\paragraph{Illustrative example.} Four MBO messages in the first second of a session
might read
\begin{verbatim}
09:30:00.001234567   Add,    bid 21432.50, size 3
09:30:00.001234571   Add,    ask 21432.75, size 5
09:30:00.017891234   Cancel, order X, bid 21432.25
09:30:00.238491702   Trade,  ask 21432.75, size 2
\end{verbatim}
The point process is the four timestamps; message type, side, price and size are marks.

\subsection{The probability space}
\label{app:pp-probspace}

An arrival stream is modelled as a \textbf{stochastic process}, a family of random
variables indexed by time, on a probability space $(\Omega, \mathcal{F}, \Pr)$:
\begin{itemize}
    \item $\Omega$, the \emph{sample space}, is the set of all possible outcomes; here an
outcome is one complete realisation of a session's arrivals.
    \item $\mathcal{F}$, the \emph{$\sigma$-algebra}, is the collection of yes-or-no
questions about the outcome to which a probability can be assigned, such as ``did at
least one packet arrive between 09:30:00 and 09:30:01?''. It is closed under negation and
countable union; for timestamps it is the standard Borel $\sigma$-algebra.
    \item $\Pr : \mathcal{F} \to [0, 1]$ assigns a probability to each such event, with
$\Pr(\emptyset) = 0$, $\Pr(\Omega) = 1$ and countable additivity over disjoint events.
\end{itemize}

\subsection{Filtration: information accumulated over time}
\label{app:pp-filtration}

A \textbf{filtration} is an increasing family of $\sigma$-algebras
$$\{\mathcal{F}_t\}_{t \geq 0}, \qquad \mathcal{F}_s \subset \mathcal{F}_t \text{ for } s \leq t.$$
$\mathcal{F}_t$ is everything observable up to and including time $t$, and
$\mathcal{F}_{t^-}$ everything observable strictly before $t$. ``Conditional on
$\mathcal{F}_{t^-}$'' means ``given everything seen strictly before $t$''. A process is
\emph{adapted} to the filtration if its value at every $t$ is knowable from the history
up to $t$.

\subsection{The counting process $N(t)$}
\label{app:pp-counting}

The observable object is the random step function
$$N(t) = \#\{\, i : t_i \leq t \,\},$$
which starts at $N(0) = 0$, is non-decreasing and right-continuous, and jumps by one at
each arrival. The number of arrivals in the half-open interval $(t, t + \delta]$ is
$N(t + \delta) - N(t)$.

\subsection{Conditional intensity}
\label{app:pp-intensity}

At time $t$, ask for the probability that at least one arrival occurs in the next short
interval of length $\delta$, given everything seen so far. The \textbf{conditional
intensity} is that probability divided by $\delta$, in the limit of small $\delta$:
$$\lambda(t \mid \mathcal{F}_{t^-}) = \lim_{\delta \downarrow 0} \frac{1}{\delta}\;
\Pr\!\big(\, N(t + \delta) - N(t) \geq 1 \,\big|\, \mathcal{F}_{t^-} \big).$$
It is the instantaneous arrival rate given the history, in arrivals per second. At an
intensity of $300$ per second the probability of an arrival in the next millisecond is
about $\lambda(t)\,\delta = 300 \times 0.001 = 0.3$. For small $\delta$,
$$\Pr(\text{arrival in }(t, t+\delta]) \approx \lambda(t)\, \delta, \qquad
\mathrm{E}[N(b) - N(a)] = \mathrm{E}\!\int_a^b \lambda(s)\, ds.$$
Specifying $\lambda(\cdot)$ as a function of the history specifies the point process
completely. For a Poisson process $\lambda(t) = \mu$ and the history carries no
information. For a Hawkes process the intensity is the background $\mu$ plus a decaying
contribution $\phi(t - t_i)$ from each past arrival (Section~\ref{app:pp-canon}).

\subsection{Three canonical models: Poisson, renewal and Hawkes}
\label{app:pp-canon}

\paragraph{(a) Homogeneous Poisson.} $\lambda(t) = \mu > 0$. Gaps are independent and
exponential with mean $1/\mu$; the count in a bin of width $\tau$ is
$\mathrm{Poisson}(\mu \tau)$. Every summary statistic takes its reference value: the
coefficient of variation of the gaps is $1$, the Fano factor is $1$ at every $\tau$, the
Hurst exponent is $0.5$, and the gap autocorrelation is zero at every lag.

\paragraph{(b) Renewal process.} Gaps $\Delta_i = t_i - t_{i-1}$ are independent draws
from a common distribution $G$, not necessarily exponential. The intensity depends only on
the time since the last arrival, $\lambda(t) = r(t - t_{\mathrm{last}})$, with $r$ the
hazard function of $G$. A renewal process is Poisson only when $G$ is exponential. A
heavy-tailed $G$ produces many short gaps and a Fano factor well above $1$, but the Fano
factor levels off at large $\tau$ (at the squared coefficient of variation of $G$, when
that is finite) instead of growing, and the gaps are uncorrelated. Permuting the gap
sequence at random (a Fisher--Yates shuffle) leaves a renewal process statistically
unchanged. The gap shuffle $G$ of Section~\ref{sec:setup-nulls} is therefore a renewal
stream with the real gap distribution, and comparing a statistic on the real stream and on
its shuffle separates what the order of the gaps carries from what their lengths carry.

\paragraph{(c) Hawkes (self-exciting) process.} Each past arrival raises the intensity
for a while:
$$\lambda(t) = \mu + \sum_{t_i < t} \phi(t - t_i).$$
$\mu \geq 0$ is the \textbf{background rate}, the intensity with no prior arrivals.
$\phi : (0, \infty) \to [0, \infty)$ is the \textbf{kernel}: an arrival $u$ seconds ago
still contributes $\phi(u)$ to the intensity now. The kernel is the shape of one event's
influence on the future. Dropping several stones into a pond is the usual picture: the
disturbance at any moment is the sum of the ripples still spreading, each faded by the
time since its stone. Two conditions apply: $\phi$ is non-negative, since in a
self-exciting model a past event can only raise future intensity, and its integral
$n = \int \phi$ is finite, so that the past does not accumulate without bound.

\paragraph{Kernel families.} The \emph{exponential kernel}
$\phi(u) = \alpha\, e^{-\beta u}$ has two parameters: $\alpha \geq 0$, the jump in
intensity immediately after an arrival, and $\beta > 0$, the decay rate, whose inverse
$1/\beta$ is the kernel's decay time. It is the kernel fitted in the paper. It is
Markovian: the current intensity summarises everything needed from the past. Right after
an arrival the intensity rises by $\alpha$; between arrivals it relaxes toward $\mu$ as
$\mu + (\lambda - \mu) e^{-\beta \Delta t}$; so it can be updated in constant time per
arrival. The \emph{power-law kernel} $\phi(u) = a\,(u + c)^{-(1+p)}$, $p > 0$, decays more
slowly; \citet{hardimanbercotbouchaud2013} find power-law kernels on E-mini S\&P
mid-price changes over lags from seconds to days. A power-law fit on this corpus is
listed in Section~\ref{sec:future}.

\subsection{Branching ratio $n$}
\label{app:pp-branching}

$$n = \int_0^\infty \phi(u)\, du, \qquad n = \alpha/\beta \text{ for the exponential kernel}.$$
$n$ is the mean number of arrivals that one arrival triggers directly. The Hawkes process
is equivalently a \textbf{branching process} (a Galton--Watson process): every arrival is
either an immigrant, generated by the background rate $\mu$, or the child of an earlier
arrival, and each arrival has on average $n$ children. The expected size of the cluster
descended from one immigrant is the geometric series
$$1 + n + n^2 + n^3 + \cdots = \frac{1}{1 - n} \qquad (n < 1),$$
and the distribution of cluster sizes is the Borel distribution of
Appendix~\ref{app:borel}. For $n < 1$ the process is \emph{subcritical}: clusters are
finite and the long-run rate is $\mu / (1 - n)$. As $n \to 1$ it becomes
\emph{near-critical}: cluster sizes and durations grow without bound and long-range
dependence appears. For $n \geq 1$ it is \emph{critical or supercritical}: with positive
probability one arrival starts a cascade that never ends, and the rate is not stationary.

On this corpus the exponential-kernel fit gives $n = 0.798$ on packets and $0.795$ on
transactions, with a kernel decay time $1/\beta$ of about $167\,\mu\mathrm{s}$
(Table~\ref{tab:hawkes-diag}). \citet{filimonovsornette2012} and
\citet{hardimanbercotbouchaud2013} report values near $0.9$ and near $1$ on E-mini S\&P
mid-price changes; \citet{filimonovsornette2015} show that a fit with a constant
background over a window in which the background varies overstates $n$.

\paragraph{The critical boundary in practice.} At $n = 1$ the stationary rate
$\mu/(1-n)$ is infinite. On a finite window the process still fires finitely often, but
the variance of the count grows faster than linearly with the window. A maximum-likelihood
fit that lands at $n = 1$ on real data signals non-stationarity, a mis-specified kernel
or a numerical problem rather than a genuinely critical market. On this corpus one window
of 3512 is pinned at that boundary.

\subsection{Simulated Hawkes processes across rate and branching ratio}
\label{app:pp-visualise}

The mean rate $\bar\lambda = \mu/(1-n)$ and the branching ratio $n$ are the two
parameters by which Section~\ref{sec:results-conditioning} conditions the tail. They
look different in an event raster. Figure~\ref{fig:hawkes-grid} simulates an
exponential-kernel Hawkes process over $15$ seconds at each of nine combinations,
$\bar\lambda \in \{2, 10, 50\}$ events per second and $n \in \{0.30, 0.60, 0.90\}$, with a
kernel decay time of one second. At low rate and low branching ratio the stream is
nearly Poisson: ticks spread evenly, a straight count curve, a flat intensity. At high
rate and low branching ratio it is dense but still nearly Poisson. At low rate and high
branching ratio the mean is $2$ per second but arrivals come in unmistakable bursts, with
long empty stretches between them and the intensity reaching $20$--$30$ per second
inside a burst. At high rate and high branching ratio the intensity swings by an order of
magnitude within the window.

\begin{figure}[H]
    \centering
    \includegraphics[width=\textwidth,height=0.72\textheight,keepaspectratio]{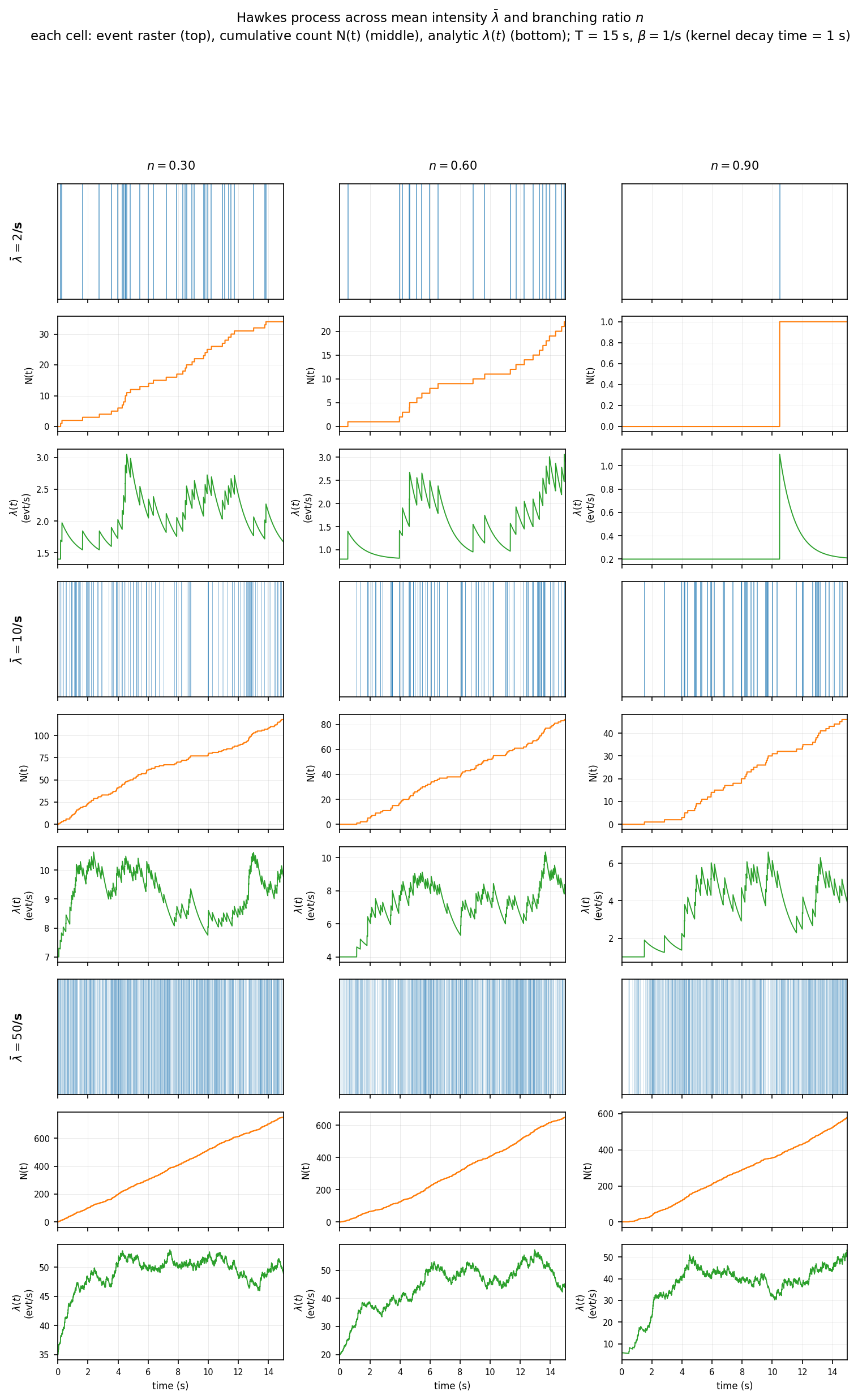}
    \caption{Simulated exponential-kernel Hawkes processes. Each cell stacks the event
    raster (top), the cumulative count $N(t)$ (middle) and the intensity $\lambda(t)$
    (bottom). Rows: mean rate $\bar\lambda \in \{2, 10, 50\}$ events per second. Columns:
    branching ratio $n \in \{0.30, 0.60, 0.90\}$. Ogata thinning over $15$\,s with
    $\beta = 1/\mathrm{s}$; generator \texttt{make\_hawkes\_grid.py}, seed fixed per
    cell. The time scale is chosen for legibility; the kernel fitted on the corpus decays
    about four orders of magnitude faster.}
    \label{fig:hawkes-grid}
\end{figure}

\subsection{Three signatures that separate Poisson, renewal and Hawkes processes}
\label{app:pp-why-hawkes}

\emph{(i) An excess of short gaps.} On the corpus $73\%$ of packet gaps are shorter than a
tenth of the mean gap; any exponential gives $9.52\%$. A Hawkes process produces such an
excess, but so does a renewal process with a heavy-tailed gap distribution, so this
signature rejects Poisson without identifying self-excitation. On this feed the gap
shuffle keeps the whole excess (Table~\ref{tab:hawkes-diag}), and
Section~\ref{sec:exchange} traces the tightest gaps to the publisher sending consecutive
transactions one period apart.

\emph{(ii) Correlated gaps.} Under self-excitation, short gaps tend to follow short gaps
and long gaps follow long ones, so the gap autocorrelation is positive. A renewal process
has independent gaps and zero autocorrelation at every lag. A rate that varies slowly
over the window also produces positive gap correlation; the one-second-binned null $B$ of
Section~\ref{sec:setup-nulls} separates that case.

\emph{(iii) Burstiness that grows with the time scale.} The Fano factor $F(\tau)$ rises
with $\tau$ on a Hawkes process and levels off on a Poisson or renewal process.
Clustering is then not a feature of one time scale but looks similar over decades of
$\tau$. A slowly varying rate can also make $F(\tau)$ rise, which is the effect
\citet{filimonovsornette2015} warn about.

\subsection{Fano factor}
\label{app:pp-fano}

For a bin width $\tau$, cut the observation interval into $K$ disjoint bins and let
$C_k = N(k\tau) - N((k-1)\tau)$ be the count in bin $k$. The Fano factor is
$$F(\tau) = \frac{\mathrm{Var}[C_k]}{\mathrm{E}[C_k]}.$$
For a Poisson process $F(\tau) = 1$ at every $\tau$. For a renewal process with
finite-variance gaps $F(\tau)$ tends to the squared coefficient of variation of the gaps
as $\tau$ grows, a constant. For a clustered process $F(\tau) > 1$ and, under
self-excitation or a varying rate, it keeps growing with $\tau$. The measure is named
after U.~Fano, who introduced it in 1947 for fluctuations in counts of ions. On the
corpus $F(\tau)$ grows from $52$ at $0.1$\,s to about a thousand at $60$\,s, and the
shuffled stream is flat at about $17$ (Table~\ref{tab:hawkes-diag}).

\begin{figure}[H]
    \centering
    \includegraphics[width=\textwidth]{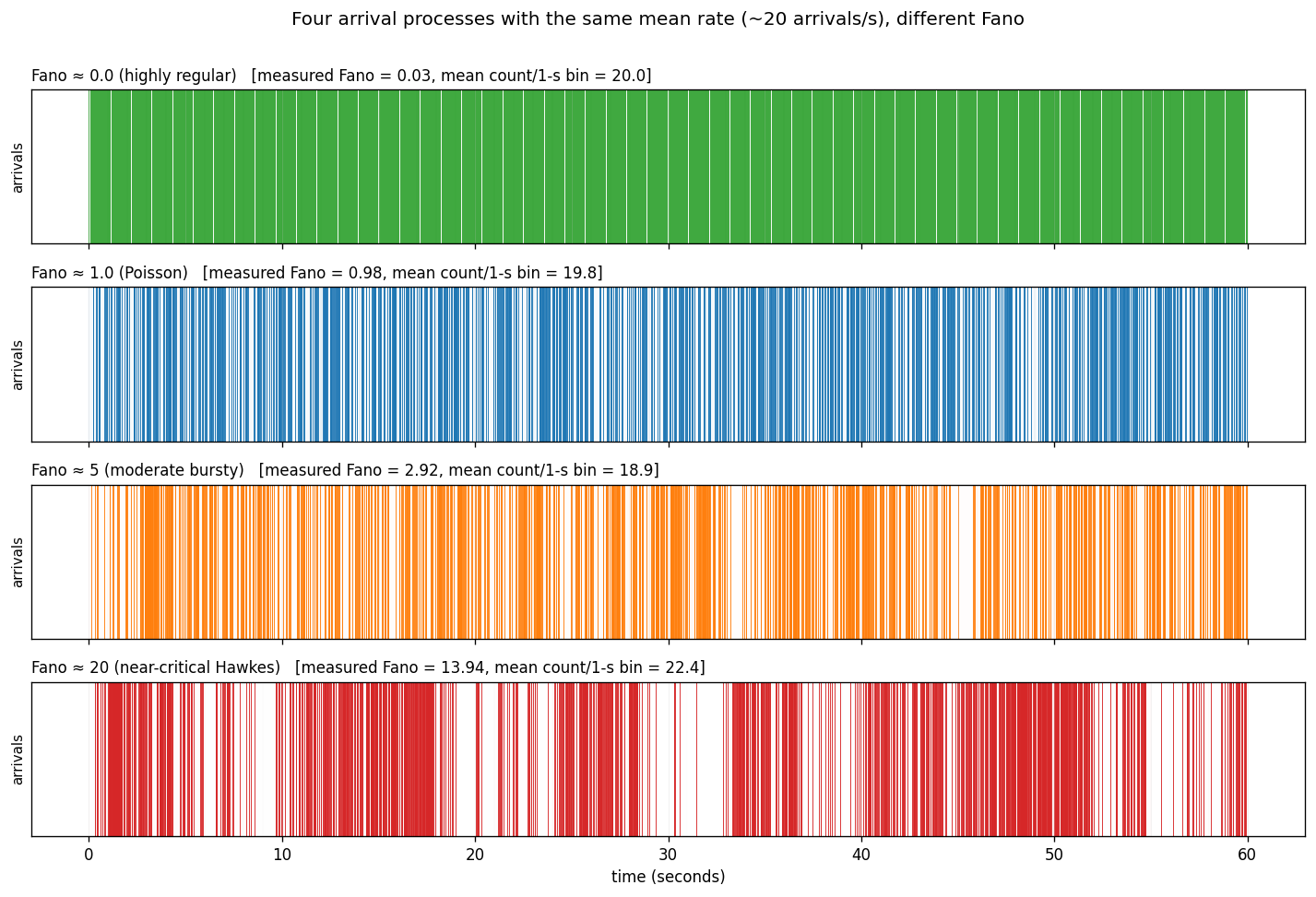}
    \caption{Four simulated processes at the same mean rate of about $20$ per second.
    Top row ($F \approx 0$, highly regular): near-lattice arrivals, nearly equal bin
    counts. Second row ($F \approx 1$, Poisson): the reference. Third row ($F \approx 3$,
    moderate Hawkes): visible clumping. Bottom row ($F \approx 14$, near-critical Hawkes):
    pronounced bursts and long quiet stretches, bin counts from near zero to over $50$.}
    \label{fig:rasters-fano}
\end{figure}

\begin{figure}[H]
    \centering
    \includegraphics[width=\textwidth]{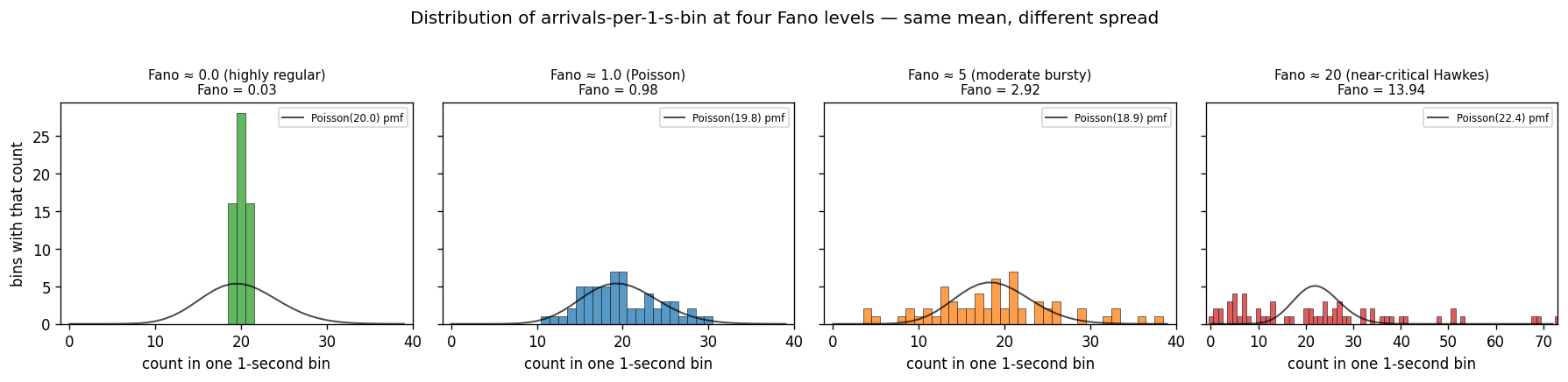}
    \caption{Distribution of arrivals per one-second bin (bars) for the four processes of
    Figure~\ref{fig:rasters-fano}, with the Poisson distribution at the same mean
    overlaid (curve). $F \approx 0$: narrower than Poisson. $F \approx 1$: matches
    Poisson. $F > 1$: wider than Poisson, with extra mass in both tails. The NQ packet
    stream of the corpus has $F(1\,\mathrm{s}) = 108$ (Table~\ref{tab:hawkes-diag}),
    far wider than any panel shown.}
    \label{fig:counthist-fano}
\end{figure}

\subsection{Hurst exponent from the scaling of the Fano factor}
\label{app:pp-hurst}

When the Fano factor grows as a power of the bin width,
$$F(\tau) \sim \tau^{2H - 1},$$
the exponent $H$ is the Hurst exponent. $H = 0.5$ means no long-range dependence (Poisson,
and renewal with finite-variance gaps); $0.5 < H < 1$ means long-range dependence, with
clustering that looks similar across time scales; $H$ near $1$ means very long memory. It
is estimated by least squares on the log-log plot,
$$\log F(\tau) = (2H - 1)\log \tau + c + \varepsilon,$$
so $H = (\text{slope} + 1)/2$. On the corpus $H = 0.74$ (Table~\ref{tab:hawkes-diag}).

\begin{figure}[H]
    \centering
    \includegraphics[width=0.75\textwidth]{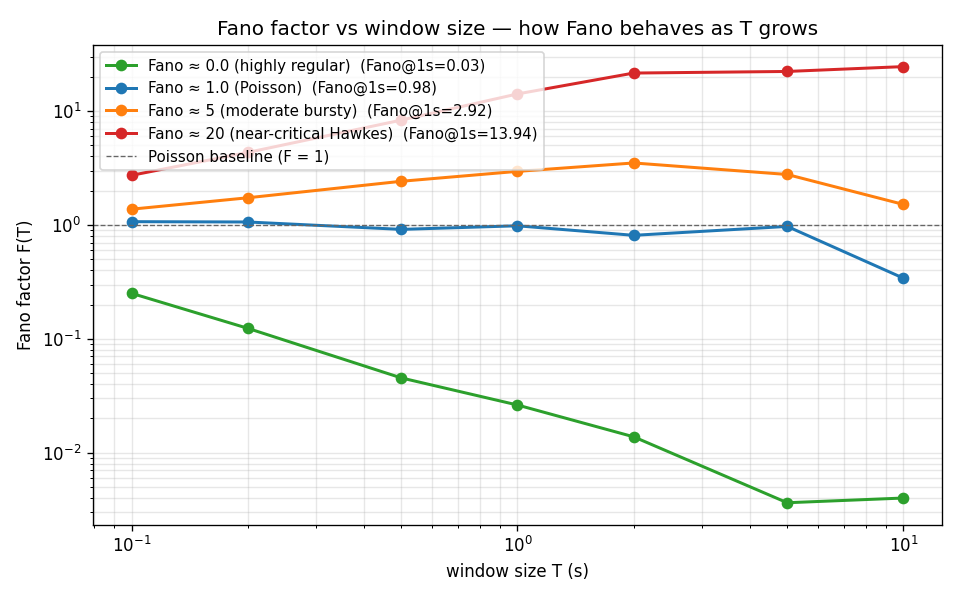}
    \caption{$F(\tau)$ against $\tau$ on logarithmic axes for four simulated processes.
    Regular: $F$ near zero. Poisson: flat at $1$. Moderate Hawkes: rising. Near-critical
    Hawkes: steeper rise, with slope about $0.5$, i.e.\ $H \approx 0.75$.}
    \label{fig:fano-scaling}
\end{figure}

Shuffling the gap sequence preserves the gap distribution exactly and destroys the
ordering, and with it any long-range dependence: on the shuffled stream $H \to 0.5$. On
the corpus the shuffle takes $H$ from $0.74$ to $0.50$, so the long-range dependence
lives in the order of the gaps (Section~\ref{sec:setup-hawkes}).

\subsection{Fitting an exponential-kernel Hawkes process}
\label{app:mle}

The input is a sorted list of arrival times $t_1 < \cdots < t_m$ over a window
$[0, t_{\mathrm{end}}]$, and there are three parameters:
\begin{itemize}
    \item $\mu$, the \textbf{background rate}: arrivals per second with no recent history;
    \item $\alpha$, the \textbf{jump}: the rise in intensity, in arrivals per second, at
each arrival;
    \item $\beta$, the \textbf{decay rate}: each jump fades as $e^{-\beta u}$, with decay
time $1/\beta$.
\end{itemize}
Fitting means choosing the parameters under which the observed timestamps are most
probable, i.e.\ maximising the log-likelihood \citep{ozaki1979}
\begin{equation}
\log L(\mu, \alpha, \beta) = \sum_{i=1}^m \log \lambda(t_i) - \int_0^{t_{\mathrm{end}}} \lambda(s)\, ds.
\label{eq:hawkes-loglik}
\end{equation}
The first term rewards high intensity at the times arrivals actually occurred; the second
penalises a model that predicts more arrivals than occurred. If $\alpha$ is too small the
model cannot explain the bursts; if $\alpha$ approaches $\beta$ it predicts more clustering
than there is; a wrong $\mu$ mis-explains the quiet periods.

For the exponential kernel both terms have closed forms computable in one pass. With
$$R_i = e^{-\beta(t_i - t_{i-1})}(1 + R_{i-1}), \qquad R_1 \equiv 0,$$
the intensity at each arrival is $\lambda(t_i) = \mu + \alpha R_i$, and the compensator is
$$\int_0^{t_{\mathrm{end}}} \lambda(s)\, ds = \mu\, t_{\mathrm{end}} + \frac{\alpha}{\beta}
\sum_{i=1}^m\!\big(1 - e^{-\beta(t_{\mathrm{end}} - t_i)}\big).$$
Timestamps are shifted to the start of the window before fitting, so that the
exponentials do not underflow. The paper maximises the log-likelihood per window with
Nelder--Mead over $(\mu, \alpha, \beta)$ and checks the optimum with a reparameterised
L-BFGS-B fit (Section~\ref{sec:setup-hawkes}). From the fitted parameters follow the
branching ratio $n = \alpha/\beta$, the kernel decay time $1/\beta$, and the long-run
rate $\bar\lambda = \mu/(1 - n)$: the background rate amplified by $1/(1-n)$ through
self-excitation.

\subsection{Notation}
\label{app:pp-notation}

\begin{center}
\begin{tabular}{@{}ll@{}}
\toprule
symbol & meaning \\
\midrule
$\Omega$, $\mathcal{F}$, $\Pr$ & sample space, $\sigma$-algebra, probability \\
$\mathcal{F}_t$, $\mathcal{F}_{t^-}$ & everything observable up to $t$ inclusive, strictly before $t$ \\
$t_1, t_2, \ldots$ & arrival times \\
$N(t)$ & counting process, $\#\{i : t_i \leq t\}$ \\
$\lambda(t \mid \mathcal{F}_{t^-})$ & conditional intensity, arrivals per unit time \\
$\mu$ & background rate of a Hawkes process \\
$\phi(u)$ & Hawkes kernel \\
$\alpha$, $\beta$ & exponential-kernel jump and decay rate; decay time $1/\beta$ \\
$n = \int \phi$ & branching ratio; $n = \alpha/\beta$ for the exponential kernel \\
$\bar\lambda = \mu/(1-n)$ & long-run mean rate \\
$\Delta_i = t_i - t_{i-1}$ & interarrival gap \\
$\tau$ & bin width for counts (the service time is $T$) \\
$C_k$ & arrivals in bin $k$ \\
$F(\tau)$ & Fano factor, $\mathrm{Var}[C_k] / \mathrm{E}[C_k]$ \\
$H$ & Hurst exponent, $F(\tau) \sim \tau^{2H-1}$ \\
\bottomrule
\end{tabular}
\end{center}

\section{The Borel batch-size tail at fixed branching ratio}
\label{app:borel}

This appendix records the calculation behind Equation~\eqref{eq:borel-asymp} and the
numerical check that the critical $k^{-1/2}$ survival tail of the Borel distribution is
confined to a finite window at every $n < 1$.

Stirling's formula $k! = \sqrt{2\pi k}\, k^k e^{-k}\,(1 + O(1/k))$ in
Equation~\eqref{eq:borel} gives
\[
\Pr(K = k) = \frac{(nk)^{k-1} e^{-nk}}{k!} = \frac{n^{k-1} k^{k-1} e^{-nk}}{\sqrt{2\pi
k}\, k^k e^{-k}}\,(1 + O(1/k)) = \frac{1}{n\sqrt{2\pi}}\, k^{-3/2}\, \big(n
e^{1-n}\big)^{k}\,(1 + O(1/k)).
\]
Writing $n = 1 - \epsilon$, $\log(n e^{1-n}) = \log(1-\epsilon) + \epsilon = -\epsilon^2/2
- \epsilon^3/3 - \cdots$, so the geometric factor is $\exp(-k\epsilon^2/2\,(1 +
O(\epsilon)))$. The pmf is therefore $\propto k^{-3/2}$ for $k \ll k^\star =
\epsilon^{-2}$ and geometric beyond; summing, $\Pr(K > k) \asymp k^{-1/2}$ on the same
window and $\asymp k^{-3/2} e^{-k\epsilon^2/2}$ beyond it. At $n = 1$ exactly, $k^\star =
\infty$, $\Pr(K > k) \sim \sqrt{2/\pi}\, k^{-1/2}$ (Borel--Tanner), and $\mathrm{E}[K] =
\infty$; that is the unstable case $\rho = \infty$, which no stationary queue attains.

Because the tail is geometric, $K$ has finite exponential moments $\mathrm{E}[e^{\theta
K}] < \infty$ for $\theta < \epsilon^2/2\,(1 + O(\epsilon))$ and is therefore outside the
subexponential class $\mathcal{S}$. Heavy-tail queueing results that assume subexponential
batch work, including the single-big-batch principle for $M^{[K]}/G/1$ and the
monotone-separable tail theorem of \citet{baccellifosslelarge2005}, do not apply, and the
stationary waiting time of the batch-arrival queue has a Cram\'er--Lundberg (exponential)
tail at every fixed $n < 1$.

Table~\ref{tab:borel-slope} reports the local log-log slope of the exact pmf, $\Delta \log
\Pr(K = k) / \Delta \log k$ between successive $k$, which a pure $k^{-3/2}$ law would hold
at $-1.5$.

\begin{table}[H]
\centering
\begin{tabular}{@{}rrrrrrr@{}}
\toprule
$n$ & $k^\star$ & $10 \to 30$ & $30 \to 100$ & $100 \to 300$ & $300 \to 1000$ & $1000 \to 3000$ \\
\midrule
0.80 & 25     & $-1.92$ & $-2.84$ & $-5.71$ & $-14.96$ & $-43.63$ \\
0.90 & 100    & $-1.59$ & $-1.81$ & $-2.48$ & $-4.62$  & $-11.26$ \\
0.95 & 400    & $-1.52$ & $-1.57$ & $-1.73$ & $-2.25$  & $-3.85$  \\
0.99 & 10000  & $-1.50$ & $-1.50$ & $-1.51$ & $-1.53$  & $-1.59$  \\
\bottomrule
\end{tabular}
\caption{Local log-log slope of the Borel pmf, Equation~\eqref{eq:borel}, between
successive cluster sizes $k$. The critical slope $-1.5$ holds while $k \ll k^\star =
(1-n)^{-2}$ and steepens geometrically beyond. Figure~\ref{fig:bars-borel} shows the same numbers as bars.}
\label{tab:borel-slope}
\end{table}

\begin{figure}[H]
\centering
\includegraphics[width=0.9\textwidth]{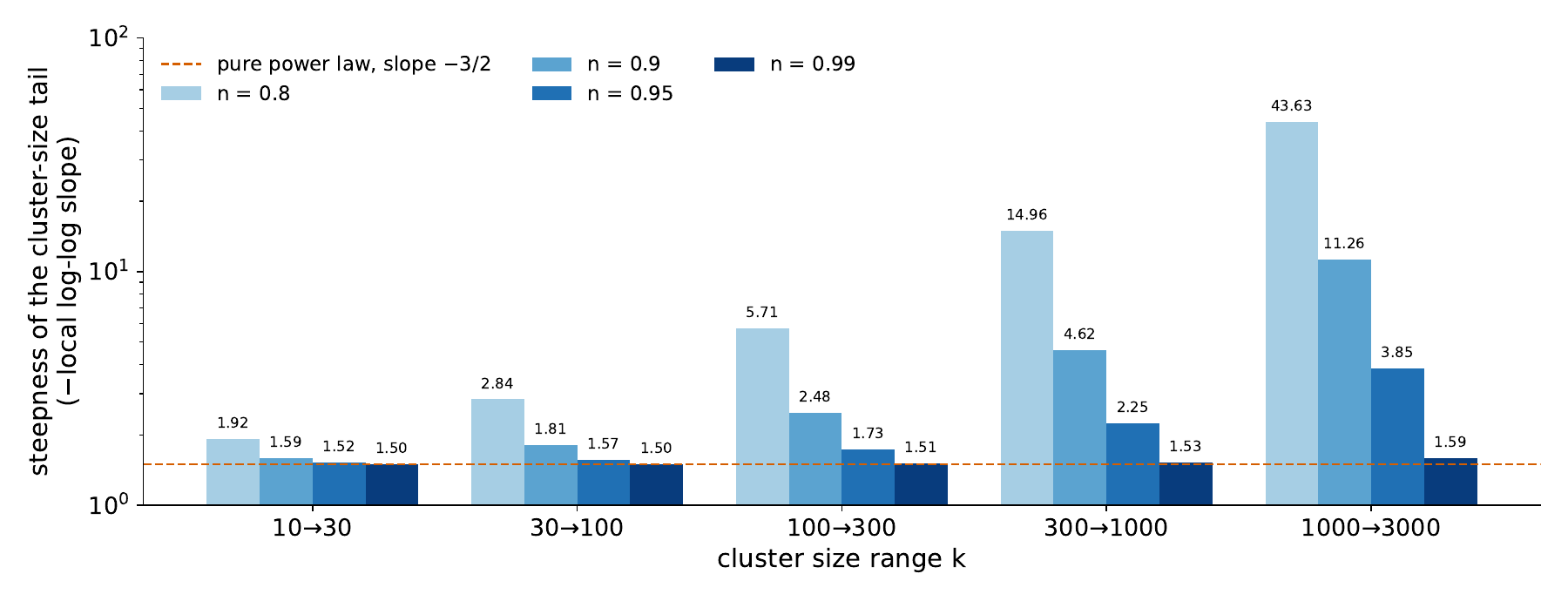}
\caption{Table~\ref{tab:borel-slope} as bars, log scale: how steeply the cluster-size distribution
falls over successive size ranges, for four branching ratios. Near criticality ($n = 0.99$)
it keeps the power-law slope of $3/2$ (dashed) over the whole range; at $n = 0.8$ it falls
off exponentially beyond about $25$, so very large clusters are rare.}
\label{fig:bars-borel}
\end{figure}

The practical reading is that at the branching ratios fitted on the corpus
($n_{\mathrm{pkt}} \approx 0.8$, Section~\ref{sec:setup-hawkes}), $k^\star \approx 25$:
clusters of a few times that size are the largest that occur with non-negligible
frequency, and the single-stage wait behind the largest clusters has a scale of about
twenty-five service times. The corpus-median $p_{99}$ tail excess of
Section~\ref{sec:results} is zero below the publisher period and grows to tens of service
times at the longest service times.
This is the object that Theorem~\ref{thm:reduction} divides by $N$, and the reason the $n
= 0.80$ row of Table~\ref{tab:borel-slope} is the relevant one.

\section{Numerical check of Theorem~\ref{thm:reduction}}
\label{app:check}

A direct simulation of the model of Theorem~\ref{thm:reduction} checks both the identity
Equation~\eqref{eq:reduction} and the bound Equation~\eqref{eq:bound}. Immigrants arrive
as a Poisson process at rate $\mu = 0.004\,\mu\mathrm{s}^{-1}$, each spawning a
Galton--Watson cluster with Poisson($0.9$) offspring, so $K \sim \mathrm{Borel}(0.9)$ and
$\mathrm{E}[K] \approx 10$; clusters are placed as instantaneous batches, giving $\rho =
\bar\lambda T \approx 0.40$ at $T = 10\,\mu\mathrm{s}$. The tandem has $N = 4$ equal
stages of $2.5\,\mu\mathrm{s}$ and $h = 1.7\,\mu\mathrm{s}$; $2 \times 10^{5}$ messages
are simulated with an explicit per-stage FIFO recursion (independent of the paper's sweep code). Over all messages, $\max_j |W_j^{(4)} - (w_j(T/4) + 3h)| <
10^{-9}\,\mu\mathrm{s}$ and Equation~\eqref{eq:bound} is violated by no message. The
quantiles are

\begin{center}
\begin{tabular}{@{}rrrrr@{}}
\toprule
$q$ & $W_q^{(1)}$ & $W_q^{(4)}$ & $W_q^{(1)}/4 + 3h$ & $(W_q^{(4)} - 3h)/W_q^{(1)}$ \\
\midrule
0.5   & 350.0  & 55.1   & 92.6   & 0.14 \\
0.9   & 2066.6 & 355.1  & 521.8  & 0.17 \\
0.99  & 5342.1 & 1127.6 & 1340.6 & 0.21 \\
0.999 & 7586.2 & 1585.1 & 1901.7 & 0.21 \\
\bottomrule
\end{tabular}
\end{center}

in microseconds. The tandem lands below the bound at every quantile, with the ratio in the
last column between $0.14$ and $0.21$ against the bound's $0.25$; the slack is the
$\rho/N$ utilisation effect described in the remark after Theorem~\ref{thm:reduction}.
Under the Poisson null at the same message rate ($\rho = 0.40$, which is far above the
corpus utilisation and is used here only to make the Poisson wait visible) the
single-stage wait quantiles at $q = 0.9, 0.99, 0.999$ are $10.3, 24.8,
38.5\,\mu\mathrm{s}$ and the tandem's are $5.1, 7.5, 9.2\,\mu\mathrm{s}$, again inside the
bound. At the corpus utilisation ($\rho = 0.0045$ at this appendix's $T =
10\,\mu\mathrm{s}$) the same two collapse to $0$ and $3h$ at $q = 0.9$ and $q = 0.99$,
which is Corollary~\ref{cor:poisson}; at $q = 0.999$ they do not, since $\rho$ there
exceeds $1 - q = 0.001$ and the corollary does not apply.

\section{Numerical check of the unequal-stage claims}
\label{app:unequal}

The same generator as Appendix~\ref{app:check} (Borel($0.9$) clusters, $\rho = 0.40$, $T =
10\,\mu\mathrm{s}$, $h = 1.7\,\mu\mathrm{s}$), with each cluster's members spread
uniformly over $3\,\mu\mathrm{s}$ instead of placed as an instantaneous batch and $4
\times 10^{5}$ messages, run through the partitions of Remark~\ref{rem:unequal} with an
explicit per-stage FIFO recursion. The last two columns are the tail-excess ratio
$(\Delta_{99}(P) - \mathrm{hops}\cdot h)/\Delta_{99}(\text{single})$ and the bound
$s_{\max}/T$ from Theorem~\ref{thm:reduction}.

\begin{center}
\begin{tabular}{@{}lrrrrrrr@{}}
\toprule
partition & $s_{\max}$ & hops & $p_{50}$ & $p_{99}$ & $p_{99.9}$ & ratio & bound \\
\midrule
$10$                 & 10.0 & 0 & 388.2 & 5719.0 & 8249.3 & 1.00 & 1.00 \\
$7 + 3$              &  7.0 & 1 & 213.6 & 3429.5 & 5297.1 & 0.60 & 0.70 \\
$3 + 7$              &  7.0 & 1 & 213.6 & 3429.5 & 5297.1 & 0.60 & 0.70 \\
$5 + 5$              &  5.0 & 1 & 134.9 & 2265.5 & 3574.2 & 0.39 & 0.50 \\
$7 + 2 + 1$          &  7.0 & 2 & 215.3 & 3431.2 & 5298.8 & 0.60 & 0.70 \\
$4 + 3 + 3$          &  4.0 & 2 & 104.6 & 1746.5 & 2714.6 & 0.30 & 0.40 \\
$4 + 3 + 2 + 1$      &  4.0 & 3 & 106.3 & 1748.2 & 2716.3 & 0.30 & 0.40 \\
$2.5 \times 4$       &  2.5 & 3 &  66.4 & 1049.8 & 1562.2 & 0.18 & 0.25 \\
\bottomrule
\end{tabular}
\end{center}

Latencies are end-to-end in microseconds. The three claims of Remark~\ref{rem:unequal}
read off directly: the tail excess depends on the partition only through $s_{\max}$ and
sits below the bound at every row; $7 + 3$ and $3 + 7$ are identical to the reported
precision; and $7 + 2 + 1$ differs from $7 + 3$, and $4 + 3 + 2 + 1$ from $4 + 3 + 3$, by
exactly one hop cost $h = 1.7\,\mu\mathrm{s}$ at every quantile. The $3\,\mu\mathrm{s}$
intra-cluster spread, which is above $s_{\max}$ for the last row, is why the equal-stage
ratio here ($0.18$) sits slightly below the batch-input value of Appendix~\ref{app:check}
($0.21$): a cluster that arrives spread out lets the faster tandem drain part of it before
the next member lands.

\section{A synthetic feed on which the tail at $8\,\mu\mathrm{s}$ depends on intensity and branching ratio}
\label{app:synthetic}

Section~\ref{sec:results-conditioning} finds the normalised tail on CME MDP3 invariant to
packet rate and branching ratio at HFT service times. This appendix constructs, with the
same simulator, a feed on which both dependences appear at $T = 8\,\mu\mathrm{s}$, and
identifies what distinguishes it from CME.

Arrivals are generated as a Hawkes process through its cluster representation:
immigrants arrive as a Poisson process of rate $\mu = \bar\lambda(1-n)$, and each spawns a
Galton--Watson tree with Poisson($n$) offspring at exponential delays of mean $1/\beta$.
Three conditions are run. In the first the kernel is fast, $1/\beta = 1\,\mu\mathrm{s}$,
and there is no minimum period. In the second the same fast kernel is followed by a
minimum period equal to the publisher period, $\spub = 7.5\,\mu\mathrm{s}$: each arrival is
delayed until at least $\spub$ after its predecessor, modelling the exchange's publisher
(Section~\ref{sec:exchange-publisher}). In the third the kernel is slow, $1/\beta =
80\,\mu\mathrm{s}$, with the same minimum period; this is the condition closest to the CME gap
structure, though the kernel fitted on the corpus is slower still, about $160\,\mu\mathrm{s}$
(Section~\ref{sec:setup-hawkes}). The grid is $n \in \{0.5, 0.7, 0.8, 0.9, 0.95\}$ and $\bar\lambda \in \{500,
2000, 8000, 32000, 64000\}$ packets per second, so that $\rho = \bar\lambda T$ at $T =
8\,\mu\mathrm{s}$ runs from $0.004$ to $0.51$; $2 \times 10^5$ events per cell, three
seeds, median reported. The last column of Table~\ref{tab:synthetic} is the fraction of interarrival gaps at
the minimum period, measured at $\bar\lambda = 8000$; on the CME corpus $2.15\%$ of packet
gaps are shorter than $7.9\,\mu\mathrm{s}$, just above the publisher period (Table~\ref{tab:tx-gaps}).

\begin{table}[H]
\centering
\small
\begin{tabular}{@{}r rrrrr r@{}}
\toprule
 & \multicolumn{5}{c}{$p_{99}/T$ at $T = 8\,\mu\mathrm{s}$, $N = 1$} & gaps at $\spub$ \\
\cmidrule(lr){2-6}
$n$ & $\bar\lambda{=}500$ & $2000$ & $8000$ & $32000$ & $64000$ & at $\bar\lambda{=}8000$ \\
\midrule
\multicolumn{7}{l}{\emph{fast kernel, no minimum period}} \\
0.50 &   12.2 &   12.5 &   13.2 &   15.5 &   22.6 & $52\%$ \\
0.70 &   38.9 &   40.7 &   41.2 &   47.3 &   66.6 & $71\%$ \\
0.80 &   91.4 &  100.0 &   97.1 &  105.0 &  140.7 & $80\%$ \\
0.90 &  421.3 &  429.4 &  411.3 &  413.7 &  620.4 & $90\%$ \\
0.95 & 1544.2 & 1482.6 & 1721.3 & 1802.4 & 2487.2 & $95\%$ \\
\midrule
\multicolumn{7}{l}{\emph{fast kernel, publisher period}} \\
0.50 &   1.75 &   1.75 &   1.88 &   2.44 &   4.38 & $53\%$ \\
0.70 &   3.50 &   3.62 &   3.88 &   5.44 &  11.87 & $72\%$ \\
0.80 &   6.81 &   7.38 &   7.56 &  10.88 &  23.12 & $81\%$ \\
0.90 &  28.06 &  28.50 &  30.44 &  41.25 & 100.19 & $91\%$ \\
0.95 &  98.06 &  94.75 & 116.36 & 193.96 & 382.18 & $95\%$ \\
\midrule
\multicolumn{7}{l}{\emph{slow kernel, publisher period}} \\
0.50 &   1.12 &   1.12 &   1.19 &   1.51 &   3.56 & $13\%$ \\
0.70 &   1.25 &   1.31 &   1.50 &   3.18 &  10.44 & $20\%$ \\
0.80 &   2.38 &   2.71 &   3.12 &   8.12 &  19.56 & $28\%$ \\
0.90 &  23.50 &  22.57 &  26.25 &  37.25 &  95.75 & $47\%$ \\
0.95 &  90.12 &  86.19 & 112.19 & 188.78 & 364.30 & $67\%$ \\
\midrule
\multicolumn{7}{l}{\emph{Poisson, matched rate}} \\
---  &   1.00 &   1.39 &   1.87 &   2.69 &   4.45 & --- \\
\midrule
\multicolumn{2}{l}{$\rho$ at $T = 8$} & $0.016$ & $0.064$ & $0.256$ & $0.512$ & \\
\bottomrule
\end{tabular}
\caption{Synthetic Hawkes feeds at the sweep point just above the publisher period,
$T = 8\,\mu\mathrm{s}$: $p_{99}/T$ by branching ratio and packet rate, for three kernel and
minimum-period conditions and a Poisson stream; the last column is the share of gaps at
the minimum period. The $\rho$ row gives the utilisation. What the table shows: on every
synthetic feed the tail grows steeply with the branching ratio, and the share of gaps at
the minimum period grows with it; the Poisson stream has a tail only at high utilisation.
What follows: on these feeds the branching ratio matters because it controls how many
arrivals are packed at the minimum period, which on CME the publisher fixes. Figure~\ref{fig:bars-synthetic} shows the same numbers as bars.}
\label{tab:synthetic}
\end{table}

\begin{figure}[H]
\centering
\includegraphics[width=0.95\textwidth]{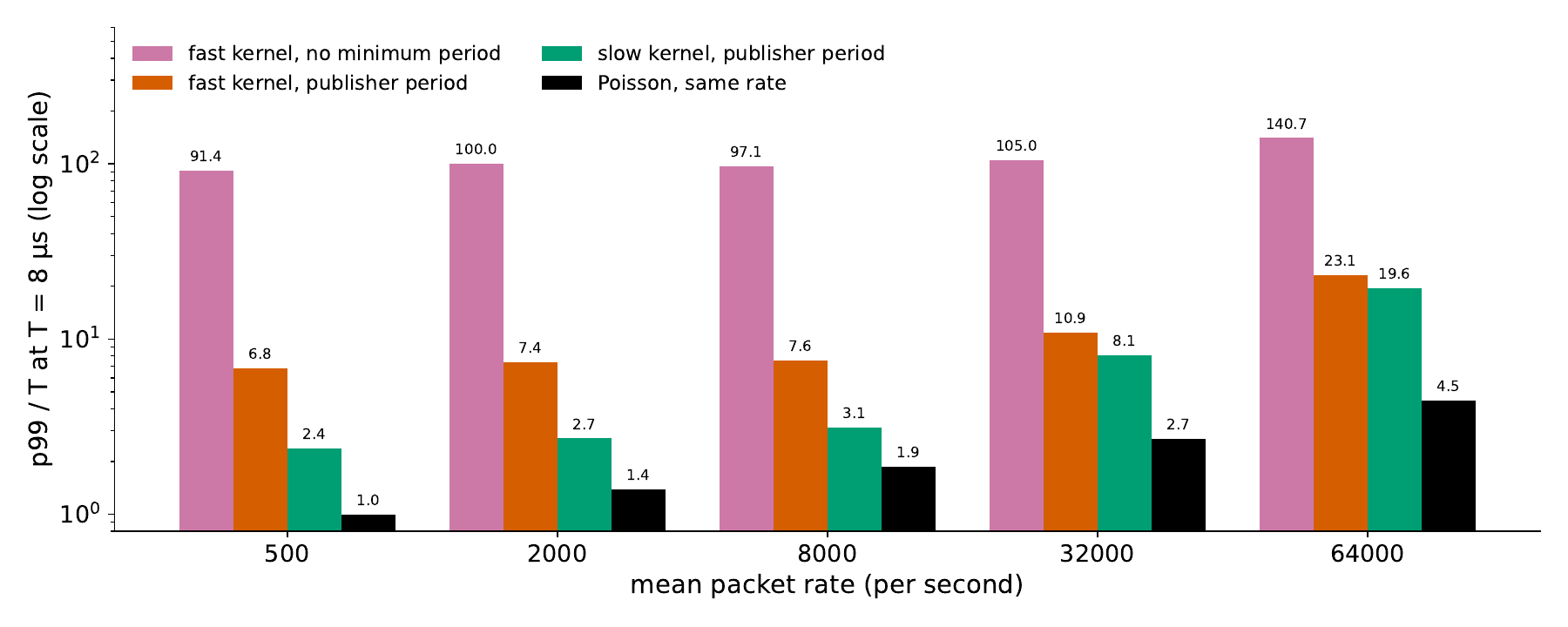}
\caption{Table~\ref{tab:synthetic} as bars, log scale, for branching ratio $0.8$: the normalised
tail $p_{99}/T$ at $T = 8\,\mu\mathrm{s}$ on synthetic feeds, by packet rate. A fast kernel
with no minimum period (pink) has a tail of about a hundred times $T$; adding a
minimum period equal to the publisher period (vermillion) cuts it more than tenfold; a slow kernel with the minimum period (green) is lower still, closer to the Poisson stream (black). What sets
the tail at $8\,\mu\mathrm{s}$ is how many gaps sit at the minimum period, which the kernel and the minimum period control and the branching ratio alone does not.}
\label{fig:bars-synthetic}
\end{figure}

Three observations.

\paragraph{Intensity.} In every synthetic Hawkes condition the tail does not change with
the packet rate while utilisation stays low, and it rises once utilisation reaches a few
percent, steeply by the time the server is half busy. The Poisson row rises earlier,
because for it there is nothing but utilisation. The CME corpus, at the sweep point just
above the publisher period, sits deep inside the flat region even in its busiest windows.
The invariance to rate found in Section~\ref{sec:results-conditioning} is therefore what a
Hawkes model predicts at that utilisation, and it ends where utilisation queueing begins,
as the CME sweep also finds at long service times.

\paragraph{Criticality.} The synthetic feeds depend strongly on the branching ratio in
every condition: the tail grows by about two orders of magnitude from the least to the
most self-exciting feed, and roughly doubles between the two values nearest the corpus's.
The CME corpus shows no dependence over its own, narrow, range. The last column of the
table explains the difference. A synthetic Hawkes feed at the corpus's branching ratio
places between a quarter and four fifths of its gaps at the minimum period, depending on
the kernel; on CME only about one gap in fifty is that short. On the synthetic feeds the
branching ratio controls how many arrivals are packed at the minimum period, and that
packing is what queues on a server at the production service time. On CME the packed
fraction is set by the publisher and does not vary with the fitted branching ratio; the
fitted ratio describes the ordering of arrivals at longer timescales, consistent with
Section~\ref{sec:setup-hawkes}, where the corpus is found to be self-exciting with
heavy-tailed gaps rather than a pure exponential-kernel Hawkes process.

\paragraph{Feed dependence of the design rule.} At the sweep point just above the
publisher period, at the corpus's branching ratio and a moderate packet rate, a two-stage
tandem has the lower $p_{99}$ on every synthetic condition: it halves the tail on the feed
with no minimum period, and on the two feeds with a minimum period it takes a tail of
several times the service time down to almost nothing. On CME at the same service time
one stage is best (Section~\ref{sec:results-threshold}). The service time at which
splitting begins to pay is not a constant; it is a property of the feed.

\paragraph{Summary.} Whether splitting a stage reduces its tail depends on how
many packets arrive closer together than the stage's service time, not on the packet rate
or the branching ratio as such. On CME about $2\%$ of packet gaps are shorter than
$7.9\,\mu\mathrm{s}$, just above the publisher period. That fraction does
not change when the market becomes busier, since a busier market produces more clusters
with the same spacing inside them, nor when the fitted branching ratio moves within its
observed range. At the $7$--$8\,\mu\mathrm{s}$ service time of a production decode path there is
therefore little to queue behind across packets, and one thread is the right design as far as
that queue goes; the tail from long packets and variable service at that service time is
separate (Sections~\ref{sec:crossval-spanrun} and~\ref{sec:fungibility-measured}). On a feed that
places a large share of its gaps at the minimum period, the result inverts: a Hawkes process at the
same $n = 0.8$ places between a quarter and four fifths of its gaps at the minimum period instead of about $2\%$, the tail
depends strongly on the branching ratio, splitting reduces the tail at $T =
8\,\mu\mathrm{s}$ where on CME it does not, and at high enough load the packet rate
matters as well. The Poisson null shows only the load effect, and only at high
utilisation, as a null should. The rule is therefore neither one thread nor two. The
split threshold, the service time above which cutting a stage lowers its $p_{99}$, is a
property of the feed's publisher period; it cannot be read off the fitted parameters, and
it must be measured under the load the system will see. On the CME feed a
$7\,\mu\mathrm{s}$ stage sits below the split threshold of $9$--$10\,\mu\mathrm{s}$; another exchange, or the same
exchange after a change to its publisher, may not. Section~\ref{sec:recommendations} states this as
a design principle.

\end{document}